\documentclass[11pt]{article}
\usepackage[utf8]{inputenc}
\usepackage[T1]{fontenc}
\usepackage[a4paper]{geometry}
\usepackage{graphicx,amsmath,amsthm,amsfonts,amssymb,bbm,stmaryrd}
\usepackage[colorlinks]{hyperref}		
\usepackage{xcolor}			
\usepackage[all]{xy}
\usepackage{bm}
\usepackage{color}
\newcommand{\comment}[1]{\textcolor{red}{#1}}
\def\added#1{{\textcolor{blue}{#1}}}

\usepackage{textcomp} 

\newtheorem{theorem}{Theorem}
\newtheorem{lemma}{Lemma}
\newtheorem{corollary}{Corollary}
\newtheorem{definition}{Definition}

\newtheorem{remark}{Remark}

\newtheorem{claim}{Claim} 
\newenvironment{proofclaim}{\noindent \emph{Proof.}\ }{\hfill
    $\Diamond$\vspace{1em}}

  \newcommand{\ccw}{counterclockwise }

\newcommand{\cw}{clockwise }

\newcommand{\Cov}{\mathop{\mathrm{Cov}}}
\newcommand{\diam}{\mathop{\mathrm{diam}}}

\begin{document}

\title{Scaling limits for random triangulations on the torus}
\author{Vincent Beffara\thanks{Univ. Grenoble Alpes, CNRS, IF, 38000
    Grenoble, France} $^\ddagger$, Cong Bang Huynh$^{*\dagger\ddagger}$, Benjamin
  Lévêque\thanks{Univ. Grenoble Alpes, CNRS, Grenoble INP, G-SCOP,
    38000 Grenoble, France} \thanks{\texttt{emails
      :vincent.beffara@univ-grenoble-alpes.fr,
      cong-bang.huynh@univ-grenoble-alpes.fr,
      benjamin.leveque@cnrs.fr}}}

\maketitle
\begin{abstract}
  We study the scaling limit of essentially simple triangulations on
  the torus. We consider, for every $n\geq 1$, a uniformly random
  triangulation $G_n$ over the set of (appropriately rooted)
  essentially simple triangulations on the torus with $n$ vertices. We
  view $G_n$ as a metric space by endowing its set of vertices with
  the graph distance denoted by $d_{G_n}$ and show that the random
  metric space $(V(G_n),n^{-1/4}d_{G_n})$ converges in distribution in
  the Gromov–Hausdorff sense when $n$ goes to infinity, at least along
  subsequences, toward a random metric space. One of the crucial steps
  in the argument is to construct a simple labeling on the map and
  show its convergence to an explicit scaling limit. We moreover show
  that this labeling approximates the distance to the root up to a
  uniform correction of order $o(n^{1/4})$.
\end{abstract}



\section{Introduction}

\subsection{Some definitions}

Recall  that  the  \emph{Hausdorff  distance}  between  two  non-empty
subsets $X$ and $Y$ of a metric space $(M,d)$ is defined as
\begin{equation*}
d_{Haus}(X,Y)  = \inf\{\epsilon\geq  0:  X  \subset Y_\epsilon  \quad
\textrm{and} \quad Y\subset X_{\epsilon} \},
\end{equation*}
where $Z_{\epsilon}$  denotes $\{m\in  M: d(m,Z)\leq  \epsilon\}$. The
\emph{Gromov-Hausdorff  distance} between  two  compact metric  spaces
$(S,\delta)$ and $(S',\delta')$ is defined as
\begin{equation*}
d_{GH}((S,\delta),(S',\delta'))=\inf \{d_{Haus}(\varphi(S),\varphi'(S') \},
\end{equation*} 
where   the   infimum  is   taken   over   all  isometric   embeddings
$\varphi:S\rightarrow S"$ and $\varphi':  S'\rightarrow S"$ of $S$ and
$S'$   into  a   common   metric  space   $(S",\delta")$.  Note   that
$d_{GH}((S,\delta), (S',\delta'))$ is equal to  $0$ if and only if the
metric spaces $S$  and $S'$ are isometric to each  other. We refer the
reader  to  \emph{e.g.}~\cite[Section~3]{berry2013scaling} for  a
detailed investigation of the Gromov-Hausdorff distance.

In this paper, we are considering some random graphs seen as random
metric spaces and consider their convergence in distribution in the
sense of the Gromov-Hausdorff distance.  In general, graphs may
contain loops and multiple edges. A graph is called \emph{simple} if
it contains no loop nor multiple edges. A graph embedded on a surface
is called a \emph{map} on this surface if all its faces are
homeomorphic to open disks. In this paper we consider orientable
surface of genus $g$ where the plane is the surface of genus $0$, the
torus the surface of genus $1$, etc.  For $p\geq 3$, a map is called a
$p$-angulation if all its faces have size $p$. For $p=3$
(resp. $p=4$), such maps are respectively called triangulations
(resp. quadrangulations).

\subsection{Random planar maps}

Let us  first review some  results on  random planar maps.  Consider a
random  planar  map  $G_n$  with   $n$  vertices  which  is  uniformly
distributed  over  a  certain  class   of  planar  maps  (like  planar
triangulations, quadrangulations or $p$-angulations). Equip the vertex
set $V(G_n)$ with  the graph distance $d_{G_n}$. It is  known that the
diameter of the resulting metric space  is of order $n^{1/4}$ (see for
example~\cite{chassaing2004random} for the  case of quadrangulations).
Thus  one   can  expect  that   the  rescaled  random   metric  spaces
$(V(G_n),n^{-1/4}d_{G_n})$ converge  in distribution  as $n$  tends to
infinity   toward   a  certain   random   metric   space.  In   2006,
Schramm~\cite{schramm2006conformally} suggested  to use the  notion of
Gromov-Hausdorff distance 
to formalize this  question by
specifying  the topology  of this  convergence.  He was  the first  to
conjecture the  existence of a  scaling limit for large  random planar
triangulations.  In 2011,  Le Gall~\cite{le2013uniqueness}  proved the
existence of  the scaling limit  of the rescaled random  metric spaces
$(V(G_n),n^{-1/4}d_{G_n})$   for  $p$-angulations   when  $p=3$,   or,
$p\geq 4$  and $p$ is  even. The case  $p=3$ solves the  conjecture of
Schramm.  Miermont~\cite{miermont2013brownian}   gave  an  alternative
proof  in  the  case  of quadrangulations  $(p=4)$.  Addario-Berry  and
Albenque~\cite{berry2013scaling}  prove  the  case  $p=3$  for  simple
triangulations (\emph{i.e.}\ triangulations with  no loop nor multiple
edges). An  important aspect  of all  these results is  that, up  to a
constant rescaling factor, all these classes converge toward the same
object called the Brownian map.

It is  natural to address the  question of the existence  of a scaling
limit of random maps on higher genus oriented surfaces. Chapuy, Marcus
and Schaeffer \cite{chapuy2009bijection}  extended the bijection known
for planar  bipartite quadrangulations to any  oriented surfaces. This
led  Bettinelli  \cite{bettinelli2010scaling}  to  show  that  random
quadrangulations  on oriented  surfaces converge  in distribution,  at
least along a subsequence. More formally:

\begin{theorem}[Bettinelli~\cite{bettinelli2010scaling}]
  \label{th:Bettinelli}
  For $g\geq 1$ and $n\geq 1$, let $G_n$ be a uniformly random element
  of the set of all  angle-rooted bipartite quadrangulations with $n$
  vertices  on the  oriented  surface  of genus  $g$.  Then, from  any
  increasing  sequence  of integers,  one  can  extract a  subsequence
  $(n_k)_{k\geq 0}$ along which the rescaled metric spaces
  $$\left(V(G_{n_k}),n_k^{-1/4}d_{G_{n_k}}\right)_{k\geq0}$$
  converge in distribution for the Gromov-Hausdorff distance.
\end{theorem}

Contrary to the planar case, the uniqueness of the subsequential limit
is not  proved there.  Nevertheless, a  phenomenon of  universality is
expected: it is  conjectured that the sequence does  converge and that
moreover,  up  to  a  deterministic  multiplicative  constant  on  the
distance, the limit  is the same for  many models of random  maps of a
given  genus.  In  genus  $1$,  the  conjectured  limit  is  described
in~\cite{bettinelli2010scaling} and referred  to as the \emph{toroidal
  Brownian map}.

The present article extends Theorem~\ref{th:Bettinelli} to the case of
(essentially simple) triangulations of the torus. In that respect, it
is comparable to the paper of Addario-Berry and
Albenque~\cite{berry2013scaling} which did the same in the planar
setup and thus our work contributes to the understanding of
universality for random toroidal maps.

\subsection{Main results}

A \emph{contractible loop} is an edge enclosing a region homeomorphic
to an open disk. A pair of \emph{homotopic multiple edges} is a pair
of edges that have the same extremities and whose union encloses a
region homeomorphic to an open disk.  A graph $G$ embedded on the
torus is called \emph{essentially simple} if it has no contractible
loop nor homotopic multiple edges. Being essentially simple for a
toroidal map is the natural generalization of being simple for a
planar map.

In this paper, we distinguish paths and cycles
from walks and closed walks as the firsts have no repeated vertices.
A \emph{triangle} of a toroidal map is a closed walk of size $3$
enclosing a region that is homeomorphic to an open disk.  This region
is called the \emph{interior} of the triangle. Note that a triangle is
not necessarily a face of the map as its interior may be not empty. 
We say that a triangle is \emph{maximal} (by inclusion) if its
interior is not strictly contained in the interior of another
triangle. 
  We
define the \emph{corners} of a triangle as the three angles
that appear in the interior of this triangle when its interior is
removed (if non empty).

Our main result is the following convergence result:

\begin{theorem}
\label{main3}
For $n\geq  1$, let  $G_n$ be  a uniformly random  element of  the set
 of all   essentially simple toroidal triangulations
on $n$  vertices that are  rooted at a  corner of a  maximal triangle.
Then,  from any  increasing sequence  of integers,  one can  extract a
subsequence $(n_k)_{k\geq 0}$ along which the rescaled metric spaces 
$$\left(V(G_{n_k}),n_k^{-1/4}d_{G_{n_k}}\right)_{k\geq0}$$
  converge in distribution for the Gromov-Hausdorff distance.
\end{theorem}

\begin{remark}
  The reason for the particular choice of rooting in
  Theorem~\ref{main3} is of a technical nature due to the
    bijection that we use in Section~\ref{sec:preliminaries}. It is a natural conjecture that
  compactness, and thus also the existence of subsequential scaling
  limits, would still hold \emph{e.g.} for triangulations rooted at a
  uniformly random angle. This is based on the following reasoning:
  if the inside of every maximal triangle has diameter of smaller
  order than $n^{1/4}$, then rooting inside such a triangle rather
  than at one of its corners would affect distances by a quantity that
  would be smoothed out by the normalization. On the other hand,
  having one maximal triangle containing $\alpha n$ vertices has very
  small probability, because of the relative growths of the number of
  triangulations of genus $0$ and $1$. The remaining obstruction would
  be the existence of a maximal triangle with an inside containing
  much fewer than $n$ vertices but having diameter of order $n^{1/4}$,
  which would presumably be ruled out by a precise control of the
  geometry of simple triangulations of genus $0$. This is a possible
  direction for future work, but we chose not to investigate it
  further due to the already large size of the present paper.
\end{remark}

We also show in an appendix that with high probability, the labeling
function that we define as a crucial tool in our argument (see
Section~\ref{sec:label} for a formal definition) approximates the
distance to the root up to a uniform $o(n^{1/4})$ correction (see
Theorem~\ref{main2bis}).  Such a comparison estimate is an essential
step in proving the uniqueness of the subsequential scaling limit, and
thus the convergence, in frameworks similar to that of our main result
--- see \cite{berry2013scaling} for the case of genus $0$, it is also
likely that a similar argument would be applicable to quadrangulations
of the torus \cite{BM19} (those two quantities are actually equal in
the case of bipartite quadrangulations on any surface with positive
genus, but it seems that a bound of the order $o(n^{1/4})$ is enough).

The overall strategy for the  proof of Theorem~\ref{main3} is the same
as     in~\cite{bettinelli2010scaling},     as      well     as     in
\cite{le2013uniqueness}   and~\cite{miermont2013brownian}:  obtain   a
bijection between  maps and  simpler combinatorial  objects (typically
decorated  trees),  then  show  the convergence  of  these  objects  to  a
non-trivial continuous  random limit  from which  relevant information
can then be  extracted about the original model. As  a result, most of
the    structure     of    the     paper    is     largely    inspired
by~\cite{bettinelli2010scaling}     (for     the    main     argument)
and~\cite{berry2013scaling} (for methods specific to triangulations).

The bijection that we use here is based on a recent generalization of
Schnyder woods to higher
genus~\cite{GL14,GKL15,leveque2017generalization}.  One issue when
going to higher genus is that the set of Schnyder woods of a given
triangulation is no longer a single distributive lattice like in the
planar case, it is rather a collection of distributive lattices.
Nevertheless, it is possible to single out one of these distributive
lattices, in the toroidal case, by requiring an extra property, called
\emph{balanced}, that defines a unique minimal element used as a
\emph{canonical orientation} for the toroidal triangulation.  The
particular properties of this canonical orientation leads to a
bijection between essentially simple toroidal triangulation and
particular toroidal unicellular maps~\cite{despre2017encoding} (a
\emph{unicellular} map is a map with only one face, i.e. the natural
generalization of trees when going to higher genus). Then the main
difficulty that we have to face is that the metric properties of the
initial map are less apparent in the unicellular map than in the
planar case or in the bipartite quadrangulations setup.

\subsection*{Structure of the paper}

The bijection between toroidal triangulations and particular
unicellular maps is presented in Section~\ref{sec:preliminaries} with
some related properties.  In Section~\ref{sec:label}, we define a
labeling function of the angles of a unicellular map and prove some
relations with the graph distance in the corresponding triangulation.
In Section~\ref{somebijection} we explain how to decompose the
particular unicellular maps given by the bijection into simpler
elements with the use of Motzkin paths and well-labeled forests.  In
Section~\ref{sec:brownian}, we review some results on variants of the
Brownian motion.  Then the proof of Theorem~\ref{main3} then proceeds
in several steps. In Section~\ref{section5}, we study the convergence
of the parameters of the discrete map in the scaling limit.  In
Sections~\ref{sec:motzkin},~\ref{section6} and~\ref{section7} we
review and extend classical convergence results for conditioned random
walks and random forests.  Finally, in Section~\ref{sec:proof}, we
combine the previous ingredients to build the proof of the main
theorem. In Appendix~\ref{rightmostpath}, we exploit the canonical
orientation of the triangulation to define rightmost paths and relate
them to shortest paths, thus obtaining the announced upper bound on
the difference between distances and labels.

This work has been partially supported by the LabEx PERSYVAL-Lab
(ANR-11-LABX-0025-01) funded by the French program Investissement
d’avenir and the ANR project GATO (ANR-16-CE40-0009-01) funded by the French
Agence National de la Recherche.

\section{Bijection between toroidal triangulations and unicellular maps}
\label{sec:preliminaries}

For $n\geq 1$, let $\mathcal G(n)$ be the set of essentially simple
toroidal triangulations on $n$ vertices that are rooted at a corner of
a maximal triangle.

Consider an element $G$ of $\mathcal G(n)$. The corner of the maximal
triangle where  $G$ is rooted  is called the \emph{root  corner}. Note
that, since  $G$ is  essentially simple, there  is a  unique triangle,
called the \emph{root  triangle}, whose corner is the  root corner (and
this root triangle  is maximal by assumption). The vertex  of the root
triangle corresponding  to the  root corner  is called  the \emph{root
  vertex}. We also define, in a  unique way, a particular angle of the
map, called the \emph{root angle}, that is the angle of $G$ that is in
the interior of the root triangle, incident to the root vertex and the
last  one in  \ccw  order around  the  root vertex.  Note  that it  is
possible to retrieve the root corner  from the root angle in a unique
way  (indeed, the  root angle  defines already  one edge  of the  root
triangle and  the side of  its interior, thus  it remains to  find the
third vertex of the root triangle  such that the interior is maximal).
Thus rooting  $G$ on its root  corner or root angle  is equivalent. We
call \emph{root face},  the face of $G$ containing the  root angle. We
introduce in  the rest  of this section  some terminology  and results
adapted               from~\cite{despre2017encoding}              (see
also~\cite{leveque2017generalization}).

\subsection{Toroidal unicellular maps}

Recall that a unicellular map is a map with only one face. There are
two types of toroidal unicellular maps since two cycles of a toroidal
unicellular map may intersect either on a single vertex (square case)
or on a path (hexagonal case). On the first row of
Figure~\ref{fig:hexasquare} we have represented these two cases into a
square box that is often use to represent a toroidal object (its
opposite sides are identified). On the second row of
Figure~\ref{fig:hexasquare} we have represented again these two cases
by a square and hexagon by copying some vertices and edges of the map
(here again the opposite sides are identified).  Depending on what we
want to look at we often move from one representation to the other in
this paper.  We call \emph{special} the vertices of a toroidal
unicellular map that are on all the cycles of the map. Thus the number
of special vertices of a square (resp. hexagon) toroidal unicellular
map is exactly one (resp. two).

\begin{figure}[h]
\center
\begin{tabular}{cc}
\includegraphics[scale=0.3]{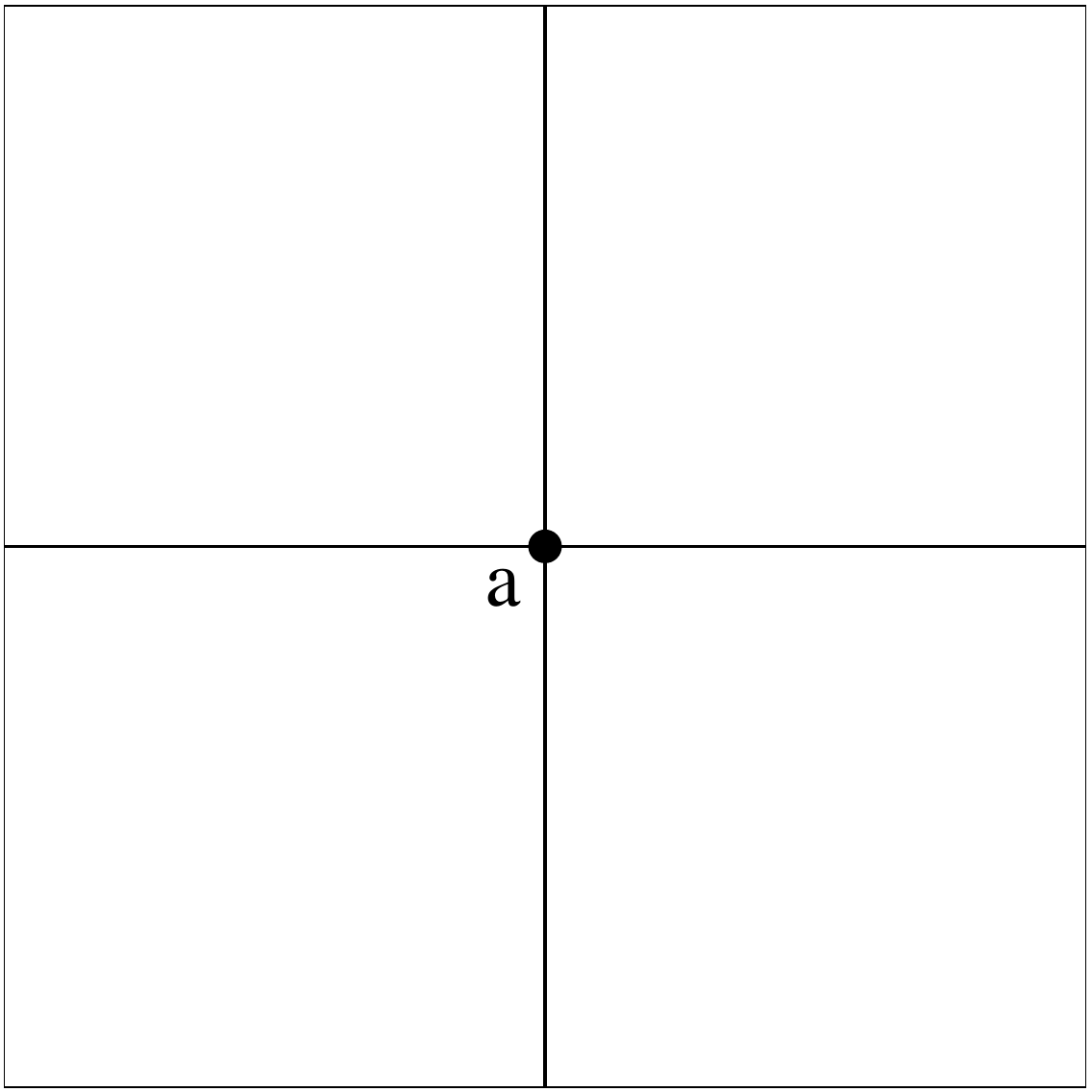} \ \ & \ \
\includegraphics[scale=0.3]{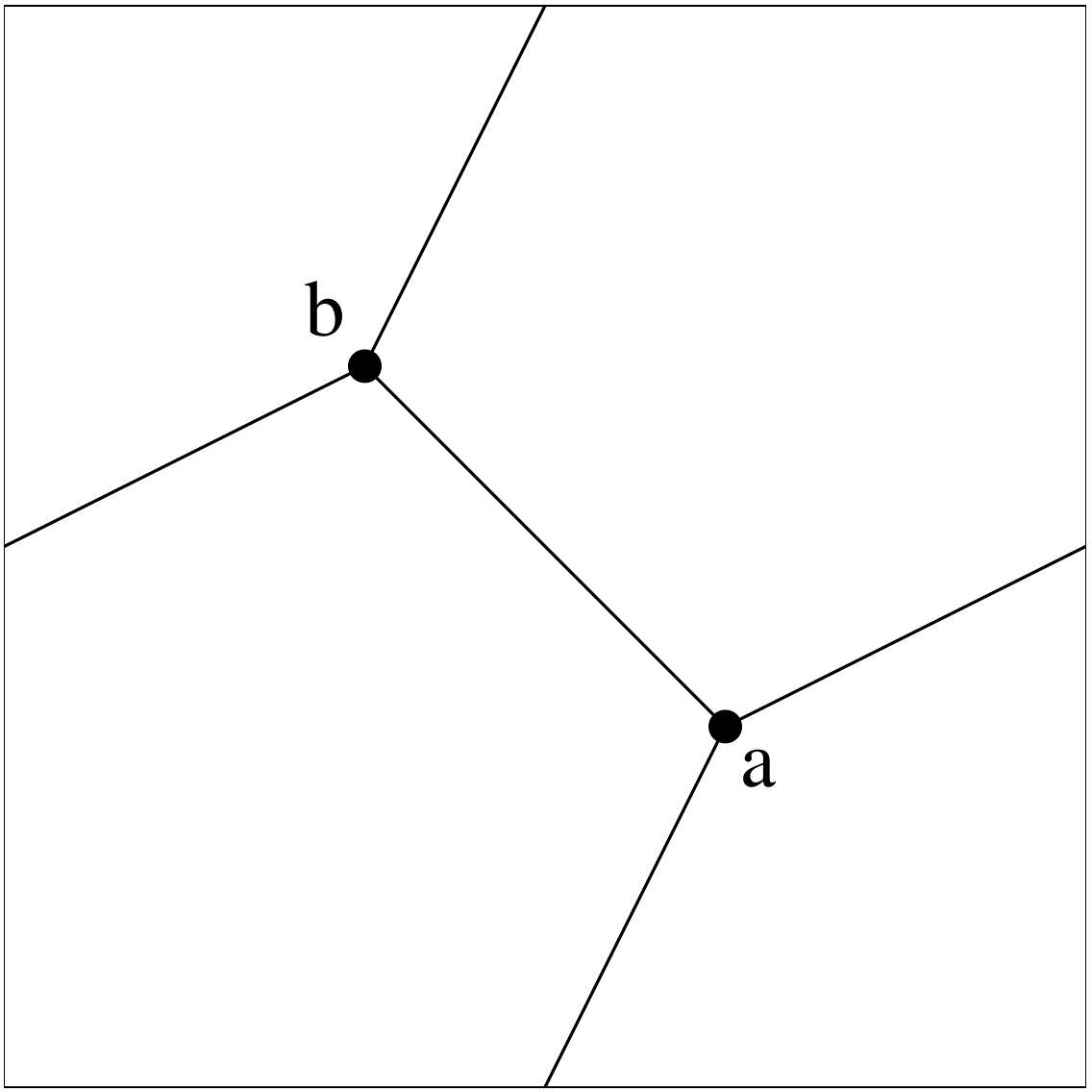}\\
\includegraphics[scale=0.3]{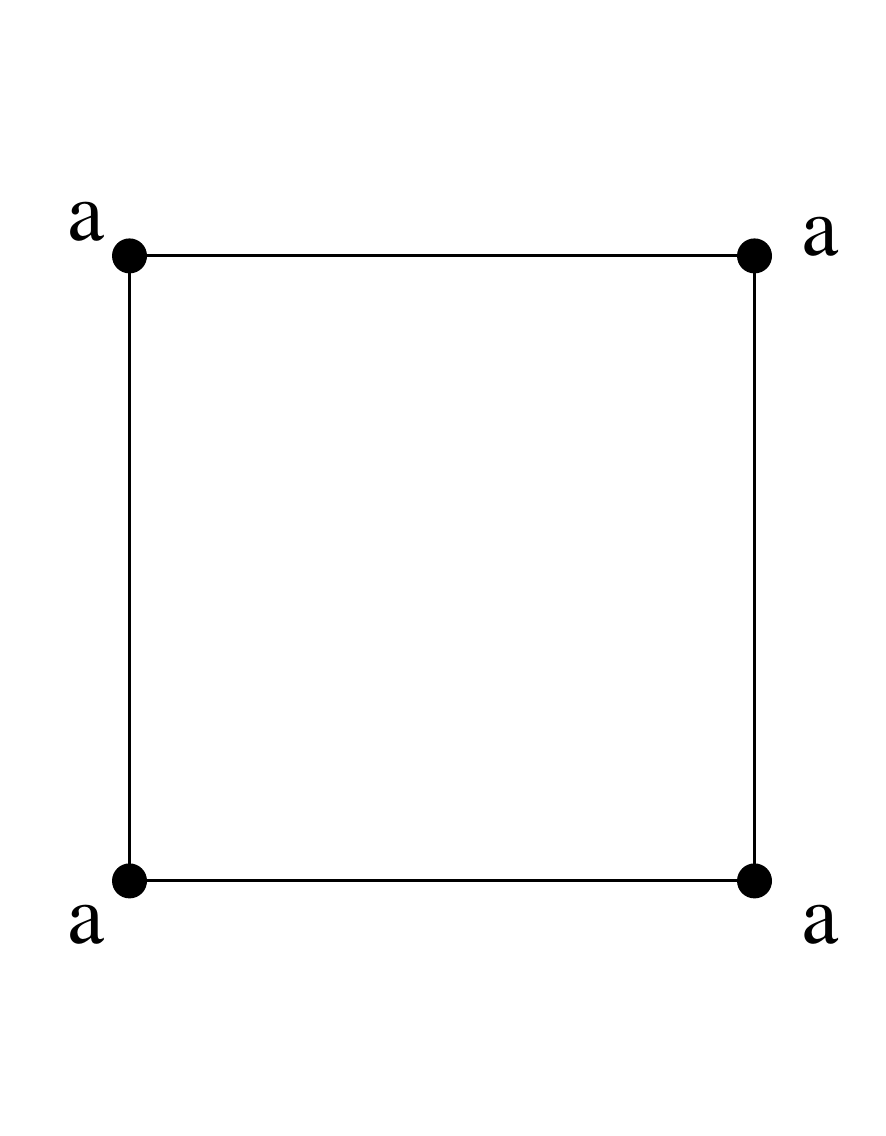} \ \ & \ \
\includegraphics[scale=0.3]{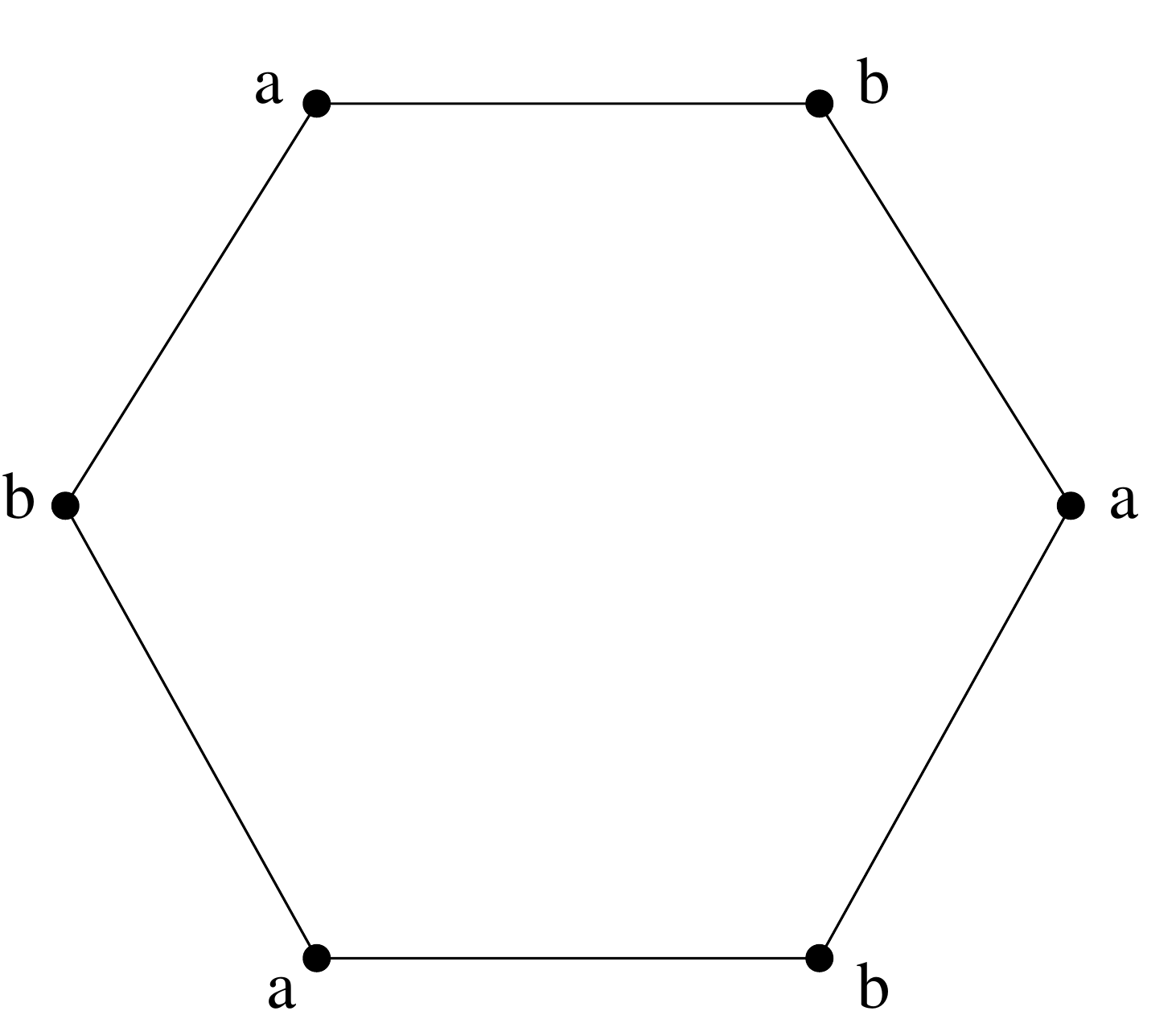}\\
Square case\ \ & \ \ Hexagonal case 
\end{tabular}

\caption{The two types of toroidal unicellular maps with two different representations for each case.} 
\label{fig:hexasquare}
\end{figure}

Given a map, we call \emph{stem}, a half-edge that is added to the map, attached to an angle of a vertex and whose other extremity is dangling in the face incident to this angle.

For $n\geq 1$, let $\mathcal T_r(n)$ denote the set of toroidal
unicellular maps  rooted on a particular angle, with exactly $n$
vertices, $n+1$ edges and $2n-1$ stems distributed as follows (see
figure~\ref{fig:k7-squarebox-nodir} for an example in
$\mathcal T_r(7)$ where the root angle is represented with the usual
"root" symbol in the whole paper.). The vertex incident to the root
angle is called the \emph{root vertex}.  A vertex that is not the root
vertex, is incident to exactly $2$ stems if it is not a special
vertex, $1$ stem if it is the special vertex of a hexagon and $0$ stem
if it is the special vertex of a square. The root vertex is incident
to $1$ additional stem, i.e.  it is incident to exactly $3$ stems if
it is not a special vertex, $2$ stems if it is the special vertex of a
hexagon and $1$ stem if it is the special vertex of a
square. Moreover, one of the stem incident to the root vertex, called
the \emph{root stem}, is incident to the root angle and just after the
root angle in \ccw order around the root vertex.

\begin{figure}[h]
\center
\includegraphics[scale=0.4]{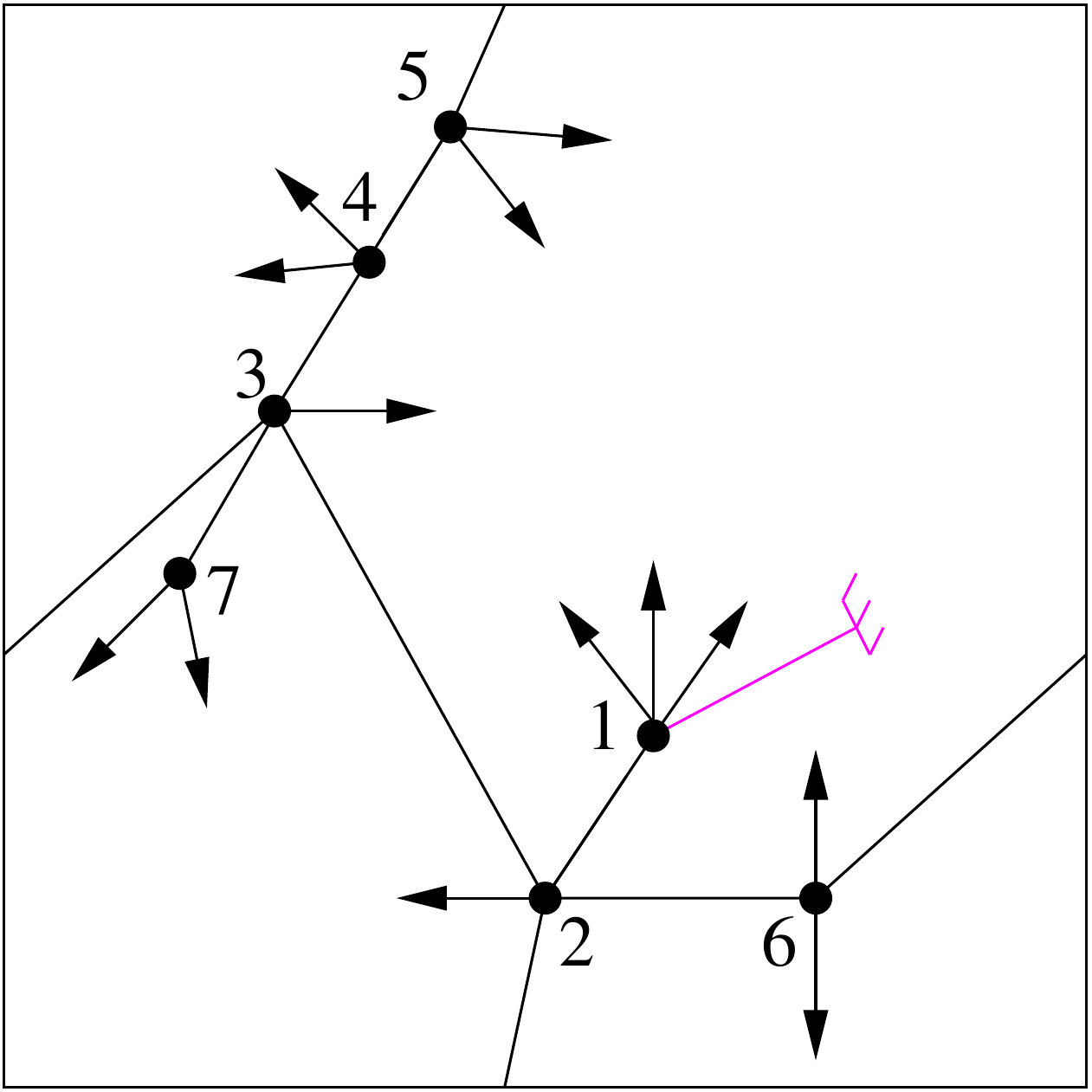} \caption{Example
  of an element of $\mathcal T_r(7)$.}
\label{fig:k7-squarebox-nodir}
\end{figure}

\subsection{Closure procedure}

Given an element $T$ of $\mathcal T_r(n)$, there is a generic way to
attach step by step all the dangling extremities of the stems of $T$
to build a toroidal triangulation.  Let $T_0=T$, and, for
$1\leq k \leq 2n-1$, let $T_{k}$ be the map obtained from $T_{k-1}$ by
attaching the extremity of a stem to an angle of the map (we explicit
below which stems can be attached and how). The \emph{special face of
  $T_0$} is its only face. For $1\leq k \leq 2n-1$, the \emph{special
  face of $T_{k}$} is the face on the right of the stem of $T_{k-1}$
that is attached to obtain $T_{k}$ (the stem is by convention oriented
from its incident vertex toward its dangling part).  For
$0\leq k\leq 2n-1$, the border of the special face of $T_k$ consists
of a sequence of edges and stems. We define an \emph{admissible
  triple} as a sequence $(e_1,e_2,s)$, appearing in \ccw order along
the border of the special face of $T_k$, such that $e_1=(u,v)$ and
$e_2=(v,w)$ are edges of $T_k$ and $s$ is a stem attached to $w$. The
\emph{closure} of this admissible triple consists in attaching $s$ to
$u$, so that it creates an edge $(w,u)$ oriented from $w$ to $u$ and
so that it creates a triangular face $(u,v,w)$ on its left side.  The
\emph{complete closure} of $T$ consists in closing a sequence of
admissible triples, i.e.  for $1\leq k \leq 2n-1$, the map $T_{k}$ is
obtained from $T_{k-1}$ by closing any admissible triple.

Figure~\ref{fig:closure} is the hexagonal representation of the
example of Figure~\ref{fig:k7-squarebox-nodir} on which a complete
closure is performed. We have represented here the unicellular map as
an hexagon since it is easier to understand what happen in the unique
face of the map. The map obtained by performing the complete closure
procedure is the clique on seven vertices $K_7$.

\begin{figure}[h]
\center
\begin{tabular}{cc}
\includegraphics[scale=0.4]{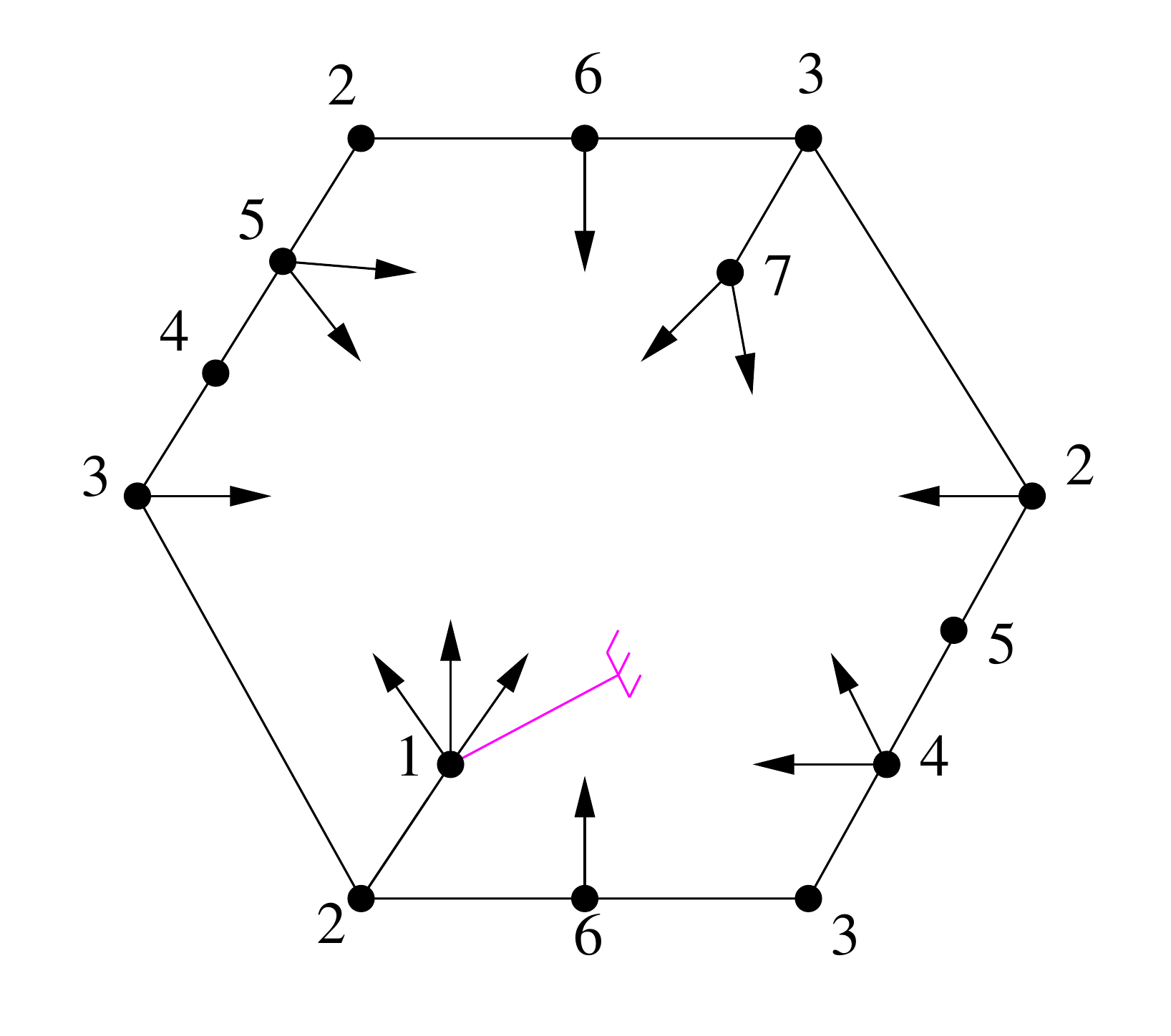}  
 & 
\includegraphics[scale=0.4]{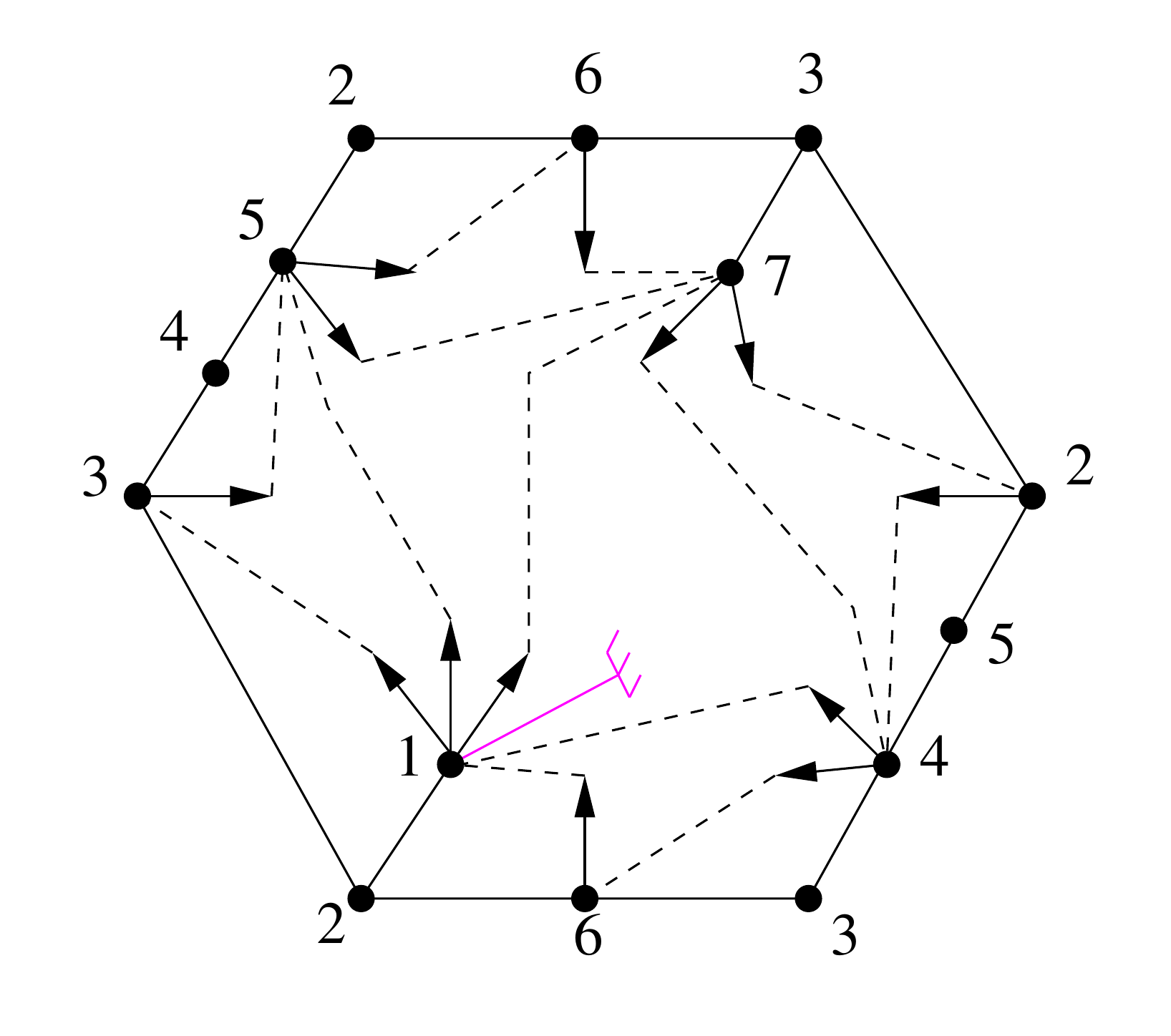}\\
A unicellular map of $\mathcal T_{r,s,b}(7)$ & The complete closure gives $K_7$ \\
\end{tabular}
\caption{Example of the complete closure procedure.}
\label{fig:closure}
\end{figure}

Note that, for $0\leq k\leq 2n-1$, the special face of $T_k$ contains
all the stems of $T_k$. The closure of a stem reduces the number of
edges on the border of the special face and the number of stems by
$1$. At the beginning, the unicellular map $T_0$ has $n+1$ edges and
$2n-1$ stems. So along the border of its special face, there are
$2n+2$ edges and $2n-1$ stems. Thus there is exactly three more edges
than stems on the border of the special face of $T_0$ and this is
preserved while closing stems. So at each step there is necessarily at
 least one admissible triple and the sequence $T_k$ is well defined.
Since the difference of three is preserved, the special face of
$T_{2n-2}$ is a quadrangle with exactly one stem. So the attachment
of the last stem creates two faces that have size three and at the end
$T_{2n-1}$ is a toroidal triangulation.  Note that at a given step
there might be several admissible triples but their closure are
independent and the order in which they are closed does not modify
the obtained triangulation $T_{2n-1}$.

When a stem is attached on the root angle, then, by convention, the
new root angle is maintained on the right side of the extremity of the
stem, i.e. the root angle is maintained in the special face.  A
particularly important property when attaching stems is when the
complete closure procedure described here never \emph{wraps over the
  root angle}, i.e. when a stem is attached, the root angle is
always on its right side in the special face. The property of never
wrapping over the root angle is called \emph{safe} (an analogous
property is sometimes called "balanced" in the planar case but we
prefer to keep the word "balanced" for something else in the current
paper).  Let $\mathcal T_{r,s}(n)$ denote the set of elements of
$\mathcal T_r(n)$ that are safe.

Consider an element $T$ of $\mathcal T_{r,s}(n)$ with root angle
$a_0$. Then for $0\leq k\leq 2n-2$, let $s$ be the first stem met while walking
\ccw from $a_0$ in the special face of $T_k$. An essential property
from~\cite{despre2017encoding} is that before $s$, at least two edges
are met and thus the last two of these edges form an admissible triple
with $s$. So one can attach all the stems of $T$ by starting from
the root angle $a_0$ and walking along the face of $T$ in \ccw order
around this face: each time a stem is met, it is attached in order
to create a triangular face on its left side. Note that in such a
sequence of admissible triples closure, the last stem that is
attached is the root stem of $T$.

\subsection{Canonical orientation and balanced property}

For $n\geq 1$, consider an element $T$ of $\mathcal T_{r}(n)$ whose
edges and stems are oriented w.r.t.~the root angle $a_0$ as follows
(see Figure~\ref{fig:orientation} that corresponds to the example of
Figure~\ref{fig:k7-squarebox-nodir}): the stems are all outgoing, and
while walking \cw around the unique face of $T$ from $a_0$, the first
time an edge is met, it is oriented \ccw w.r.t.~the face of $T$. This
orientation plays a particular role and is called the \emph{canonical
  orientation} of $T$.

\begin{figure}[h]
\center
\includegraphics[scale=0.4]{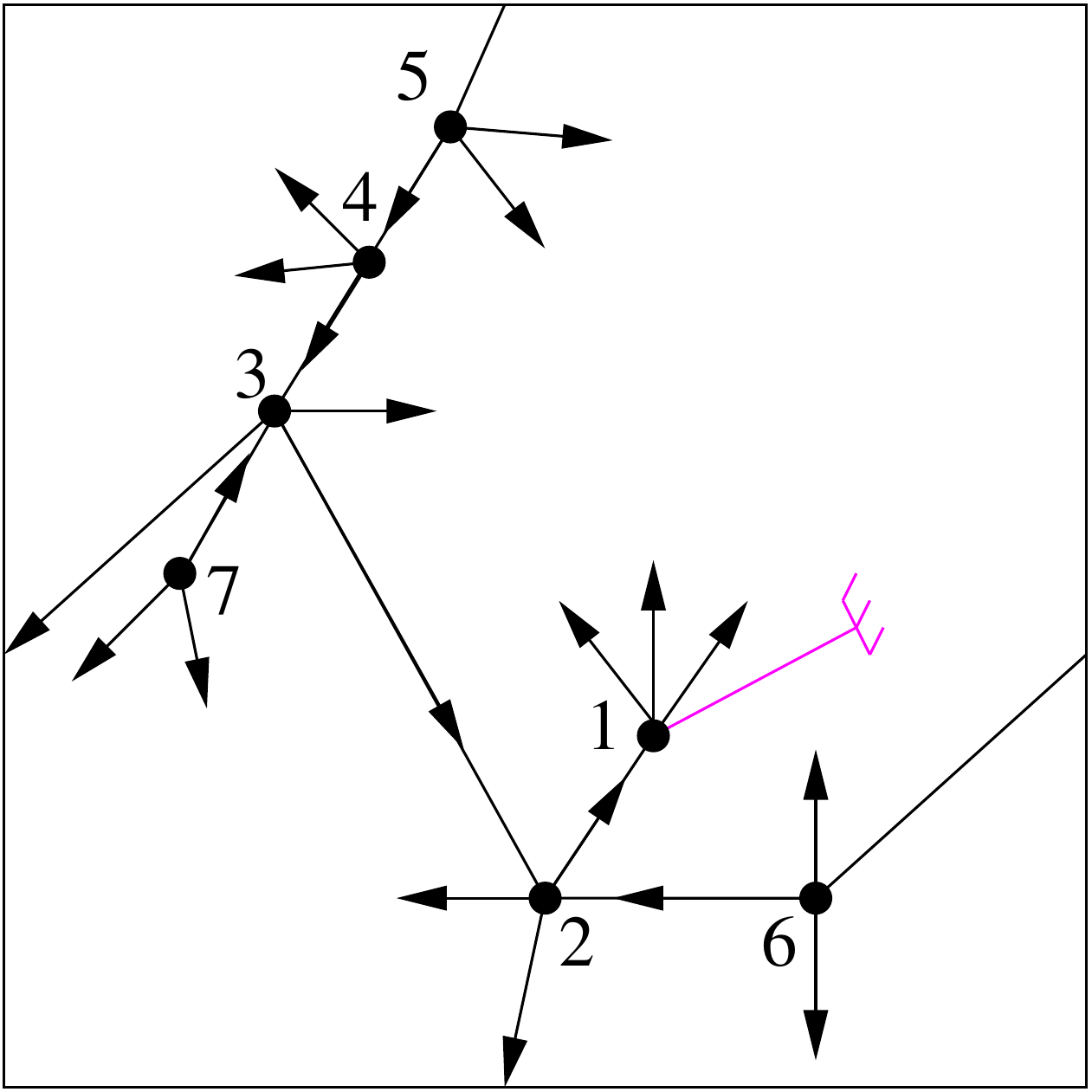}
\ \ \ \includegraphics[scale=0.4]{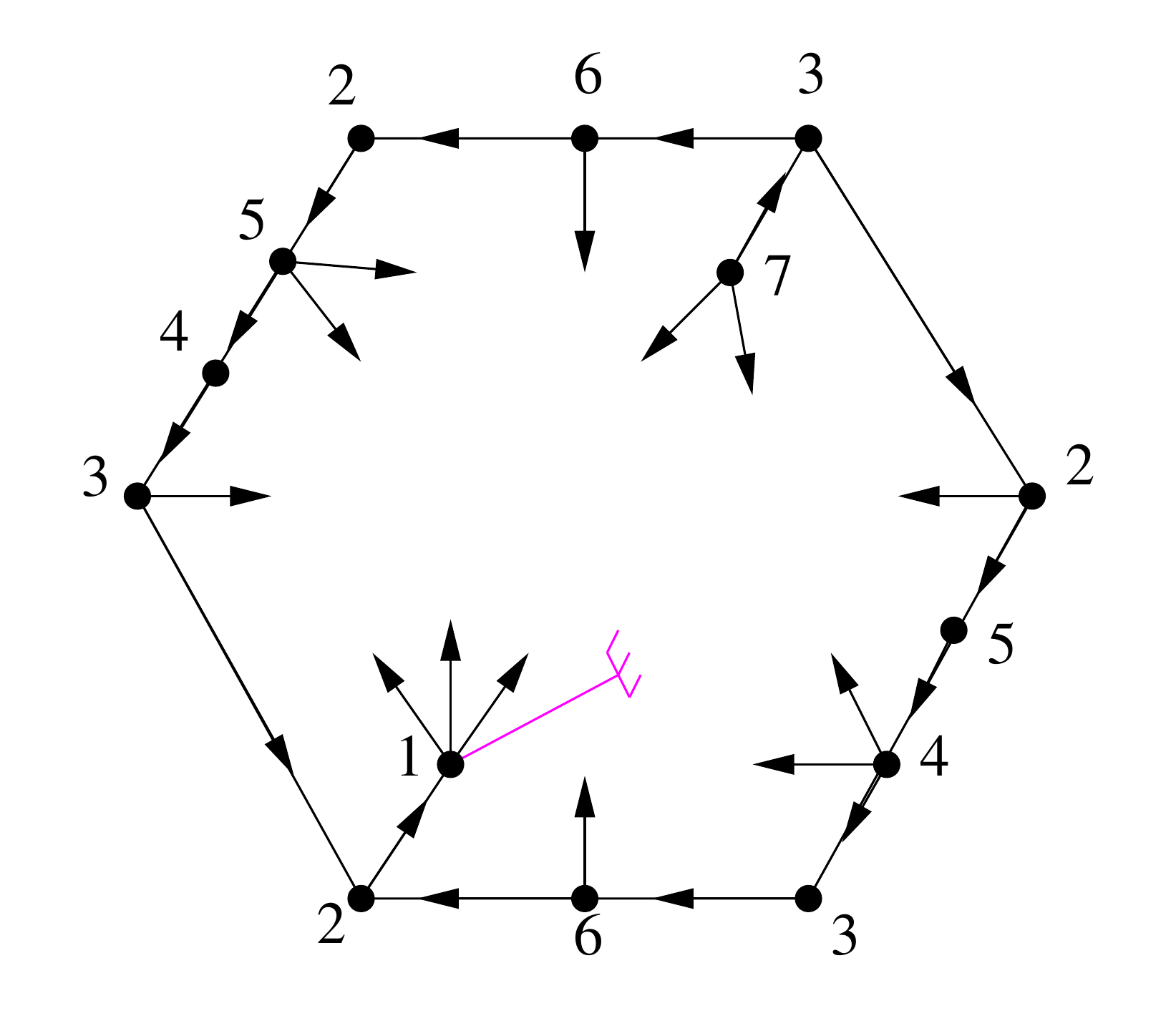}
\caption{Orientation of the edges and stems of an element of $\mathcal T_r(7)$.} 
\label{fig:orientation}
\end{figure}
 
For a cycle $C$ of $T$, given with a traversal direction, let
$\gamma(C)$ be the number of outgoing edges and stems that are
incident to the right side of $T$ minus the number of outgoing edges
and stems that are incident to its left side.  A unicellular map of
$\mathcal T_{r}(n)$ is said to be \emph{balanced} if $\gamma(C)=0$ for
all its (non-contractible) cycles $C$. Let us call
$\mathcal T_{r,s,b}(n)$ the set of balanced elements of
$\mathcal T_{r,s}(n)$.
 
Figure~\ref{fig:orientation} is an example of an element of
$\mathcal T_{r,s,b}(7)$.  The values $\gamma$ of the cycles of the
unicellular map are much more easier to compute on the left
representation.
 
 A consequence of~\cite{despre2017encoding} (see the proof of Theorem
 7 where $\mathcal T_{r,s,b}(n)$ is called
 $\mathcal U'_{r,b,\gamma_0}(n)$ and $\mathcal G(n)$ is called
 $\mathcal T'_r(n)$), is that, for $n\geq 1$, the complete closure
 procedure is indeed a bijection between elements of
 $\mathcal T_{r,s,b}(n)$ and $\mathcal G(n)$, that we denote $\Phi$
 in the curent paper:

\begin{theorem}[\cite{despre2017encoding}]
\label{them:bijectionbenjamin}
For $n\geq 1$, there is a bijection between $\mathcal T_{r,s,b}(n)$ and $\mathcal G(n)$.
\end{theorem}

The left of Figure~\ref{fig:closure} gives an example of a hexagonal
unicellular map in $\mathcal T_{r,s,b}(7)$. Note that on the right of
Figure~\ref{fig:closure}, the face containing the root angle, after
the closure procedure, is indeed
a maximal triangle, so the obtained triangulation is an element of
$\mathcal G(7)$ if rooted on the corner of the face corresponding to
the root angle.

Given an element $T$ of $\mathcal T_{r,s,b}(n)$, the canonical
orientation of $T$, defined previously, induces an orientation of the
edges of the corresponding triangulation $G$ of $\mathcal G(n)$ that
is also called the \emph{canonical orientation} of $G$. Note that in
this orientation of $G$, all the vertices have outdegree exactly $3$,
we call such an orientation a \emph{$3$-orientation}.  In fact this
orientation corresponds to a particular $3$-orientation that is called
the \emph{minimal balanced Schnyder wood of $G$ w.r.t. to the root
  face} (see~\cite{leveque2017generalization} for more on Schnyder
woods in higher genus).  We extend the definition of function $\gamma$
to $G$ by the following.  For a cycle $C$ of $G$, given with a
traversal direction, let $\gamma(C)$ be the number of outgoing edges
that are incident to the right side of $T$ minus the number of
outgoing edges that are incident to its left side. As shown
in~\cite{leveque2017generalization}, the canonical orientation of $G$
as the particular property that $\gamma(C)=0$ for all its
non-contractible cycles $C$, we call this property \emph{balanced}.
 
Figure~\ref{fig:canonical}, gives the canonical orientation of $K_7$
obtained from the canonical orientation of its corresponding element
in $\mathcal T_{r,s,b}(7)$ after a complete closure procedure.
 
\begin{figure}[h]
\center
\includegraphics[scale=0.4]{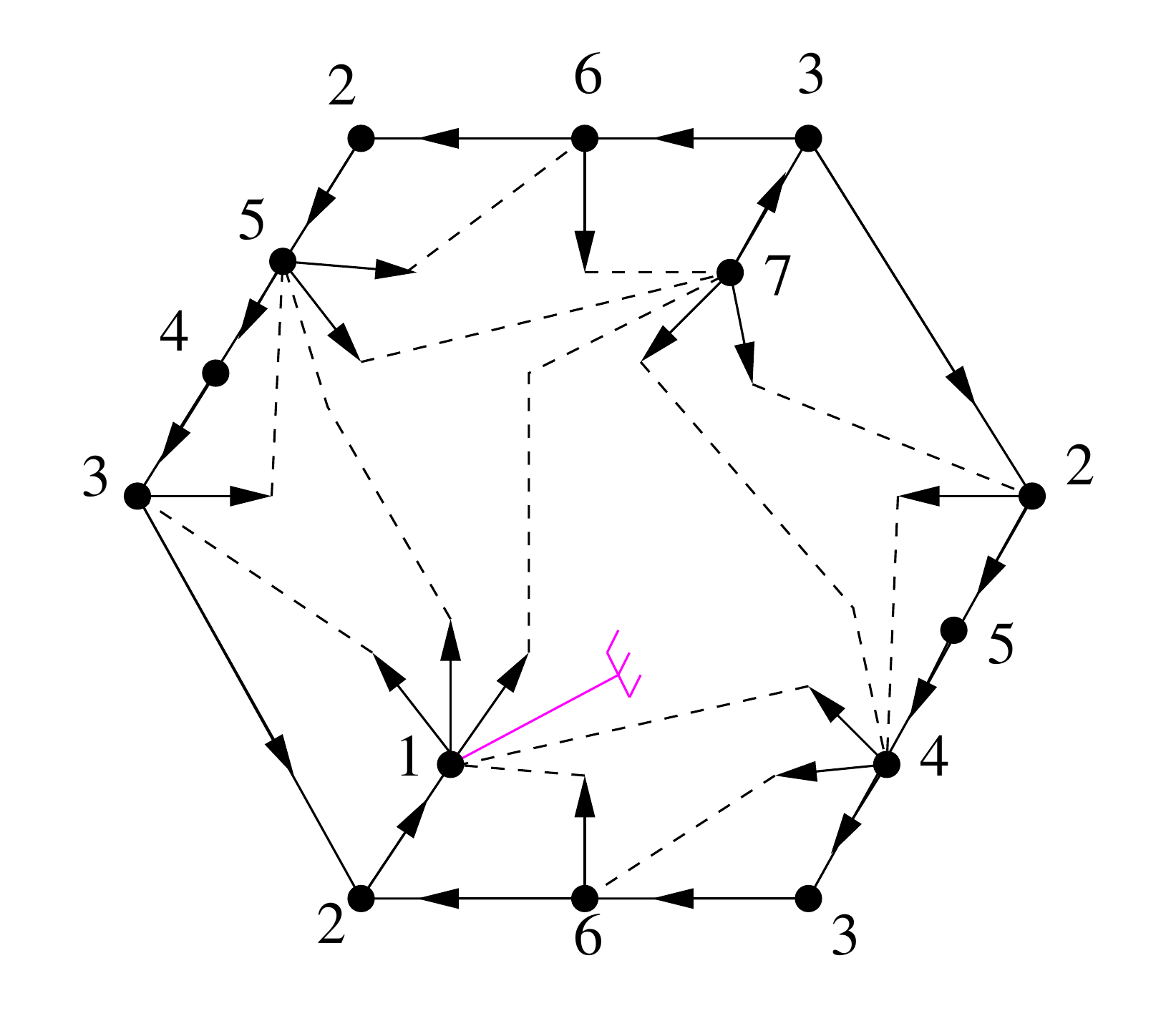}
\ \  
\includegraphics[scale=0.4]{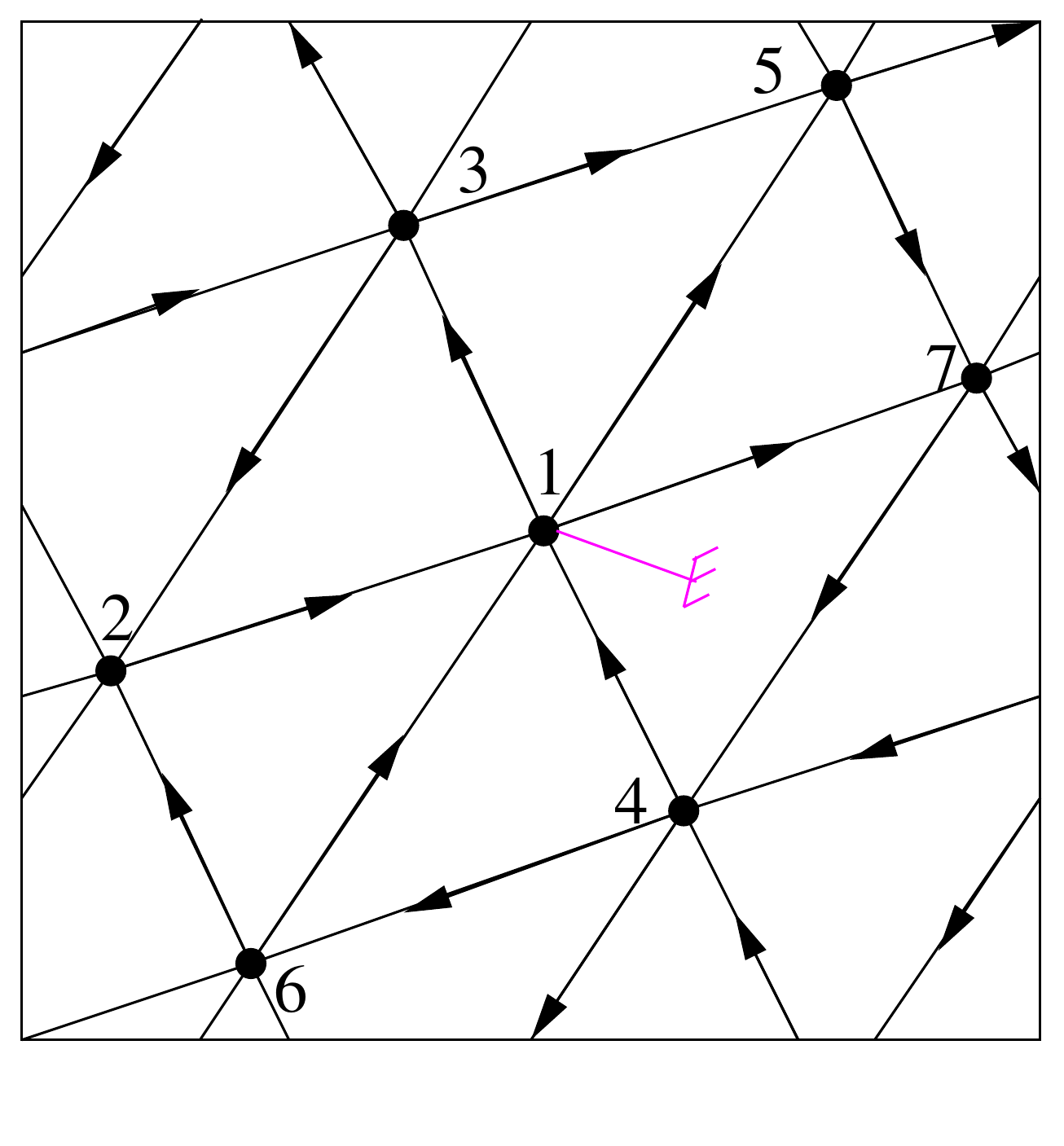}\\

\caption{The canonical orientation of $K_7$.}
\label{fig:canonical}
\end{figure}

\subsection{Unrooted unicellular maps}
\label{sec:unrootuni}

Given an element $T$ of $\mathcal T_{r,s,b}(n)$, we have seen that the
root stem $s_0$ can be the last stem that is attached by the
complete closure procedure.  Consequently, if one removes the root
stem $s_0$ from $T$ to obtain an unicellular map $U$ with $n$
vertices, $n+1$ edges and $2n-2$ stems, one can recover the graph
$T_{2n-2}$ by applying the closure procedure on $U$.

For $n\geq 1$, let $\mathcal U(n)$ denote the set of (non-rooted)
toroidal unicellular maps, with exactly $n$ vertices, $n+1$ edges and
$2n-2$ stems satisfying the following: a vertex is incident to exactly
$2$ stems if it is not a special vertex, $1$ stem if it is the special
vertex of a
hexagon and $0$ stem if it is the special vertex of a square. Thus, given an
element $T$ of $\mathcal T_{r}(n)$, the element $U$ obtained from $T$
by removing the root angle and the root stem is an element of
$\mathcal U(n)$.

Since an element $U$ of $\mathcal U(n)$ is non-rooted, it has no
"canonical orientation" as define previously for elements of
$\mathcal T_r(n)$.  Nevertheless one can still orient all the stems as
outgoing and compute $\gamma$ on the cycles of $U$ by considering
only its stems in the counting (and not the edges nor the root stem
anymore). For a cycle $C$ of $U$, given with a traversal direction,
let $\gamma(C)$ be the number of outgoing stems that are incident to
the right side of $U$ minus the number of outgoing stems that are
incident to its left side.  A unicellular map of $\mathcal U(n)$ is
said to be \emph{balanced} if $\gamma(C)=0$ for all its
(non-contractible) cycles $C$.  Let us call $\mathcal U_{b}(n)$ the
set of elements of $\mathcal U(n)$ that are balanced.

As remarked in~\cite{despre2017encoding}, an interesting property is
that an element $U$ of $\mathcal U(n)$ is balanced if and only if any
element $T$ of $\mathcal T_{r}(n)$ obtained from $U$ by adding a root
stem anywhere in $U$ is balanced (recall that in $U$ we use the
canonical orientation to compute $\gamma$).  Moreover, given an
element $T$ of $\mathcal T_{r,b}(n)$, then the element $U$ of
$\mathcal U(n)$, obtained by removing the root angle, (the canonical
orientation,) and the root stem is balanced.

Figure~\ref{fig:noroot} is the element of $\mathcal U_{b}(7)$
corresponding to Figure~\ref{fig:orientation}.

\begin{figure}[h]
\center
\includegraphics[scale=0.4]{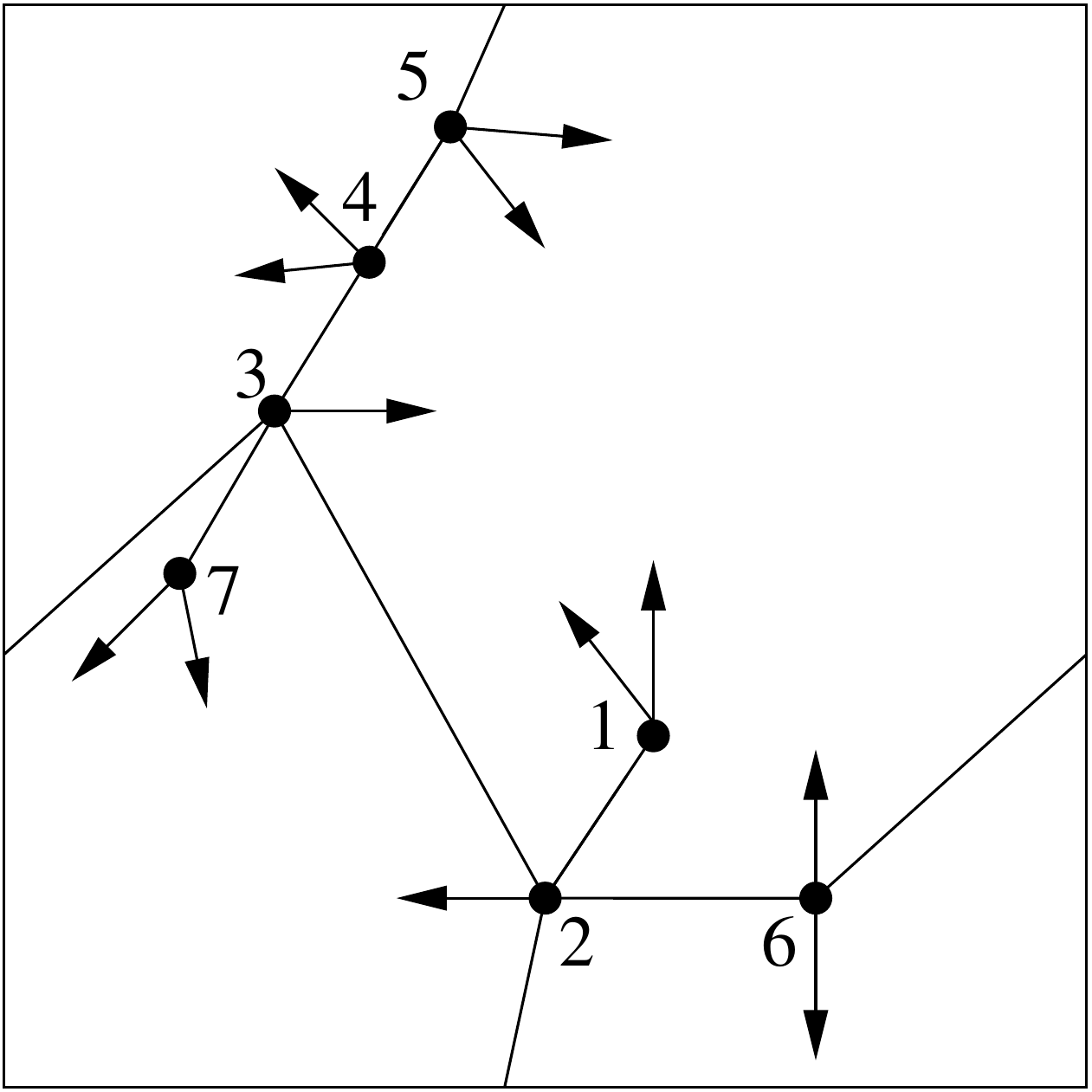}
\caption{Example of an element of $\mathcal U_b(7)$.}
\label{fig:noroot}
\end{figure}

\section{Labeling of the angles and distance properties}
\label{sec:label}

For $n\geq 1$, let $T$ be an element of $\mathcal T_{r,s,b}(n)$, and
$G=\Phi(T)$ the corresponding element of $\mathcal G(n)$ by
Theorem~\ref{them:bijectionbenjamin}. Let $V$ (resp. $E$) denotes the
set of vertices (resp. edges) of $G$. Let $a_0$ be the root angle of
$T$ and $v_0$ be its root vertex. We use the same notations for the
root angle and vertex of $G$ (while maintaining the root angle on the
right side of every stem during the complete closure procedure, as
explained in Section~\ref{sec:preliminaries}).  In this section, we
prove some relations between the graph distance in the triangulation
$G$ and a particular labeling of the vertices defined on the
unicellular map $T$.

\subsection{Definition and properties of the labeling function}

Let $\ell=4n+1$ be the number of angles of $T$. We add a special
dangling half-edge incident to the root angle of $T$, called the
\emph{root half-edge} (and not considered as a stem).  Let $\Gamma$ be
the obtained unicellular map. We define the \emph{root angle} of
$\Gamma$ as the angle of $\Gamma$ just after the root half-edge in
counterclockwise order around its incident vertex.  Let
$A=(a_0,\ldots,a_\ell)$ be the sequence of consecutive angles of
$\Gamma$ in clockwise order around the unique face of $\Gamma$ such
that $a_0$ is the root angle. Note that $a_\ell$ is incident to the
root half-edge. For $0\leq i \leq l-1$, two angles $a_i$ and $a_{i+1}$
are either consecutive around a stem or consecutive around an edge of
$\Gamma$.  We define a labeling function
$\lambda: A \rightarrow \mathbb{Z}$ as follows. Let
$\lambda(a_0)=3$. For $0\leq i \leq l-1$, let
$\lambda(a_{i+1})=\lambda(a_i)+1$ if $a_i$ and $a_{i+1}$ are
consecutive around a stem, and let $\lambda(a_{i+1})=\lambda(a_i)-1$
if they are consecutive around an edge.  By definition, the
unicellular map $\Gamma$ has $n+1$ edges and $2n-1$ stems.  While
going clockwise around the unique face of $\Gamma$, each edge is
encountered twice, so $\lambda(a_\ell)=2n-1-2(n+1)+\lambda(a_0)=0$.
Figure~\ref{fig:labeling} gives an example of the labeling function of
the unicellular map of Figure~\ref{fig:orientation}.

\begin{figure}[h]
\center
\includegraphics[scale=0.4]{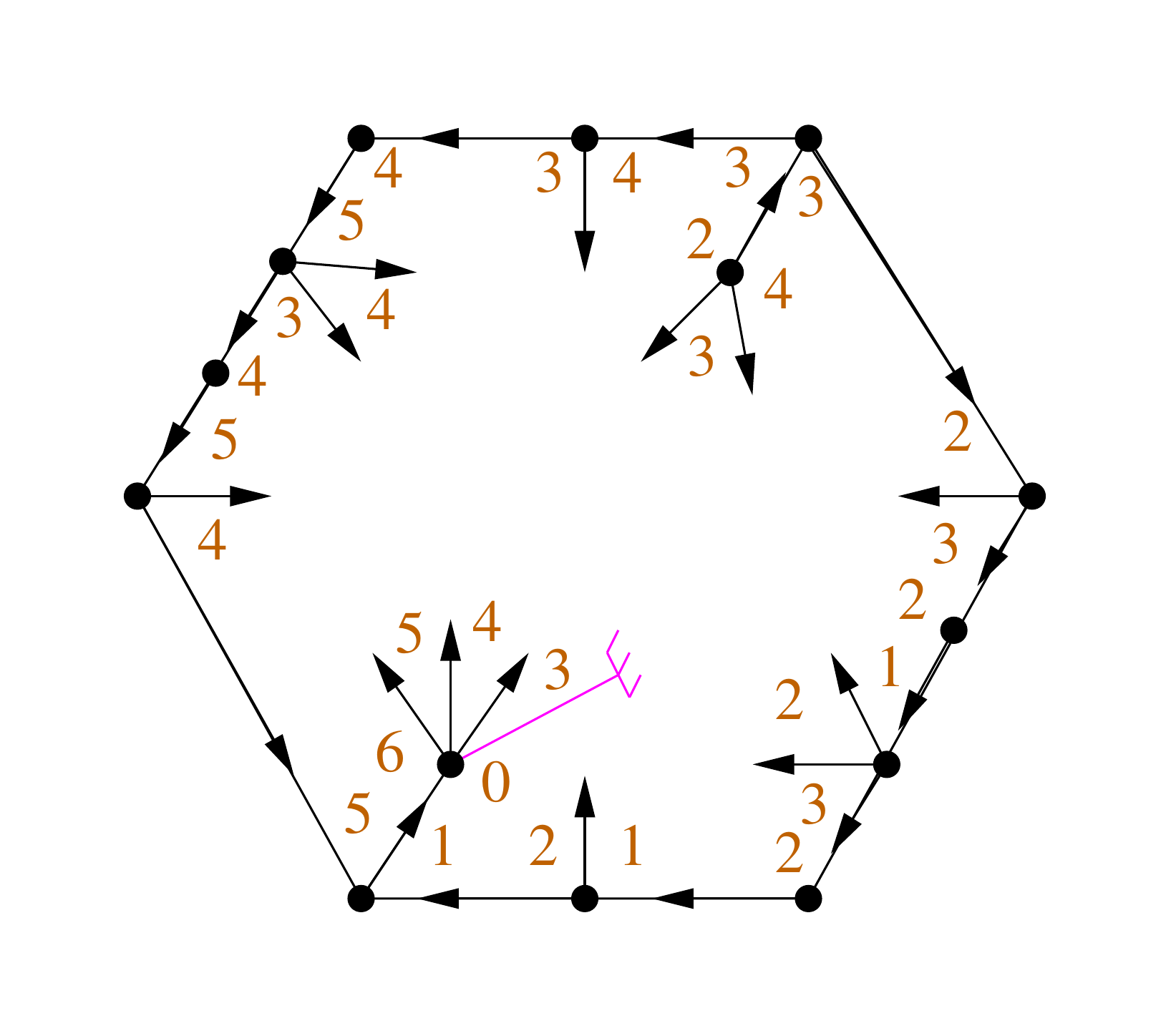}
\caption{Labeling of the angles of the unicellular map.}
\label{fig:labeling}
\end{figure}

Given a stem $s$ of $\Gamma$, we define the label $\lambda(s)$ of $s$ as the label of the angle that is just before $s$ in counterclockwise order around its incident vertex.

The complete closure procedure is formally defined on $T$ but we can
consider that it behaves on $\Gamma$ since the presence of the root
half-edge in $\Gamma$ does not change the procedure as $T$ is safe
(the root half-edge is maintained on the right of every stem during
the closure).  Let $\Gamma_0=\Gamma$, and, for $1\leq k \leq 2n-1$,
let $\Gamma_{k}$ be the map obtained from $\Gamma_{k-1}$ by closing an
admissible triple of $\Gamma_{k-1}$. By the bijection $\Phi$ we have
that $\Gamma_{2n-1}$ is the graph $G$ with an additional dangling
half-edge incident to the root angle, we call this graph $G^+$.  We
propagate the labeling $\lambda$ of $\Gamma$ during the closure
procedure by the following.  For $1\leq k \leq 2n-1$, when the stem
$s$ of $\Gamma_{k-1}$ is attached, it splits an angle $a$ of
$\Gamma_{k-1}$ into two angles of $\Gamma_{k}$ that both inherit the
label of $a$ in $\Gamma_{k-1}$. In other words, the complete closure
procedure just splits some angles that keeps the same label on each
side of the split. We still note $\lambda$ the labeling of the angles
of $\Gamma_k$. It is clear that the labeling of $G^+=\Gamma_{2n-1}$
that is obtained is independent from the order in which the admissible
triples are closed. We denote $\mathcal{A}(i)$ the set of angles of
$G^+$ which are splited from $a_i$ by the complete closure
procedure. Note that for all $a\in \mathcal{A}(i)$, we have
$\lambda(a)=\lambda(a_i)$.  Given a stem $s$ of $\Gamma$, we denote
$a(s)$ the angle of $\Gamma$ corresponding to where $s$ is attached
during the complete closure procedure (i.e. $s$ is attached to an
angle that comes from some splittings of $a(s)$).

Consider a stem $s$ of $T$. Let $i,j$, be such that $a_i$ is the angle just before $s$ in \ccw order around its incident vertex and $a_j=a(s)$. The fact that $T$ is safe implies that $0\leq i<j\leq \ell$.

\begin{lemma}
\label{lemma1}
For $0\leq k \leq 2n-1$,
the rules that are used to define the labeling function $\lambda$ are still valid around the special face of  $\Gamma_{k}$, i.e. the root angle of $\Gamma_{k}$ is labeled $3$, and while walking \cw around the special face of $\Gamma_{k}$, the labels are increasing by one around a stem and decreasing by one along an edge until finishing at label $0$ at the last angle.

In particular, for each stem $s$ of $\Gamma$, we have $\lambda(a(s))=\lambda(s)-1$.
Moreover, all the angles of $\Gamma$ that appear strictly between $s$ and $a(s)$ in \cw order along the unique face of $\Gamma$ have labels that are greater or equal to $\lambda(s)$.
\end{lemma}

\begin{proof}
  We prove the first part of the lemma by induction on $k$. Clearly
  the statement is true for $k=0$ by definition and properties of
  $\lambda$. Suppose now that for $1\leq k \leq 2n-1$, the statement
  is true for $\Gamma_{k-1}$.  Let $s$ be the stem of $\Gamma_{k-1}$
  that is attached to obtained $\Gamma_{k}$. Let $(e_1,e_2,s)$ be the
  admissible triple of $\Gamma_{k-1}$ involving $s$, when $s$ is
  attached. Let $\alpha_0,\alpha_1,\alpha_2,\alpha_3$ be the angles of
  the special face of $\Gamma_{k-1}$ that appears along the admissible
  triple $(e_1,e_2,s)$, such that
  $\alpha_0,s,\alpha_1,e_1,\alpha_2,e_2,\alpha_3$ appears
  consecutively in \cw order around the special face. So we have that
  the dangling part of $s$ is attached to the angle $\alpha_3$ to form
  $\Gamma_{k}$.  Since $T$ is safe, the root angle of $\Gamma_{k-1}$
  is distinct from $\alpha_1, \alpha_2, \alpha_3$. So, by induction,
  the rules of the labeling function applies in $\Gamma_{k-1}$ from
  $\alpha_0$ to $\alpha_3$. Thus
  $\lambda(\alpha_1)=\lambda(\alpha_0)+1$,
  $\lambda(\alpha_2)=\lambda(\alpha_1)-1$,
  $\lambda(\alpha_3)=\lambda(\alpha_2)-1$. So
  $\lambda(\alpha_3)=\lambda(\alpha_1)-1$, and the rules still apply
  in the special face of $\Gamma_{k}$.

A direct consequence of the above paragraph, is that for each stem $s$ of $\Gamma$, we have $\lambda(a(s))=\lambda(s)-1$.

Suppose by contradiction that there is a stem $s$ and an angle of $\Gamma$ that appear strictly between $s$ and $a(s)$ in \cw order along the unique face of $\Gamma$ whose label is less or equal to $\lambda(a(s))$. We choose such an angle $\alpha$ whose label is minimum. With the same notations of the angles  $\alpha_1,\alpha_2$ as above,
since $\lambda(\alpha_2)=\lambda(a(s))+1$ and $\lambda(\alpha_1)=\lambda(a(s))+2$, we have that neither $\alpha_1$ nor $\alpha_2$ comes from a splits of $\alpha$. So there exists an admissible triple $s'$, closed before $s$ is the complete closure procedure, and whose one of the two internal angles $\alpha'_1,\alpha'_2$ (with analogous notations as above)
is $\alpha$ (or comes from a split of $\alpha$). By the rule of the labeling, we have $\lambda(\alpha)\in\{\lambda(a(s'))+1,\lambda(a(s'))+2\}$ (depending on which internal angle it is, either $\alpha'_1$ or $\alpha'_2$). Thus by minimality of $\alpha$, we have $a(s')=a(s)$, but then $\lambda(\alpha)\in\{\lambda(a(s))+1,\lambda(a(s))+2\}$, a contradiction.
\end{proof}

\begin{lemma}
\label{lem:stemcycle}
Consider a (non-contractible) cycle $C$ of $\Gamma$ of length $k$ that does not contain the root vertex. Then there is exactly $k-1$ stems attached to each side of $C$.
\end{lemma}

\begin{proof}
  As explained in Section~\ref{sec:unrootuni}, when one remove from
  $T$ the root stem, the canonical orientation and the root angle, one
  obtain an element of $\mathcal U_b(n)$. So we have that the number
  of stems attached to the left and right side of $C$ are the same. In
  both cases, whether $\Gamma$ is a square or hexagonal unicellular
  map, we have that $C$ is incident to exactly $2(k-1)$ stems, so
  there is exactly $k-1$ stems attached to each side of $C$.
\end{proof}

Note that if $v_0\in C$ then the conclusion of Lemma~\ref{lem:stemcycle} is not true since there is an additional stem attached to the root vertex.

\begin{lemma}
\label{lemma2}

For $0\leq i\leq \ell-1$, we have $\lambda(a_i) > 0$.
\end{lemma}
\begin{proof}

Assume that there exists $0\leq i\leq \ell-1$, such that $\lambda(a_i) \leq  0$.
 Let $k=\max\left \{ 0\leq i\leq l-1: \lambda(a_i)\leq 0 \right \}$.  If $a_k$ and $a_{k+1}$ are consecutive along an edge, then we have $\lambda(a_{k+1})=\lambda(a_k)-1<0$. 
  If $a_k$ and $a_{k+1}$ are separated by a stem, then, by Lemma~\ref{lemma1}, we have $\lambda(a(s))=\lambda(a_k)-1$
  , so there exists $k'>k$ such that $\lambda(a_{k'})<0$. In both cases, there is a contradiction to the definition of $k$.
\end{proof}

Let $V_S$  be the set of special vertices of $\Gamma$ (defined in Section~\ref{sec:preliminaries}).
 We call \emph{proper} the edges and vertices of $\Gamma$ that are on at least one cycle of $\Gamma$. Let $V_P$ (respectively $E_P$) be the set of proper vertices (respectively edges) of $\Gamma$. Note that $V_S\subseteq V_P$.

We call \emph{root path} the (unique) shortest path of $\Gamma$ from the root vertex to a proper vertex. Note that the root path might have length $0$ if $v_0$ is proper. 
The sequence of vertices along the root path is denoted $V_R=(r_0,r_1,...,r_s)$, with $s\geq 0$, $r_0=v_0$ and $r_s$ is proper. The set of edges of the root path is denoted $E_R$.
Let $V_N=V \setminus (V_P \bigcup V_R) $ be the set of \emph{normal vertices} of $\Gamma$ and $E_N=E \setminus (E_P \bigcup E_R) $ be the set of \emph{normal edges} of $\Gamma$. 

The \emph{canonical orientation} of $\Gamma$ is the orientation of the edges and stems of $\Gamma$ that corresponds to the canonical orientation of $T$ (the root half edge added has no particular orientation). Consider an edge $e$ of $\Gamma$ with its orientation in the canonical orientation, then by the orientation rule, the angles of $\gamma$ incident to $e$ that are on its right side have greater indices in the set $A$ than the angles that are on its left side, i.e. they are seen after while going in \cw order around the unique face of $\Gamma$ starting from the root angle.

\begin{lemma}
\label{lem:labelvariation}
Consider an edge $e=uv$ of $\Gamma$ that is oriented from $u$ to $v$ in the canonical orientation of $\Gamma$. Let $0\leq i<j< \ell$ such that $a_i,a_{i+1},a_j, a_{j+1}$ appear in this order in \ccw order around $e$ with $a_i, a_{j+1}$ incident to $v$ and $a_{i+1}, a_j$ incident to $u$. Then we have the following (see Figure~\ref{fig:labelvariation}):
\begin{tabular}{ccc}
 $ \lambda(a_{j+1})-\lambda(a_i) = 
  \begin{cases}
    0 & \text{if } e\in E_N \\
    -3 & \text{if } e\in E_P \\
    -6  & \text{if } e\in E_R
  \end{cases}
  $
& and &
$ \lambda(a_{i+1})-\lambda(a_{j}) =
  \begin{cases}
    -2 & \text{if } e\in E_N \\
    1 & \text{if } e\in E_P \\
    4  & \text{if } e\in E_R
  \end{cases}
$
\end{tabular}

\end{lemma}

\begin{proof}
Note first that by the labeling rule we have $\lambda(a_{i+1})=\lambda(a_i)-1$ and  $\lambda(a_{j+1})=\lambda(a_j)-1$. So $(\lambda(a_{i+1})-\lambda(a_{j}))+ (\lambda(a_{j+1})-\lambda(a_i))=-2$. 

Suppose first that $e\in E_N$.
While going clockwise around the unique face of $\Gamma$ starting from $a_{i}$ to $a_{j+1}$, we encounter only normal vertices and edges.
So we go around a planar tree whose
 edges are encountered twice and whose number of stems is equal to twice the number of edges. This implies that $\lambda(a_{j+1})-\lambda(a_i)=0$ and so $\lambda(a_{i+1})-\lambda(a_{j})=-2$. 
 
 The case where $e\in E_R$ is quite similar. While going clockwise around the unique face of $\Gamma$ starting from $a_{j}$ to $a_{i+1}$, we are in the same situation as above except that we go over the root vertex. The root vertex is incident to $1$ more stem than normal vertices and there is a jump of $3$ from the label of $a_\ell$ to $a_0$ around the root vertex. This implies that $\lambda(a_{i+1})-\lambda(a_{j})=4$ and so $\lambda(a_{j+1})-\lambda(a_{i})=-6$.
 
It only remains to consider the case where $e\in E_P$. We suppose here that $\Gamma$ is hexagonal. The case where $\Gamma$ is square can be proved similarly.

 The value $\lambda(a_{j+1})-\lambda(a_{i})$ is equal to the number of stems minus the number of edges that are encountered while
going clockwise around the unique face of $\Gamma$ starting from $a_{i}$ to  $a_{j+1}$, with $i<j$. Each normal edge that is met is encountered twice and the number of stems that are met and attached to normal vertices is equal to exactly twice this number of edges. So there number does not affect the value  $\lambda(a_{j+1})-\lambda(a_{i})$. Thus we just have to look at proper edges and stems attached to proper vertices.

Let $s$ be the first special vertex that is encountered. Note that $s$ is encountered twice along the computation and the other special vertex only once. Let $P$ be the unique path of $\Gamma$ between $v$ and $s$ with no special inner vertices. Let $k$ be the length of $P$.  All the stems attached to inner vertices of $P$ are encountered exactly once and all the edges of $P$ are encountered exactly twice. Since each inner vertex of $P$ is incident to exactly two stems, and there one more edges in $P$ than inner vertices, this part results in value $-2$ in the computation of  $\lambda(a_{{j+1}})-\lambda(a_{i})$.

It remains to look at the part encountered between the two copies of $s$. This corresponds to exactly a cycle $C$ of $\Gamma$ of length $k'$, where all its edges and all the stems incident to one of its side are encountered exactly once. Note that $v_0$ does not belong to $C$ since $i<j$. Then by Lemma~\ref{lem:stemcycle}, there are exactly $k'-1$ stems attached to each side of $C$. So this part results in value $(k'-1)-k'=-1$ is the computation of  $\lambda(a_{{j+1}})-\lambda(a_{i})$.

Finally, in total we obtain $\lambda(a_{{j+1}})-\lambda(a_{i})=-3$  and so $\lambda(a_{i+1})-\lambda(a_{j})=1$.  
\end{proof}

\begin{figure}[h]
\center
\scalebox{0.7}{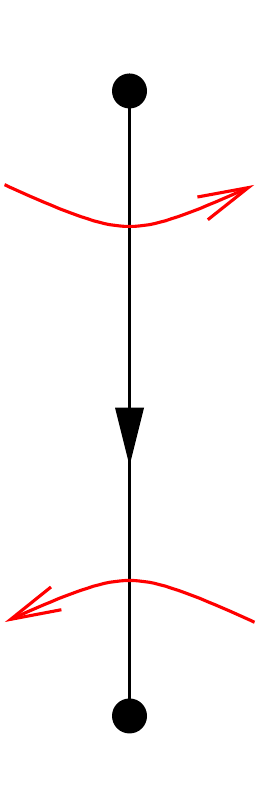}
 \hspace{4em}
 \scalebox{0.7}{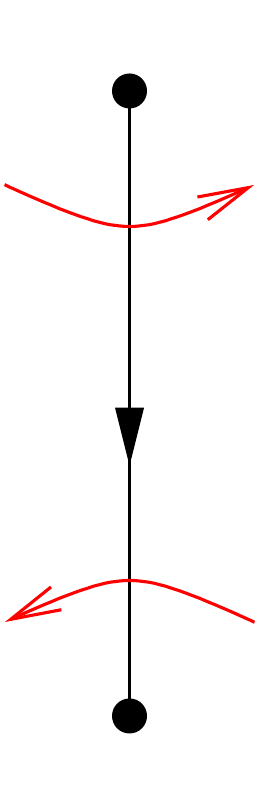}
 \hspace{4em}
 \scalebox{0.7}{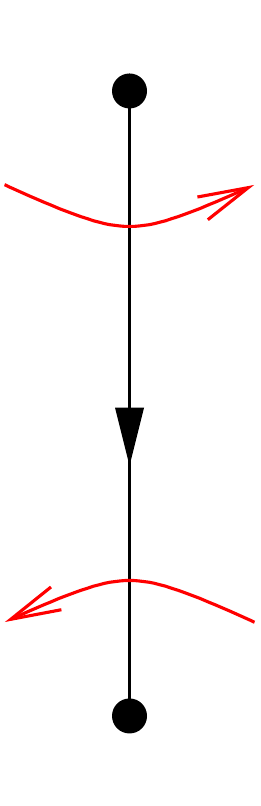}
\caption{Variations of the labeling around the three different kind of edges of $\Gamma$.}
\label{fig:labelvariation}
\end{figure}

One can remark on Figure~\ref{fig:labelvariation} that an incoming edge of $\Gamma$ corresponds to a variation of the labeling in \ccw order around its incident vertex that is always $\leq 0$.

By Lemma~\ref{lem:labelvariation}, we can deduce the variation of the
labels around the different kind of possible vertices that may appear
on $\Gamma$. They are many different such vertices, the $12$ different
cases are represented on Figures~\ref{fig:allcases}.(a) to
($\ell$). The stems are not represented on the figures, except the
root stem, but their number is indicated below each figure. These
stems can be incident to any angle of the figures, except the angles
incident to the root half-edge that are marked with an empty set. Recall
that each of this stem results in a $+1$ in the variation of the
labels while going \ccw around their incident vertex. The incoming
normal edges are not represented either. There can be an arbitrary
number of such edges incident to each angle of the figures. By
Lemma~\ref{lem:labelvariation}, there is no variation of the labels
around them. When $v=v_0$, i.e. $v$ is the root vertex, we have
represented the root stem and the root half-edge. In
this particular case, there is no stem nor incoming normal edges
incident to the angles incident to the root half-edge by the safe
property.

\begin{figure}[hp]
\center
\tiny
\begin{tabular}{ccc}
\scalebox{0.5}{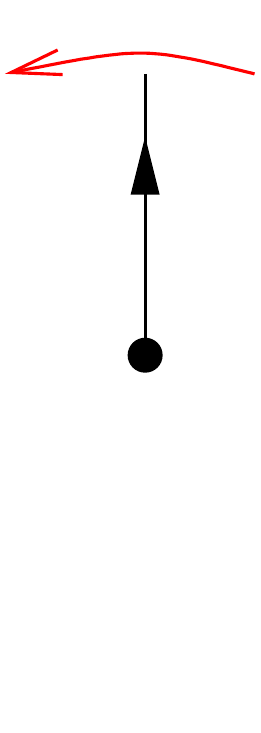} & \scalebox{0.5}{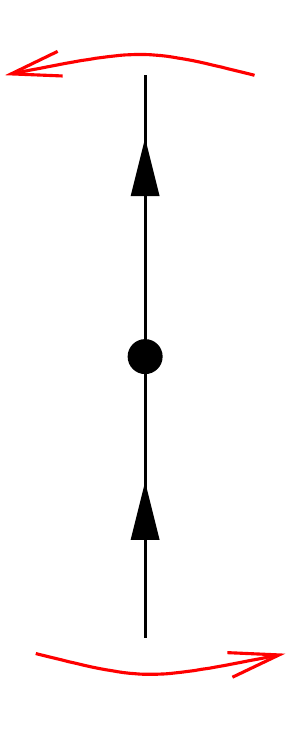} & \scalebox{0.5}{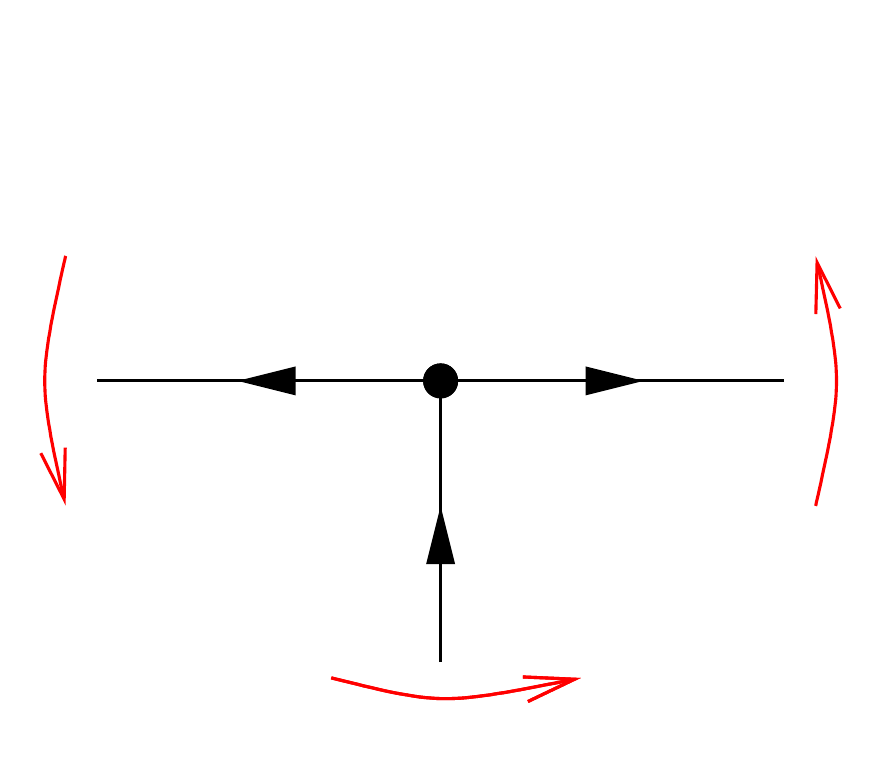} \\
$v\in V_N$& 
$v\in V_P$, $v\notin V_S$, $v\notin V_R$ & 
$v\in V_S$, $v\notin V_R$, hexagonal \\
2 additional stems  & 2 additional stems & 1 additional stem\\
$M(v)-m(v)=2$ & $M(v)-m(v)=3$& $M(v)-m(v)=3$ \\
(a) & (b) & (c) \\ \\
\scalebox{0.5}{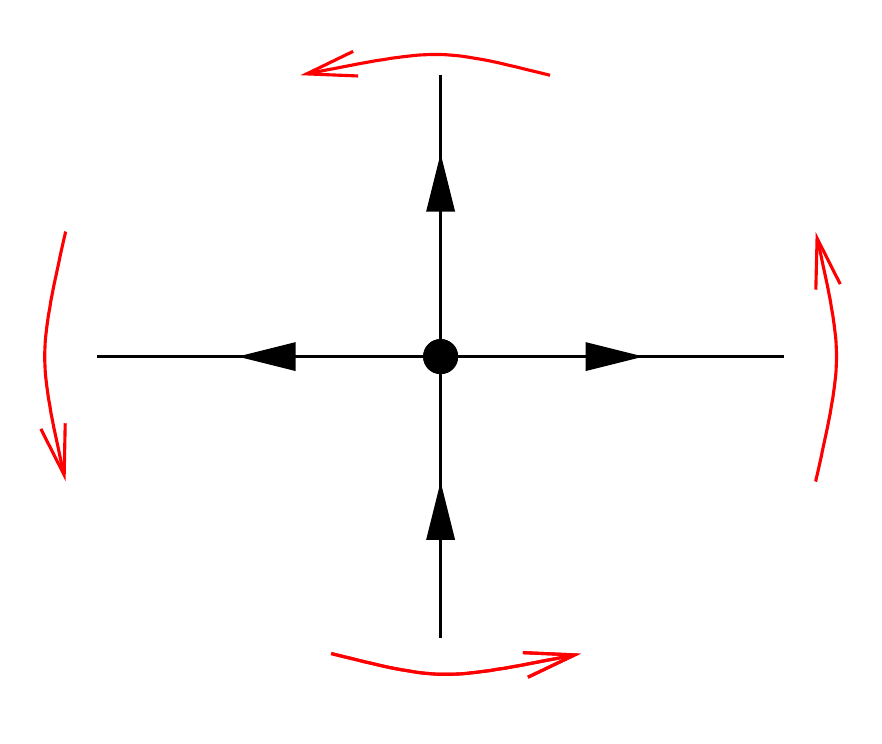} & \scalebox{0.5}{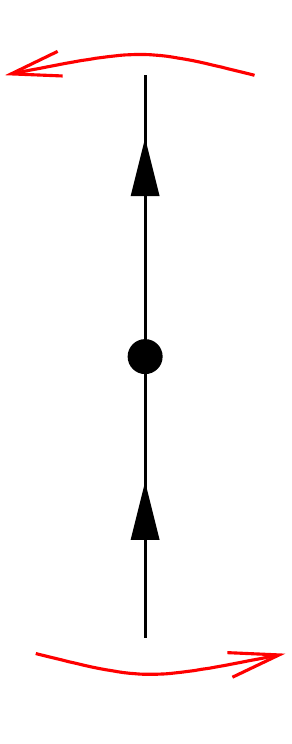} & \scalebox{0.5}{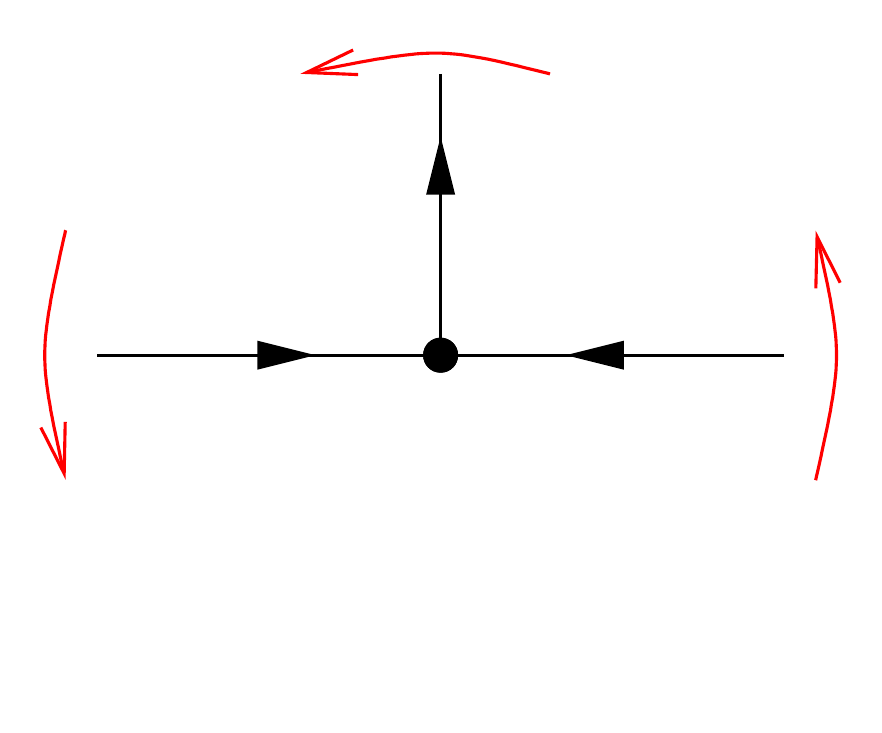} \\
$v\in V_S$, $v\notin V_R$, square  & $v\in V_R$, $v\neq r_0$, $v\neq r_s$& $v= r_s$, $s>0$, $v\notin V_S$\\
0 additional stem & 2 additional stems & 2 additional stems \\
$M(v)-m(v)=3$  & $M(v)-m(v)= 6$&  $4\leq M(v)-m(v)\leq 6$\\
(d) & (e) & (f) \\
\\

\scalebox{0.5}{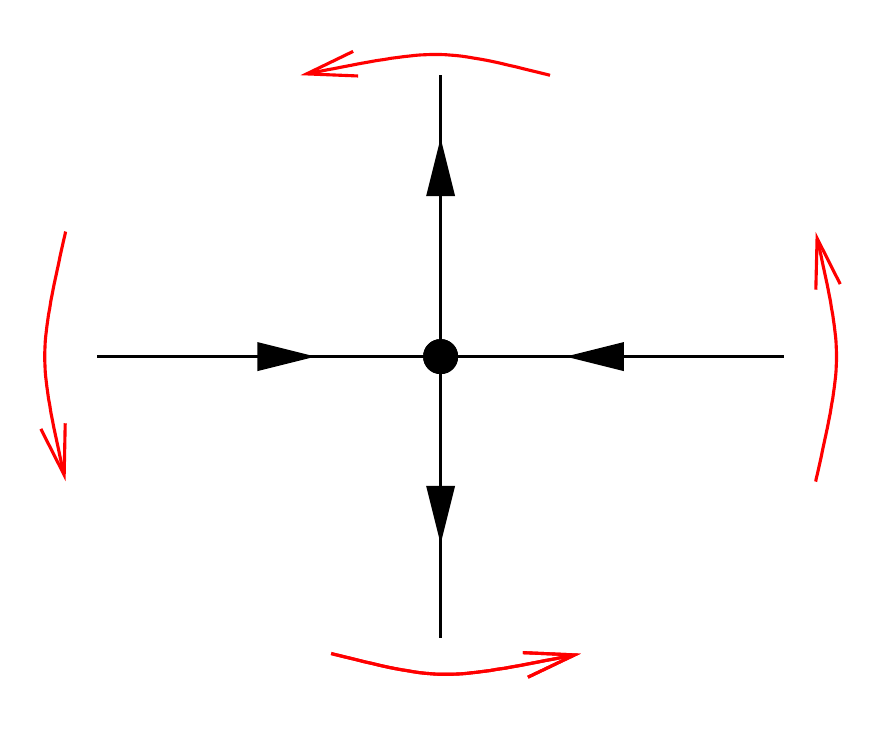} & \scalebox{0.5}{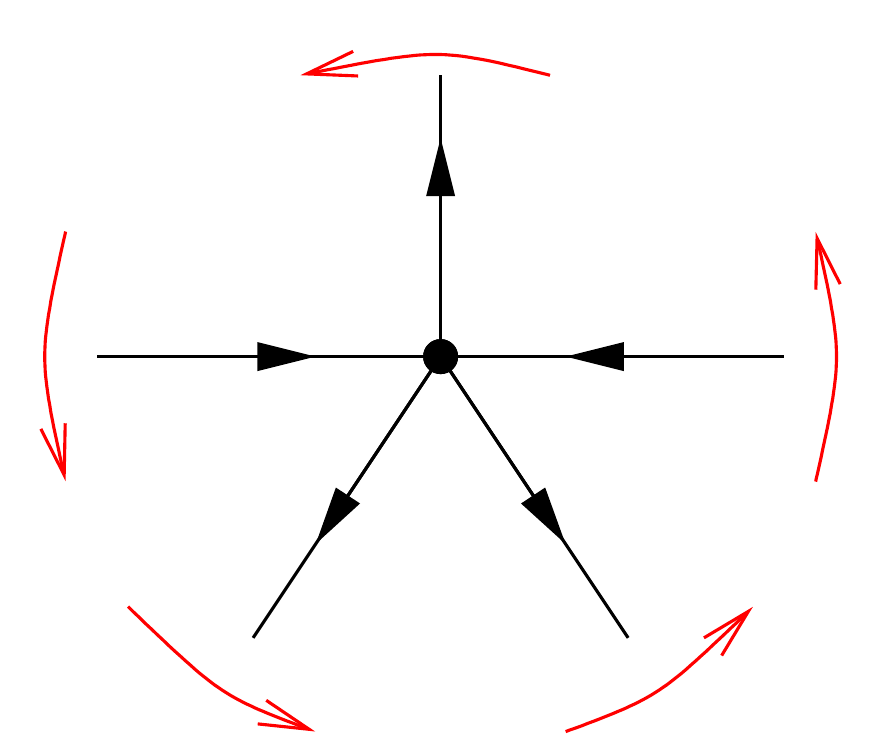} & \scalebox{0.5}{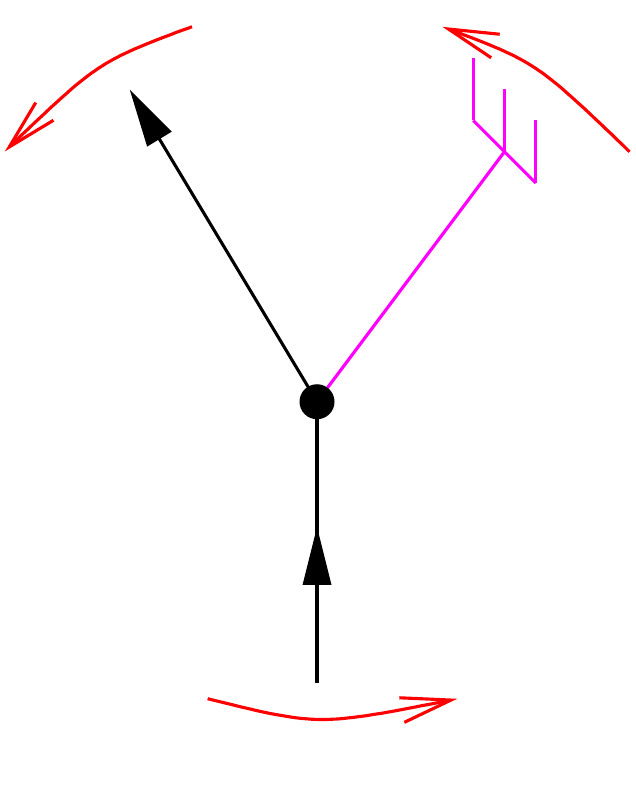} \\
$v= r_s$, $s>0$, $v\in V_S$, hexagonal& $v= r_s$, $s>0$, $v\in V_S$, square
& $v= v_0$, $v\notin V_P$ \\
1 additional stem & 0 additional stems 
& 2 additional stems \\
$4\leq M(v)-m(v)\leq 5$ 
& $M(v)-m(v)=4$
& $M(v)-m(v)=6$ \\
(g) & (h) & (i) \\
\\

\scalebox{0.5}{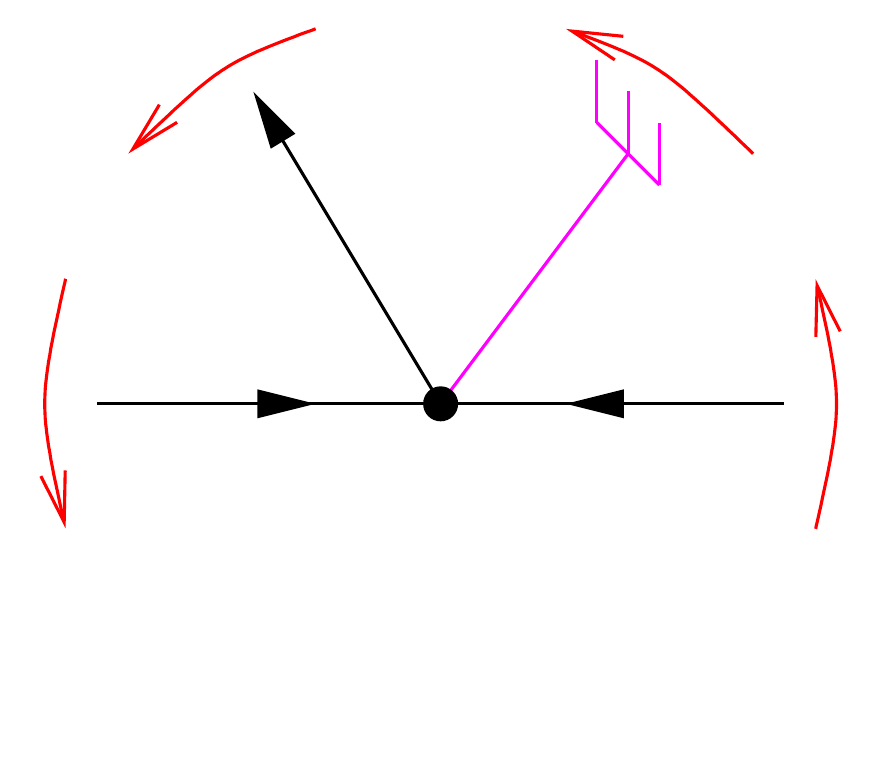} & \scalebox{0.5}{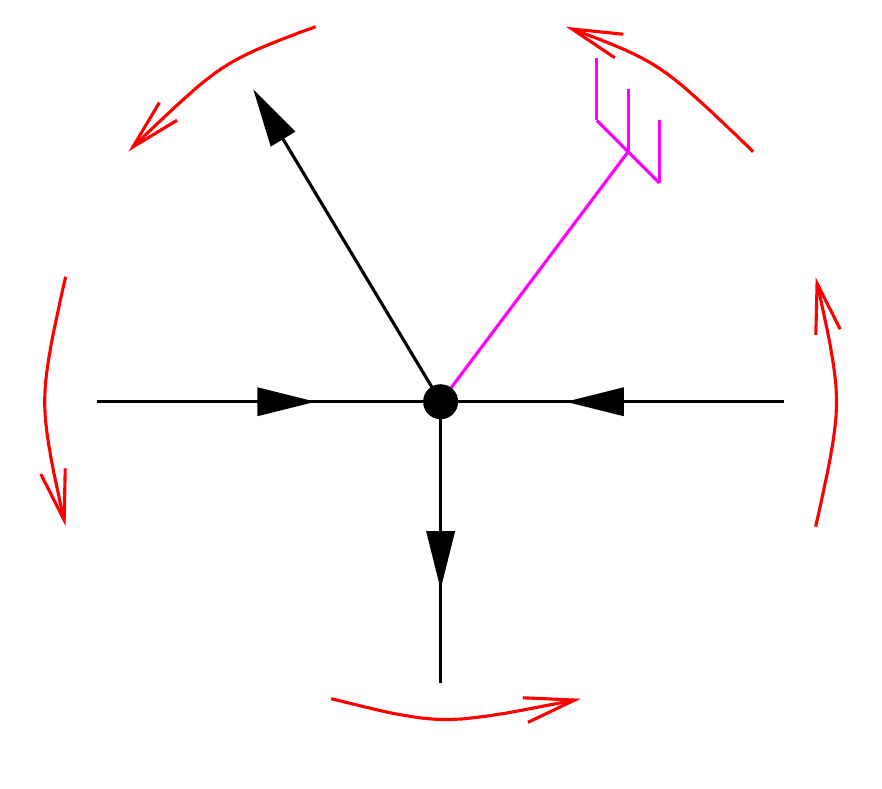} & \scalebox{0.5}{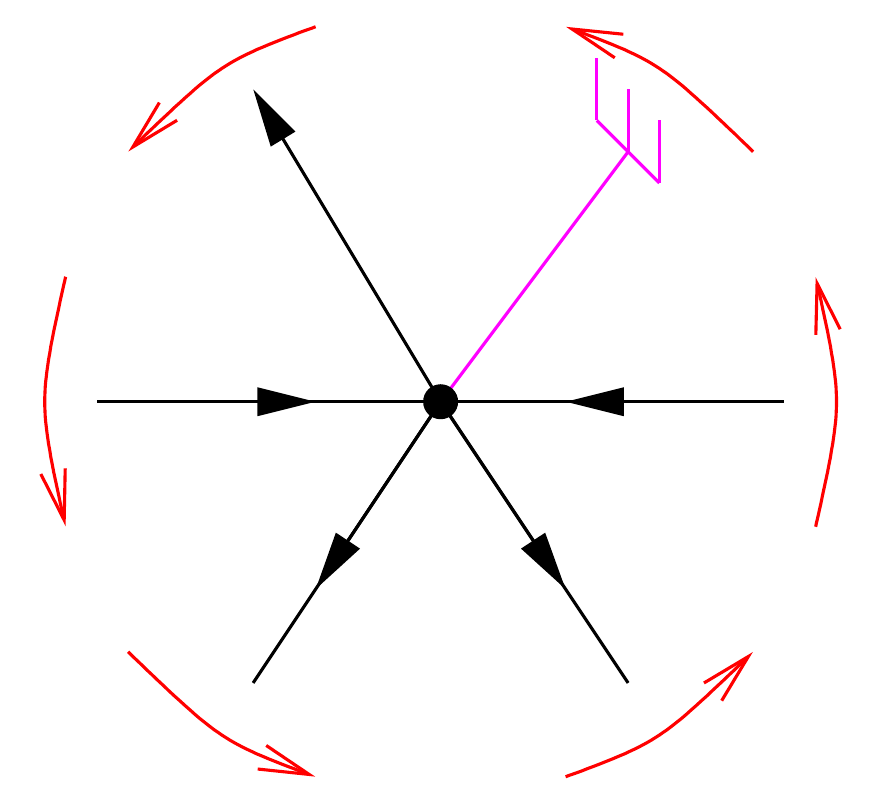} \\
$v= v_0$, $v\in V_P$, $v\notin V_S$ & $v= v_0$, $v\in V_S$, hexagonal& $v= v_0$, $v\in V_S$, square \\
2 additional stems & 1 additional stem & 0 additional stems \\
$4\leq M(v)-m(v)\leq 6$ & $4\leq M(v)-m(v)\leq 5$& $M(v)-m(v)=4$ \\
(j) & (k) & ($\ell$) \\

\end{tabular}
\caption{Variations of the labeling around the different kind of possible vertices of $\Gamma$.}
\label{fig:allcases}
\end{figure}

For each $u\in V$, let $A(u)$ be the set of angles incident to $u$, let $m(u)=\min_{a\in A(u)}\lambda(a)$, and let $M(u)=\max_{a\in A(u)}\lambda(a)$.
On Figures~\ref{fig:allcases}.(a) to ($\ell$)  we have represented the position of the label $M(v)$ and $m(v)$ wherever the missing stems are. We also have given the value of $M(v)-m(v)$ or an inequality on it. This case analysis gives the following lemma :

\begin{lemma}
\label{lem:Mm6}
For all $v \in V$, we have $M(v)-m(v)\leq 6$. 
\end{lemma}

From Lemma~\ref{lem:Mm6}, we obtain the following lemma.

\begin{lemma}
\label{lem:mm7}
For all $\{u,v\}\in E(G)$, we have $ | m(u)-m(v)  |\leq 7$. 
\end{lemma}

\begin{proof}
Let $e\in E(G)$ with extremities $u$ and $v$. 
We consider two cases whether $e$ is an edge of $\Gamma$ or not.

\begin{itemize}
\item \emph{$e$ is an edge of $\Gamma$:}
While walking \cw around the special face of $\Gamma$ from the root angle, there is an angle $\alpha$ incident to $u$ and an angle $\beta$ incident to $v$ that appears consecutively. By definition of the labels, we have $\lambda(\beta)=\lambda(\alpha)-1$.
Moreover by Lemma~\ref{lem:Mm6}, we have
$m(u) \in \left [\![ \lambda(\alpha)-6,\lambda(\alpha) \right ]\!] \text{ and } m(v) \in  \left [\![  \lambda(\beta)-6,\lambda(\beta) \right ]\!]$.
This implies that $\left | m(u)-m(v) \right |\leq 7$.  
\item \emph{$e$ is not an edge of $\Gamma$:}
Thus $e$ comes from the attachment of a stem $s$ of $\Gamma$ by the complete closure procedure. W.l.o.g., we may assume that $s$ is incident to $u$.
By Lemma~\ref{lemma1}, we have $\lambda(a(s))=\lambda(s)-1$. By lemma~\ref{lem:Mm6}, we have
$m(u)\in \left [\![ \lambda(s)-6,\lambda(s) \right ]\!]$ and  $m(v)\in \left [\![ \lambda(s)-7,\lambda(s)-1 \right ]\!]$.
This implies that $|m(u)- m(v)|\leq 7$. 
\end{itemize}
\end{proof}

\subsection{Relation with the graph distance}
\label{subsection:Relationwiththegraphdistance}
For $(u,v) \in V$, we denoted by $d_G(u,v)$ the length (i.e. the number of edges) of a shortest path in $G$ starting at $u$ and ending at $v$.  

Given an angle $\alpha$ of $\Gamma$, let $v(\alpha)$ denote the vertex of $\Gamma$ incident to $\alpha$.

\begin{lemma}
\label{prop:bornesupdistance}
For all $v\in V$, we have $\frac{m(v)}{7} \leq d_G(v_0,v)\leq m(v)$.
\end{lemma}

\begin{proof}
 We first prove the left inequality.
Let $P=(w_0,w_1,...,w_k)$ be a shortest path in $G$ starting at $w_0=v$ and ending at $w_k=v_0$, thus $d_G(v_0,v)=k$. We want to prove that $k\geq \frac{m(v)}{7}$. 
By Lemma~\ref{lem:mm7}, for all $0\leq i\leq k-1$, we have $m(w_{i+1})\geq m(w_i)-7$. 
Thus we have $m(w_k)-m(w_0)=\sum_{i=0}^{k-1}(m(w_{i+1})-m(w_{i}))\geq -7k$. Moreover $m(w_k)=m(v_0)=0$ and $m(w_0)=m(v)$. This implies that $k\geq \frac{m(v)}{7}$.

We now proof the right inequality.
We define a walk $W=(w_i)_{i\geq 0}$ of $G$, starting at $v$ by the following. Let $w_0=v$ and assume that $w_i$ is defined for $i\geq 0$. If $w_i=v_0$, then the procedure stops. If $w_i$ is distinct from $v_0$, we consider an angle $\alpha$  incident to $w_i$ such that $\lambda(\alpha)=m(w_i)$. 
Let $\alpha'$ be the angle of the unique face of $\Gamma$, just after
$\alpha$ in clockwise order around this face.  If $\alpha$ and
$\alpha'$ are separated by a stem $s$, we set $w_{i+1}=v(a(s))$. If
$\alpha$ and $\alpha'$ are consecutive along an edge of $\Gamma$, we
set $w_{i+1}=v(\alpha')$.  In both cases, we prove that
$m(w_{i+1})\leq m(w_{i})-1$.  When $\alpha$ and $\alpha'$ are
separated by a stem $s$, then, by Lemma~\ref{lemma1}, we have
$m(w_{i+1})\leq \lambda(a(s))=\lambda(\alpha)-1=m(w_{i})-1$. When
$\alpha$ and $\alpha'$ are consecutive along an edge of $\Gamma$,
then, by the definition of the labeling function, we have
$m(w_{i+1})\leq \lambda(\alpha')=\lambda(\alpha)-1=m(w_{i})-1$.  So,
the sequence $(m(w_i))_{i\geq 0}$ is strictly decreasing along the walk $W$. By
Lemma~\ref{lemma2}, the function $m$ is $\geq 0$, and equal to zero
only for $v_0$. So the procedure ends on $v_0$. Let $k$ be the length
of $W$, we have $k\leq m(v)$.  So finally, we have
$d_G(v_0,v)\leq k \leq m(v)$.
\end{proof}

Recall that $A=(a_0,a_1,...,a_\ell)$ is the set of angles of $\Gamma$ and for $v\in V$, we have $A(v)$ is the set of angles incident to $v$. For $v\in V$, let $b(v)=\min\{i:a_i\in A(v)\}$.

For $v\in V$, we define the sequence $J(v)=({j(i)})_{i\geq 0}$ of
elements of $\mathbb N$ by the following.  Let ${j(0)}={b(v)}$ and
assume that ${j(i)}$ is defined for $i\geq 0$. If ${j(i)}=\ell$, then
the procedure stops.  If ${j(i)}\neq\ell$, then we define ${j(i+1)}$
by the following.  If the two consecutive angles $a_{j(i)}$ and
$a_{j(i)+1}$ of $A$ are separated by a stem $s$, then let $j(i+1)$ be
such that $a_{j(i+1)}=a(s)$. If $a_{j(i)}$ and $a_{j(i)+1}$ are
consecutive along an edge of $\Gamma$, then let ${j(i+1)}={j(i)+1}$.
Note that in both cases, by Lemma~\ref{lemma1} or the labeling rule,
we have $\lambda(a_{j(i+1)})=\lambda(a_{j(i)})-1$. So
$(\lambda(a_{j(i)}))_{i\geq 0}$ is decreasing by exactly one at each
step. Let $k=\lambda(a_{b(v)})$. Then for $i\geq 0$, we have
$\lambda(a_{j(i)})=k-i$.  Thus the procedure ends on $\ell$ after $k$
steps, i.e. $J(v)=({j(i)})_{0\leq i\leq k}$. Moreover we have that the
sequence $J(v)$ is strictly increasing since, as already remarked, by
the safe property, a stem $s$ is always attached to an angle with
greater index than the index of the angles incident to $s$.  We also
define the corresponding walk $W_J(v)=(v(a_{j(i)}))_{0\leq i\leq k}$
of $G$.

We have the following lemma:

 \begin{lemma}
 \label{lem:minindex}
 Consider $v\in V$ with $k=\lambda(a_{b(v)})$ and $J(v)=({j(i)})_{0\leq i\leq k}$. Then, $k>0$, and for $0\leq i\leq k$, we have $j(i)=\min\{z\geq b(v) :\lambda(a_z)=k-i\}$.
 \end{lemma}
  \begin{proof}
 First, suppose by contradiction that $k=0$. Then we have ${b(v)}=\ell$, so $v=v_0$ and thus ${b(v)}=0$. This contradicts $\ell = 4n+1$ and $n\geq 1$. So $k> 0$.
 
 Let $y$ be such that $0\leq y< k$. We claim that for all $z$ such that $j(y)\leq z < j(y+1)$, we have $\lambda(a_z)\geq k-y$.
 Recall that we have $\lambda(a_{j(y)})=k-y$ so the claim is true for $z=j(y)$.
 If the two consecutive angles $a_{j(y)}$ and $a_{j(y)+1}$ of $A$ are consecutive along an edge of $\Gamma$, then we are done since ${j(y+1)}={j(y)+1}$. Suppose now that $a_{j(y)}$ and $a_{j(y)+1}$ are separated by a stem $s$, then we have $a_{j(y+1)}=a(s)$.
 By Lemma~\ref{lemma1}, for $j(y)< z < j(y+1)$, we have $\lambda(a_z)\geq \lambda(a_{j(y)})=k-y$. This concludes the proof of the claim.
 
 Let $i$ be such that $0\leq i< k$. So, by the claim applied for $0\leq y \leq i$, we have the following: for  $b(v)\leq z < j(i+1)$, we have $\lambda(a_z)\geq k-i$. Since $\lambda(a_{j(i+1)})=k-i-1$, we have $j(i+1)=\min\{z\geq b(v) :\lambda(a_z)=k-(i+1)\}$. Moreover, we clearly have $j(0)=\min\{z\geq b(v) :\lambda(a_z)=k\}$.
  \end{proof}

We say that a vertex $v$ is the successor of a vertex $u$ if $b(u)\leq b(v)$ and denote this by $u \preceq v$.  Then for all $u,v \in V$, we define 

$$\overline{m}(u,v)=
\begin{cases}
\label{definitionofmcheck}
\min \{ \lambda(a_k):b(u)\leq k\leq b(v)\} & \text{ if } u \preceq v \\ 
\min \{ \lambda(a_k):b(v)\leq k\leq b(u)\} & \text{ if } v \preceq u \\ 
\end{cases}.$$

\begin{lemma}
\label{distancebetweentwovertices}
For all $u,v \in V$, we have
$d_G(u,v) \leq m(u)+m(v)-2\overline{m}(u,v)+14$.
\end{lemma}

\begin{proof}
By symmetry, we can assume that $u \preceq v$. If $u=v$,
then, by Lemma~\ref{lem:Mm6}, we have $\overline{m}(u,v)\leq m(u)+6$ and the lemma is clear since $d_G(u,v)=0$.
If $u$ is equal to $v_0$, then $\overline{m}(u,v)\leq \lambda (b(v_0))= \lambda (a_0)=3$ and the lemma is clear by Lemma~\ref{prop:bornesupdistance}. We now assume that $u$ is distinct from $v$ and $v_0$. Thus $v$ is also distinct from $v_0$ since $u \preceq v$. 
Then, by Lemma~\ref{lemma2}, we have $\overline{m}(u,v)>0$.

 Let $k=\lambda(b(u))$ and $k'=\lambda(b(v))$.
 Consider the two sequences $J(u)=({j(i)})_{0\leq i\leq k}$ and 
 $J(v)=({j'(i)})_{0\leq i\leq k'}$. By definition, we have $\overline{m}(u,v)\leq k$ and $\overline{m}(u,v)\leq k'$. Moreover we have $\overline{m}(u,v)>0$.
  Let $t>0$ and $t'>0$ be such that $k-t=k'-t'=\overline{m}(u,v)-1$.
By Lemma~\ref{lem:minindex}, we have 
 $j(t)=\min\{z\geq b(u) :\lambda(a_z)=k-t\}$
 and $j'(t')=\min\{z\geq b(v) :\lambda(a_z)=k'-t'\}$.
By definition of $\overline{m}(u,v)$, we have $j(t)> b(v)$ and so $j(t)=j'(t')$.
So the two walks $W_J(u)$ and $W_J(v)$ of $G$ are  reaching vertex $v(a_{j(t)})=v(a_{j'(t')})$ in respectively $t$ and $t'$ steps.
So  $d_G(u,v)\leq t+t'\leq k+k'-2\overline{m}(u,v)+2$.

By Lemma~\ref{lem:Mm6}, we have $k\leq m(u)+6$ and $k'\leq m(v)+6$. So finally we obtain 
$d_G(u,v)\leq m(u)+m(v)-2\overline{m}(u,v)+14$
\end{proof}

\section{Decomposition of unicellular maps}
\label{somebijection}

In this section we decompose the unicellular map considered in the
bijection given by Theorem~\ref{them:bijectionbenjamin} into simpler
objects, namely well-labelled forest and Motzkin paths.

\subsection{Forests and well-labelings}
\label{subsection:forestandwell-labelings}

We first introduce a formal definition of forest from~\cite{neveu1986arbres}. 
  
Let  $\mathbb{N}=\{0,1,2,...\}$ and $\mathbb{N^*}=\mathbb{N}\setminus \{0\}$. Let $\mathcal F$ be the set of all $n$-uplets of elements of $\mathbb{N^*}$ for $n\geq 1$, i.e.:

$$\mathcal F=\bigcup_{n=1}^{\infty}\mathbb{(N^*)} ^n,$$ 

For $n\geq 1$, if $u\in \mathbb{(N^*)}^n$, we write $|u|=n$. Let
$u=u_1\,u_2\,...\,u_n$ and $v=\,v_1\,v_2\,...\,v_p$ be two elements of
$\mathcal F$, then $u\, v=u_1\,u_2\,...\,u_n\,v_1\,v_2\,...\,v_p$ is
the \emph{concatenation} of $u$ and $v$. If $w=u\, v$ for some
$u,v \in \mathcal F$, we say $u$ is an \emph{ancestor} of $w$.
In the
particular case where $|v|=1$, we say that $u$ is the \emph{parent} of
$w$, denoted by $pa(w)$, and $w$ is a \emph{child} of $u$.

For $F\subseteq \mathcal F$ and $i\geq 1$, we denote
$F_i=\{u\in F\, :\, |u|=i\}$ and
$F_{\geq i}=\{u\in F\, :\, |u|\geq i\}$.

\begin{definition}
\label{def:forest}
A \emph{forest} is a non-empty finite subset $F$ of $\mathcal F$ satisfying the following (see example of Figure~\ref{fig:forest}):
\begin{enumerate}
\item There exists $t(F)\in \mathbb N$ such that 
$F_1=[\![1,t(F)+1]\!]$. 

\item If $u\in F_{\geq 2}$, then $pa(u) \in F$.

\item For all $u\in F$, there exists $c_u(F) \in \mathbb N$ such that:
  for all $i\in \mathbb{N^*}$, we have $u\, i \in F$ if and only if  $i\leq c_u(F)$.

\item $c_{t(F)+1}(F)=0$. 
\end{enumerate}
\end{definition}

Given a forest $F\in \mathbb F$.
The integer $t(F)$ of Definition~\ref{def:forest} is called the
\emph{number of trees} of $F$. The set $F_1$ is called the set of
\emph{floors of $F$}. For $n\geq 1$, if $u=u_1\,u_2\,...\,u_n$ is an element of $F$, then we denote $fl(u)=u_1$. Note that $fl(u)\in F$ by Definition~\ref{def:forest} (item 2.). So $fl(u)$ is a floor of the forest that we call the \emph{floor of $u$}.
The set of ancestor of $u$ in $F$ is denoted $A_u(F)$.
For $1\leq j\leq t(F)$, the $j$-th \emph{tree} of $F$, denoted by $F^j$,  
is the set of elements of $F$ that have floor $j$.
We say that $j$ is the \emph{floor of $F^j$}.  
For $\rho \in \mathbb N$ and $\tau\in\mathbb N^*$, the set of all forests $F$ with $\tau$ trees and $\rho+\tau+1$ elements is denoted by $\mathbb{F}^\rho_{\tau}$. 

A plane rooted tree is a connected acyclic graph represented in the plane that is rooted at a particular angle.
We represent a forest as a plane rooted tree by the following (see example of Figure~\ref{fig:forest}).
The set of vertices are the elements of $F$.
The set of oriented edges are the couples $(u,v)$, with $u,v$ in $F$, such that $pa(v)=u$, or there exists $i\in [\![ 1,t(F)]\!]$  such that $u=i$ and $v=i+1$. The tree is embedded in the plane such that it satisfies the following: 

\begin{itemize}
\item Around the vertex $1$ appear in \ccw order : the root angle, then, if $c_1(F)\geq 1$, the vertices $1\, 1$ to $1\, c_1(F)$, then vertex $2$.
\item Around a vertex $i\in [\![2,t(F)]\!]$  appear in \ccw order : the vertex $(i-1)$, then, if $c_i(F)\geq 1$, the vertices $i\, 1$ to $i\, c_i(F)$, then vertex $(i+1)$.
\item Around a vertex $u\in F_{\geq 2}$   appear in \ccw order : the vertex $pa(u)$, then, if $c_u(F)\geq 1$, the vertices $u\, 1$ to $u\, c_u(F)$.
\end{itemize}

One can recover the set of floors of $F$ from the plane rooted tree by
considering, as on figure~\ref{fig:forest}, the left most path
starting from the root angle.  A vertex which is not a floor, is
called a \emph{tree-vertex}. An edge between two floors is called
\emph{floor-edge}. An edge which is not a floor-edge is called
\emph{tree-edge}

Note that there is indeed a bijection between  $\mathbb{F}^\rho_{\tau}$, and, plane rooted trees with $\tau+1$ floors and $\rho$ tree-vertices.

\begin{figure}[h]
  \centering
  \scalebox{0.5}{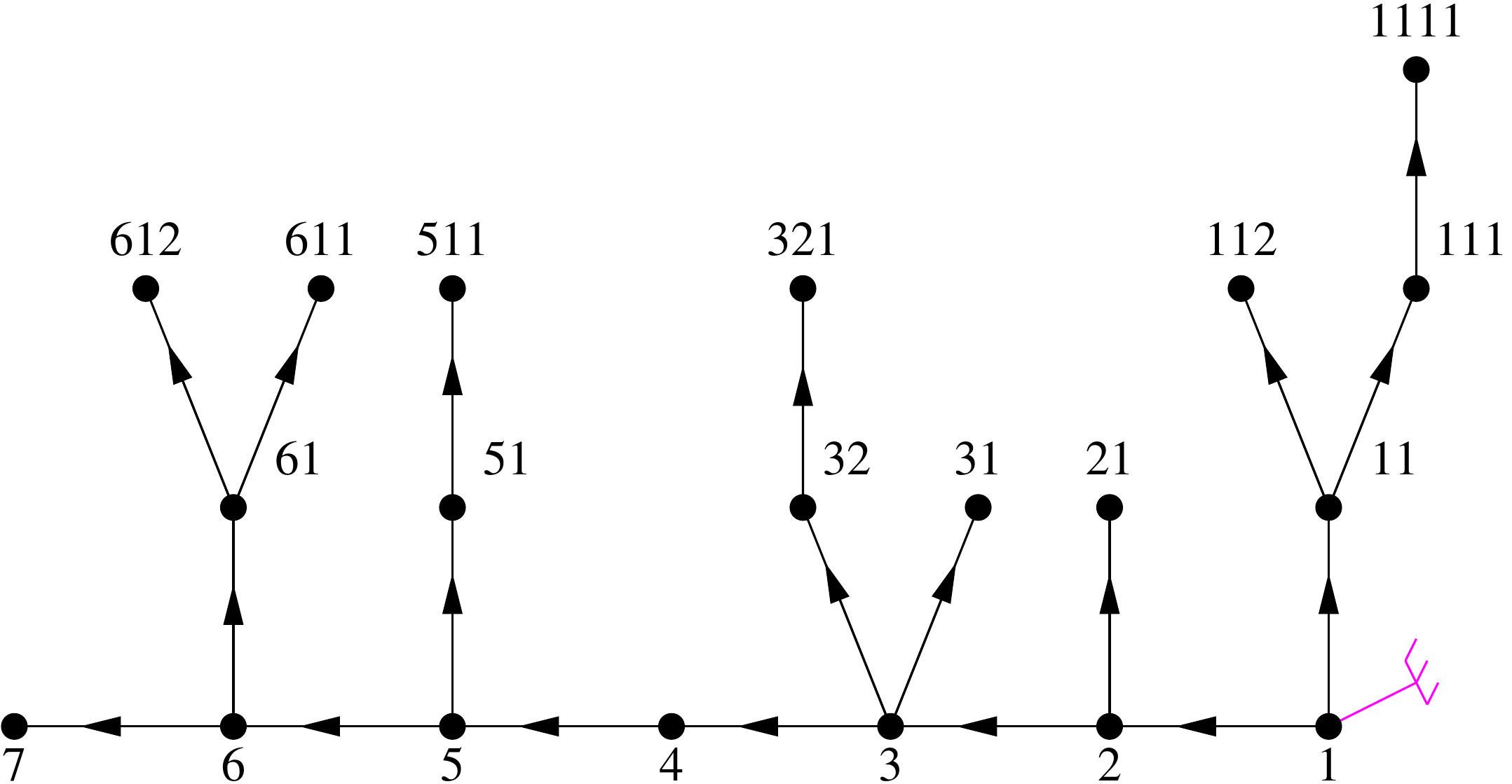}
 
  \ \\
  
  $F=\{1,11,111,1111,112,2,21,3,31,32,321,4,5,51,511,6,61,611,612,7\}$
  \caption{Representation of a forest
    of $\mathbb{F}^{13}_6$.}
    \label{fig:forest}
  \end{figure}

We now equip the considered forest with a label function on the set of vertices.


\begin{definition}
  A \emph{well-labeled forest} is a pair $(F,\ell)$, where $F$ is a forest
   and $\ell: F \rightarrow \mathbb{Z}$ is such that
  $\ell$ satisfies the following conditions (see example of
  Figure~\ref{wellforest}):

  \begin{enumerate}
    
  \item For all $u\in F_1$, we have $\ell(u)=0$
\item For all $u\in F_2$, we have $\ell(u)=-1$, 
 \item For all $u\in F_{\geq 2}$ and $c_u(F)\geq 1$, we have $\ell(u)-1\leq \ell(u\, 1) \leq \ell(u\, 2)\leq \cdots \leq \ell(u\, c_u(F))\leq \ell(u)+1$.
  
\end{enumerate}

\end{definition}

The set of all well-labeled forests $(F,\ell)$ such that $F \in \mathbb{F}^\rho_{\tau}$  is denoted by $\mathcal F^\rho_{\tau}$.

\begin{figure}[h]
    \centering
 \scalebox{0.5}{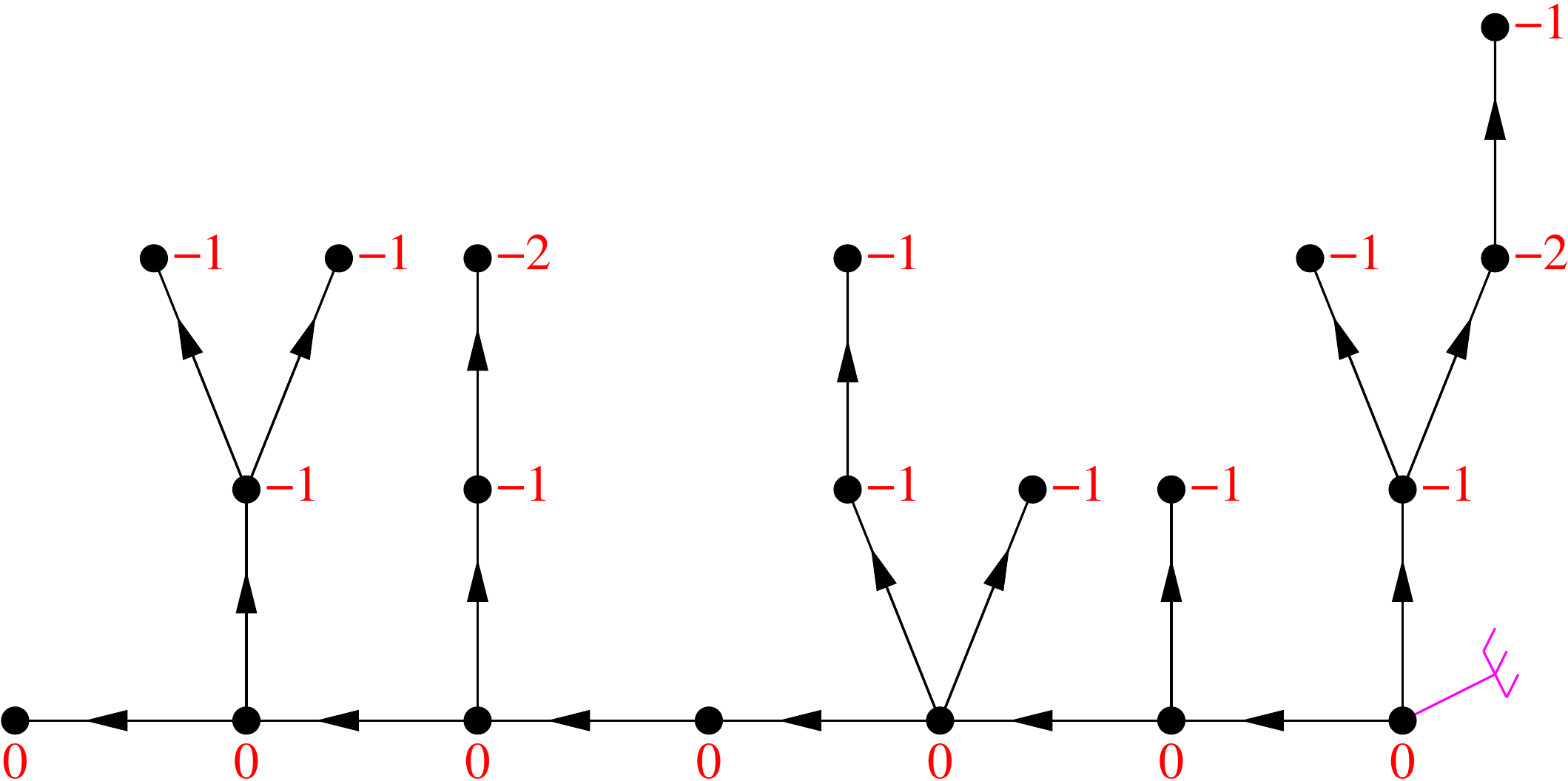}
  \caption{Example of a well-labeled forest of $\mathcal F^{13}_6$.}
    \label{wellforest}
  \end{figure}

  The function $\ell$ of a well-labeled forest $(F,\ell)$ can be
  represented on the plane rooted tree representing $F$ by adding two
  stems incident to each tree-vertex of $F$ (see
  figure~\ref{fig:welllabelstem}). A variation into the value $\ell$
  of two consecutive children of vertex $u$ indicates the presence of one (or two)
  stems incident to $u$ in the corresponding angle (assuming that we add two virtual children, one on the right
  having label $\ell(u)-1$ and one on the left having  label $\ell(u)+1$.

Note that there is a bijection between $\mathcal F^\rho_{\tau}$, and, plane rooted tree with $\tau+1$ floors and $\rho$ tree-vertices each being incident to two additional stems.

\begin{figure}[h]
  \centering
 \scalebox{0.5}{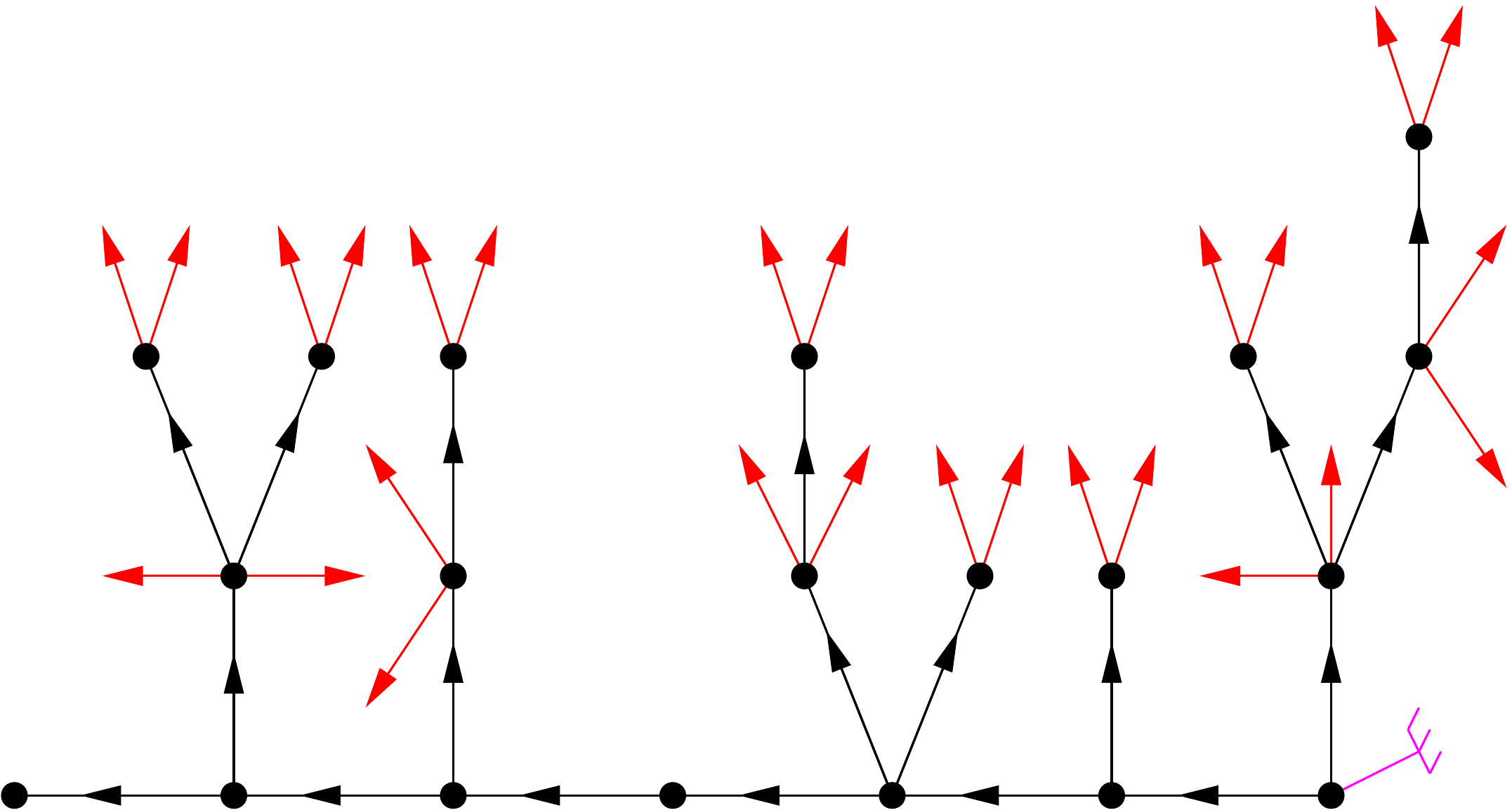}
  \caption{Representation of the well-labeled forest of Figure~\ref{wellforest} by a  plane rooted tree with two additional stems incident to each tree-vertex.}
    \label{fig:welllabelstem}
  \end{figure}

We now encode forests and well-labeled forest similarly as in \cite{bettinelli2010scaling}. To do this, we need to define the contour and labeling functions.

Consider a forest $F$ of $\mathbb{F}^\rho_{\tau}$. 

We  define  the \emph{vertex contour  function}  of  $F$ as  the  function
$r_F:[\![0,2\rho+\tau]\!] \rightarrow  F$, such  that $r_F(0)=1$  and for
$0\leq i< 2\rho+\tau$, we have the following:
\begin{itemize}
\item If $r_F(i)$ have children which do not belong to the set $\{r_F(0),\ldots,r_F(i-1) \}$, then $r_F(i+1)=r_F(i)\, j$ where $j=\min\{k\in \mathbb N^*: r_F(i)\, k \notin \{r_F(0),\ldots,r_F(i-1) \} \}$. 
\item If all children of $r_F(i)$ belong to $\{r_F(0),\ldots,r_F(i-1) \}$ 
then, $r_F(i+1)=pa(r_F(i))$ if $|r_F(i)|\geq 2$, and,  $r_F(i+1)=r_F(i)+1$ otherwise 
\end{itemize}

Note that $r_F(2\rho+\tau)= \tau +1$ by a simple counting argument. 

Unformaly, the vertex contour function of a forest corresponds to a \ccw
walk around its representation, starting from the root angle.  For
the example of Figure~\ref{fig:forest}, one obtain the following
vertex contour function:

\begin{equation*}
\begin{split}
r_F([\![0,2\rho+\tau]\!])= ( & 
1,11,111,1111,111,11,112,11,1,2,21,2,3,31,3,32,321, 32,3,\\ & 4,5,51,511,51,5,6,61,611,61,612,61,6,7)
\end{split}
\end{equation*}

We now define the \emph{contour function} of $F$ as the function
$C_F:[0,2\rho+\tau] \rightarrow \mathbb{R}$ such that for
$i\in[\![0,2\rho+\tau]\!]$ $$C_F(i)=fl(r_F(i))-|r_F(i)|.$$ Note that
$C_F(0)=0$ and $C_F(2\rho+\tau)=\tau$.

For example, the contour function of the forest of Figure~\ref{fig:forest} is:
\begin{equation*}
\begin{split}
C_F([\![0,2\rho+\tau]\!])= ( & 
0,\text{-} 1,\text{-} 2,\text{-} 3,\text{-} 2,\text{-} 1,\text{-} 2,\text{-} 1,0,1,0,1,2,1,2,1,0,1,2,\\ & 3,4,3,2,3,4,5,4,3,4,3,4,5,6)
\end{split}
\end{equation*}

Note that one can recover a forest $F$ from its contour function $C_F$.

Now consider $(F,\ell)$ a well-labeled forest with $F\in \mathbb
F^\rho_{\tau}$.

We defined the \emph{labeling function} of $(F,\ell)$ as the function
$L_{(F,\ell)}:[0,2\rho+\tau]\rightarrow \mathbb{R}$ such that for  $i \in [\![0,2\rho+\tau]\!]$ by $$L_{(F,\ell)}(i)=\ell(r_F(i)).$$

For example, the labeling function of the well-labeled forest of
Figure~\ref{wellforest} is: 
\begin{equation*}
\begin{split}
L_F([\![0,2\rho+\tau]\!])= ( & 
0,\text{-}1,\text{-}2,\text{-}1,\text{-}2,\text{-}1,\text{-}1,\text{-}1,0,0,\text{-}1,0,0,\text{-}1,0,\text{-}1,\text{-}1,\text{-}1,0, \\
& 0,0,\text{-}1,\text{-}2,\text{-}1,0,0,\text{-}1,\text{-}1,\text{-}1,\text{-}1,\text{-}1,0,0)
\end{split}
\end{equation*}

Note that one can recover $(F,\ell)$ from
 the pair $(C_F,L_{(F,\ell)})$. This pair is called the \emph{contour pair} of $(F,\ell)$.

\subsection{Relation between well-labeled forests and $3$-dominating binary words}
\label{section:Relation between well-labeled forests and $3$-dominating binary words}
In this section, we show how to compute the value of
$|\mathcal F^{\rho}_{\tau}|$ for $\rho \in \mathbb N$ and
$\tau\in\mathbb N^*$.

Consider $b\in \{0,1\}^p$. If $b=b_1\, \ldots\, b_p$, then we define the \emph{inverse of $b$} by $b^{-1}=b_p\, \ldots\, b_1$. For $x\in \{0,1\}$, we denote $|b|_x=|\{1\leq i \leq p : b_i=x\}|$.
  We say that $b$ is \emph{$k$-dominating}, for $k>0$, if for  $1\leq i\leq p$,  we have
 $|b_1\ldots b_i|_0> k\, |b_1\ldots b_i|_1$. 
  For example, the sequence  $01001$ is not $1$-dominating and the sequence
 $000011001$ is $1$-dominating but not $2$-dominating. 
 We have the following lemma from~\cite{dershowitz1990cycle}:

\begin{lemma}[\cite{dershowitz1990cycle}]
\label{cyclelemma}
Consider $b\in \{0,1\}^{p+q}$ with $|b|_0=p$ and $|b|_1=q$. For $k\in \mathbb N^*$, if $p \geq k\, q$, then there exist exactly $p-k\,q$ elements of
$\{b_j\,b_{j+1}\,...\, b_{p+q}\, b_1\,b_2\,...\,b_{j-1}
: 1\leq j \leq p+q\}$
that are $k$-dominating. 
\end{lemma}

The set of elements $b\in \{0,1\}^{p+q}$ with $|b|_0=p$ and $|b|_1=q$
that are $3$-dominating is denoted $\mathcal D_{3,p,q}$.  The elements
whose inverse is in $\mathcal D_{3,p,q}$ are called \emph{inverse
  $3$-dominating binary words} and their set is denoted
$\mathcal D_{3,p,q}^{-1}$

\begin{lemma}
\label{bijection3dominating}
There is a bijection between $\mathcal F^\rho_{\tau}$ and $\mathcal D_{3,3\rho+\tau,\rho}^{-1}$.
\end{lemma}

\begin{proof}
As already mentioned $\mathcal F^\rho_{\tau}$ is in bijection with 
plane rooted trees with $\tau$ floors and $\rho$ tree-vertices each being incident to two stems.

Similarly as in~\cite{poulalhon2006optimal}, we encode these plane
rooted trees by the following method.  Let $\alpha$ be the (unique)
angle of the last vertex of the {left} most path from the root angle.
We walk around the tree starting from the root angle in \ccw order,
and ending at $\alpha$. We write a $"1"$ when going along an outgoing
tree-edge, and a $"0"$ when going along an ingoing tree-edge, or
around a stem of $F$, or along an outgoing floor-edge (see
Figure~\ref{fig:codewellforest}).  By doing so, we obtain an element
$b$ of $\{0,1\}^{4\rho+\tau}$ with $|b|_1=\rho$ such that $b$ is the inverse
of a $3$-dominating word. Indeed, while walking around the tree in
reverse order, i.e. starting from $\alpha$, walking in \cw order
around the tree and ending at the root angle, we go along an outgoing
tree-edge $e$, and the two stems incident to its terminal vertex
before going along this tree-edge $e$ in the other direction. Thus we
have seen three $"0"$ before the $"1"$ corresponding to edge
$e$. Moreover, this walk starts by going along an ingoing floor-edge,
therefore we start with an additional $"0"$. Thus $b^{-1}$ is
$3$-dominating  so $b\in\mathcal D_{3,3\rho+\tau,\rho}^{-1}$.  As
in~\cite{poulalhon2006optimal}, one can see that the rooted plane tree
can be recovered from $b$.  Moreover, it is easy to see that any
$b\in \mathcal D_{3,3\rho+\tau,\rho}^{-1}$ corresponds to such a tree. So there
is a bijection between $\mathcal F^\rho_{\tau}$ and
$\mathcal D_{3,3\rho+\tau,\rho}^{-1}$.
\end{proof}

\begin{figure}[h]
  \centering
 \scalebox{0.5}{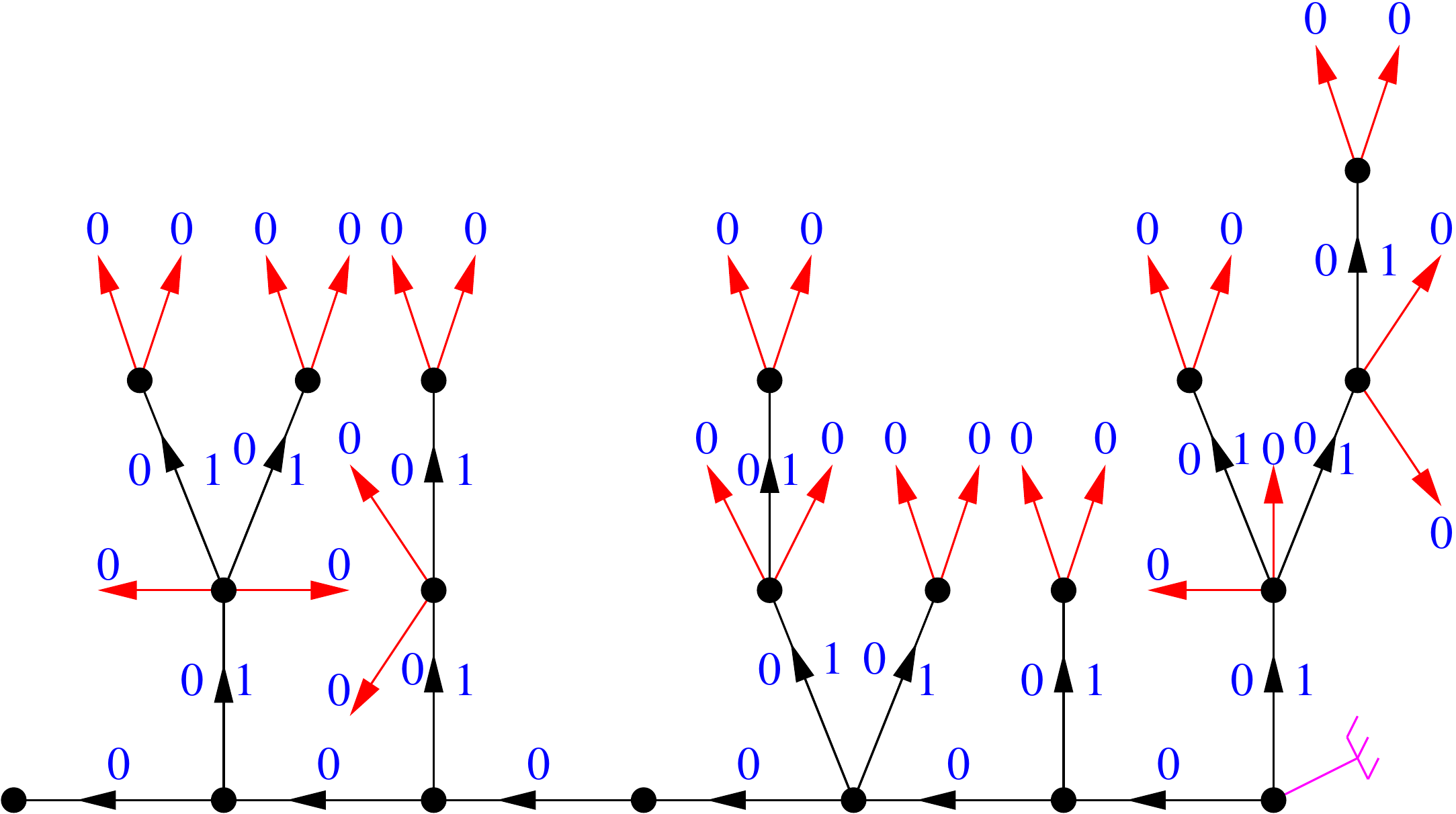}
  
  \ \\
  
  $\color{blue}b=1100100000100000010000100010100000001100000001010001000000$
  \caption{Encoding a forest with two stem at each tree-vertex.}
    \label{fig:codewellforest}
  \end{figure}

\begin{lemma}
\label{mnforest}
For $\rho\in \mathbb N$ and $\tau \in \mathbb N^*$, we have:
$$ |\mathcal F^\rho_\tau|=
\frac{\tau}{4\rho+\tau}\begin{pmatrix}
4\rho+\tau\\ 
\rho
\end{pmatrix}.$$
\end{lemma}

\begin{proof}
By Lemma~\ref{bijection3dominating}, it is suffices to prove that
$$|\mathcal D_{3,3\rho+\tau,\rho}|=
\frac{\tau}{4\rho+\tau}\begin{pmatrix}
4\rho+\tau\\ 
\rho
\end{pmatrix}. $$

The number of elements $b\in \{0,1\}^{4\rho+\tau}$ with $|b|_0=3\rho+\tau$ and $|b|_1=\rho$ is $\begin{pmatrix}
4\rho+\tau\\ 
\rho
\end{pmatrix}$.

By Lemma~\ref{cyclelemma}, for each such element $b$, there are exactly
$3\rho+\tau-3\rho=\tau$ elements
of
$\{b_j\,b_{j+1}\,...\, b_{4\rho+\tau}\, b_1\,b_2\,...\,b_{j-1}
: 1\leq j \leq 4\rho+\tau\}$
that are $3$-dominating. 
Thus we obtain the result.
\end{proof}

\subsection{Motzkin paths}
\label{sec:motz}

A \emph{Motzkin path} of length $\sigma\in \mathbb N$, from $0$ to
$\gamma \in \mathbb{Z}$, with $|\gamma|\leq\sigma$, is a sequence of integers
$M=(M_i)_{0\leq i\leq \sigma}$, such that $M_0=0$, $M_{\sigma}=\gamma$, and
for all $0\leq i\leq \sigma-1$, we have $M_{i+1}-M_i\in \{-1,0,1\}$.
The set of Motzkin path of length $\sigma$ from $0$ to $\gamma$ is denoted
$\mathcal M_{\sigma}^{\gamma}$.

An example of a Motzkin path in $\mathcal M_5^{-2}$ is the
following:

\begin{equation}
\label{eq:MotzkinExample}
  M=(0,1,0,0,-1,-2)
\end{equation}

Consider $M\in \mathcal M_\sigma^\gamma$.

We define the \emph{extension of
  $M$} as a sequence of integers denoted
$\widetilde M=(\widetilde M_i)_{0\leq i\leq 2\sigma+\gamma}$ and
defined by the following.  We obtain $\widetilde M$ from
$M=(M_0,\ldots, M_\sigma)$ by considering consecutive values
$M_i, M_{i+1}$, for $0\leq i < \sigma$. When $M_{i+1}=M_i$ we add the
value $(M_i+1)$ between $M_i$ and $M_{i+1}$ in the sequence of
$\widetilde M$.  When $M_{i+1}=M_i+1$ we add the two values
$(M_i+1),(M_i+2)$ between $M_i$ and $M_{i+1}$ in the sequence of
$\widetilde M$.  When $M_{i+1}=M_i-1$ we add nothing between $M_i$ and
$M_{i+1}$ in the sequence of $\widetilde M$. So at each step $i$, the
number of values that are added to obtain  $\widetilde M$ is exactly 
$M_{i+1}-M_i+1$.
Note that the extension
of an element of $\mathcal M_\sigma^\gamma$ is an element of $\mathcal M_{2\sigma+\gamma}^\gamma$.

With this definition, the extension of the example of Motzkin
path $M$ given by~(\ref{eq:MotzkinExample}) is the following element
of $\mathcal M_8^{-2}$ (where added
values from $M$ are
represented in red):

\begin{equation}
\label{eq:MotzkinExtensionExample}
  \widetilde M=(0,\comment{1},\comment{2},1,0,\comment{1},0,-1,-2)
\end{equation}

We also define the \emph{inverse of
  $M$} as a sequence of integers denoted
$\underline M=((\underline M)_i)_{0\leq i\leq
  \sigma}$ and equal to
$(M_\sigma-\gamma,M_{\sigma-1}-\gamma,\ldots,M_0-
\gamma)$.
Thus informally, $\underline M$ is the Motzkin path obtained by "reading" the
variation of $M$ in reverse order.
Note that the
inverse of an element of $\mathcal
M_\sigma^\gamma$ is an element of $\mathcal M_{\sigma}^{-\gamma}$.

With this definition, the inverse of the example of Motzkin path $M$
given by~(\ref{eq:MotzkinExample}) is the following element of
$\mathcal M_5^{2}$:

\begin{equation}
\label{eq:MotzkinExampleInverse}
  \underline M=(0,1,2,2,3,2)
\end{equation}

Then one can consider the \emph{extension of the inverse} of $M$, that is
defined by the composition of the inverse then the extension of a
Motzkin path. It is  thus denoted by $\widetilde{(\underline M)}$ or
$\widetilde{\underline M}$ for simplicity.
Note that the extension of the
inverse of an element of $\mathcal
M_\sigma^\gamma$ is an element of $\mathcal M_{2\sigma-\gamma}^{-\gamma}$.

The extension of the inverse of the example of Motzkin path $M$ given
by~(\ref{eq:MotzkinExample}) is thus the extension of the Motzkin path
$\underline M$ given by~(\ref{eq:MotzkinExampleInverse}), and thus the
following element of $\mathcal M_{12}^{2}$ (where added values from
$\underline M$ are  represented in red):

\begin{equation}
\label{eq:MotzkinExampleInverseExtension}
  \widetilde{\underline M}=(0,\comment{1},\comment{2},1,\comment{2},\comment{3},2,\comment{3},2,\comment{3},\comment{4},3,2)
\end{equation}

\subsection{Decomposition of unicellular maps into well-labeled forests and Motzkin paths}
\label{dec}

Consider $n\geq 1$, and $U$ an element of $\mathcal U(n)$ (or
$\mathcal T_r(n)$). As in Section~\ref{sec:label}, we call proper the
set of vertices of $U$ that are on at least one cycle of $U$.  The
\emph{core} $C$ of $U$ is obtained from $U$ by deleting all the
vertices that are not proper (and keeping all the stems attached to
proper vertices). In $C$, or $U$, we call \emph{maximal chain} a path $P$ whose
extremities are special vertices and all inner  vertices of
$P$ are not special. Then the \emph{kernel} $K$ of $U$ is obtained
from $C$ by replacing every maximal chain $P$ by an edge (and thus
removing the inner vertices and the stems incident to them). Note that
we keep the stems incident to special vertices in the kernel.

Let $\mathcal U_r(n)$ be the set of elements $U$ of $\mathcal U(n)$
that are rooted at a half-edge of the kernel that is not a stem.  Note
that if $U\in \mathcal U(n)$ is hexagonal there is $6$ such
half-edges, and if $U$ is square there is $4$ such half-edges. Let
$\mathcal U_{r,b}(n)$ be the set of elements of $\mathcal U_r(n)$ that
are balanced.  Finally, let $\mathcal U^H_{r,b}(n)$,
$\mathcal U^S_{r,b}(n)$, $\mathcal T^H_{r,s,b}(n)$ and
$\mathcal T^S_{r,s,b}(n)$ be the elements of $\mathcal U_{r,b}(n)$ and
$\mathcal T_{r,s,b}(n)$ that are respectively hexagonal and square.

Next lemma enables to avoid the safe property while studying $\mathcal T_{r,s,b}(n)$.

\begin{lemma}
\label{prop:bijrootuni}
There is a bijection between $ [\![1,3]\!]\times  \mathcal T_{r,s,b}(n) $ and $$([\![1,3]\!] \times \mathcal U^S_{r,b}(n) )\,\bigcup\,([\![1,2]\!] \times\mathcal U^H_{r,b}(n) ) .$$
\end{lemma}

\begin{proof}
Let $Z(n)$ be the set of elements of $\mathcal T_{r,s,b}(n)$ that are moreover rooted at a half-edge of the kernel that is not a stem. 
Let $Z^H(n)$ (resp. $Z^S(n)$) be the set of elements of $Z$ that are hexagonal (resp. square).
Given an element of $\mathcal T^H_{r,s,b}(n)$,
there are $6$ possible such roots.
So there is a bijection between $Z^H(n)$
and $ [\![1,6]\!]\times  \mathcal T^H_{r,s,b}(n)$.
Given an element of 
 $\mathcal T^S_{r,s,b}(n)$,
there are $4$ possible roots.
So there is a bijection between $Z^S(n)$
and $[\![1,4]\!]\times  \mathcal T^S_{r,s,b}(n)$.

Given an element $U$ of $\mathcal U_{r,b}(n)$, there are four angles
where a root stem can be added to obtain an element of $Z(n)$. Indeed,
these four angles corresponds to the four angles remaining in the
special face when the complete closure procedure is applied on $U$.
So there is a bijection between $Z^S(n)$ and
$[\![1,4]\!]\times \mathcal U^S_{r,b}(n)$ and a bijection between
$Z^H(n)$ and $[\![1,4]\!]\times \mathcal U^H_{r,b}(n)$.  Finally
$ \mathcal T_{r,s,b}(n) = \mathcal T^S_{r,s,b}(n) \cup \mathcal
T^H_{r,s,b}(n)$ and we obtain the result.
\end{proof}

Let $n\geq 1$.  There are different possible kernels for element of
$\mathcal U_{r}(n)$, depending on the position of the possible
stems. All the possible kernels of elements of $\mathcal U_{r}(n)$ are
depicted on Figure~\ref{fig:kernels} where the root half-edge of the
kernel is depicted in pink.  There are exactly $10$ such possibilities
and, for $0\leq k \leq 9$, we say that an element of
$\mathcal U_{r}(n)$ is of type $k$ if its kernel corresponds to type
$k$ of Figure~\ref{fig:kernels}.  We decompose the elements
of $\mathcal U_{r,b}(n)$ depending on their types.

\begin{figure}[h!]
\center
\begin{tabular}{ccc}
  &
    \scalebox{0.3}{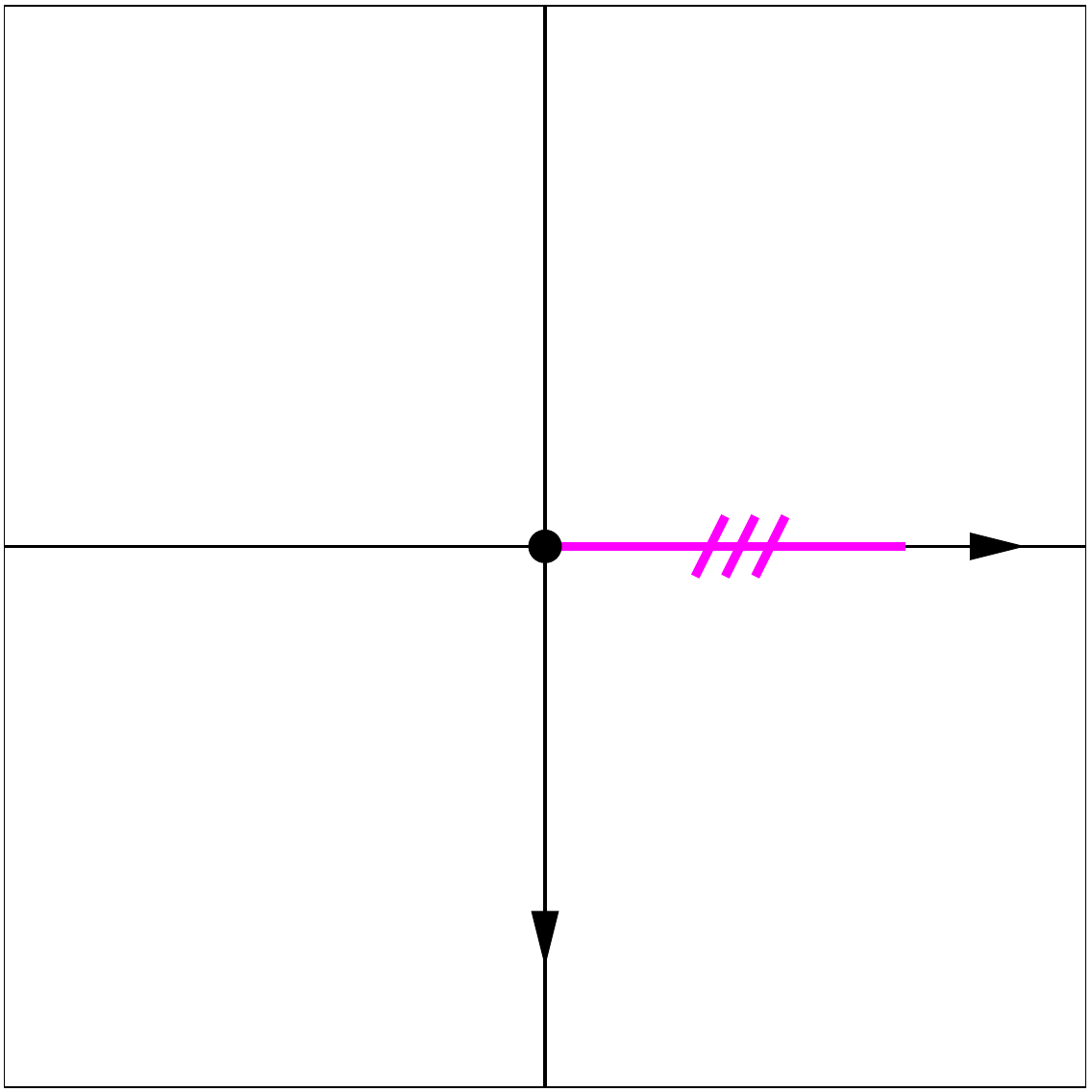}
  \\ &Type $0$
 \\
& \\

  \scalebox{0.3}{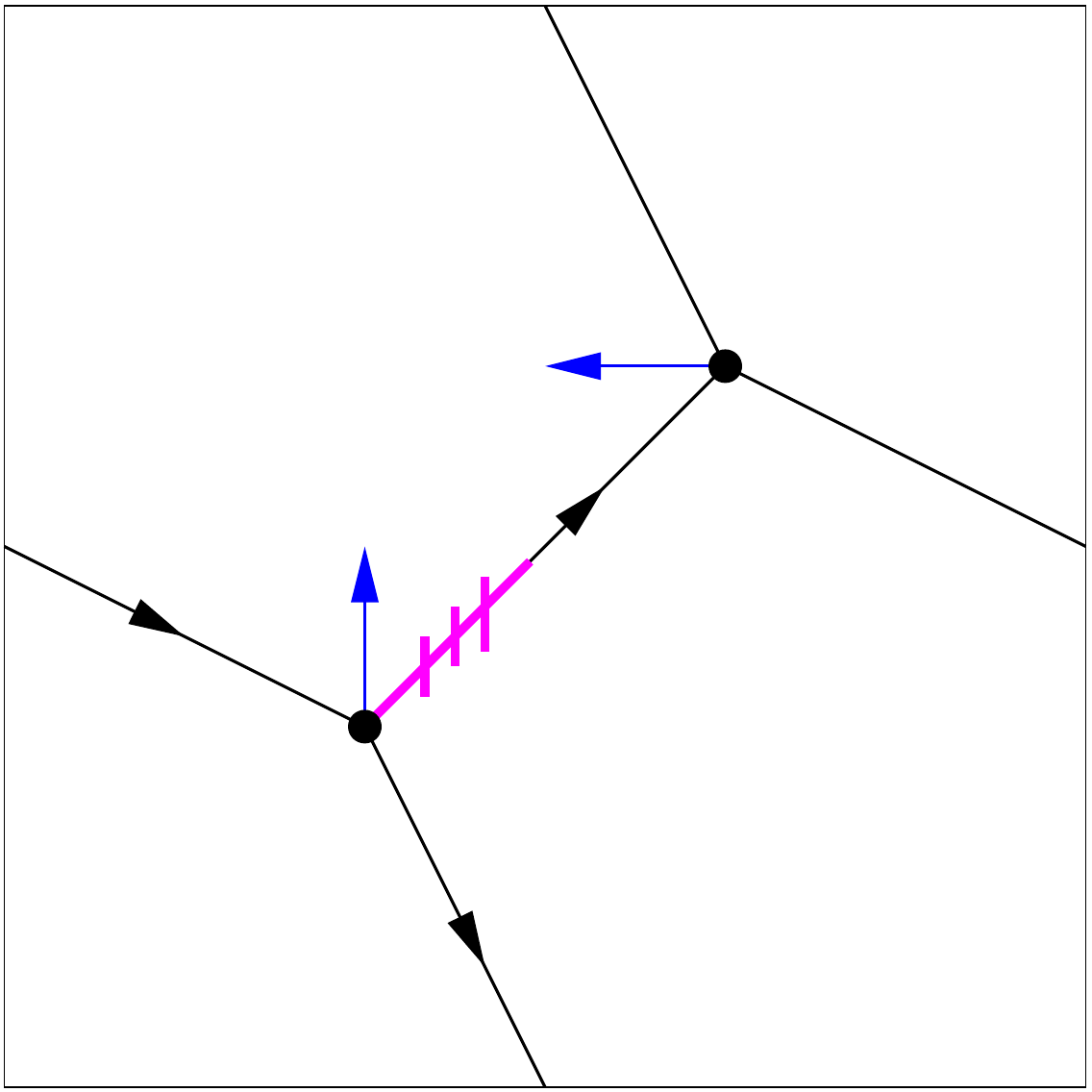} &
  \scalebox{0.3}{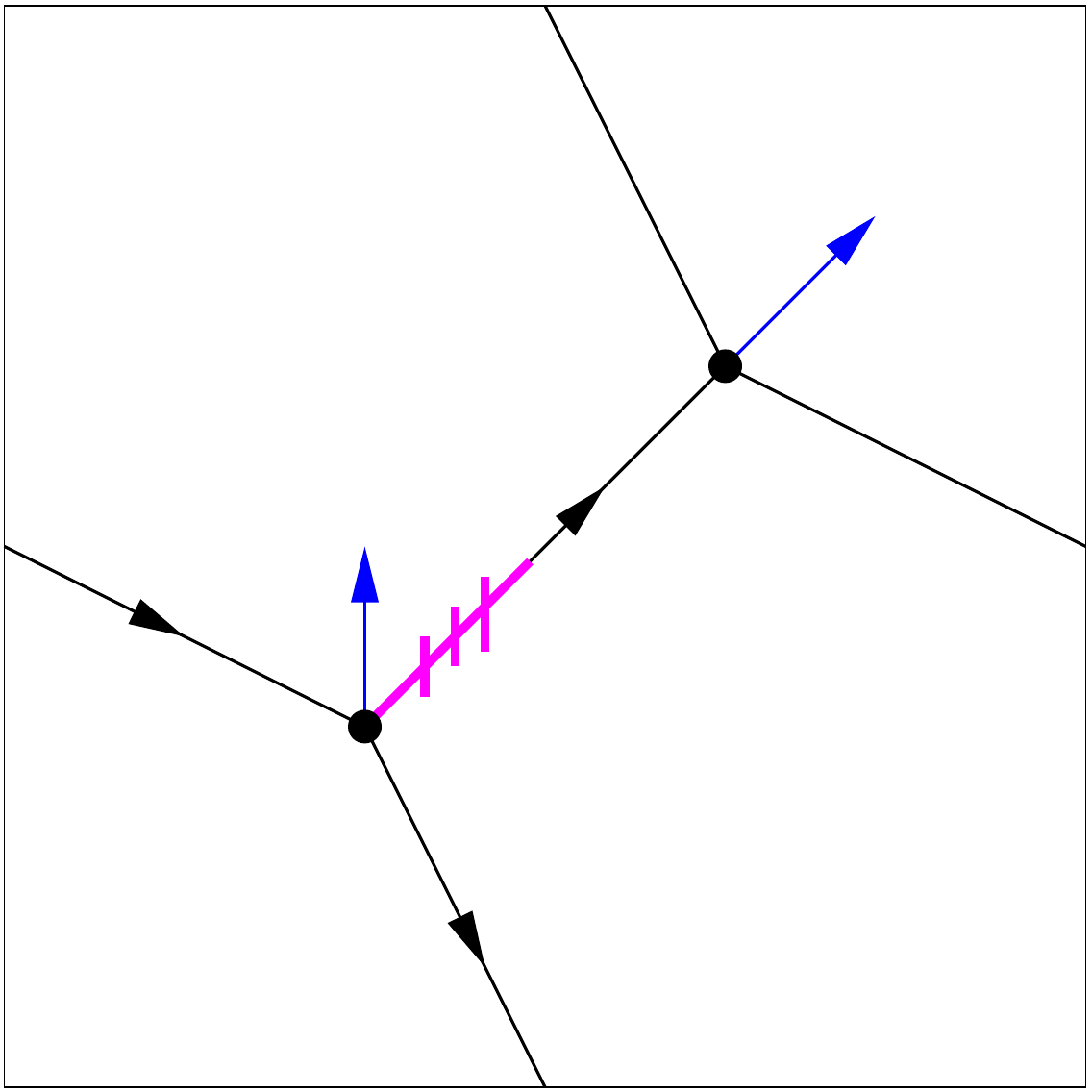} &
  \scalebox{0.3}{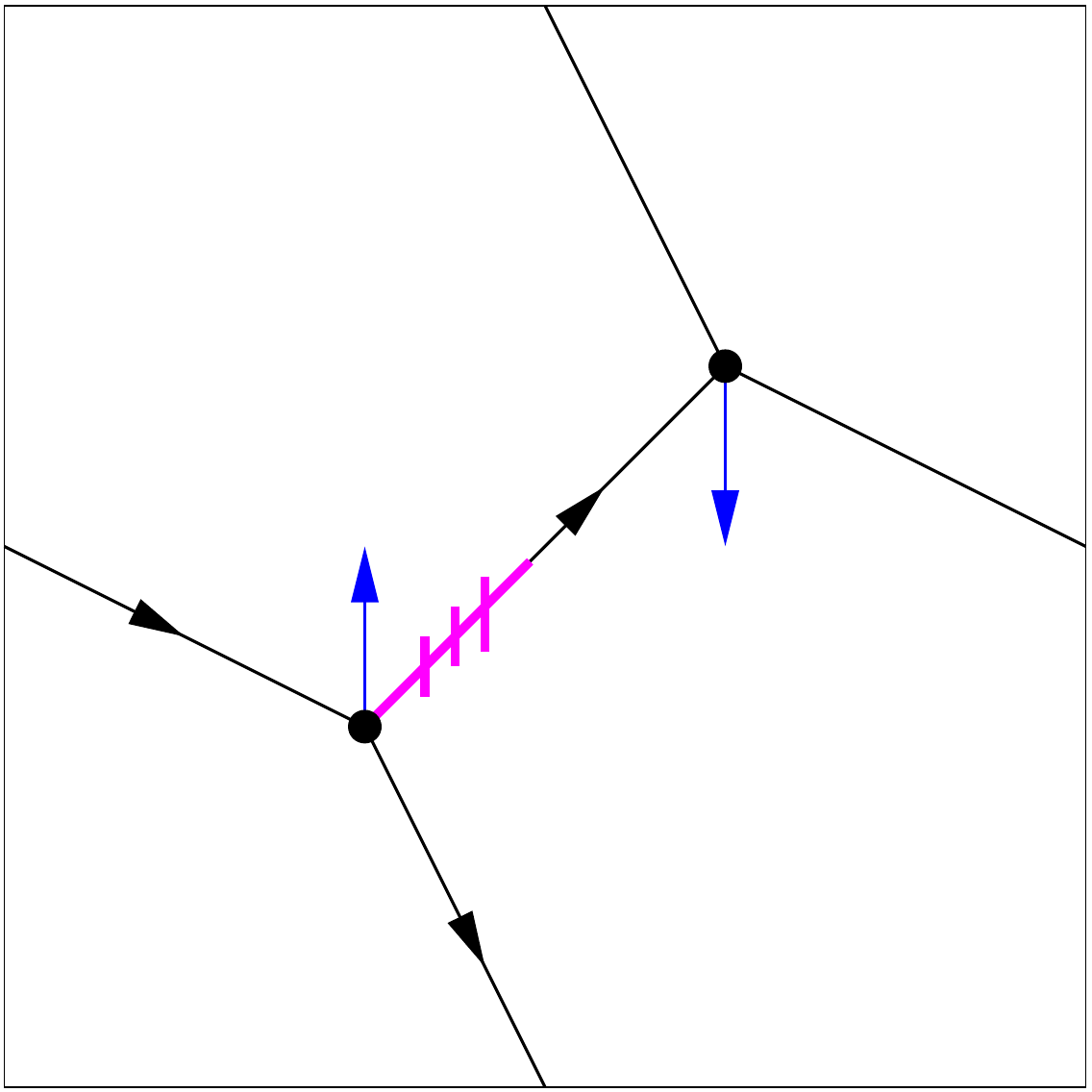}
  \\
  Type $1$ & Type $2$ & Type $3$ \\
  & \\

  \scalebox{0.3}{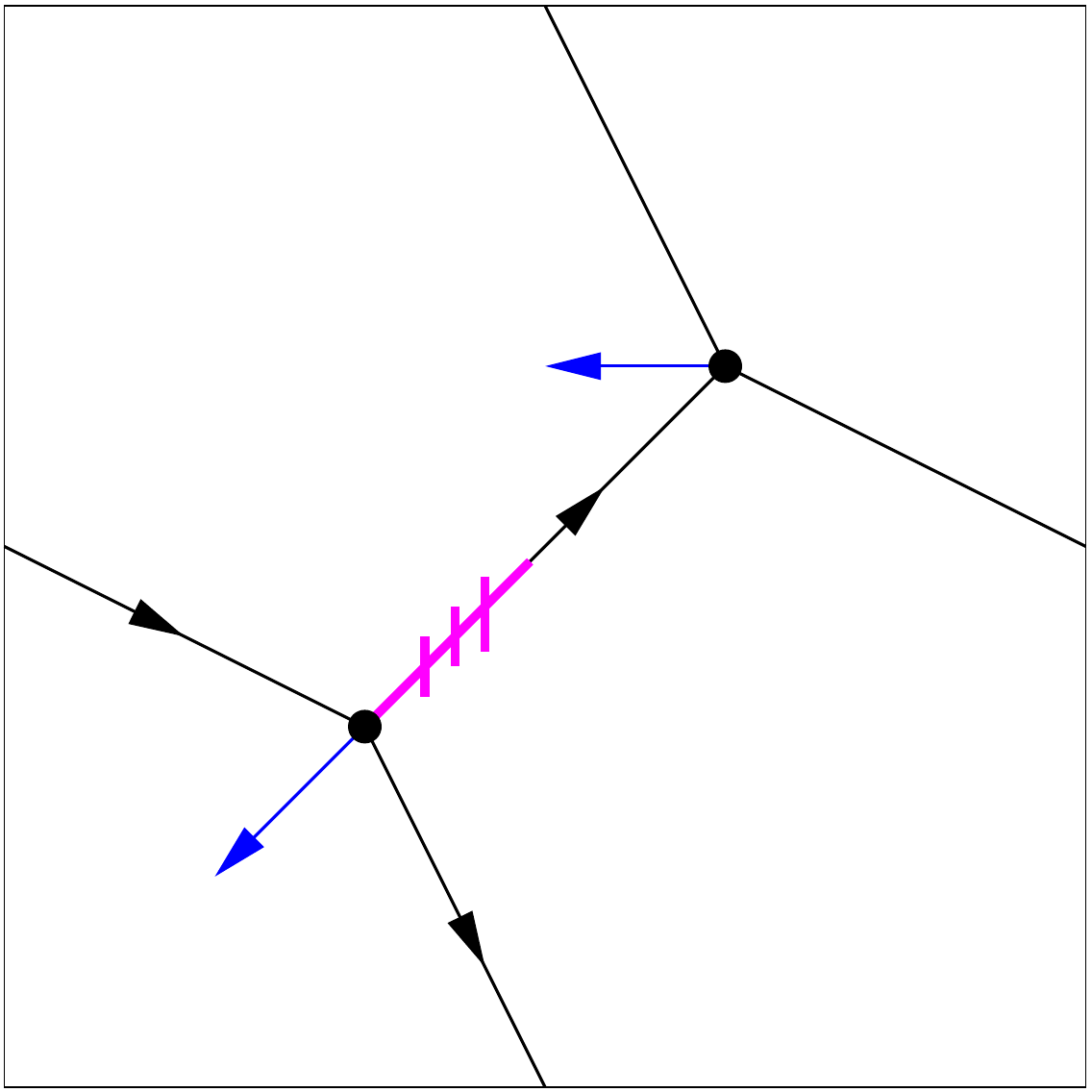} &
  \scalebox{0.3}{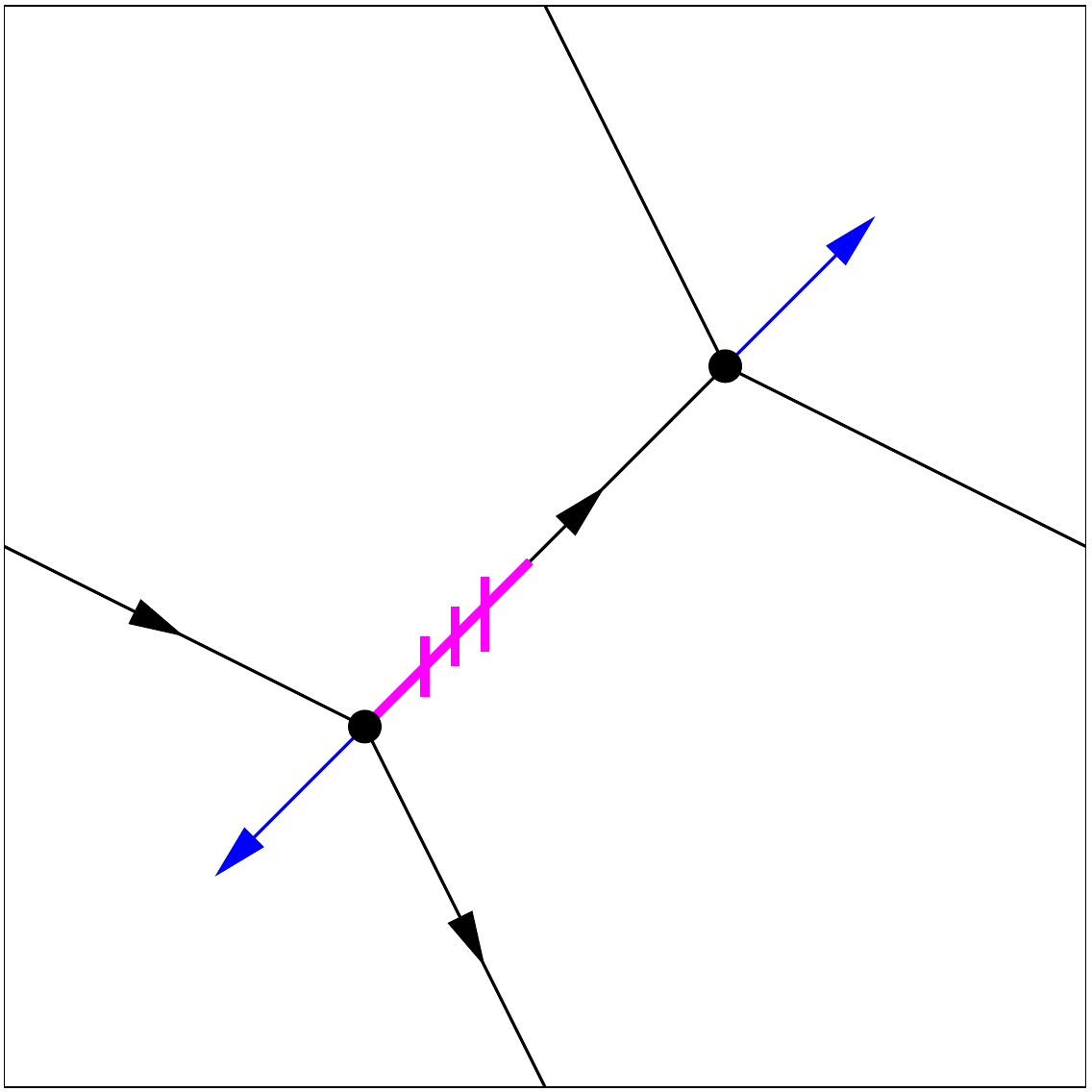} &
  \scalebox{0.3}{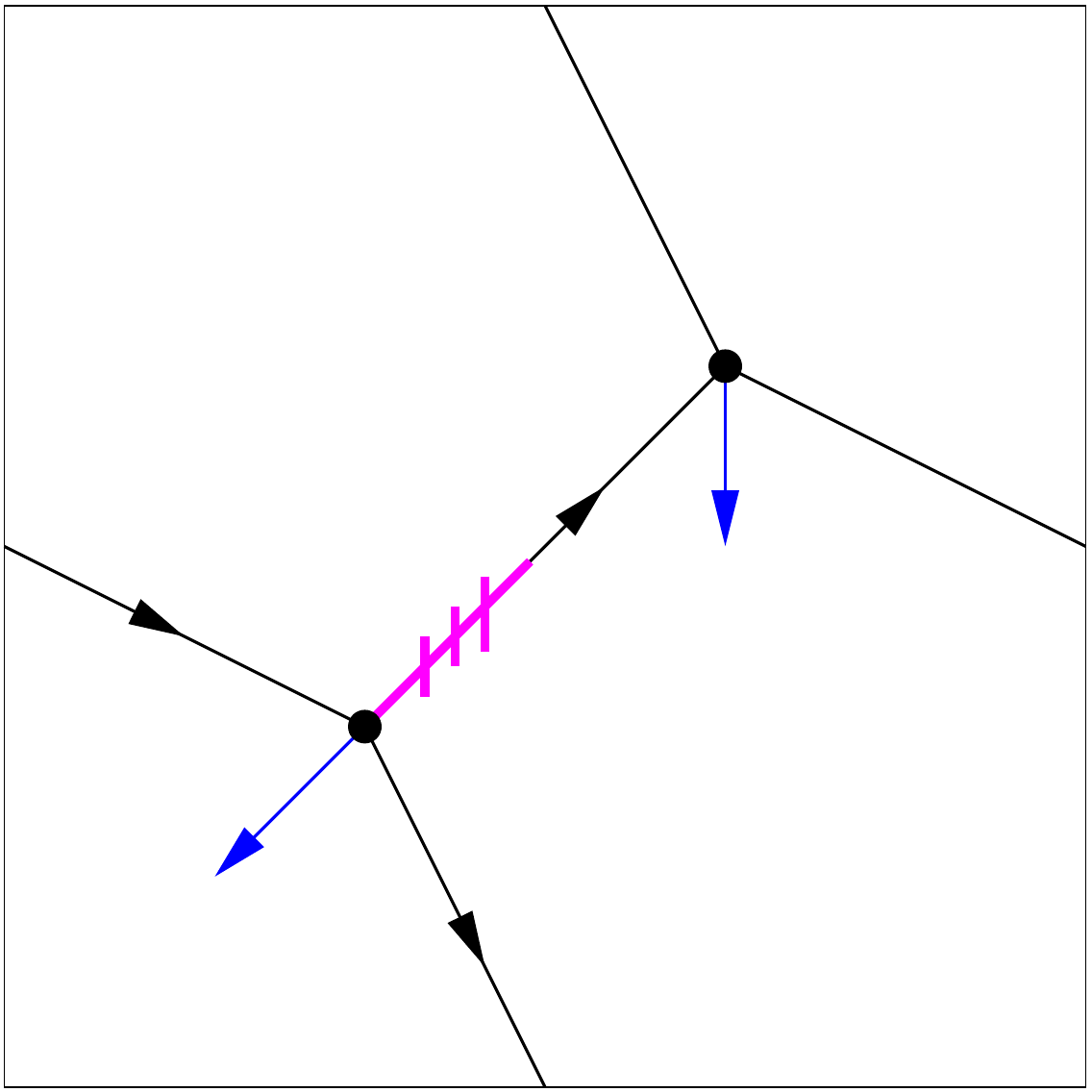}
\\ Type $4$ & Type $5$ & Type $6$\\
& \\

  \scalebox{0.3}{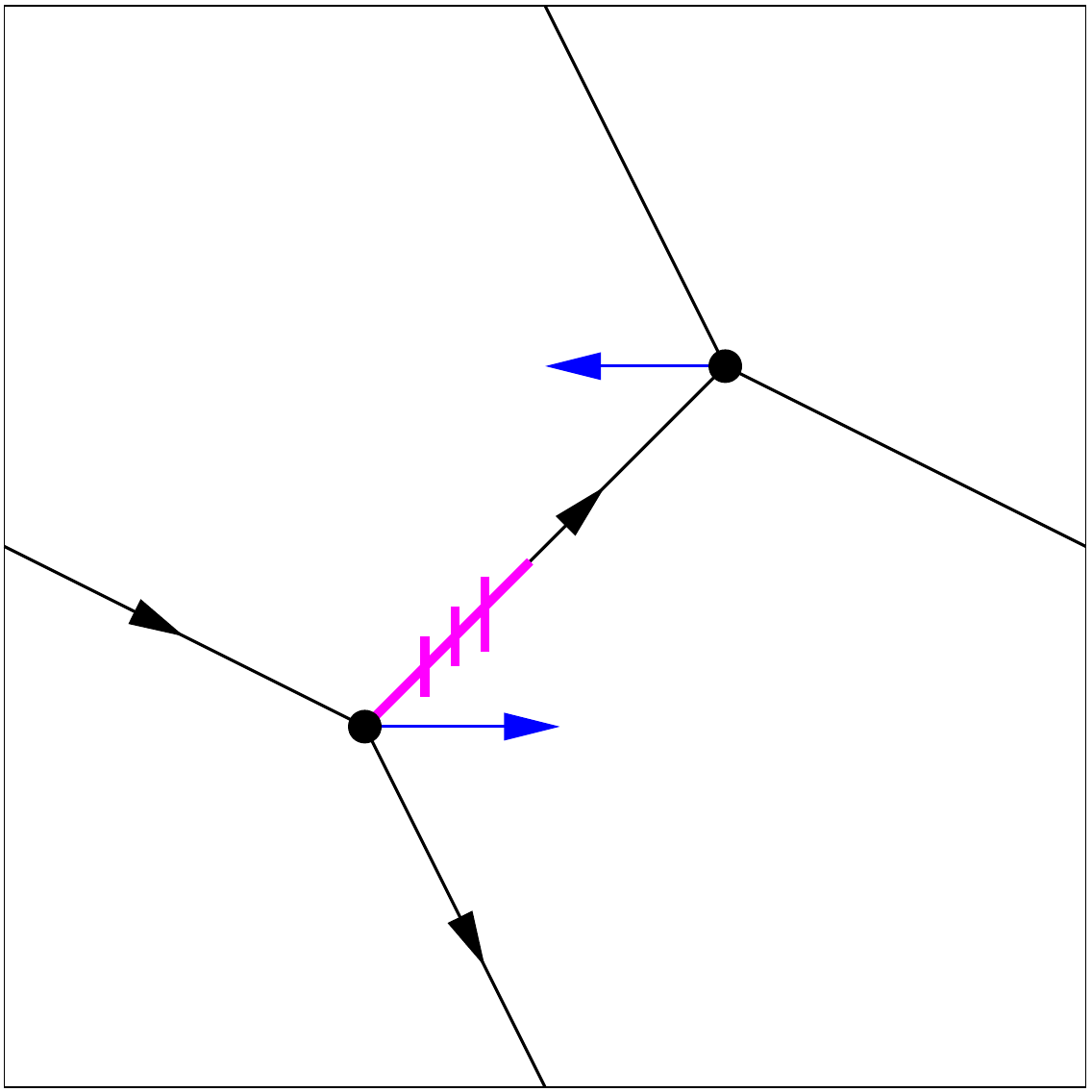} &
  \scalebox{0.3}{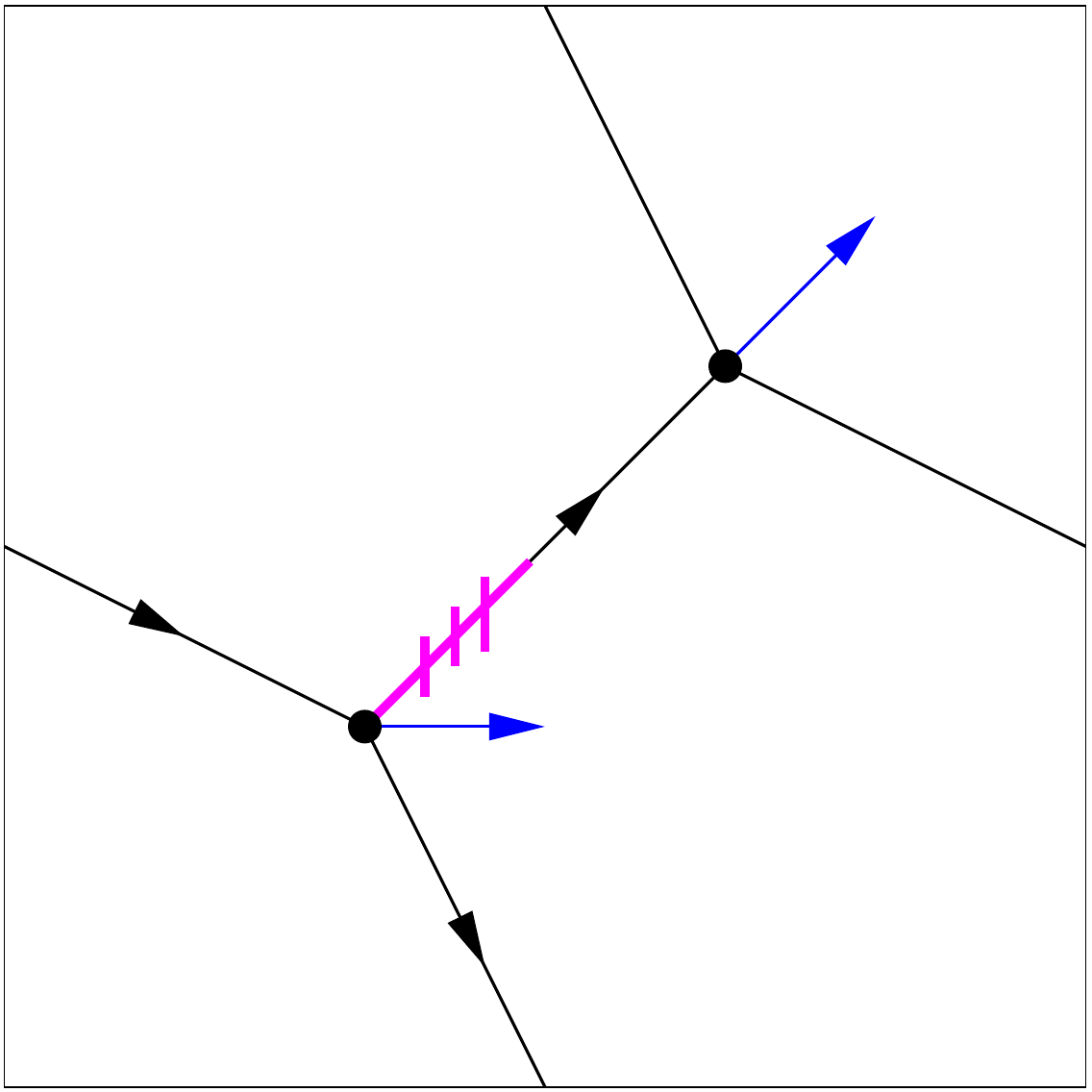} &
  \scalebox{0.3}{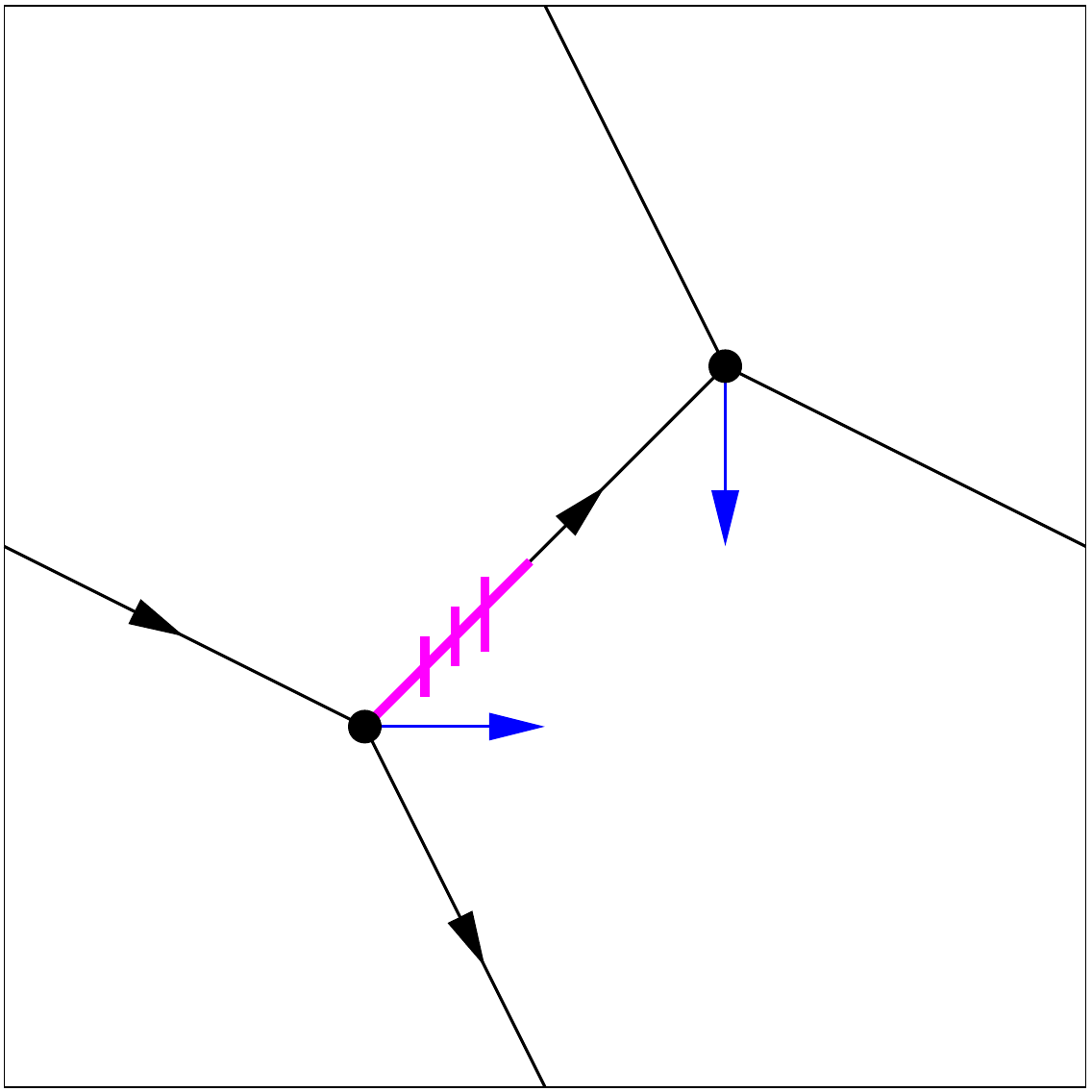}
\\ Type $7$ & Type $8$ & Type $9$
\end{tabular}
\caption{The ten possible types of kernels for an element of $\mathcal U_{r}(n)$. The pink half-edge indicates the root half-edge.}
    \label{fig:kernels}
\end{figure}

Given an element  $U\in \mathcal U_{r,b}(n)$ of a given type, we
decompose it into its core $C$ and a set of forests. We orient and
denote the maximal chains of $U$ as on Figure~\ref{fig:kernels}. Each
of these maximal chain as two sides. For $t=3$ when $U$ is hexagonal
and $t=2$ when $U$ is square, we define $2t$ particular angles
$\alpha_1,\ldots, \alpha_{2t}$ of $U$ as depicted on
Figure~\ref{fig:anglecore} and moreover we set
$\alpha_{2t+1}=\alpha_1$. Note that the angles
$\alpha_1,\ldots, \alpha_{2t+1}$ are formally defined on $U$ but with a
slight abuse of notations, we also consider them to be defined on $C$
(with exactly the same definition as Figure~\ref{fig:anglecore}).

\begin{figure}[h]
\center
\begin{tabular}{cc}
      \scalebox{0.5}{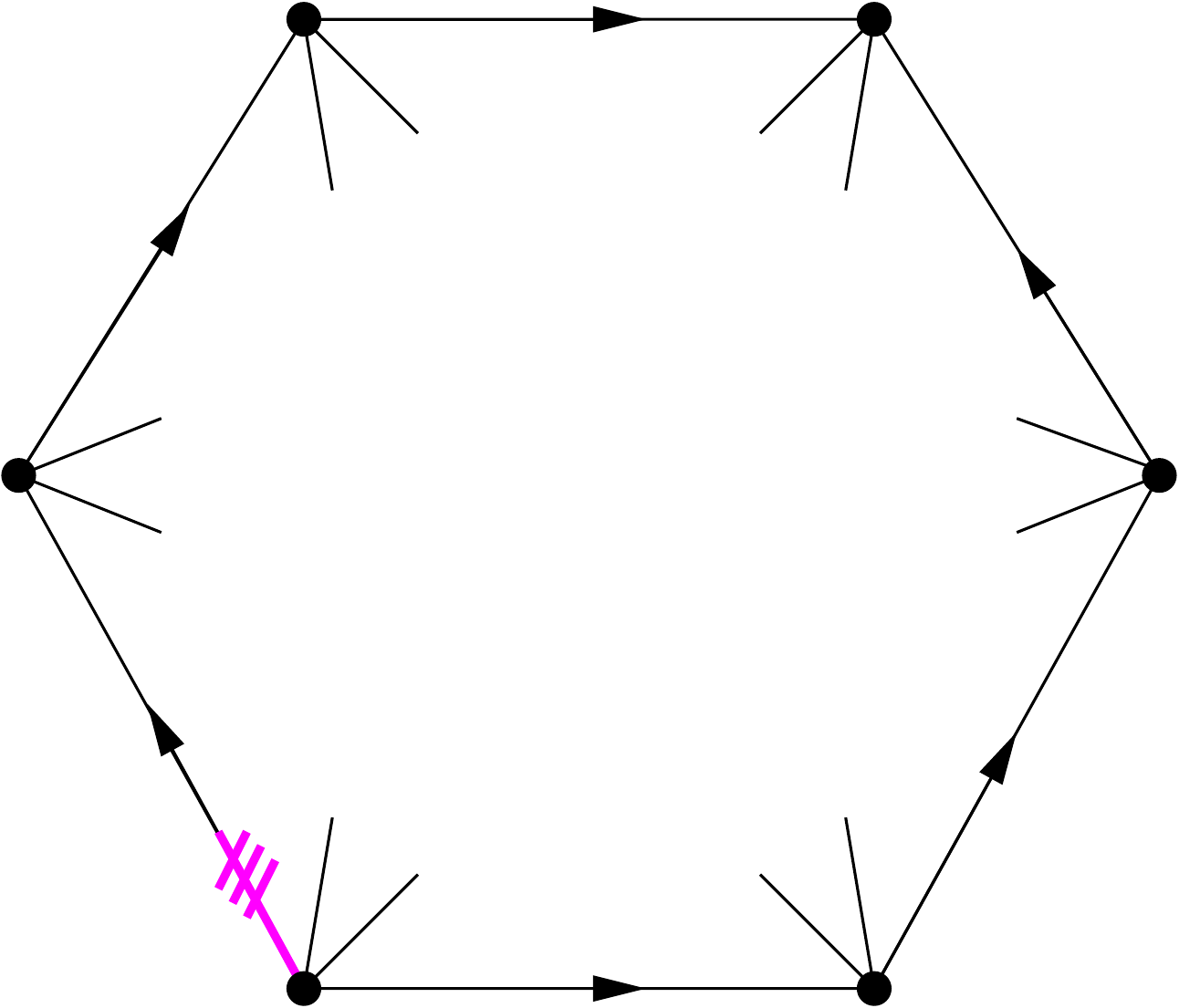}
\hspace{1em} & \hspace{1em} 
               \scalebox{0.5}{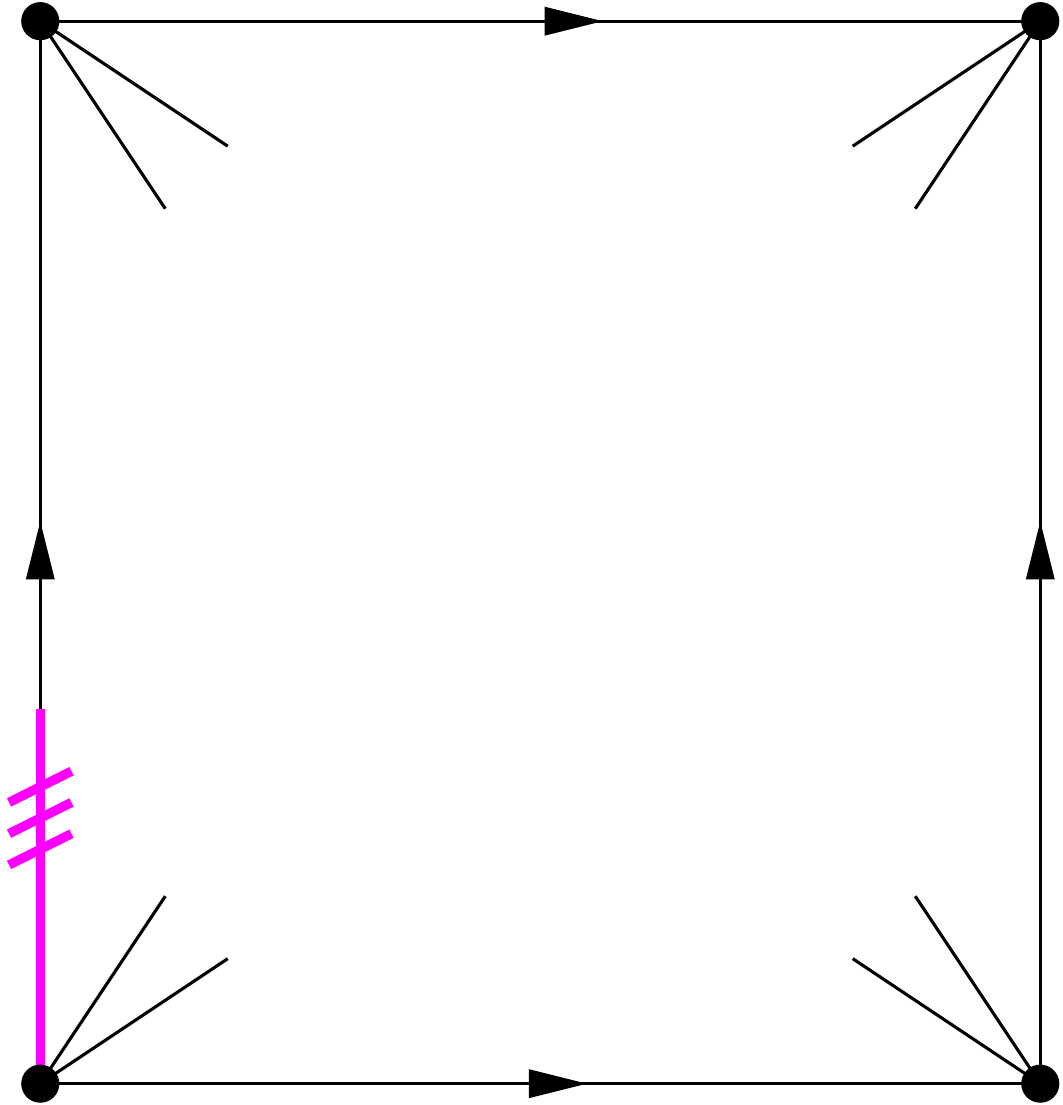}
                
\end{tabular}
\caption{Definition of the angles $\alpha_1,\ldots, \alpha_{2t}$.}
\label{fig:anglecore}
\end{figure}

Let $[\alpha,\beta[$ denote the set of angles of $U$ between $\alpha$
and $\beta$, while walking along the border of the unique face of $U$
in \cw order, including $\alpha$ and excluding $\beta$.  Let
$[\alpha,\beta[\cap C$ denote the set of angles of $[\alpha,\beta[$
that are also incident to the core $C$. For $1\leq i\leq t$, let $S_i$
(resp.  $S_{i+t}$) be the maximal chain $W_i$ with all the stems of
$U$ that are incident to an angle of $[\alpha_i,\alpha_{i+1}[\cap C$
(resp.  $[\alpha_{i+t},\alpha_{i+t+1}[\cap C$).  Then $U$ is
decomposed into its core $C$ plus $2t$ parts where the $i$-th part is
the part of $U$ ``attached to (the right side of) $S_i$''. More
formally, for $1\leq i\leq t$, the $i$-th part (resp. the $(i+t)$-th
part) corresponds to all the components of $U\setminus C$ that are
connected to the rest of $U$ via an edge of $U$ that is incident to an
angle of $[\alpha_i,\alpha_{i+1}[\cap C$ (resp.
$[\alpha_{i+t},\alpha_{i+t+1}[\cap C$).  Each of these $2t$ parts can
be represented by one well-labeled forest (see
Figure~\ref{fig:exampledecompforest} where $S_i$ is represented in
green): the floor vertices of the forest corresponds to the angles of
$C$ in $[\alpha_i,\alpha_{i+1}[$ and the tree-vertices, tree-edges and
stems of the forest represents the part of $U$ attached to $S_i$.
Thus, the unicellular map $U$ is decomposed into its core $C$ plus
$2t$ well-labeled forests $((F_i,\ell_i))_{1\leq i \leq 2t}$.  For
$1\leq i\leq 2t$, let $\tau_i$ be the number of angles
$[\alpha_i,\alpha_{i+1}[\cap C$ and $\rho_i$ be the number of vertices of
the part of $U$ attached to $S_i$. So we have
$(F_i,\ell_i) \in \mathcal F^{\rho_i}_{\tau_i}$ for $1\leq i\leq 2t$.

\begin{figure}[h]
  \centering
 \scalebox{0.5}{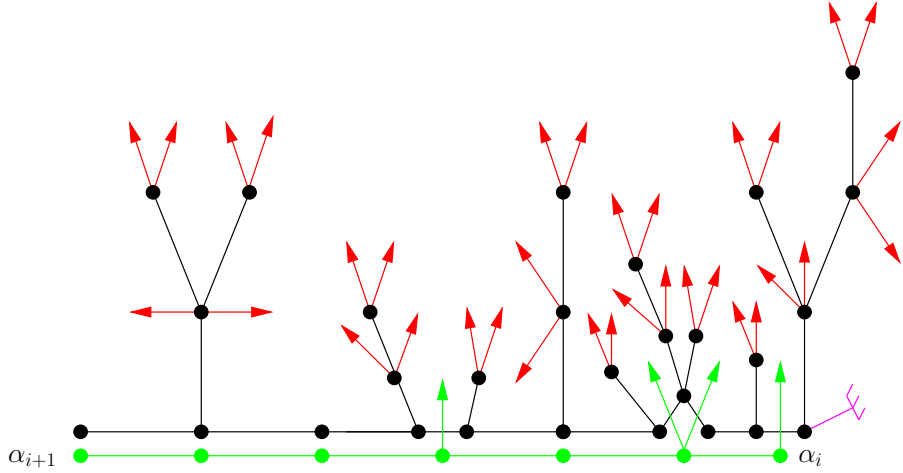}
  \caption{The forest attached to $S_i$.}
    \label{fig:exampledecompforest}
  \end{figure}

  We now decompose the core $C$ of $U$.

  For $1\leq i\leq t$, we define $R_i$ as the maximal chain $W_i$ of
  $U$ with all the stems of $U$ that are incident to an inner vertex of
  $W_i$. Note that the ``union'' of $S_i$ and $S_{i+t}$, almost gives
  $R_i$ except that $R_i$ contains no stems incident to special
  vertices.
  Then we decompose $C$  into the type of its kernel (see
  Figure~\ref{fig:kernels}) plus $(R_i)_{1\leq i \leq t}$.

  For $1\leq i \leq t$, all the inner vertices of $R_i$ are incident
  to exactly $2$ stems. Let $\gamma_i$ be half of the number of stems
  incident to the right side of $R_i$ minus half of the number of
  stems incident to the left side of $R_i$. Note that $\gamma_i$ is an
  integer. Let $\sigma_i$ be the number of inner vertices of $R_i$.

When $U$ is square, we have $\gamma_1=\gamma_2=0$ by the balanced
property of $U$.  In this case, for $1\leq i \leq 2$, the total number
of angles of $R_i$ and incident to inner vertices of $R_i$ is
$4\sigma_i$. So the number of angles of
$R_i$ on one of its side and incident to inner vertices is
$2\sigma_i$. So for $1\leq i\leq 2$,
$\tau_i = \tau_{i+2}= 2\sigma_i+1$.

When $U$ is hexagonal, the value of $\gamma_1+\gamma_2$ and
$\gamma_2+\gamma_3$ is given by the type of $U$ and the fact that $U$
is balanced, see Table~\ref{tab:gamma}.  As for the square case, we
have a relation between $\tau$ and $\sigma$, but this times it depends
on the type and of the $\gamma_i$'s.  For $1\leq i\leq 6$, let
$c_i\in\{0,1\}$ such that $c_i=1$ if and only if there is a stem incident to the
angle $\alpha_i$. The value of $c_1,\ldots,c_6$
is given in Table~\ref{tab:gamma}.  For $1\leq i\leq 3$, we have
$\tau_i = 2\sigma_i+1+\gamma_i+c_i$, and
$\tau_{3+i} = 2\sigma_i+1-\gamma_i+c_{3+i}$.

\begin{table}[h!]
\center
\begin{tabular}{|c|c|c|c|c|c|c|c|c|}
\hline
&$\gamma_1+\gamma_2$ & $\gamma_2+\gamma_3$ & $c_1$  & $c_2$  & $c_3$ & $c_4$  & $c_5$  & $c_6$\\
\hline
 Type 1 & 1& 0 & 0 & 0& 0 & 1 & 1 & 0 \\  \hline
 Type 2 & 1& 1 & 0 & 0& 0 & 0 & 1 & 1 \\ \hline 
 Type 3 & 0& 0 & 0 & 1& 0 & 0 & 1 & 0 \\ \hline
 Type 4 & 0& -1 & 0 & 0& 1 & 1 & 0 & 0 \\ \hline
 Type 5 & 0& 0 & 0 & 0& 1 & 0 & 0 & 1 \\ \hline 
 Type 6 & -1& -1 & 0 & 1& 1 & 0 & 0 & 0 \\ \hline
 Type 7 & 0& 0 & 1 & 0& 0 & 1 & 0 & 0 \\ \hline 
 Type 8 & 0& 1 & 1 & 0& 0 & 0& 0 & 1 \\ \hline  
 Type 9 & -1& 0 & 1 & 1& 0 & 0 & 0 & 0 \\ 
\hline
\end{tabular}
\caption{Values of $\gamma_1+\gamma_2$,  $\gamma_2+\gamma_3$,  $c_1$, \ldots, $c_6$, depending of the type.}
\label{tab:gamma}
\end{table}

For $1\leq i \leq t$, we represent $R_i$ by a Motzkin path $M_i$ of
length $\sigma_i$ from $0$ to $\gamma_i$, thus
$M_i\in\mathcal M_{\sigma_i}^{\gamma_i}$. Two stems on the right
(resp. left) side of $R_i$ corresponds to a step of $1$ (resp. $-1$)
in the Motzkin path. A stem on each side of $R_i$ corresponds to a
step of $0$ in the Motzkin path.

The path $R_i$ corresponding to the example $S_i$ of
Figure~\ref{fig:exampledecompforest} is represented on
Figure~\ref{fig:exampledecompmotzkin} with the corresponding Motzkin
path in $\mathcal M_5^{-2}$ (from right to
left). This Motzkin path is precisely the example given
in~(\ref{eq:MotzkinExample}). Note that from
Figure~\ref{fig:exampledecompforest}, the stem that was incident to
$\alpha_i$ has been removed since $R_i$ contains no stems incident to
special vertices (the Motzkin path $M_i$ represents only the stems
incident to inner vertices of $W_i$).
 
\begin{figure}[h]
  \centering 
 \scalebox{0.5}{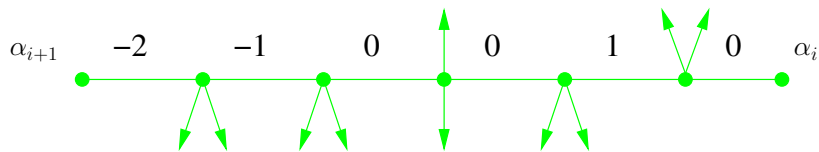}
  \caption{The Motzkin path corresponding to $R_i$ (from right to left).}
    \label{fig:exampledecompmotzkin}
  \end{figure}

Finally, we have a relation between the number of vertices $n$ and the value $\sigma_i$ and $\rho_i$:
$$n=\rho_1+\cdots+\rho_{2t}+\sigma_1+\cdots+\sigma_t+(t-1)$$

\begin{definition}
\label{def:defofR0}
For $n\geq 1$, let $$\mathcal R^0(n)=\bigcup_{\substack{(\rho_1,\ldots,\rho_4)\in \mathbb N^4\\ (\tau_1,\ldots,\tau_4)\in (\mathbb N^*)^4\\
  (\sigma_1,\sigma_2)\in \mathbb N^2}}
 \mathcal F^{\rho_1}_{\tau_1}\times \cdots \times \mathcal F^{\rho_4}_{\tau_4}\times \mathcal M_{\sigma_1}^{0}\times M_{\sigma_2}^{0}$$
where
$$n=\rho_1+\cdots+\rho_{4}+\sigma_1+\sigma_2+1$$
$$\textrm{and for }1\leq i\leq 2,\textrm{ we have }
\tau_i = \tau_{2+i}= 2\sigma_i+1.$$ 
\end{definition}

Thus, by above discussion, for $n\geq 1$, there is a bijection between the set of (square) unicellular maps 
$\mathcal U^S_{r,b}(n)$ and $\mathcal R^0$.

\begin{definition}
\label{def:defofR1to9}
For $n\geq 1$ and $1\leq k\leq 9$, let
 $$\mathcal R^k(n)=\bigcup_{\substack{(\rho_1,\ldots,\rho_6)\in \mathbb N^6\\ (\tau_1,\ldots,\tau_6)\in (\mathbb N^*)^6\\
     (\gamma_1,\gamma_2,\gamma_3)\in \mathbb Z^3\\
     (\sigma_1,\sigma_2,\sigma_3)\in \mathbb N^3}} \mathcal
 F^{\rho_1}_{\tau_1}\times \cdots \times \mathcal F^{\rho_6}_{\tau_6}\times
 \mathcal M_{\sigma_1}^{\gamma_1}\times \mathcal
 M_{\sigma_2}^{\gamma_2} \times \mathcal M_{\sigma_3}^{\gamma_3}$$
 where
$$n=\rho_1+\cdots+\rho_{6}+\sigma_1+\sigma_2+\sigma_3+2$$
$$\textrm{for } 1\leq i\leq 3\textrm{, we have }
\tau_i = 2\sigma_i+1+\gamma_i+c_i \textrm{ and }
\tau_{3+i} = 2\sigma_i+1-\gamma_i+c_{3+i}$$
$$\textrm{for } 1\leq i\leq 3\textrm{, we have }
|\gamma_i| \leq \sigma_i$$
$$\textrm{with }\gamma_1+\gamma_2, \gamma_2+\gamma_3, c_1, \ldots, c_6  \textrm{ given by line $k$ of Table~\ref{tab:gamma}}.$$
\end{definition}

Thus, by above discussion, for $n\geq 1$ and $1\leq k\leq 9$, there exists a bijection between elements of
$\mathcal U^H_{r,b}(n)$ with kernel of type k and
$\mathcal R^k(n)$.

So by Lemma~\ref{prop:bijrootuni} we have the following:

\begin{lemma}
\label{bij:decomposition}
For $n\geq 1$,  there exists a bijection between
$[\![1,3]\!]\times \mathcal T_{r,s,b}(n)$ and
$$\left([\![1,3]\!]\times \mathcal R^0(n)\right)\ \bigcup\  \left([\![1,2]\!] \times \bigcup_{1\leq k\leq 9} \mathcal R^k(n)\right)$$
\end{lemma}

\subsection{Relation with labels of the unicellular map}
\label{sec:labeldistance}

We use the same notations as in previous section where
$U$ is an element of $\mathcal U_{r,b}(n)$ that is decomposed into the
type $k$ of its kernel, $2t$ well-labeled forests
$((F_i,\ell_i))_{1\leq i\leq 2t}$, with $(F_i,\ell_i)\in \mathcal F^{\rho_i}_{\tau_i}$, and $t$ Motzkin paths
$(M_i)_{1\leq i\leq t}$, with $M_i\in\mathcal M_{\sigma_i}^{\gamma_i}$.

We explain in this section how the well-label forests, the Motzkin
paths and the type are linked to the labeling function $\lambda$
defined in Section~\ref{sec:label}.

As in the proof of Lemma~\ref{prop:bijrootuni}, there are four angles
of $U$ where a root stem can be added to obtain an element of
$\mathcal T_{r,s,b}(n)$ from $U$ (after also forgetting the initial
root of $U$). Consider one such element $T\in\mathcal
T_{r,s,b}(n)$. Let $G$ be the image of $T$ by the bijection $\Phi$ of
Theorem~\ref{them:bijectionbenjamin} and $V$ the set of vertices of
$G$.  Let $\Gamma$ be the unicellular map obtained from $T$ by adding
a dangling root half-edge incident to its root angle. Let $\lambda$ be
the labeling function of the angles of $\Gamma$ as defined in
Section~\ref{sec:label}.  For all $u,v \in V$, let $m(u)$ and
$\overline m(u,v)$ be  defined as in 
Section~\ref{sec:label}.

Recall that the labeling function $\lambda$ is defined on
the angles of $\Gamma$ by the following: while going clockwise around
the unique face of $\Gamma$ starting from the root angle with
$\lambda$ equals to $3$, the variation of $\lambda$ is ``+1'' while going
around a stem and ``-1'' while going along an edge.

Recall that, for $1\leq i\leq t$, the Motzkin path $M_i$ is used to
represent the part $R_i$ of the unicellular map $U$ (see
Section~\ref{dec}).  Consider the extension $\widetilde {M_i}$ of
$M_i$, defined in Section~\ref{sec:motz}. Note that $\widetilde{M_i}$
can be used to encode the variation of the labels along the path $R_i$
between $\alpha_i$ (excluded) and $\alpha_{i+1}$ (included) as if we
were computing $\lambda$ around $R_i$.
Figure~\ref{fig:exampledecompmotzkinextention} is an example obtained
by superposing the example $R_i$ of
Figure~\ref{fig:exampledecompmotzkin} and the extension of the
corresponding Motzkin path given
by~(\ref{eq:MotzkinExtensionExample}). One can check that, from
$\alpha_i$ (excluded) to $\alpha_{i+1}$ (included), we get ``+1''
around a stem and ``-1'' along an edge, like in the definition of
$\lambda$.

\begin{figure}[h]
  \centering 
 \scalebox{0.5}{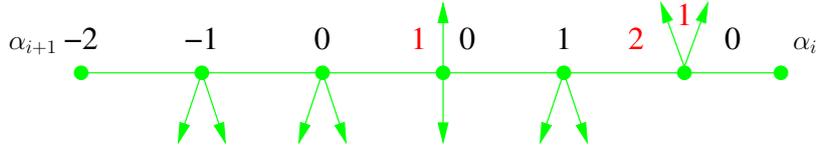}
 \caption{The extension of the Motzkin path (from right to left).}
    \label{fig:exampledecompmotzkinextention}
  \end{figure}

  Note also that $\widetilde{\underline{M_i}}$ encode the variation of
  the labels along the path $R_i$ between $\alpha_{i+t}$ (excluded)
  and $\alpha_{i+t+1}$ (included).
  Figure~\ref{fig:exampledecompmotzkinextentioninverse} is an example
  obtained by superposing the example $R_i$ of
  Figure~\ref{fig:exampledecompmotzkin} and the extension of the
  inverse of the corresponding Motzkin path given
  by~(\ref{eq:MotzkinExampleInverseExtension}).

\begin{figure}[h]
  \centering 
 \scalebox{0.5}{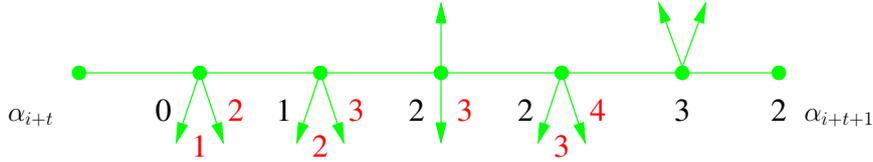}
  \caption{The inverse of the extension of the Motzkin path (from left
    to right).}
    \label{fig:exampledecompmotzkinextentioninverse}
  \end{figure}

  For convenience, we define 
  $\widetilde {M_i}=\widetilde{\underline{M_{i-t}}}$  for $t+1\leq
  i\leq 2t$. So the sequence $\widetilde {M_1}, \ldots,\widetilde
  {M_{2t}}$ corresponds to the parts of the $R_i$ appearing
  consecutively while going \cw around the unique face of $U$.

  Now we need to extend a bit more $\widetilde {M_i}$ so
  it also encodes $\alpha_i$ and a possible stem incident to $\alpha_i$.
  For a Motzkin path $\widetilde M\in\mathcal M_{2\sigma+\gamma}^{\gamma}$ and
  $c\in\{0,1\}$, we define the \emph{$c$-shift of $\widetilde M$} as the following
  Motzkin path in $\mathcal M_{2\sigma+\gamma+c+1}^{\gamma+c-1}$: 
  $$
  \widetilde M^c=
  \begin{cases}
    (   0, (\widetilde M)_1-1,\ldots,(\widetilde M)_{2\sigma+\gamma}-1) &\text{ if } c=0\\
    (   0, 1, (\widetilde M)_1,\ldots,(\widetilde M)_{2\sigma+\gamma}) &\text{ if } c=1\\
  \end{cases}
  $$

  For $1\leq k\leq 9$ and $1\leq i\leq 6$, let $c_i(k)$ be the value
  of $c_i$ given by line $k$ of Table~\ref{tab:gamma}. We also define
  $c_1(0)=c_2(0)=c_3(0)=c_4(0)=0$. For $t+1\leq i\leq 2t$, let
  $\gamma_i=-\gamma_{i-t}$ and $\sigma_i=\sigma_{i-t}$.  With these
  notations, for $1\leq i\leq 2t$, we can consider the Motzkin path
  $\widetilde {M_i}^{c_i(k)}$ that is an element of
  $\mathcal M_{\tau_i}^{\gamma_i+c_i(k)-1}$ (see
  Definitions~\ref{def:defofR0} and~\ref{def:defofR1to9} for the
  relation between $\tau$, $\gamma$, $\sigma$, $c$). Now
  $\widetilde {M_i}^{c_i(k)}$ encode ``completely'' $R_i$ from
  $\alpha_i$ to $\alpha_{i+1}$ (both included) with also the stems
  incident to special vertices depending on the type.

Now we explain the links between $\lambda$ and the well-labeled forests.
Consider a tree of a well-labeled forest $(F,\ell)$. 
 Figure~\ref{fig:tree-label} gives an example represented either with
 its labels (on the left side) or with its stems (on the right
 side). Note that it is  the
first tree of the well-labeled forest of
Figures~\ref{wellforest} and~\ref{fig:welllabelstem} (i.e. the one on the right).

\begin{figure}[h]
    \centering 
    \scalebox{0.5}{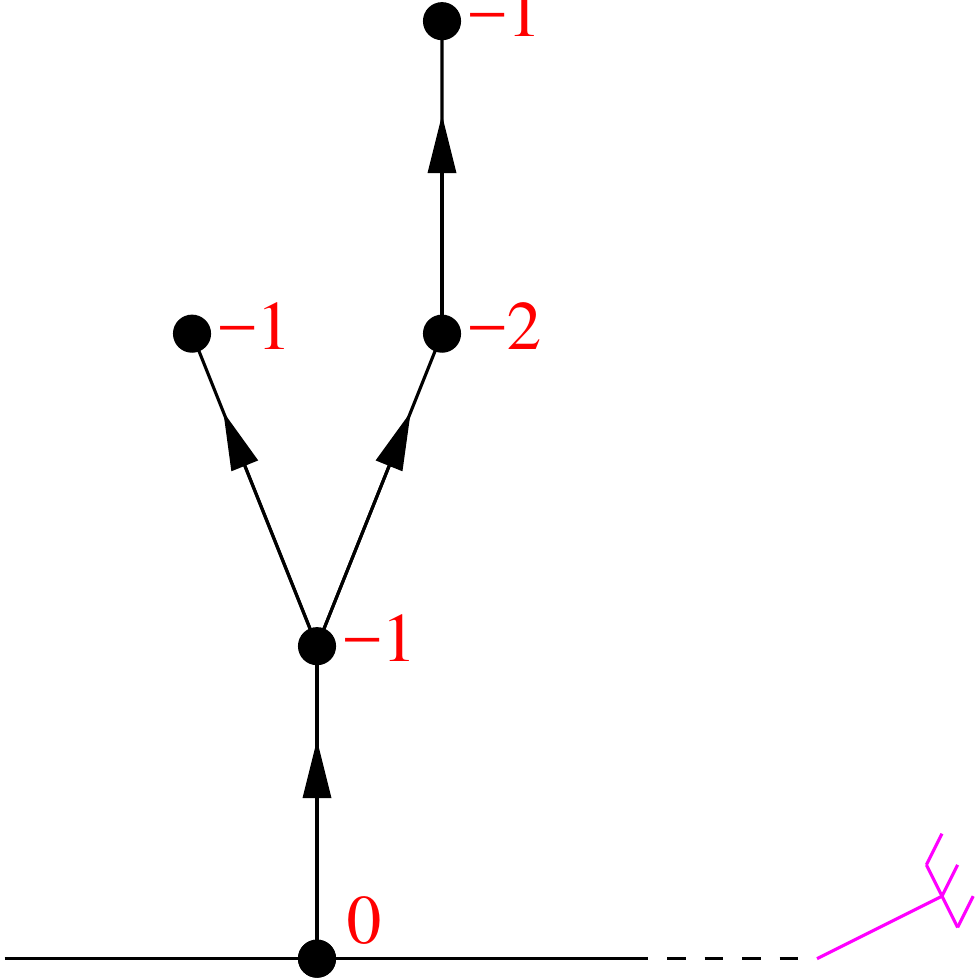}\hspace{4em}
     \scalebox{0.5}{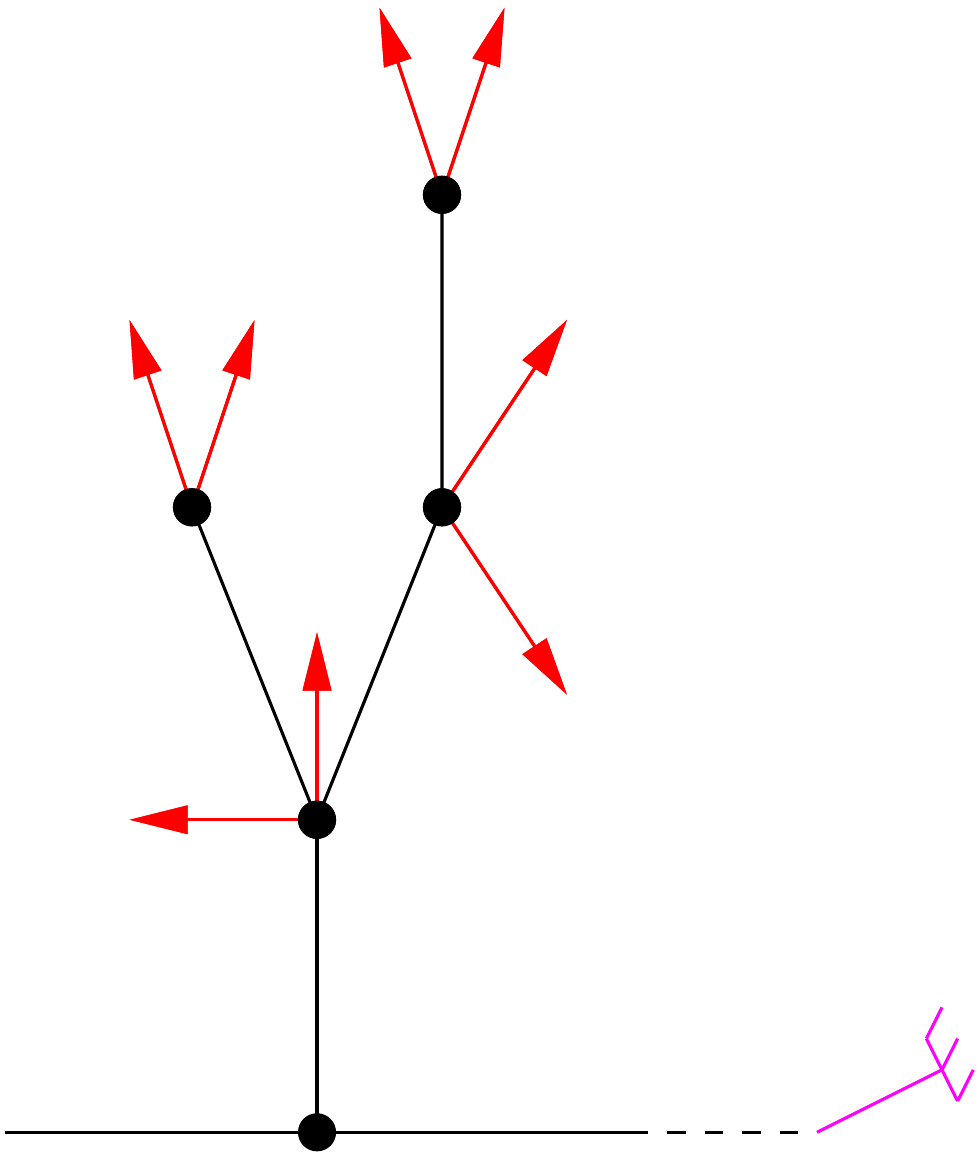}
  \caption{Example of a tree of a well-labeled forest.}
    \label{fig:tree-label}
  \end{figure}

  If one computes the variation of $\lambda$ on the angles of the tree
  ``above the floor line''. Then one can note that the first angle of
  each vertex that is encountered receive precisely the label given by
  the function $\ell$ of $(F,\ell)$.  Figure~\ref{fig:tree-stem}, show
  this computation on the example of Figure~\ref{fig:tree-label} where
  the correspondence with the values of $\ell$ is represented in red.

\begin{figure}[h]
    \centering 
     \scalebox{0.5}{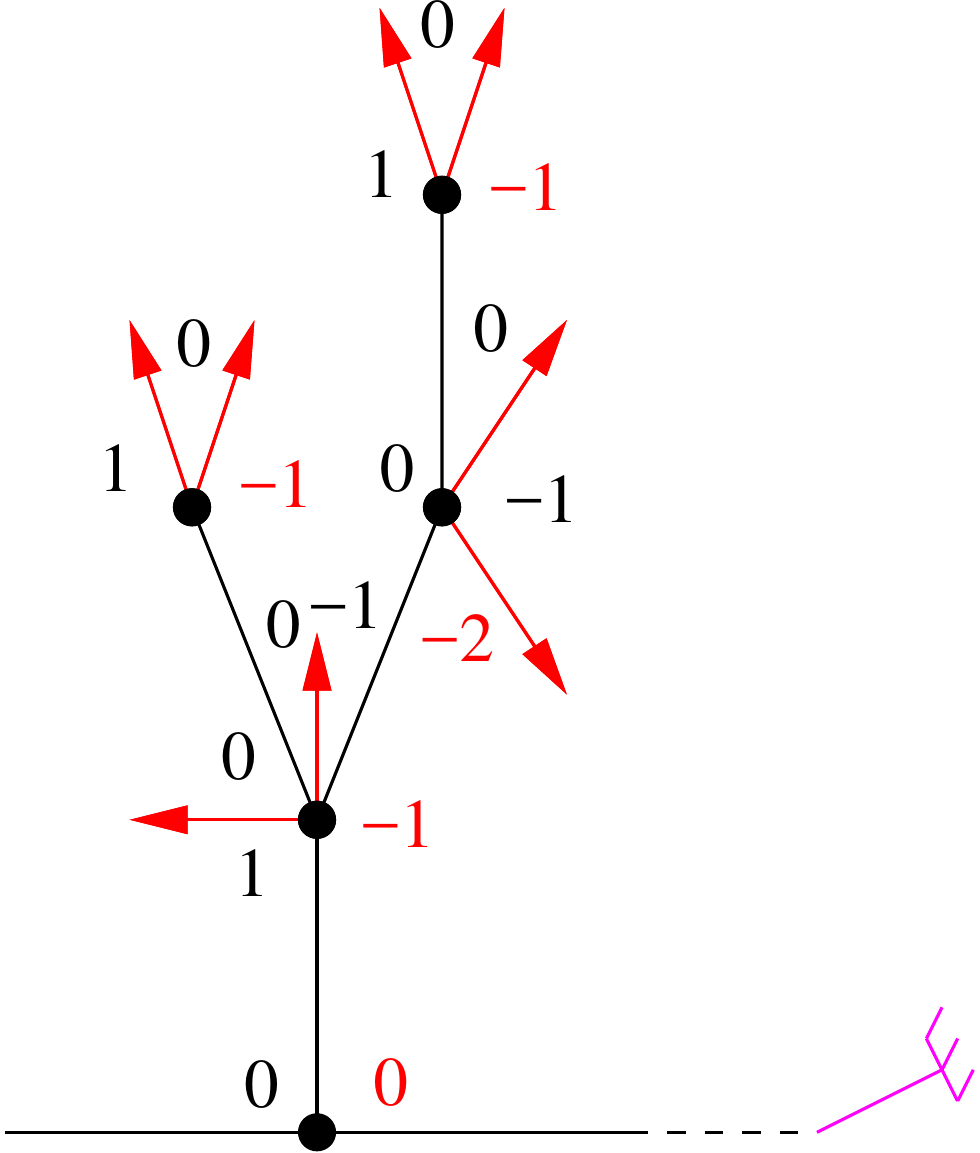}
  \caption{Computation of the label $\lambda$ around a tree of a well-labeled forest.}
    \label{fig:tree-stem}
  \end{figure}

  Now with the help of the $c$-shift extensions of Motzkin paths we
  can encode completely the variation of the labels around the
  well-labeled forests.  For $1\leq i\leq 2t$, consider the vertex
  contour function $r_{F_i}$ and contour pair
  $(C_{F_i},L_{(F_i,\ell_i)})$ of $(F_i,\ell_i)$.  For
  $0\leq s \leq 2\rho_i+\tau_i$, let
  $\overline{C_{F_i}}(s)=\max_{x\leq s}C_{F_i}(x)$. Note that for
  $0\leq s \leq 2\rho_i+\tau_i$ the value of $\overline{C_{F_i}}(s)+1$
  is  the floor of the vertex $r_{F_i}(s)$.  For
  $0\leq s \leq 2\rho_i+\tau_i$, we define
  $$S_i(s)=L_{(F_i,\ell_i)}(s)+\widetilde
  {M_i}^{c_i(k)}(\overline{C_{F_i}}(s)).$$ With this definition, if one
  computes the variations of $\lambda$ around $F_i$, starting from
  $\alpha_i$ with value $0$, an ending at $\alpha_{i+1}$ then the
  first angle of each vertex $v$ that is encountered receives the
  value $S_i(t)$ where $t$ is any value $0\leq t\leq 2\rho_i+\tau_i$
  such that $r_{F_i}(t)=v$.

  For $f,g$ two functions defined on $[\![0,s]\!]$ and $[\![0,s']\!]$
  respectively, taking values in $\mathbb Z$ and such that
  $g(0)=0$. We define the \emph{concatenation} of $f,g$, denoted
  $f\bullet g $, as the function defined on $[\![0,s+s']\!]$ by the
  following:
$$(f\bullet g)(i)=
\begin{cases}
f(i) & \text{ if } 0\leq i\leq s \\
f(s)+g(i-s) & \text{ if } s\leq i\leq s+s'
\end{cases}$$

Let $I=\sum_{1\leq i\leq
   2t}(2\rho_i+\tau_i)$.
 Let $S^\bullet=S_1\bullet \cdots \bullet S_{2t}$ be the function
 defined on $[\![0,I]\!]$.
 Note that $S^\bullet(0)=S^\bullet(I)=0$.
 Note also that $I=\sum_{1\leq i\leq
   2t}(2\rho_i+\tau_i)=(2n+2)+2\times(\sigma_1+\cdots+\sigma_{t})+2\times
 \mathbbm 1_{k\neq 0}$.

 As in Section~\ref{sec:label}, we call \emph{proper}, the vertices of
 $U$ that are on at least one cycle of $\Gamma$. Let $P$ be the 
 unicellular map obtained from $U$ by removing all the stems that are
 not incident to proper vertices.  We still denote by
 $\alpha_1,\ldots, \alpha_{2t}$ the angles of $P$ corresponding to the
 angles $\alpha_1,\ldots, \alpha_{2t}$ of $U$. Note that $P$ has
 precisely $I$ angles. So we see $S^\bullet$ as a function from the
 angles of $P$ to $\mathbb Z$ by starting at $\alpha_1$ and walking
 \cw around the unique face of $P$.
 
 We define \emph{the vertex contour function of $P$} as the function
 $r_{P}:[\![0,I-1]\!] \rightarrow V$ as follows: while walking \cw
 around the unique face of $P$, starting at $\alpha_1$,
 let $r_{P}(i)$ denote the $i$-th vertex of $P$ that is encountered.

 Recall that for $u\in V$, $m(u)$ is the minimum of the values of
 $\lambda$ that appears in the angles incident to $u$.

 We explain that for $i\in[\![0,I-1]\!]$, $S^\bullet(i)$ is almost
 equal to $m(r_P(i)) - m(r_P(0))$. On one hand, we have explain above
 that $S^\bullet$ almost acts as computing a ``variation'' of
 $\lambda$ around $U$ from $\alpha_1$. On the other hand the value of
 $m$ is obtained by computing $\lambda$ around $\Gamma$ from its root
 angle $a_0$. This angle $a_0$ can be anywhere in $U$. Since we are
 considering $m(r_P(i)) - m(r_P(0))$ we have shifted $m$ so its
 corresponds to ``computing $\lambda$ from $\alpha_1$.  Let
 $a_0,a_1,\ldots,a_\ell$ denote the angles of $\Gamma$ as in
 Section~\ref{sec:label}. There is a jump of $4$ in the computation of
 $\lambda$ from $a_\ell$ to $a_1$.  Thus in the ``variation'' of
 $\lambda$ computed around the well-labeled forests we can get a $+4$
 at some place. Moreover in such computations, we match the
 computation of $\lambda$ just at the first angle of each vertex that
 is encountered around the forest. By Lemma~\ref{lem:Mm6}, it can
 differ from $m$ by $\pm 6$.  Thus in total we have, for
 $i\in [\![0,I-1]\!]$,
 $|S^\bullet(i)-(m(r_P(i)) - m(r_P(0)))|\leq 4+6+6=16.$

 Note that $P$ contains exactly $2\times(\sigma_1+\cdots+\sigma_{t})+2\times
 \mathbbm 1_{k\neq 0}$ stems.
Let $Q$ be the unicellular map obtained from $P$ by removing all its
stems. We also denote by $\alpha_1,\ldots, \alpha_{2t}$ the
corresponding angles of $Q$. Note that $Q$ has exactly $2n+2$ angles.
We now define the \emph{vertex contour function of $Q$} as the function
$r_Q:[\![0,2n+1]\!] \rightarrow V$ as follows: while walking \cw
around the unique face of $W$, starting at $\alpha_1$, let
$r_Q(i)$ denote the $i$-th vertex of $Q$ that is encountered.

We define the sequence $(S(i))_{0\leq i\leq 2n+1}$ as the sequence
that is obtained from $(S^\bullet(i))_{0\leq i\leq I}$ by removing all
the values of $(S^\bullet)$ that appear in an angle of $P$ that is
just after a stem of $P$  in \cw order around its incident vertex.
 So we see $S$ as a function from the
 angles of $Q$ to $\mathbb Z$ by starting at $\alpha_1$ and walking
 \cw around the unique face of $Q$.
We call $S$ the \emph{shifted labeling function} of the unicellular map
  $U$.

By above arguments, for $i\in [\![0,2n+1]\!]$, we have
\begin{equation}
  \label{eq:Setm}
  |S(i)-(m(r_{Q}(i))-m(r_{Q}(0)))|\leq 16.
  \end{equation}

  We now introduce the following pseudo-distance function.
For $i,j\in[\![0,2n+1]\!]$, let
$$d^{o}(i,j)=m(r_Q(i))+         m(r_Q(j))-
2\overline{m}(r_Q(i),r_Q(j))$$
By (\ref{eq:Setm}), we obtained the following: for $i,j\in[\![0,2n+1]\!]$,
\begin{equation}
  \label{eq:d0S}
|d^{o}(i,j)-(S(i)+S(j)-2\overline S(i,j))| \leq 64
\end{equation}
where $\overline S(i,j)=\min_{i\leq t\leq j}S(t)$.

\section{Some variants of the Brownian motion}
\label{sec:brownian}

We start with a some definitions
Let

 $$\mathcal H=\bigcup_{x\in \mathbb{R}_+} C([0,x],\mathbb{R}),$$
 where $C([0,x],\mathbb{R})$ is  the  set of  continuous functions
 from 
 $[0,x]$ to $\mathbb{R}$.

 We use the following standard notation: $x\wedge y=\min (x,y)$ for
 $x,y\in \mathbb R^2$.  For an element $f\in \mathcal H$, let
 $\sigma(f)$ be the only $x$ such that $f \in
 C([0,x],\mathbb{R})$. Then we
 define the following metric on $\mathcal H$:
$$d_{\mathcal H}(f,g)=|\sigma(f)-\sigma(g)|+\sup_{y\geq 0}|f(y\wedge
\sigma(f))-g(y\wedge \sigma(g))|.$$

Given a function $f:[0,x] \rightarrow \mathbb{R}$,
for $0\leq t\leq x$, let {$\overline{f}(t)=\sup_{r\in
    [0,t]}f(r)$.
 
Let $p$ (resp. $p_a$) denote the density of the standard Gaussian
random variable (resp. the centered Gaussian random variable with
variance $a$), i.e. for $x\in \mathbb R$,
$p(x)=\frac{1}{2\pi}e^{\frac{-x^2}{2}}$
(resp. $p_a(x)=\frac{1}{\sqrt{a}}p(\frac{x}{\sqrt{a}})$).  Let $p'_a$
denotes the derivative of $p_a$.

Let $\beta$
be  the   standard   Brownian  motion.

Consider $\tau,\rho\in \mathbb R^*_+$.
Intuitively,  the Brownian bridge $B^{0\rightarrow \tau}_{[0,\rho]}$ is the
standard Brownian motion on $[0,\rho]$ conditioned to take value $\tau$
at time $\rho$ and the first-passage Brownian bridge $F^{0\rightarrow \tau}_{[0,\rho]}$ is the
Brownian bridge conditioned to take value $\tau$ at time $\rho$  for the
first time. Since the probabilities of these conditioning events
are equal to $0$, these processes need to be more formally defined.
There are many equivalent definitions (see for example
\cite{bertoin2003path,billingsley2013convergence,revuz2013continuous})
and we use the following one (as   explained   in   
\cite{fitzsimmons1993markovian}, lemma $1$).

The      \emph{Brownian   bridge}
$B^{0\rightarrow \tau}_{[0,\rho]}$  is  the  unique  continuous  process
$(B_t)_{t\in[0,\rho]}$ taking  value $\tau$  at time $\rho$  and satisfying,
for every  $\rho' \in  [0,\rho[$ and  every continuous $f  : \mathcal  H \to
\mathbb R$, the identity
$$\mathbb{E}[f(B\vert_{[0,\rho']})]  =  \mathbb{E} \left[  f(\beta\vert_{[0,\rho']})
  \frac{p_{\rho-\rho'}(\tau-\beta_{\rho'})}{p_\rho(\tau)} \right].$$
Similarly, the \emph{first-passage Brownian bridge}
$F^{0\rightarrow \tau}_{[0,\rho]}$ is the unique continuous process
$(F_t)_{t \in [0,\rho]}$ taking value $\tau$ at time $\rho$ for the
first time and satisfying, for every $\rho' \in [0,\rho[$ and every
continuous function $f : \mathcal H \to \mathbb R$, the identity
$$\mathbb{E}[f(F\vert_{[0,\rho']})]  =  \mathbb{E} \left[  f(\beta\vert_{[0,\rho']})
  \frac{p'_{\rho-\rho'}(\tau-\beta_{\rho'})}{p'_\rho(\tau)}              {\mathbbm
    1}_{\overline{\beta}_{\rho'}<\tau} \right].$$

For convenience we define: 
$$\widetilde{F}^{0\rightarrow
    \tau}_{[0,\rho]}=\frac{1}{2}\left(F^{0\rightarrow
    \tau}_{[0,\rho]}+\overline {F^{0\rightarrow \tau}_{[0,\rho]}}\right).$$

Given a function $f:[0,\rho] \rightarrow \mathbb{R}$,
for $0\leq s\leq t\leq \rho$, let  
  $\check{f}(s,t)=\inf_{r\in [s,t]} \left(\overline f(r)-f(r)\right)$.

We now define the Brownian snake's head driven by a first-passage Brownian
bridge.
To simplify the notation, let $F$ denote the first-passage Brownian
bridge $F^{0 \to \tau}_{[0,\rho]}$.
The \emph{Brownian snake's head $Z=Z_{[0,\rho]}^\tau$ driven by $F$} is,
conditionally on $F$, define as the centered Gaussian process
satisfying, for $0\leq s\leq t\leq \rho$:
$$\Cov(Z(s),Z(t)){=\check{F}(s,t})$$

We can assume that  $Z_{[0,2\rho]}^\tau$ is
almost surely (a.s.) continuous.

Now, define an  equivalence relation as follows: for
any   $0\leq  s\leq   t\leq   \rho$,  we   say   that  $s\sim_{F}t$   if
{$\check{F}(s,s)=\check{F}(t,t)=\check{F}(s,t)$}. Then  the \emph{Brownian  continuum random
  forest} $(\mathcal{T}_F,d_{\mathcal{T}_F})$ is  defined as the space
$\mathcal{T}_F=[0,\rho]/\!\sim_F$    equipped     with    the    distance function
$d_{\mathcal{T}_F}(s,t)=\check{F}(s,s)+\check{F}(t,t)-2\check{F}(s,t)$   for   any   pair
$(s,t)$ such that $0\leq s\leq t\leq 2\rho$.

\begin{remark}
  \label{domainofZ}
  Note that      if       $s\sim_{F}t$      then
  $\mathbb{E}[(Z^\tau_{[0,\rho]}(s)-Z\tau_{[0,\rho]}(t))^2]=0$,  meaning   that  as
  usual  $Z^\tau_{[0,\rho]}$ can  be seen  as a  continuous Gaussian  process
  defined on $\mathcal T_{F}$.
\end{remark}

We now give  some definitions and results from (\cite{bettinelli2010scaling}, see also~\cite{petrov2012sums}):

The \emph{maximal span} of an integer-valued random variable $X$
is the greatest $h\in \mathbb{N}$ for which there exists an integer $a$ such that almost surely $X \in a+h\mathbb{Z}$.

Consider $(X_i)_{i\geq 0}$ a sequence of independent and identically
distributed $i.i.d.$ integer-valued centered
random variables with a moment of order $r_0$ for some $r_0\geq
3$. Let $\eta^2=Var(X_1)$, $h$ be the maximal span of $X_i$ and $a$ be
the integer such that $a.s.$ $X_i \in a+h\mathbb{Z}$. Let
$\Sigma_k=\sum_{i=0}^{k}X_i$ and $Q_k(i)=\mathbb{P}(\Sigma_k=i)$.

\begin{lemma}[\cite{bettinelli2010scaling}]
\label{sumiid}
We have:
    
    $$ 
\sup_{i\in ka+h\mathbb{Z}}\left | \frac{\eta}{h}\sqrt{k}\,Q_k(i)-p\left(\frac{i}{\eta \sqrt{k}}\right) \right |=o(k^{-\frac{1}{2}}), $$
    
and, for all $2 \leq r \leq r_0$, there exists a constant $C$ such that for all $i \in \mathbb{Z}$ and $k \geq 1$, 
    
    $$\left | \frac{\eta}{h}\sqrt{k}\, Q_k(i) \right | \leq \frac{C}{1+\left| \frac{i}{\eta \sqrt{k}}\right |^r}. $$
\end{lemma}

Consider $(\rho_n)\in \mathbb{N}^{\mathbb{N}}$ and
$(\tau_n)\in \mathbb{Z}^{\mathbb{N}}$  two sequences of integers
such that there exists $\rho,\tau\in \mathbb
R_+^*$ satisfying:
$$\frac{\rho_n}{n}\rightarrow \rho \text{ and } \frac{\tau_n}{\eta\sqrt{n}}\rightarrow \tau$$

Let $(B_n(i))_{0\leq i\leq \rho_n}$ be the process whose law is the law of $(\Sigma_i)_{0\leq i\leq \rho_n}$ conditioned on the event
$$\Sigma_{\rho_n}=\tau_n,$$
which we suppose occurs with positive probability.

We  write $B_n$ the linearly interpolated version of $B_n$
and define its rescaled version by: 
 $$B_{(n)}=\left(\frac{B_n(ns)}{\eta\sqrt{n}}\right)_{0\leq s\leq \frac{\rho_{n}}{n}}$$

\begin{lemma}[\cite{bettinelli2010scaling}]
\label{sumiid-2}
 There exists an integer $n_0\in \mathbb{N}$
 such that, for every $2\leq q \leq q_0$, there exists a constant
 $C_q$ satisfying, for all $n\geq n_0$ and
 $0\leq s\leq t \leq \frac{\rho_{n}}{n}$,
 $$\mathbb{E}[|B_{(n)}(t)-B_{(n)}(s)|^q]\leq C_q|t-s|^{\frac{q}{2}}.$$
\end{lemma}

\begin{theorem}[\cite{bettinelli2010scaling}]
  \label{convergencemotzkinbridge}
  The  process   $B_{(n)}$  converges   in  law  toward   the  process
  $B^{0\rightarrow       \tau}_{[0,\rho]}$,       in      the       space
  $(\mathcal H,d_\mathcal H)$, when $n$ goes to infinity.
\end{theorem}

\section{Convergence of the parameters in  the decomposition}
\label{section5}

For  all ${n\geq  1}$,  consider  a random  pair  $(z_n,T_n)$ that  is
uniformly          distributed          over          the          set
$[\![1,3]\!]\times\mathcal{T}_{r,s,b}(n)$.  Let   $(r_n,R_n)$  be  the
image      of      $(z_n,T_n)$      by      the      bijection      of
Lemma~\ref{bij:decomposition}. Let  $k_n\in [\![0,9]\!]$ be  such that
$R_n\in \mathcal R^{k_n}(n)$. We have $r_n\in [\![1,3]\!]$ if
$k_n=0$  (i.e. $T_n$  is a  square) and  $r_n\in [\![1,2]\!]$  otherwise
(i.e. $T_n$  is hexagonal). In  what follows,  we  need  some rather
heavy additional notation, and the cases $k_n=0$ and $k_n>0$  have
to be treated  slightly differently, even though  the general approach
is parallel between both.

If    $k_n=0$,    let   $(\rho^1_n,\ldots,\rho^4_n)\in    \mathbb    N^{4}$,
$(\tau^1_n,\ldots,\tau^4_n)\in           (\mathbb           N^*)^{4}$,
$(\sigma^1_n,\sigma^2_n)\in               \mathbb              N^{2}$,
$((F^{1}_{n},\ell^{1}_{n}), \ldots ,(F^{4}_{n},\ell^{4}_{n}))\in \mathcal F^{\rho^1_n}_{\tau^1_n}\times
\cdots       \times      \mathcal       F^{\rho^4_n}_{\tau^4_n}$      and
\smash{$(M^{1}_{n},               M^{2}_{n})\in               \mathcal
  M_{\sigma^1_n}^{\gamma^1_n}\times                           \mathcal
  M_{\sigma^2_n}^{\gamma^2_n}$}          be         such          that
$R_n=((F^{1}_{n},\ell^{1}_{n}),   \ldots   ,(F^{4}_{n},\ell^{4}_{n}),M^{1}_{n},   M^{2}_{n})$   (see
Definition~\ref{def:defofR0}).      If      $k_n\neq      0$,      let
$(\rho^1_n,\ldots,\rho^6_n)\in                \mathbb                N^{6}$,
$(\tau^1_n,\ldots,\tau^6_n)\in           (\mathbb           N^*)^{6}$,
$(\gamma^1_n,\gamma^2_n,\gamma^3_n)\in         \mathbb         Z^{3}$,
$(\sigma^1_n,\sigma^2_n,\sigma^3_n)\in         \mathbb         N^{3}$,
$((F^{1}_{n},\ell^{1}_{n}), \ldots ,(F^{6}_{n},\ell^{6}_{n}))\in \mathcal F^{\rho^1_n}_{\tau^1_n}\times
\cdots         \times          \mathcal         F^{\rho^6_n}_{\tau^6_n}$,
$(M^{1}_{n},M^{2}_{n},              M^{3}_{n})\in             \mathcal
M_{\sigma^1_n}^{\gamma^1_n}\times \mathcal M_{\sigma^2_n}^{\gamma^2_n}
\times    \mathcal   M_{\sigma^3_n}^{\gamma^3_n}$    be   such    that
$R_n=((F^{1}_{n},\ell^{1}_{n}),  \ldots  ,(F^{6}_{n},\ell^{6}_{n}),M^{1}_{n},M^{2}_{n},  M^{3}_{n})$
(see Definition~\ref{def:defofR1to9}).

We define $t(k)$, for $k\in [\![0,9]\!]$, such that $t(0)=2$ and
$t(k)=3$ if $k\in [\![1,9]\!]$. For convenience again, we write $t_n$
for $t(k_n)$.  

When $k_n=0$, let $\gamma^1_n=\gamma^2_n=0$; for $k_n\in [\![0,9]\!]$ and
$i\in [\![t_n+1,2t_n]\!]$, let $\gamma^i_n=-\gamma^{i-t_n}_n$ and
$\sigma^i_n=\sigma^{i-t_n}_n$.

We  often  denote simply by $x$  the vector $(x^1,\ldots,x^{2t})$;
in particular, $\rho_n$, $\tau_n$,  $\gamma_n$, $\sigma_n$
denote
the        families        $(\rho^i_n)_{1\leq        i\leq        2t_n}$,
$(\tau^i_n)_{1\leq  i\leq 2t_n}$,  $(\gamma^i_n)_{1\leq i\leq  2t_n}$,
$(\sigma^i_n)_{1\leq  i\leq  2t_n}$,
respectively.
For $k\in [\![1,9]\!]$, let
$(c^1(k),\ldots,c^6(k))\in \{0,1\}^6$ denote the constants given by
line $k$ of Table~\ref{tab:gamma}. Moreover, let
$c^1(0)=\cdots=c^4(0)=0$.  Let $c(k)=(c^1(k),\ldots,c^{t(k)}(k))$.
For convenience, we write $c_n=c(k_n)$,
i.e. $(c^1_n,\ldots,c^{t_n}_n)= (c^1(k_n),\ldots,c^{t_n}(k_n))$.

With  these notations,  by Definitions~\ref{def:defofR0}
and~\ref{def:defofR1to9}, we have the following equality:
\begin{equation}
  \label{eq:tau}
  \tau_n=2\sigma_n+\gamma_n+c_n+1.
  \end{equation}
  Conditionally on the  vector $(k_n, \rho_n, \tau_n,\gamma_n,\sigma_n)$,
  the                 forests                and                 paths
  $F^{1}_{n},  \ldots  ,F^{2t_n}_{n},M^{1}_{n}, M^{2}_{n},M^{t_n}_{n}$  are
  independent and:

\begin{itemize}
\item for every $i \in [\![1,2t_n]\!]$, the well-labeled forest $(F^i_n,\ell^i_n)$
 is uniformly
distributed  over the set $\mathcal{F}_{\tau^i_n}^{\rho^i_n}$,
\item for  every $i  \in [\![1,t_n]\!]$, the  Motzkin path  $M^i_n$ is
  uniformly          distributed         over          the         set
  $\mathcal{M}_{\sigma^i_n}^{\gamma^i_n}$.
\end{itemize}

For every $n>0$, we define the renormalized version
$\rho_{(n)},\gamma_{(n)},\sigma_{(n)}$ by letting
$\rho_{(n)}=\frac{\rho_n}{n}$,
$\gamma_{(n)}=(\frac{9}{8n})^{1/4}\gamma_n$ and
$\sigma_{(n)}=\frac{\sigma_n}{\sqrt{2n}}$.

For $k  \in \{0,\ldots,9\}$,  we repeatedly use  two vector  spaces in
what           follows,           a          ``small           space''
$(\mathbb R_{+})^{2t(k)-1}  \times \mathbb R^{t(k)-2}  \times (\mathbb
R_{+})^{t(k)}$          and          a          ``big          space''
$(\mathbb  R_{+})^{2t(k)}  \times  \mathbb R^{2t(k)}  \times  (\mathbb
R_{+})^{2t(k)}$,  and use  the  terms ``small''  and  ``big'' in  what
follows as shortcuts for these spaces. The small space can be seen as
a subspace of the big one  by imposing the following relations between
coordinates      in     the      big      space.     Every      triple
$(\rho,\gamma,\sigma)  \in  (\mathbb R_{+})^{2t(k)-1}  \times  \mathbb
R^{t(k)-2}  \times  (\mathbb R_{+})^{t(k)}$  can  be  extended into  a
triple                                                              in
$(\mathbb  R_{+})^{2t(k)}  \times  \mathbb R^{2t(k)}  \times  (\mathbb
R_{+})^{2t(k)}$ by letting:
\begin{itemize}
\item $\rho^{2t(k)}=1-\sum_{i=1}^{2t(k)-1}\rho^i$
\item for  $i\in [\![2,2t(k)]\!]$,  $\gamma^{i}=(-1)^{i-1}\gamma^1$,
\item for $i\in [\![t(k)+1,2t(k)]\!]$, $\sigma^{i}=\sigma^{i-t(k)}$,
\end{itemize}
The  idea  is  that  combinatorial  constraints coming  from  our
previous  constructions will  impose  these relations  on the  scaling
limits:  the natural  limit  takes place  in the  big  space, but  the
degrees of  freedom correspond to  the coordinates in the  small space
and  so will    the  integration   variables  in  what  follows.  As  a
particularly useful  notation, we  several  times extend functions
from  the   small  space  to   the  big  space,  more   precisely:  if
$(\rho,  \gamma,  \sigma)\in  (\mathbb  R_{+})^{2t(k)-1}  \times  \mathbb
R^{t(k)-2}  \times (\mathbb  R_{+})^{t(k)}$ is  a point  in the  small
space                                                              and
$f : (\mathbb R_{+})^{2t(k)}  \times \mathbb R^{2t(k)} \times (\mathbb
R_{+})^{2t(k)}    \to    \mathbb    R$,     we        denote    by
$f(\rho, \gamma, \sigma)$ the value of $f$  at the point in the big space
obtained by computing the extra coordinates as above.

  Now, define  a probability
measure             $\mu$            on             the            set
$\mathcal   L=\bigcup_{k\in   [\![0,9]\!]}\left(\{k\}\times   (\mathbb{R}_{+})^{2t(k)}
  \times \mathbb{R}^{2t(k)} \times (\mathbb{R}_{+})^{2t(k)}\right)$ as
follows:  for  every  non-negative measurable  function  $\varphi$  on
$\mathcal L$, let
\begin{align*}
\mu(\varphi)=&\frac{1}{\Upsilon}\sum_{k=1}^{9}\int
_{\mathcal      X}      
 \left(  \mathbbm      1_{\rho^{6}\geq      0} \times \varphi(k,\rho,\gamma,\sigma)\times
               \right.
  \\ &\left.\prod_{i=1}^{6}\left(\frac{\sigma^i}{\sqrt{2}\,             \rho^i}\times
  \frac{2}{\sqrt{6\pi    \rho^i}}\times    e^{\frac{-(\sigma^i)^2}{3\rho^i}}
  \times
     \left(\frac{4}{3}\right)^{c^i(k)+1}\right)\times
  \prod_{i=1}^{3}p_{\sigma^i}(\gamma^i)\right   )\mathrm  d   X
\end{align*}
  where  like above  $(c^1(k),\ldots,c^6(k))$ is  given by  line $k$  of
Table~\ref{tab:gamma},    where $\mathrm  d  X$  is the  Lebesgue
measure on
$$\mathcal X=(\mathbb R_+)^{5} \times
\mathbb    R    \times    (\mathbb    R_{+})^{3},$$
and where   the
renormalization constant
\begin{align*}
\label{eq:defofUpsilon}
\Upsilon=\sum_{k=1}^{9}\int
_{\mathcal      X}      
  \left(  \mathbbm      1_{\rho^{6}\geq      0} 
  \prod_{i=1}^{6}\left(\frac{\sigma^i}{\sqrt{2}\,             \rho^i}\times
  \frac{2}{\sqrt{6\pi    \rho^i}}\times    e^{\frac{-(\sigma^i)^2}{3\rho^i}}
  \times
      \left(\frac{4}{3}\right)^{c^i(k)+1}\right)\times
  \prod_{i=1}^{3}p_{\sigma^i}(\gamma^i)\right   )\mathrm  d   X
\end{align*}
is  chosen so  that  $\mu$ has  total  mass $1$.  Note  that $\mu$  is
supported on a subspace of the big  space. The goal of this section is
to prove the following convergence result:

\begin{lemma}
\label{structureofu}
The      law       $\mu_n$      of      the       random      variable
$(k_n, \rho_{(n)},  \gamma_{(n)}, \sigma_{(n)})$ converges  weakly toward
the probability measure $\mu$.
\end{lemma}

We  say that  a random,  infinite  Motzkin path  $(M_i)_{i\geq 0}$  is
\emph{uniform} if its steps  are independent and uniformly distributed
in $\{-1,0,1\}$ (which means that for every $\sigma>0$, the restricted
path $(M_i)_{0\leq i\leq\sigma}$ is uniformly distributed among Motzkin
paths of length  $\sigma$). There is a relation  between Motzkin paths
with prescribed final value and uniform Motzkin paths:
\begin{equation}
  \label{eq:motzkin}
  |\mathcal M_{\sigma}^{\gamma}|=3^{\sigma}\mathbb{P}(M_{\sigma}=\gamma).\end{equation}
Consider $n \geq 1$ and $k\in [\![0,9]\!]$. 
Let $\mathcal C_n^k \subseteq
\mathbb N^{2t(k)}\times(\mathbb N^*)^{2t(k)}\times\mathbb Z^{2t(k)}\times
\mathbb N^{2t(k)}$ be the set of t-uples
$(\rho,\tau,\gamma,\sigma)$ 
satisfying the following conditions:
\begin{align}
\label{eq:condition1}
  \text{when $k=0$:} \quad & \gamma^1=\gamma^2=0 ;\\
\label{eq:condition2}
  \text{when $k\neq 0$:} \quad & \gamma^1+\gamma^2, \gamma^2+\gamma^3,
                                 \text{~are given by
                                 line $k$ of Table~\ref{tab:gamma}}; \\
\label{eq:condition3}
  \text{for      $i\in       [\![t(k)+1,2t(k)]\!]$:}      \quad      &
                                                                       \gamma^i=-\gamma^{i-t(k)} \text{~and~}
                                                                       \sigma^i=\sigma^{i-t(k)};
  \\
\label{eq:condition4}
                           &  n=\rho^1+\cdots  +\rho^{2t(k)}+\sigma^1+\cdots
                             +\sigma^{t(k)}+t(k)-1 \\
\label{eq:condition5}
                           & \tau=2\sigma +\gamma+c(k)+1 \\
\label{eq:condition6}
  \text{for $i\in [\![1,2t(k)]\!]$:} \quad & |\gamma^i|\leq \sigma^i.
\end{align}

For 
$(k,\rho,\tau,\gamma,\sigma)\in [\![0,9]\!]\times \mathcal C_n^k$, we define:
$$\mathbb{P}_n(k,\rho,\tau,\gamma,\sigma)=\mathbb{P}\big(\left(k_n,\rho_n,\tau_n,\gamma_n,\sigma_n)=(k,\rho,\tau,\gamma,\sigma\right)\big)$$

Then, by          Lemmas~\ref{mnforest}         and~\ref{bij:decomposition},
Definitions~\ref{def:defofR0}                and~\ref{def:defofR1to9},
Equations~(\ref{eq:tau})  and~(\ref{eq:motzkin}),  we   have:
\begin{equation}
\label{eq:proba}
\begin{split}
  \mathbb{P}_n(k,\rho,\tau,\gamma,\sigma)
  &=\frac{2+\mathbbm{1}_{k=0}}{3|\mathcal{T}_{r,s,b}(n)|}\prod _{i=1}^{2t(k)}|\mathcal F^{\rho^i}_{\tau^i}| \prod
  _{i=1}^{t(k)}|\mathcal M^{\gamma^i}_{\sigma^i}|
  \\
  &=\frac{2+\mathbbm{1}_{k=0}}{3|\mathcal{T}_{r,s,b}(n)|}\times
  \prod_{i=1}^{2t(k)}\frac{\tau^i}{4\rho^i+\tau^i}\binom{4\rho^i+\tau^i}{\rho^i}\times
  \prod_{i=1}^{t(k)} 3^{\sigma^i}\mathbb{P}(M_{\sigma^i}=\gamma^i)
\end{split}
\end{equation} 
where $(M_i)_{i\geq 0}$  is a uniform Motzkin path. To  get a grasp on
this quantity, we now collect a few combinatorial results.

\begin{lemma}
\label{binomial} For $a,b\in \mathbb N$, we have
$$ \begin{pmatrix}
  4a+b\\
  a
\end{pmatrix}=
\begin{pmatrix}
  4a\\
  a
\end{pmatrix}
\times               \left(\frac{4}{3}\right)^{b}               \times
\frac{\prod_{p=1}^{b}\left(1+\frac{p}{4a}\right)}{\prod_{p=1}^{b}\left(1+\frac{p}{3a}\right)}.$$
\end{lemma}

\begin{proof}
A straightforward computation shows that
\begin{equation*}
\begin{split}
\frac{\begin{pmatrix}
4a+b\\ 
a
\end{pmatrix}}{\begin{pmatrix}
4a\\ 
a
\end{pmatrix}}
&=\frac{(4a+b)!}{(a)!\,(3a+b)!}\times
\frac{(a)!\,(3a)!}{(4a)!}=\frac{(4a+b)!}{(4a)!}\times \frac{(3a)!}{(3a+b)!}=\frac{\prod_{p=1}^{b}(4a+p)}{\prod_{p=1}^{b}(3a+p)}\\
&=\frac{(4a)^{b}\prod_{p=1}^{b}\left(1+\frac{p}{4a}\right)}{(3a)^{b}\prod_{p=1}^{b}\left(1+\frac{p}{3a}\right)}=\left(\frac{4}{3}\right)^{b} \frac{\prod_{p=1}^{b}\left(1+\frac{p}{4a}\right)}{\prod_{p=1}^{b}\left(1+\frac{p}{3a}\right)}
\qedhere
\end{split}
\end{equation*}
\end{proof}

By Lemma~\ref{binomial},  the
binomial term in (\ref{eq:proba}) can be rewritten:
\begin{equation}
\label{eq:probabinom}
\begin{split}
\binom{4\rho^i+\tau^i}{\rho^i}
&=\binom{4\rho^i+2\sigma^i+\gamma^i+c^i(k)+1}{\rho^i}\\
&=\binom{4\rho^i+2\sigma^i+\gamma^i}{\rho^i}\prod_{p=1}^{c^i(k)+1}\frac{4\rho^i+2\sigma^i+\gamma^i+p}{3\rho^i+2\sigma^i+\gamma^i+p}\\
&=\binom{4\rho^i}{\rho^i}\left(\frac{4}{3}\right)^{2\sigma^i+\gamma^i}\frac{\prod_{p=1}^{2\sigma^i+\gamma^i}\left(1+\frac{p}{4\rho^i}\right)}{\prod_{p=1}^{2\sigma^i+\gamma^i}\left(1+\frac{p}{3\rho^i}\right)}\prod_{p=1}^{c^i(k)+1}\frac{4\rho^i+2\sigma^i+\gamma^i+p}{3\rho^i+2\sigma^i+\gamma^i+p}.
\end{split}
\end{equation}

For $x \in\mathbb R$, let $\left  \lfloor x \right \rfloor$ denote the
largest integer that is bounded above by $x$.

\begin{lemma}
\label{stirling}
For
$(\rho,\gamma,\sigma)       \in      \mathbb{R}^*_+\times\mathbb{R}\times
\mathbb{R}^*_+$,       as       $n$       goes       to       infinity
\[\frac{\prod_{p=1}^{2 \left \lfloor  \sqrt{2n}\sigma \right \rfloor +
      \left     \lfloor      (8n/9)^{1/4}\gamma     \right     \rfloor
    }\left(1+\frac{p}{4\left         \lfloor         n\rho         \right
        \rfloor}\right)}{\prod_{p=1}^{2 \left  \lfloor \sqrt{2n}\sigma
      \right \rfloor + \left \lfloor (8n/9)^{1/4}\gamma \right \rfloor
    }\left(1+\frac{p}{3\left  \lfloor n\rho  \right \rfloor}\right)}  \to
  e^{\frac{-\sigma^2}{3\rho}}.\]
\end{lemma}

\begin{proof}

  For $n\geq 1$, let $a_n$ denote  the left-hand term in the statement
  of the lemma. By Lemma~\ref{binomial}, we have:
$$a_n=\begin{pmatrix}
4\left \lfloor n\rho \right \rfloor+2\left \lfloor \sqrt{2n}\sigma \right \rfloor+\left \lfloor (8n/9)^{1/4}\gamma \right \rfloor \\ 
\left \lfloor n\rho \right \rfloor
\end{pmatrix} \times \left(\frac{3}{4}\right)^{2\left \lfloor \sqrt{2n}\sigma \right \rfloor+\left \lfloor (8n/9)^{1/4}\gamma \right \rfloor}/\begin{pmatrix}
4\left \lfloor n\rho \right \rfloor\\ 
\left \lfloor n\rho \right \rfloor
\end{pmatrix}$$
$$
=\frac{(4\left \lfloor n\rho \right \rfloor+2\left \lfloor \sqrt{2n}\sigma \right \rfloor+\left \lfloor (8n/9)^{1/4}\gamma \right \rfloor)!\,(3\left \lfloor n\rho \right \rfloor)!}{(3\left \lfloor n\rho \right \rfloor+2\left \lfloor \sqrt{2n}\sigma \right \rfloor+\left \lfloor (8n/9)^{1/4}\gamma \right \rfloor)!\,(4\left \lfloor n\rho \right \rfloor)!}\times \left(\frac{3}{4}\right)^{2\left \lfloor \sqrt{2n}\sigma \right \rfloor+\left \lfloor (8n/9)^{1/4}\gamma \right \rfloor}.$$
Using the Stirling formula, we obtain: 
\begin{align*}
a_n \sim &\frac{(4\left \lfloor n\rho \right \rfloor+2\left
      \lfloor \sqrt{2n}\sigma \right \rfloor+\left \lfloor
      (8n/9)^{1/4}\gamma \right \rfloor)^{4\left \lfloor n\rho \right
      \rfloor+2\left \lfloor \sqrt{2n}\sigma \right \rfloor+\left
        \lfloor (8n/9)^{1/4}\gamma \right \rfloor}(3\left \lfloor n\rho
    \right \rfloor)^{3\left \lfloor n\rho \right \rfloor}}{(3\left
      \lfloor n\rho \right \rfloor+2\left \lfloor \sqrt{2n}\sigma \right
    \rfloor+\left \lfloor (8n/9)^{1/4}\gamma \right \rfloor)^{3\left
        \lfloor n\rho \right \rfloor+2\left \lfloor \sqrt{2n}\sigma
      \right \rfloor+\left \lfloor (8n/9)^{1/4}\gamma \right
      \rfloor}(4\left \lfloor n\rho \right \rfloor)^{4\left \lfloor n\rho
      \right \rfloor}}\\
&  \times \left(\frac{3}{4}\right)^{2\left \lfloor \sqrt{2n}\sigma \right \rfloor+\left \lfloor (8n/9)^{1/4}\gamma \right \rfloor} 
\\ \sim &\frac{\left(\frac{4\left \lfloor n\rho \right
  \rfloor+2\left \lfloor \sqrt{2n}\sigma \right \rfloor+\left \lfloor
  (8n/9)^{1/4}\gamma \right \rfloor}{4\left \lfloor n\rho \right
  \rfloor}\right)^{4\left \lfloor n\rho \right
  \rfloor}}{\left(\frac{3\left \lfloor n\rho \right \rfloor+2\left
  \lfloor \sqrt{2n}\sigma \right \rfloor+\left \lfloor
  (8n/9)^{1/4}\gamma \right \rfloor}{3\left \lfloor n\rho \right
  \rfloor}\right)^{3\left \lfloor n\rho \right \rfloor}}\\ &\times \frac{(4\left \lfloor n\rho \right \rfloor+2\left \lfloor \sqrt{2n}\sigma \right \rfloor+\left \lfloor (8n/9)^{1/4}\gamma \right \rfloor)^{2\left \lfloor \sqrt{2n}\sigma \right \rfloor+\left \lfloor (8n/9)^{1/4}\gamma \right \rfloor}}{\left(3\left \lfloor n\rho \right \rfloor+2\left \lfloor \sqrt{2n}\sigma \right \rfloor+\left \lfloor (8n/9)^{1/4}\gamma \right \rfloor\right)^{2\left \lfloor \sqrt{2n}\sigma \right \rfloor+\left \lfloor (8n/9)^{1/4}\gamma \right \rfloor}}\\
&\times \left(\frac{3}{4}\right)^{2\left \lfloor \sqrt{2n}\sigma \right
  \rfloor+\left \lfloor (8n/9)^{1/4}\gamma \right \rfloor}
\\ \sim &\frac{\left(1+\frac{2\left \lfloor \sqrt{2n}\sigma \right
          \rfloor+\left \lfloor (8n/9)^{1/4}\gamma \right
          \rfloor}{4\left \lfloor n\rho \right \rfloor}\right)^{4\left
          \lfloor n\rho \right \rfloor}}{\left(1+\frac{2\left \lfloor
          \sqrt{2n}\sigma \right \rfloor+\left \lfloor
          (8n/9)^{1/4}\gamma \right \rfloor}{3\left \lfloor n\rho \right
          \rfloor}\right)^{3\left \lfloor n\rho \right \rfloor}}\times
          \frac{\left(1+\frac{2\left \lfloor \sqrt{2n}\sigma \right
          \rfloor+\left \lfloor (8n/9)^{1/4}\gamma \right
          \rfloor}{4\left \lfloor n\rho \right \rfloor}\right)^{2\left
          \lfloor \sqrt{2n}\sigma \right \rfloor+\left \lfloor
          (8n/9)^{1/4}\gamma \right \rfloor}}{\left(1+\frac{2\left
          \lfloor \sqrt{2n}\sigma \right \rfloor+\left \lfloor
          (8n/9)^{1/4}\gamma \right \rfloor}{3[n\rho]}\right)^{2\left
          \lfloor \sqrt{2n}\sigma \right \rfloor+\left \lfloor
          (8n/9)^{1/4}\gamma \right \rfloor}}
\end{align*}

We have the following estimates as $n\rightarrow \infty$:
\begin{equation}
  \label{eb}
  \left(1+\frac{2\left \lfloor \sqrt{2n}\sigma \right \rfloor+\left \lfloor (8n/9)^{1/4}\gamma \right \rfloor}{4\left \lfloor n\rho \right \rfloor}\right)^{4\left \lfloor n\rho \right \rfloor}\sim e^{2\left \lfloor \sqrt{2n}\sigma \right \rfloor+\left \lfloor (8n/9)^{1/4}\gamma \right \rfloor-\sigma^2/\rho},
\end{equation}
\begin{equation}
  \label{ebb}
  \left(1+\frac{2\left \lfloor \sqrt{2n}\sigma \right \rfloor+\left \lfloor (8n/9)^{1/4}\gamma \right \rfloor}{3\left \lfloor n\rho \right \rfloor}\right)^{3\left \lfloor n\rho \right \rfloor} \sim e^{2\left \lfloor \sqrt{2n}\sigma \right \rfloor+\left \lfloor (8n/9)^{1/4}\gamma \right \rfloor-\frac{4\sigma^2}{3\rho}},
\end{equation}
    \begin{equation}
    \label{ebbb}
    \left(1+\frac{2\left \lfloor \sqrt{2n}\sigma \right \rfloor+\left \lfloor (8n/9)^{1/4}\gamma \right \rfloor}{4\left \lfloor n\rho \right \rfloor}\right)^{2\left \lfloor \sqrt{2n}\sigma \right \rfloor+\left \lfloor (8n/9)^{1/4}\gamma \right \rfloor} \rightarrow e^{2\sigma^2/\rho},
    \end{equation}
    \begin{equation}
    \label{ebbbb}
    \left(1+\frac{2\left \lfloor \sqrt{2n}\sigma \right \rfloor+\left \lfloor (8n/9)^{1/4}\gamma \right \rfloor}{3\left \lfloor n\rho \right \rfloor}\right)^{2\left \lfloor \sqrt{2n}\sigma \right \rfloor+\left \lfloor (8n/9)^{1/4}\gamma \right \rfloor}\rightarrow e^{\frac{8\sigma^2}{3\rho}}.
    \end{equation}
Combining these completes the proof.
\end{proof}

We are now ready to prove Lemma~\ref{structureofu}:

\begin{proof}[Proof of Lemma~\ref{structureofu}]
  Let  $\varphi$  be   a  bounded  continuous  function   on  the  set
  $\mathcal                L$                and                define
  $\mathbb{E}_n(\varphi)=\mathbb{E}\left(\varphi(k_n,        \rho_{(n)},
    \gamma_{(n)},  \sigma_{(n)})\right)$.   We  need  to   prove  that
  $\mathbb{E}_n(\varphi)$ converges toward  $\mu(\varphi)$ as $n$ goes
  to infinity.
  
  Let $n\in \mathbb{N}$. For a given value of $k$, we identify
  $(\rho,\gamma, \sigma) \in \mathbb N^{2t(k)-1}\times\mathbb
  Z^{t(k)-2}\times \mathbb N^{t(k)}$ with an element
  $\mathrm p(\rho,\gamma, \sigma)=(\rho,\tau,\gamma,\sigma)$ of
  $(\mathbb N^{2t(k)-1} \times \mathbb Z)\times(\mathbb
  N^*)^{2t(k)}\times  \mathbb  Z^{2t(k)}\times \mathbb  N^{2t(k)}$  by
  setting the missing coordinates so  that they satisfy the conditions
  (\ref{eq:condition1}) to (\ref{eq:condition5}).
  Note that  $\rho^{2t(k)}$ depends not  only on $n$ and  the $\rho^i$
  for $i  \leq 2t(k)-1$  but also  on the  $\sigma^i$. Note  also that
  $\mathrm p(\rho,\gamma,\sigma)$  is an  element of  $\mathcal C^k_n$
  provided that  the conditions lead  to $\rho^{2t(k)}\geq 0$  and for
  any  $i\in [\![1,2t(k)]\!]$  we have  $|\gamma^i|\leq \sigma^i$.  By
  Equations~\eqref{eq:proba} and \eqref{eq:probabinom} we have
  \begin{equation*}
  \begin{split}
    \mathbb{E}_n(\varphi)&=
\sum_{k=0}^{9}
\sum_{(\rho,\tau,\gamma,\sigma)\in \mathcal C_n^k}
\left( \mathbb{P}_n(k,\rho,\tau,\gamma,\sigma)
\, \varphi\left(k,\frac{\rho}{n},\left (\frac{9}{8n}\right)^{1/4} \gamma,\frac{\sigma}{\sqrt{2n}}\right)\right)\\
&=
\sum_{k=0}^{9}\frac{2+\mathbbm{1}_{k=0}}{3|\mathcal{T}_{r,s,b}(n)|}
\sum_{(\rho,\tau,\gamma,\sigma)\in \mathcal
  C_n^k}\left (
\mathrm f(k,\rho,\gamma,
\sigma)
\times \mathrm g(k,\gamma, \sigma)
\times \mathrm h(k,\rho,\gamma,\sigma)
\right)
\end{split}
\end{equation*}
where we introduced the functions
\begin{align*}\mathrm f(k,\rho,\gamma,
\sigma)=\prod_{i=1}^{2t(k)}&\left(\left(\frac{2\sigma^i+\gamma^i+c^i(k)+1}{4\rho^i+2\sigma^i+\gamma^i+c^i(k)+1}\right)
  \binom{4\rho^i}{\rho^i}\left(\frac{4}{3}\right)^{2\sigma^i+\gamma^i}\right.\\
                                   &\left.\frac{\prod_{p=1}^{2\sigma^i+\gamma^i}\left(1+\frac{p}{4\rho^i}\right)}{\prod_{p=1}^{2\sigma^i+\gamma^i}\left(1+\frac{p}{3\rho^i}\right)}\prod_{p=1}^{c^i(k)+1}\frac{4\rho^i+2\sigma^i+\gamma^i+p}{3\rho^i+2\sigma^i+\gamma^i+p}\right),
\end{align*}
$$\mathrm g(k,\gamma,
\sigma)=\prod_{i=1}^{t(k)}3^{\sigma^i}\mathbb{P}(M_{\sigma^i}=\gamma^i),$$
$$\mathrm h(k,\rho,\gamma,\sigma)=\varphi\left(k,\frac{\rho}{n},\left
    (\frac{9}{8n}\right)^{1/4}\gamma,\frac{\sigma}{\sqrt{2n}}\right).$$

In order  to derive  the asymptotic behavior  of the  discrete objects
above, we are going to compare discrete sums to integrals. To do that,
we             need      some      more       notation.

For $k\in[\![0,9]\!]$, $n\geq 0$ and
$(\rho,\gamma, \sigma) \in (\mathbb R_+)^{2t(k)-1}\times\mathbb
R^{t(k)-2}\times (\mathbb R_+)^{t(k)}$, we define
$(\lfloor \rho \rfloor,\lfloor \gamma \rfloor,\lfloor \sigma \rfloor
)\in (\mathbb N^{2t(k)-1} \times \mathbb Z)\times \mathbb
Z^{2t(k)}\times \mathbb N^{2t(k)}$  by the following.
For every
$i \in \{1, \ldots, 2t(k)-1 \}$, let
$\lfloor \rho \rfloor ^i = \lfloor \rho^i \rfloor$.
If $k\neq 0$, let $\lfloor \gamma \rfloor ^1 = \lfloor \gamma^1 \rfloor$.
For every
$i \in \{1, \ldots, t(k) \}$, let $\lfloor \sigma \rfloor ^i = \lfloor \sigma^i \rfloor$.
Then we choose $\lfloor \rho \rfloor^{2t(k)}$, $\lfloor \gamma \rfloor ^{t(k)-1}$,
\ldots,$\lfloor \gamma \rfloor ^{2t(k)}$,
$\lfloor \sigma \rfloor ^{t(k)+1}$,
\ldots,$\lfloor \sigma \rfloor ^{2t(k)}$ so that
$\lfloor \rho\rfloor$, $\lfloor \gamma\rfloor$,
$\lfloor \sigma\rfloor$ satisfies the relation~ (\ref{eq:condition1}),
(\ref{eq:condition2}), (\ref{eq:condition3}), and
\eqref{eq:condition4}.

Note  that the set of
all    preimages   of    a    given   joint    integral   value    for
$(\lfloor \rho \rfloor, \lfloor  \sigma \rfloor, \lfloor \gamma \rfloor)$
is  a  unit cube  in  the  ``small space''.  Note  as  well that  this
definition  does   not  coincide   with  first  computing   the  extra
coordinates as before and then taking integral parts coordinatewise on
the  big space:  we  choose  this particular  definition  so that  the
constraints  on  coordinates match  better  between  the discrete  and
continuous versions.

Writing the sum  over $\mathcal C_n^k$ in the form  of an integral, we
have:
\begin{align*}
  \mathbb{E}_n(\varphi)=&\sum_{k=0}^{9}\frac{2+\mathbbm{1}_{k=0}}{3|\mathcal{T}_{r,s,b}(n)|}
  \int _{{  X}^k}
  \left(
    \mathbbm {1}_{\mathcal E^k_n}(\lfloor
    \rho  \rfloor, \lfloor  \gamma  \rfloor, \lfloor  \sigma \rfloor)
    \times 
    \mathrm  f(k, \lfloor  \rho    \rfloor, \lfloor  \gamma
      \rfloor,    \lfloor  \sigma    \rfloor)\right.
  \\ &\left. \times
       \mathrm  g(k,  \lfloor  \gamma   \rfloor,    \lfloor
       \sigma       \rfloor)
       \times   \mathrm    h(k,\lfloor  \rho
       \rfloor,\lfloor  \gamma   \rfloor,\lfloor  \sigma   \rfloor)\right)
       \mathrm d {X}^k
\end{align*}
  where  $\mathrm d {X}^k$ is the
Lebesgue                           measure                          on
${\mathcal X}^k  = (\mathbb{R}_+)^{2t(k)-1}  \times \mathbb
R^{t(k)-2} \times (\mathbb R_{+})^{t(k)}$ and

\begin{align*}
  {\mathcal E}^k_n = &
  \left\{(\rho,\gamma,\sigma)\in
  (\mathbb{R}_+)^{2t(k)-1}  \times  \mathbb R^{t(k)-2}  \times  (\mathbb
    R_{+})^{t(k)}:\right.
  \\
  &\left. \lfloor    \rho     \rfloor    ^{2t(k)}\geq     0
  \,\,\text{and}\,\,  \forall i\in  [\![1,t(k)]\!],\,\, |\gamma^i|\leq
  \sigma^i  \right\}.\end{align*}

  We  now do  a change  of variables  by setting
$\rho'=  \frac{\rho}{n}$  ,  $\gamma'=  (\frac{9}{8n})^{\frac{1}{4}}\gamma$,
$\sigma'  =   \frac{\sigma}{\sqrt{2n}}$  (but  still  write   the  new
variables as  $(\rho, \gamma, \sigma)$  below for simpler  notation). The
change of variables is linear and acts like a multiplication by $n$ on
$\rho   \in    (\mathbb{R}_+)^{2t(k)-1}$,     by    $(8n/9)^{1/4}$    on
$\gamma   \in    (\mathbb{R})^{t(k)-2}$   and   by    $\sqrt{2n}$   on
$\sigma  \in  (\mathbb{R}_+)^{t(k)}$,  so  its Jacobian  is  equal  to
$n^{2t(k)-1} ({8n}/{9})^{(t(k)-2)/4}
(\sqrt{2n})^{t(k)}$.
Therefore we
obtain:
\begin{align*}
\mathbb{E}_n(\varphi)=&
\sum_{k=0}^{9}\frac{2+\mathbbm{1}_{k=0}}{3|\mathcal{T}_{r,s,b}(n)|}
  \int _{{\mathcal  X}^k} \left(
  \mathbbm 1_{\mathcal E^k_n}\left(\lfloor
  n\rho \rfloor,\left \lfloor (8n/9)^{\frac{1}{4}}\gamma \right \rfloor, \left \lfloor \sqrt{2n}\sigma
  \right \rfloor\right)\right.
\\ &
  \left (
n^{2t(k)-1}
  \left (\frac{8n}{9}\right)^{\frac{t(k)-2}{4}}
(\sqrt{2n})^{t(k)}\right )
 \\ &\times \mathrm f\left(k,\lfloor
  n\rho \rfloor,\left \lfloor (8n/9)^{\frac{1}{4}}\gamma \right \rfloor, \left \lfloor \sqrt{2n}\sigma
\right \rfloor\right)
\\ &\times \mathrm g\left(k,\left \lfloor (8n/9)^{\frac{1}{4}}\gamma \right \rfloor,
\left \lfloor \sqrt{2n}\sigma \right \rfloor\right)
 \\ &
\left. \times \mathrm h\left(k,\lfloor
   n\rho \rfloor,\lfloor (8n/9)^{\frac{1}{4}}\gamma \rfloor,\lfloor
      \sqrt{2n}\sigma \rfloor
\right)
  \right) \mathrm d {X}^k.
\end{align*}
  
Note            that,            for            every            $k\in
[\![0,9]\!]$,   due   to   the   way   we   defined   $\lfloor   \cdot
\rfloor$, we have:
\begin{align*}
  &\prod_{i=1}^{2t(k)}
 \left(\frac{4}{3}\right)^{2\left \lfloor \sqrt{2n}\sigma^i
        \right \rfloor+\left \lfloor (8n/9)^{1/4}\gamma^i \right
        \rfloor}
      \prod_{i=1}^{t(k)} 3^{\left \lfloor \sqrt{2n}\sigma^i \right
        \rfloor}
   \\   &=\left(\frac{256}{27}\right)^{\sum_{i=1}^{t(k)}\left \lfloor
          \sqrt{2n}\sigma^i \right \rfloor}
      \\ &=\left(\frac{256}{27}\right)^{
         n-\sum_{i=1}^{2t(k)}\left \lfloor n \rho \right \rfloor^i-(t(k)-1)}.
\end{align*}

   Hence,       we        can       rewrite        $\mathbb
   E_n(\varphi)$                                                             as
   \begin{align*}
     \mathbb{E}_n(\varphi)=\sum_{k=0}^{9}&
                                           \frac{2+\mathbbm{1}_{k=0}}{3|\mathcal{T}_{r,s,b}(n)|}\,n^{\frac{t(k)-3}{2}}\left(\sqrt{2}\right)^{t(k)}\left(\frac{9}{8}\right)^{1/2}\left(\frac{256}{27}\right)^{n-t(k)+1}
                                           \int
   _{{\mathcal   X}^k}\\ &  \Bigg(\mathbbm   1_{\mathcal
     E^k_n}\left(\lfloor        n\rho        \rfloor,\left        \lfloor
       (8n/9)^{\frac{1}{4}}\gamma   \right   \rfloor,  \left   \lfloor
                                     \sqrt{2n}\sigma  \right \rfloor\right)
\\ &
h\left(k,\lfloor
  n\rho \rfloor,\lfloor (8n/9)^{\frac{1}{4}}\gamma \rfloor,\lfloor \sqrt{2n}\sigma \rfloor\right)\\ & \prod_{i=1}^{2t(k)}\left(\left(\sqrt{n}\, \frac{2\left \lfloor
     \sqrt{2n}\sigma^i \right \rfloor+\left \lfloor
     (8n/9)^{1/4}\gamma^i \right \rfloor+c^i(k)+1}{4\left \lfloor n
     \rho^i \right \rfloor+2\left \lfloor \sqrt{2n}\sigma^i \right
     \rfloor+\left \lfloor (8n/9)^{1/4}\gamma^i \right
     \rfloor+c^i(k)+1}\right)\right.\\ &\times \sqrt{n} \left(\frac{27}{256}\right)^{\left
     \lfloor n\rho \right \rfloor^i} \binom{4\left \lfloor n \rho^i
     \right \rfloor}{\left \lfloor n \rho^i \right \rfloor}
     \\ &
\times \left(\frac{\prod_{p=1}^{2\left \lfloor \sqrt{2n}\sigma^i
          \right \rfloor+\left \lfloor (8n/9)^{1/4}\gamma^i \right
          \rfloor}(1+\frac{p}{4\left \lfloor n \rho^i \right
          \rfloor})}{\prod_{p=1}^{2\left \lfloor \sqrt{2n}\sigma^i
          \right \rfloor+\left \lfloor (8n/9)^{1/4}\gamma^i \right
          \rfloor}(1+\frac{p}{3\left \lfloor n \rho^i \right
          \rfloor})}\right)\\ &\left. \times \left(\prod_{p=1}^{c^i(k)+1}\frac{4\left \lfloor n \rho^i \right \rfloor+2\left \lfloor \sqrt{2n}\sigma^i \right \rfloor+\left \lfloor (8n/9)^{1/4}\gamma^i \right \rfloor+p}{3\left \lfloor n \rho^i \right \rfloor+2\left \lfloor \sqrt{2n}\sigma^i \right \rfloor+\left \lfloor (8n/9)^{1/4}\gamma^i \right \rfloor+p}\right)\right)\\
&\left.\times \prod_{i=1}^{t(k)}(8n/9)^{1/4}\mathbb{P}\left(M_{\left
     \lfloor \sqrt{2n}\sigma^i \right \rfloor}=\left \lfloor
                                                                                                                                                                                                                                                                                                                                                                           (8n/9)^{1/4}\gamma^i \right \rfloor\right)\right)
                                                                                                                                                                                                                                                                                                                                                                           \\
                                         & \mathrm   d   X^k 
     \end{align*}

We are now going to use dominated convergence to show that every integral term appearing in $\mathbb{E}_n$ converges. We have the following:
\begin{itemize}
\item
  $         \left          \lfloor         n          \rho         \right
  \rfloor^{2t(k)}=n-\sum_{i=1}^{2t(k)-1}\left   \lfloor  n   \rho  \right
  \rfloor^i-\sum_{i=1}^{t(k)}  \left \lfloor  \sqrt{2n} \sigma  \right
  \rfloor^i-(t(k)-1)$, and therefore
$$\frac{\left           \lfloor          n           \rho          \right
  \rfloor^{2t(k)}}{n}=1-\sum_{i=1}^{2t(k)-1}\frac{\left  \lfloor  n  \rho
  \right  \rfloor^i}{n}-\sum_{i=1}^{t(k)}\frac{\left   \lfloor  \sqrt{2n}
    \sigma\right                        \rfloor^i}{n}-\frac{t(k)-1}{n}\to
1-\sum_{i=1}^{2t(k)-1}\rho^i=\rho^{2t(k)}.$$  On the  other hand,  for every
$i\in           [\![1,           t(k)]\!]$          we           have:
$\mathbbm  1_{\left\{\left|\left \lfloor  (8n/9)^{\frac{1}{4}}\gamma^i
      \right \rfloor\right|\leq \left \lfloor \sqrt{2n}\sigma^i \right
    \rfloor\right\}}    \to    \mathbbm{1}_{\{\sigma^i\geq
  0\}}$, and hence,
$$\mathbbm 1_{\mathcal E^k_n}\left(\lfloor
  n\rho \rfloor,\left \lfloor (8n/9)^{\frac{1}{4}}\gamma \right \rfloor, \left \lfloor \sqrt{2n}\sigma
\right \rfloor\right) \to \mathbbm 1_{\left\{\rho^{2t(k)}\geq 0\right\}}.$$
    \item 
$ h\left(k,\lfloor
  n\rho \rfloor,\lfloor (8n/9)^{\frac{1}{4}}\gamma \rfloor,\lfloor
  \sqrt{2n}\sigma \rfloor\right)=\varphi\left(k,\frac{\left
      \lfloor n\rho \right \rfloor}{n},\frac{\left \lfloor
      (8n/9)^{\frac{1}{4}}\gamma \right \rfloor}{
    (8n/9)^{\frac{1}{4}}}, \frac{\left \lfloor \sqrt{2n}\sigma \right
    \rfloor}{\sqrt{2n}}\right)
\to \varphi\left(k,\rho,\gamma, \sigma\right)$.

    \item 
    By Lemma~\ref{stirling}, we obtain:
     $$\frac{\prod_{p=1}^{2\left \lfloor \sqrt{2n}\sigma^i \right \rfloor+\left \lfloor (8n/9)^{1/4}\gamma^i \right \rfloor}(1+\frac{p}{4\left \lfloor n \rho^i \right \rfloor})}{\prod_{p=1}^{2\left \lfloor \sqrt{2n}\sigma^i \right \rfloor+\left \lfloor (8n/9)^{1/4}\gamma^i \right \rfloor}(1+\frac{p}{3\left \lfloor n \rho^i \right \rfloor})}\longrightarrow e^{\frac{-(\sigma^i)^2}{3\rho^i}}.$$
    
    \item
    
      By Lemma~\ref{sumiid} with $(\eta,h)=\left(\sqrt{\frac{2}{3}},1\right)$ , we
      obtain (with some simple calculus) :
    $$(8n/9)^{1/4}\,\mathbb{P}\left(M_{\left \lfloor \sqrt{2n}\sigma^i \right \rfloor}=\left \lfloor (8n/9)^{1/4}\gamma^i \right \rfloor\right)\to p_{\sigma^i}(\gamma^i).$$

  \item If $\rho^i>0$, then
   $$\sqrt{n}\left(\frac{27}{256}\right)^{\left   \lfloor  n\rho   \right
      \rfloor^i}\binom{4\left  \lfloor  n  \rho^i  \right  \rfloor}{\left
        \lfloor      n     \rho^i      \right     \rfloor}\to
    \frac{2}{\sqrt{6\pi \rho^i}}.$$
    
    \item 
    $\prod_{p=1}^{c^i(k)+1}\frac{4\left \lfloor n \rho^i \right \rfloor+2\left \lfloor \sqrt{2n}\sigma^i\right \rfloor+\left \lfloor (8n/9)^{1/4}\gamma^i \right \rfloor+p}{3\left \lfloor n \rho^i \right \rfloor+2\left \lfloor \sqrt{2n}\sigma^i \right \rfloor+\left \lfloor (8n/9)^{1/4}\gamma^i \right \rfloor+p}\longrightarrow \left(\frac{4}{3}\right)^{c^i(k)+1}$.

    \item 
    
    $\sqrt{n}\, \frac{2\left \lfloor \sqrt{2n}\sigma^i \right \rfloor+\left \lfloor (8n/9)^{1/4}\gamma^i \right \rfloor+c^i(k)+1}{4\left \lfloor n \rho^i \right \rfloor+2\left \lfloor \sqrt{2n}\sigma^i \right \rfloor+\left \lfloor (8n/9)^{1/4}\gamma^i \right \rfloor+c^i(k)+1}\longrightarrow \frac{\sigma^i}{\sqrt{2}\rho^i}$.

\end{itemize}


It remains to prove domination of  the summand, which  follow from
the  following bounds:

\begin{itemize}

    \item  $\left|\varphi\left(k,\frac{\left \lfloor n\rho \right \rfloor}{n},\frac{\left \lfloor (8n/9)^{\frac{1}{4}}\gamma \right \rfloor}{ (8n/9)^{\frac{1}{4}}}, \frac{\left \lfloor \sqrt{2n}\sigma \right \rfloor}{\sqrt{2n}}\right)\right| \leq \left \| \varphi \right \|_{\infty}$.
    
    \item
If $\lfloor n\rho^i\rfloor=0$, then $\sqrt{n\rho^i}<1$. Hence,
$$\sqrt{n\rho^i}\left(\frac{27}{256}\right)^{\lfloor n\rho^i\rfloor}\binom{4\left \lfloor n \rho^i \right \rfloor}{\left \lfloor n \rho^i \right \rfloor}\leq 1.$$
If  on  the other  hand  $\lfloor  n\rho^i\rfloor>0$, by  using  Stirling
formula, there exists  a constant $c$ do not depend  on $n, \rho^i$ such
that:
$$\sqrt{n}\left(\frac{27}{256}\right)^{\lfloor n\rho^i\rfloor}\binom{4\left \lfloor n \rho^i \right \rfloor}{\left \lfloor n \rho^i \right \rfloor}\leq  \frac{c}{\sqrt{ \rho^i}}.$$
Let $C=\max \{1, c\}$. For all $n\geq 1$ and $0<\rho^i< 1$, we obtain:
$$\sqrt{n}\left(\frac{27}{256}\right)^{\lfloor n\rho^i\rfloor}\binom{4\left \lfloor n \rho^i \right \rfloor}{\left \lfloor n \rho^i \right \rfloor}\leq  \frac{C}{\sqrt{ \rho^i}}.$$

\item                                                            Since
  $\left|\left \lfloor  (8n/9)^{1/4}\gamma^i \right \rfloor\right|\leq
  \left  \lfloor \sqrt{2n}\sigma^i  \right \rfloor$,  $c^i(k)\in \{0,1\}$
  and $\left \lfloor \sqrt{2n}\sigma^i \right  \rfloor \geq 1$, we get
  $c^i+1\leq 2\leq  2\left \lfloor \sqrt{2n}\sigma^i  \right \rfloor$.
  By                using                the                inequality
  $\left \lfloor x  \right \rfloor^{-1}\leq 2/x$ for  all $x\geq1$ and
  $|\left \lfloor x \right \rfloor|\leq |x|+1$, then we obtain:
 $$\left|\sqrt{n}\frac{2\left \lfloor \sqrt{2n}\sigma^i\right \rfloor+\left \lfloor (8n/9)^{1/4}\gamma^i\right \rfloor+c^i(k)+1}{4\left \lfloor n\rho^i\right \rfloor+2\left \lfloor \sqrt{2n}\sigma^i \right \rfloor+\left \lfloor (8n/9)^{1/4}\gamma^i \right \rfloor+c^i(k)+1}\right|\leq \frac{5\sqrt{2}\sigma^i}{2\rho^i}.$$

\item 

$\prod_{p=1}^{c^i(k)+1}\frac{4\left \lfloor n \rho^i \right \rfloor+2\left \lfloor \sqrt{2n}\sigma^i \right \rfloor+\left \lfloor (8n/9)^{1/4}\gamma^i \right \rfloor+p}{3\left \lfloor n \rho^i \right \rfloor+2\left \lfloor \sqrt{2n}\sigma^i \right \rfloor+\left \lfloor (8n/9)^{1/4}\gamma^i \right \rfloor+p}\leq (\frac{4}{3})^{c^i(k)+1}$.

\item

By using {Lemma~\ref{sumiid}} with $r=2$, there exists $C_1\mathbb R_+$, such that 
$$(8n/9)^{1/4}\mathbb{P}\left(M_{\left \lfloor \sqrt{2n}\sigma^i \right \rfloor}=\left \lfloor (8n/9)^{1/4}\gamma^i \right \rfloor\right) \leq \frac{C_1}{\sqrt{\sigma^i}}\left(1+\frac{(\gamma^i)^2}{\sigma^i}\right)^{-1}. $$

\item                 For any $p\in \mathbb N$,       we          have
  $
  \frac{1+\frac{p+1}{4\left       \lfloor      n       \rho^i      \right
      \rfloor}}{1+\frac{p+1}{3\left  \lfloor n  \rho^i \right  \rfloor}}\leq\frac{1+\frac{p}{4\left       \lfloor       n      \rho^i       \right
      \rfloor}}{1+\frac{p}{3\left \lfloor  n \rho^i  \right \rfloor}}$
  and therefore, since
  $\left|\left \lfloor  (8n/9)^{1/4}\gamma^i \right \rfloor\right|\leq
  \left  \lfloor \sqrt{2n}\sigma^i  \right \rfloor$
$$
\frac{\prod_{p=1}^{2\left \lfloor \sqrt{2n}\sigma^i \right
    \rfloor+\left \lfloor (8n/9)^{1/4}\gamma^i \right
    \rfloor}(1+\frac{p}{4\left \lfloor n \rho^i \right
    \rfloor})}{\prod_{p=1}^{2\left \lfloor \sqrt{2n}\sigma^i \right
    \rfloor+\left \lfloor (8n/9)^{1/4}\gamma^i \right
    \rfloor}(1+\frac{p}{3\left \lfloor n \rho^i \right \rfloor})}
\leq 
\frac{\prod_{p=\left \lfloor \frac{\sqrt{2n}\sigma^i}{2} \right
    \rfloor}^{\left \lfloor \sqrt{2n}\sigma^i \right
    \rfloor}(1+\frac{p}{4\left \lfloor n \rho^i \right
    \rfloor})}{\prod_{p=\left \lfloor \frac{\sqrt{2n}\sigma^i}{2} \right
    \rfloor}^{\left \lfloor \sqrt{2n}\sigma^i \right
    \rfloor}
    (1+\frac{p}{3\left \lfloor n \rho^i \right \rfloor})}
$$$$\leq \left(\frac{1+\frac{\left \lfloor \frac{\sqrt{2n}\sigma^i}{2} \right \rfloor}{4\left \lfloor n \rho^i \right \rfloor}}{1+\frac{\left \lfloor \frac{\sqrt{2n}\sigma^i}{2} \right \rfloor}{3\left \lfloor n \rho^i \right \rfloor}} \right)^{\left \lfloor \frac{\sqrt{2n}\sigma^i}{2} \right \rfloor}\leq e^{\frac{-(\sigma^i)^2}{{24}\rho^i}}.$$

\end{itemize}

By  the dominated  convergence theorem,  the integral  in the  term of
index $k$ in $\mathbb{E}_n(\varphi)$ converges to
$$\int _{\mathcal{X}^k}$$$$\left ( \mathbbm 1_{\left\{\rho^{2t(k)}\geq 0\right\}}\varphi(k,\rho,\gamma,\sigma)\times
  \prod_{i=1}^{2t(k)}\left(\frac{\sigma^i}{\sqrt{2}\rho^i}\times
  \frac{2}{\sqrt{6\pi    \rho^i}}\times    e^{\frac{-(\sigma^i)^2}{3\rho^i}}
  \times                     \left(\frac{4}{3}\right)^{c^i(k)+1}\right)\times
  \prod_{i=1}^{t(k)}p_{\sigma^i}(\gamma^i)\right )\mathrm d X^k.$$ The
term $n^{\frac{t(k)-3}{2}}$ is equal to $n^{-1/2}$ if $k=0$ and $1$ if
$k\in [\![1,9]\!]$ (so in the end the case $k=0$ will not contribute).

Choosing $\varphi=1$ provides the estimate
$$|\mathcal{T}_{r,s,b}(n)|\sim 2\, \Upsilon
\left(\frac{256}{27}\right)^{n-2}.$$ Finally, we obtain the
convergence of $\mathbb{E}_n(\varphi)$ to
$$ \frac{1}{\Upsilon}\sum_{k=1}^{9}\int
_{\mathcal      X^k}       $$$$\left( \mathbbm      1_{\{\rho^{6}\geq      0\}}\times\varphi(k,\rho,\gamma,\sigma)\times
  \prod_{i=1}^{6}\frac{\sigma^i}{\sqrt{2}\,             \rho^i}\times
  \frac{2}{\sqrt{6\pi    \rho^i}}\times    e^{\frac{-(\sigma^i)^2}{3\rho^i}}
  \times                     \left(\frac{4}{3}\right)^{c^i(k)+1}\times
  \prod_{i=1}^{3}p_{\sigma^i}(\gamma^i)\right  )\mathrm   d  X^k,$$
For $k\in[\![1,9]\!]$, we have $\mathcal X^k=\mathcal X$
which completes the proof of the lemma.
\end{proof}

An immediate consequence of Lemma~\ref{structureofu} is the
following:

\begin{corollary}
\label{cor:nosymmetric}
There exists two constants $c,c'\in \mathbb R^*_+$ such that for all $n\geq 1$,
$$c\leq n\times \mathbb P \left(\exists i, i'\in  [\![1,2t_n]\!]\, :
  \, \rho^i_n=\rho^{i'}_n\right)\leq c',$$
$$c\leq \sqrt{n}\times \mathbb{P}\left(\exists i, i'\in  [\![1,t_n]\!]\,: \, \sigma^i_n=\sigma^{i'}_n\right)\leq c'.$$
\end{corollary}

In the proof  of Lemma~\ref{structureofu} we compute  an asymptotic of
$\mathcal{T}_{r,s,b}(n)$; by  Theorem~\ref{them:bijectionbenjamin}, we
obtain  a reformulation  of the  asymptotic of  the number  of rooted
essentially simple triangulations:

\begin{corollary}
\label{cor:assymptotic}
For $n\geq 1$, the set  $\mathcal{G}(n)$ of essentially simple toroidal
triangulations  on $n$  vertices  that are  rooted at  a  corner of  a
maximal triangle satisfies:
$$|\mathcal{G}(n)|\sim {2\, \Upsilon}\left(\frac{256}{27}\right)^{n-2},$$
where $\Upsilon$ is the constant defined earlier. 
\end{corollary}

It is  possible that  the formula 
defining $\Upsilon$  could be
amenable to an  explicit computation, but we did not  manage to find a
simple  way  to  do  it.

\section{Convergence of uniformly random Motzkin paths}
\label{sec:motzkin}

  Consider
$(\sigma_n)\in (\mathbb{N}^*)^{\mathbb{N}}, (\gamma_n)\in
\mathbb{Z}^{\mathbb{N}}$ such that, there exist $\sigma \in \mathbb
R^*_+$ and $\gamma \in \mathbb R$ satisfying :
$$\frac{\sigma_n}{\sqrt {2n}}\longrightarrow \sigma \text{ and }
\left(\frac{9}{8n}\right)^{1/4}\gamma_n\rightarrow \gamma.$$
Let $M_n$ be a uniformly random element of $\mathcal{M}_{\sigma_n}^{\gamma_n}$ and let $M_n$ also denote its piecewise linear interpolation which is therefore a random element of $\mathcal{H}$. Let $M_{(n)}$ denote the
rescaled process defined as:
\begin{equation*}
\label{eq:definitionofW_n}
M_{(n)}=\left(\left(\frac{9}{8n}\right)^{1/4}M_n(\sqrt{2n}s)\right)_{0\leq       s       \leq
 \frac{\sigma_n}{\sqrt {2n}}}
\end{equation*}

By Theorem~\ref{convergencemotzkinbridge} with
$(\eta,h)=(\sqrt{\frac{2}{3}},1)$, we have the following:

\begin{lemma}
\label{convergeofmotzkinpath}
The process $M_{(n)}$ converges in law toward the 
Brownian bridge $B^{0\rightarrow \gamma}_{[0,\sigma]}$ in the space
$(\mathcal H,d_{\mathcal H})$, when $n$ goes to infinity.
\end{lemma}

Recall from Section~\ref{sec:motz} that $\widetilde{M_n}$ is the extension of $M_n$ and let $\widetilde{M_n}$ also denote its piecewise linear interpolation. When $2\sigma_n+\gamma_n< 2\sqrt{2n}\sigma$, we
assume that $\widetilde{M_n}$ is extended to take value $\gamma_n$ on
$[2\sigma_n+\gamma_n, 2\sqrt{2n}\sigma]$.  Then we define
the rescaled versions:

\begin{equation*}
\label{eq:definitionofW_n}
\widetilde{M_{(n)}}=\left(\left(\frac{9}{8n}\right)^{1/4}\widetilde{M_n}(\sqrt{2n}s)\right)_{0\leq       s       \leq
  \max\left(\frac{2\sigma_n+\gamma_n}{\sqrt{2n}},\,2\sigma\right)}
\end{equation*}
\begin{lemma}
\label{convergeofmotzkinpathex} 
The process $\widetilde{M_{(n)}}$ converges in law toward the Brownian
bridge $B^{0\rightarrow \gamma}_{[0,2\sigma]}$ in the space
$(\mathcal H,d_{\mathcal H})$, when $n$ goes to infinity.
\end{lemma}
\begin{proof}

Let $t\in [\![0, \sigma_n ]\!]$. By the construction of $\widetilde{M_n}$, we have

$$ M_n(t)=\widetilde{M_n}(2t+M_n(t))$$

Let $t,s$ be distinct element of $[\![0, 2\sigma_n
+\gamma_n]\!]$. Note that there exist $t_1, s_1$ distinct element of
$[\![0, \sigma_n ]\!]$ such that
 $$|t-(2t_1+M_n(t_1))|\leq 2\,\,\, \text{and}\,\,\, |s-(2s_1+M_n(s_1))|\leq 2$$

Therefore, we obtain
\begin{equation*}
\begin{split}
&\left|\widetilde{M_n}(t)-\widetilde{M_n}(s)\right|\\=&
\left|\widetilde{M_n}(t)-\widetilde{M_n}\left(2t_1+M_n(t_1)\right)+\widetilde{M_n}\left(2t_1+M_n(t_1)\right)-\widetilde{M_n}\left(2s_1+M_n(s_1)\right)\right.\\
&\left.+\widetilde{M_n}\left(2s_1+M_n(s_1)\right)-\widetilde{M_n}(s)\right| \\
\leq& \left|\widetilde{M_n}(t)-\widetilde{M_n}\left(2t_1+M_n(t_1)\right)\right|+\left|\widetilde{M_n}\left(2t_1+M_n(t_1)\right)-\widetilde{M_n}\left(2s_1+M_n(s_1)\right)\right|\\
&+\left|\widetilde{M_n}\left(2s_1+M_n(s_1)\right)-\widetilde{M_n}(s)\right|\\
\leq& 4+\left|\widetilde{M_n}\left(2t_1+M_n(t_1)\right)-\widetilde{M_n}\left(2s_1+M_n(s_1)\right)\right|\\
\leq &4+\left|M_n(t_1)-M_n(s_1)\right|
\end{split}
\end{equation*}

The convergence of $M_{(n)}$ by
 Lemma~\ref{convergeofmotzkinpath}
 implies that there exists $\alpha   <1/2$ such that
 \begin{equation}
   \label{eq:tightm(n)}
 \forall   \epsilon  >0   \quad  \exists
  C \quad \forall n\quad \mathbb{P}(\|M_{(n)}\|_\alpha \leq
  C)> 1-\epsilon.
 \end{equation}
  
Consider $\epsilon >0$. Let $C$ be such that~\eqref{eq:tightm(n)} is satisfied.

Conditioned on $\|M_{(n)}\|_\alpha \leq  C$, we have
\begin{equation}
\label{eq:tightofextension}
\begin{split}
\left|\widetilde{M_n}(t)-\widetilde{M_n}(s)\right|
\leq 4+C
\left(\frac{8n}{9}\right)^{1/4}
\left|\frac{t_1}{\sqrt{2n}}-\frac{s_1}{\sqrt{2n}}\right|^\alpha
\end{split}
\end{equation}

Since $\alpha <1/2$, there exists a constant $C_1$ which do not depend on $t_1$ and $s_1$ such that:
\begin{equation}
\label{eq:tightofextension1}
4\leq C_1 \left(\frac{8n}{9}\right)^{1/4}\left|\frac{t_1}{\sqrt{2n}}-\frac{s_1}{\sqrt{2n}}\right|^\alpha 
\end{equation}
By using \eqref{eq:tightofextension} and \eqref{eq:tightofextension1}, there exists a constant $C_2$ such that:
$$
\left|\widetilde{M_n}(t)-\widetilde{M_n}(s)\right|\leq C_2 \left(\frac{8n}{9}\right)^{1/4}\left|\frac{t_1}{\sqrt{2n}}-\frac{s_1}{\sqrt{2n}}\right|^\alpha 
$$

Note that $|t_1-s_1|\leq |t-s|+4\leq 5|t-s|$. So there exist a constant $C_3$,
such that:
$$\left|\widetilde{M_{(n)}}\left(\frac{t}{\sqrt{2n}}\right)-\widetilde{M_{(n)}}\left(\frac{s}{\sqrt{2n}}\right)\right|\leq C_3\left|\frac{t}{\sqrt{2n}}-\frac{s}{\sqrt{2n}}\right|^\alpha.$$

This inequality is satisfied for
$0\leq x < y \leq \frac{2\sigma_n+\tau_n}{\sqrt{2n}}$ such that
$2nx, 2ny\in \mathbb N$. It is also satisfied for all
$0\leq x < y \leq \frac{2\sigma_n+\tau_n}{\sqrt{2n}}$ by linear
interpolation. So we have:
$$ \forall n\quad \mathbb{P}(\|\widetilde{M_{(n)}}\|_\alpha \leq
C_3)> 1-\epsilon.$$

Therefore the family of laws of $\left(\widetilde{M_{(n)}}\right)_{n\geq 1}$ is tight in the space of probability measures on $\mathcal{H}$.

  Let $0\leq t <2\sigma$ and $\epsilon >0$.
    Since $\frac{2\sigma_n+\gamma_n}{\sqrt{2n}}$ converge toward $2\sigma$,
    there
  exists $N$ such that $t\leq \min_{n\geq N}{\frac{2\sigma_n+\gamma_n}{\sqrt{2n}}}$. Note that there exists $0\leq s <\sigma$ such that 
  $$\left|\lfloor\sqrt{2n}t\rfloor - \left(2\lfloor\sqrt{2n}s\rfloor+ M_n\left(\lfloor\sqrt{2n}s\rfloor\right)\right)\right|\leq 2.$$
 Therefore we obtain:
  $$\left|\widetilde{M_n}\left(\lfloor\sqrt{2n}t\rfloor\right) - \widetilde{M_n}\left(2\lfloor\sqrt{2n}s\rfloor+ M_n\left(\lfloor\sqrt{2n}s\rfloor\right)\right)\right|\leq 2.$$

Since $\widetilde{M_n}\left(2\lfloor\sqrt{2n}s\rfloor+
  M_n\left(\lfloor\sqrt{2n}s\rfloor\right)\right)=M_n\left(\lfloor
  \sqrt{2n}s\rfloor\right)$ and $\lfloor
\sqrt{2n}s\rfloor=\frac{1}{2}\left(\lfloor
  \sqrt{2n}t\rfloor-M_n\left(\lfloor
    \sqrt{2n}s\rfloor\right)\right)+e$, with $e=O(1)$. We then obtain:

$$\widetilde{M_n}\left(\lfloor\sqrt{2n}t\rfloor\right) =M_n\left[\frac{1}{2}\left(\lfloor \sqrt{2n}t\rfloor-M_n\left(\lfloor \sqrt{2n}s\rfloor\right)\right)+e\right].$$
Since the family of laws of $(M_{(n)})_{n\geq 1}$ is tight,  there exists  a constant  $c_1$ such
  that
  \begin{equation}
    \label{equation2345coucou}
    \inf_{n\geq N}\mathbb{P}\left(\sup_{k\in[\![0, \sigma_n]\!]}\left | M_n(k) \right |<c_1n^{1/4}\right)\geq 1-\epsilon.
  \end{equation}

  Let $\mathcal E_n$ the event:
  $$\left \{ \sup_{k\in[\![0, \sigma_n]\!]} \left | M_n(k) \right |<c_1n^{1/4} \right \}.$$
  Now we  define  a   random   variable  $Y_n$   as   follows:
$$Y_n=M_n\left[\frac{1}{2}\left(\lfloor \sqrt{2n}t\rfloor-M_n\left(\lfloor \sqrt{2n}s\rfloor\right)\mathbbm 1_{\mathcal E_n}\right)+e\right] $$

    By Lemma~\ref{convergeofmotzkinpath}, we have
    $\left(\left(\frac{9}{8n}\right)^{1/4} Y_n\right)_{n\geq N}$ converge toward
    $B^{0\rightarrow \gamma}_{[0,\sigma]}(t/2)$ when $n$ goes to
    infinity. Let $f$ be a bounded continuous function from $\mathbb
    R$ to $\mathbb R$.  Thus by (\ref{equation2345coucou}), there exists $n_0\geq N$ such
    that for all $n\geq n_0$:

  \begin{equation*}\begin{split}
  &    \left | \mathbb{E}[f(\widetilde{M_{(n)}}(t))]
    -\mathbb{E}\left[f\left(
B^{0\rightarrow \gamma}_{[0,\sigma]}(t/2)   
          \right)\right]\right    |
\\&\leq
     \left | \mathbb{E}[f(\widetilde{M_{(n)}}(t))]
        -\mathbb{E}\left[f\left(\left(\frac{9}{8n}\right)^{1/4} Y_n\right)\right]\right    |+
       \left | \mathbb{E}\left[f\left(\left(\frac{9}{8n}\right)^{1/4} Y_n\right)\right]
         -\mathbb{E}\left[f\left(
B^{0\rightarrow \gamma}_{[0,\sigma]}(t/2)   
          \right)\right]\right    |
      \\
      &\leq 2\,\mathbb{E}[1-\mathbbm 1_{\mathcal E_n}]\,\left    \|f    \right
    \|_{\infty}+\epsilon.
      \\
      &\leq
    (2\left    \|f    \right
    \|_{\infty}+1)\epsilon.
  \end{split}
  \end{equation*}

      This implies that $\left(\mathbb{E}[f(\widetilde{M_{(n)}}(t))]\right)_{n\geq N}$
  converge toward $\mathbb{E}\left[f\left(B^{0\rightarrow \gamma}_{[0,\sigma]}(t/2)
        \right)\right]$.

        We now prove the finite dimensional convergence of
        $\widetilde{M_{(n)}}$. Let $k\geq 1$ and consider
        $0\leq t_1< t_2<...< t_k < 2 \sigma$.  Let $N$ such that
        $t_{k}\leq \min_{n\geq N}{\frac{2\sigma_n+\gamma_n}{\sqrt{2n}}}$.  By
        above arguments, for $1\leq i\leq k$, we have
        $(\widetilde{M_{(n)}}(t_i))_{n\geq N}$ converge in law toward
        $B^{0\rightarrow \gamma}_{[0,\sigma]}(t_i/2)$\\        
        It remains to
          deal with the point $2\sigma$.  
\begin{align*}
  \left |\widetilde{M_{(n)}}\left (2\sigma\right )-\gamma  \right|&
  =\left |\widetilde{M_{(n)}}\left (2\sigma\wedge \frac{2\sigma_n+\gamma_n}{\sqrt{2n}}\right)-\gamma  \right|\\
  &=\left |\widetilde{M_{(n)}}\left (2\sigma\wedge \frac{2\sigma_n+\gamma_n}{\sqrt{2n}}\right)-\widetilde{M_{(n)}}\left (\frac{2\sigma_n+\gamma_n}{\sqrt{2n}}\right)+\widetilde{M_{(n)}}\left (\frac{2\sigma_n+\gamma_n}{\sqrt{2n}}\right)-\gamma \right|\\
    &\leq \left |\widetilde{M_{(n)}}\left (2\sigma\wedge \frac{2\sigma_n+\gamma_n}{\sqrt{2n}}\right)-\widetilde{M_{(n)}}\left (\frac{2\sigma_n+\gamma_n}{\sqrt{2n}}\right)\right|+\left |\gamma_n-\gamma \right| \\
\end{align*}
Consider $\epsilon>0$. Since the family of laws of $\widetilde{M_{(n)}}$ is tight, there exists $\alpha$ and $C$ such that
for all $n$:
$\mathbb{P}\left(\|  \widetilde{M_{(n)}}   \|_\alpha   \leq
  C\right)> 1-\epsilon$. Condition on the event $\{\|  \widetilde{M_{(n)}}   \|_\alpha
\leq C\}$,
  we have
\begin{align*}\left |\widetilde{M_{(n)}}\left (2\sigma\wedge
    \frac{2\sigma_n+\gamma_n}{2n}\right)-\widetilde{M_{(n)}}\left
    (\frac{2\sigma_n+\gamma_n}{2n}\right)\right|&\leq
C\,\left|2\sigma\wedge
\frac{2\sigma_n+\gamma_n}{\sqrt{2n}} - \frac{2\sigma_n+\gamma_n}{\sqrt{2n}}\right|^\alpha\\
&\leq
C\,\left|2\sigma - \frac{2\sigma_n+\gamma_n}{\sqrt{2n}}\right|^\alpha
\end{align*}
Since  $\frac{2\sigma_n+\gamma_n}{\sqrt{2n}}\rightarrow 2\sigma$ and $\gamma_n \rightarrow
\gamma$, for $n$ large enough, we have:
\begin{align*}
  \left |\widetilde{M_{(n)}}\left (2\sigma\right )-\gamma  \right|&
                                                \leq \epsilon
\end{align*}
Therefore we obtain for $n$ large enough:
$$\mathbb{P}\left(\left |\widetilde{M_{(n)}}\left (2\sigma\right)-\gamma \right|>\epsilon\right)\leq \mathbb{P}\left(\|  \widetilde{M_{(n)}}   \|_\alpha
  > C \right)\leq \epsilon.$$ This implies that
$\widetilde{M_{(n)}}\left (2\sigma\right)$ converges in probability toward the
deterministic value $\gamma$. So Slutzky's lemma shows that
$\widetilde{M_{(n)}}\left (2\sigma\right)$ converges in law toward $\gamma$. Note that $\left(B^{0\rightarrow \gamma}_{[0,2\sigma]}(t)\right)_{0\leq t \leq 2\sigma}$ and $\left(B^{0\rightarrow \gamma}_{[0,\sigma]}(t/2)\right)_{0\leq t \leq 2\sigma}$ have the same law. Thus we have
proved the convergence of the finite-dimensional marginals of
$\widetilde{M_{(n)}}$ toward
$B^{0\rightarrow \gamma}_{[0,2\sigma]}$. Moreover,  $\widetilde{M_{(n)}}$ is tight so Prokhorov's lemma
give the result.
\end{proof}

\section{Convergence of uniformly random $3$-dominating binary words}
\label{section6}

Consider
$(\rho_n)\in \mathbb{N}^{\mathbb{N}}, (\tau_n)\in
\mathbb{N}^{\mathbb{N}}$ and recall that
$\mathcal{D}_{3,3\rho_n+\tau_n, \rho_n}^{-1}$ is the set of elements
$b\in \{0,1\}^{p+q}$ with $|b|_0=3\rho_n+\tau_n$ and $|b|_1=\rho_n$ that are
inverse of $3$-dominating binary words (see
Section~\ref{section:Relation between well-labeled forests and
  $3$-dominating binary words}).  The goal of this section is to prove
the convergence of uniform random elements of the set
$\mathcal{D}_{3,3\rho_n+\tau_n, \rho_n}^{-1}$, in which we assume
that, there exists $\rho,\tau \in \mathbb R_+$, such that:

$$\rho_{(n)}=\frac{\rho_n}{n}\longrightarrow \rho \text{ and } \tau_{(n)}=\frac{\tau_n}{\sqrt{n}}\rightarrow \tau.$$

Given a element $b$ of $\mathcal{D}_{3,3\rho_n+\tau_n, \rho_n}^{-1}$, we can
replace the bits ``$1$'' by $-3$ and the bits ``$0$'' by $1$, getting
an encoding of a (random) inverse $3$-dominating binary word of length
$4\rho_n+\tau_n$ by a (random) path of the same length
$w=(w(0),w(1),...,w(4\rho_n+\tau_n))$ in $\mathbb Z$
such that
$$w(0)=0,                         w(4\rho_n+\tau_n)=\tau_n,
\overline{w}(4\rho_n+\tau_n)<\tau_n \text{ and }
w(i+1)-w(i) \in \{-3,1\} (\forall i),$$ where
$\overline{w}(t)=\sup_{s< t} w(s)$.  If $b$ is
uniformly distributed in $\mathcal{D}_{3,3\rho_n+\tau_n, \rho_n}^{-1}$, then
$w$ is uniformly distributed in the set $\mathcal
P_{3,3\rho_n+\tau_n, \rho_n}$ of all paths of
length $4\rho_n+\tau_n$ starting at $0$, with increments in
$\left \{ -3,1 \right \}$ and taking value $\tau_n$ at their
last step  for the first time.

Let $W_n$ be a uniformly random element of $\mathcal P_{3,3\rho_n+\tau_n,\rho_n}$ and
let $W_n$ also denote its piecewise linear interpolation which is
therefore a random element of $\mathcal H$.  Let $W_{(n)}$ denote the
rescaled process defined as:
\begin{equation}
\label{eq:definitionofW_n}
W_{(n)}=\left(\frac{W_n(2ns)}{\sqrt{3n}}\right)_{0\leq       s       \leq
  \frac{4\rho_n+\tau_n}{2n}}
\end{equation}

The goal of this section is to prove the following convergence result:

\begin{lemma}
\label{convergetowardbr}
The process $W_{(n)}$ converges in law toward the first-passage
Brownian bridge $F^{0\rightarrow \tau}_{[0,2\rho]}$ in the space
$(\mathcal H,d_{\mathcal H})$, when $n$ goes to infinity.
\end{lemma}

\subsection{Review and generalization of a result of Bertoin, Chaumont
  and Pitman}
\label{sectionbertion}

We are  going to  extend a  result in  \cite{bertoin2003path}, showing
that its  proof is  still valid  for the  case of  a random  path with
increments in  $\left \{ -3,1  \right \}$  as above. Fix  two integers $\beta$ and $n$ such that 
$1  \leq \beta  \leq  n$, and  let  $(X_i)_{1\leq i  \leq  n}$ be  a
sequence of   $i.i.d.$  random variables
of law:
$$\mathbb{P}(X_i=-3)=\frac{1}{4} \text{ and } \mathbb{P}(X_i=1)=\frac{3}{4}.$$
Let $S=(S_i)_{0\leq i  \leq n}$ be the random path  started at $0$ and
with  increments  given  by  the   $X_i$,  conditioned  on  the  event
$\{S_n  = \beta\}$. For  any  $k=0,1,...,n-1$,  define the  shifted
chain:
$$\theta_k(S)_i=\begin{cases}
S_{i+k}-S_k & \text{ if } 0\leq i \leq n-k, \\ 
S_{k+i-n}+S_n-S_k & \text{ if } n-k\leq i\leq n.
\end{cases}$$  For $k=0,1,...,\beta-1$,  define  the  first time  at
which $S$ reaches its maximum minus $k$ as follows:
$$ m_k(S)=\inf\left \{ i:S_i=\max_{0\leq j\leq n}S_j-k \right \}.$$
For convenience, we write $\theta_{m_k}(S)$ for $\theta_{m_k(S)}(S)$ in
what follows.

Denote  by  $\Gamma$ the  support  of  the law  of  $S$.   For every $\gamma  \in \Gamma$, define  the sequence
$\Lambda(s)      =      (s,     \theta_1(s),      ...,
\theta_{n-1}(s))$.
 Let $\overline{\Lambda}(s)$ be  the    subsequence
  of the paths in $\Lambda(s)$ which
  first  hit their  maximum  at time  $n$.
We need  the
following lemma. 

\begin{lemma}
  \label{cyclelemmabertoin}
  For      every     $s\in      \Gamma$,   $\overline{\Lambda}(s)$   contains exactly  $\beta$
  elements and more  precisely:
  $$\overline{\Lambda}(s)=\left            (  \theta_{m_{\beta-1}}(s),...,         \theta_{m_0}(s)\right ).$$
\end{lemma}

\begin{proof}
  One  can see  that the  path $\theta_{m_k}(s)$  is contained  in
  $\overline{\Lambda}(s)$ and  the cycle lemma gives  us that the
  cardinality of $\overline{\Lambda}(s)$ is exactly $\beta$.
\end{proof}

The following is an extension of a result of 
 Bertoin, Chaumont, Pitman \cite{bertoin2003path}:

\begin{lemma}
  \label{bertoin}
  Let $\nu$ be a random variable which is independent of $S$ and
  uniformly distributed on $\{0,1,...,\beta-1\}$.  The chain
  $\theta_{m_{\nu}}(S)$ has the same law as that of $S$ conditioned on
  the event $\{m_0=n\}$ and independent from $m_{\nu}$.
\end{lemma}

\begin{proof}
  For      every     bounded      function     $f$      defined     on
  $\left  \{ 0,1,...,n  \right \}$  and every  bounded function  $F$
  defined             on            $\mathbb{Z}^{n+1}$,             we
  have
  \begin{equation}
  \mathbb E\left[F(\theta_{m_{\nu}}(S))f(m_{\nu})\right]=\sum_{s\in
    \Gamma}\mathbb P(S=s)\frac{1}{\beta}\sum_{j=0}^{\beta-1}F(\theta_{m_{j}}(s))f(m_{j}).\label{eq:cyclebertoin}
\end{equation}
  By Lemma~\ref{cyclelemmabertoin}, we obtain
$$\sum_{j=0}^{\beta-1}F(\theta_{m_{j}}(s))f(m_{j})=\sum_{k=0}^{n-1}F(\theta_{k}(s))f(k){\mathbbm
  1}_{    \{     m_0(\theta_k(s))    =n\}}.$$     Replacing    in
\eqref{eq:cyclebertoin}, we get
$$\mathbb E\left[F(\theta_{m_{\nu}}(S))f(m_{\nu})\right]=\frac{n}{\beta}\mathbb E\left[F(\theta_{U}(S))f(U){\mathbbm
  1}_{   \{  m_0(\theta_U(S))   =n\}}\right]$$  where   $U$  is   uniform  on
$\{0,1,...,n-1\}$ and independent of $S$. This can be rewritten as 
$$\mathbb E\left[F(\theta_{m_{\nu}}(S))f(m_{\nu})\right]=\mathbb E\left[F(S)|m_0(S)=n\right]\mathbb E\left[f(U)\right],$$
which concludes the proof of the lemma.
\end{proof}

\subsection{Convergence to the first-passage Brownian bridge}
\label{sec:convergencebrownianbridge}

In     this     section,     we      prove
Lemma~\ref{convergetowardbr}.
Let  $a\in (0,1)$  and  let  $(X^{a}_n)_{n\geq 1}$  be  a sequence  of
$i.i.d.$  random   variables  with distribution
$a\delta_{-3}+(1-a)\delta_{1}$ (i.e. whose steps are in $\{-3,1\}$ with
probability $a$ for ``-3'' and $(1-a)$ for ``1'').
 We define  $S^a_0=0$  and
$S^{a}_n=\sum_{i=1}^{n}X^{a}_i$.  We begin  with  the following  basic
lemma.

\begin{lemma}
  \label{lemmasurprise}
  For  all  $a   \in  (0,1)$  and  $\rho,\tau\in   \mathbb{N}$,  we  have:
    $$\mathcal          L((S^a)_{0\leq i\leq
    4\rho+\tau}|S^a_{4\rho+\tau}=\tau,\overline{S^a}_{4\rho+\tau-1}<\tau)=
  \mathcal  U(\mathcal P_{3,3\rho+\tau,\rho}),$$
  where   $\overline{S^a}_{k}=\max_{0\leq   i\leq  k}S^a_i$   and
  $ \mathcal U(\mathcal P_{3,3\rho+\tau,\rho}) $ is the uniform law on $\mathcal P_{3,3\rho+\tau,\rho}$.
\end{lemma}

\begin{proof}
  Let
  $w=(w_0=0,w_1,...,w_{4\rho+\tau}=\tau) \in \mathcal
  P_{3,3\rho+\tau,\rho}$.
  $$\mathbb{P}((S^a)_{0\leq i\leq 4\rho+\tau}=\omega|S^a_{4\rho+\tau}=\tau, \overline{S^a}_{4\rho+\tau-1}<\tau)=\frac{\left(1-a\right)^{3\rho+\tau}a^{\rho}}{\mathbb{P}(S^a_{4\rho+\tau}=\tau, \overline{S^a}_{4\rho+\tau-1}<\tau)},$$
  which  does not  depend on  $\omega$.  This concludes  the proof  of
  lemma.
\end{proof}

We are now ready to prove Lemma~\ref{convergetowardbr}:

\begin{proof}[Proof of Lemma~\ref{convergetowardbr}]

  Let $S_n=(S_n(i))_{0\leq i \leq 4\rho_n+\tau_n}$ be the random path started at $0$
  and with increments given by the $X_i$ (defined in section
  \ref{sectionbertion}).  Let $F_n$ be the random path $S_n$
   conditioned to take value $\tau_n$ at time $4\rho_n+\tau_n$  for the
   first time.
   Let
  the same notations $S_n$ and $F_n$  denote their piecewise linear
  interpolation which is therefore a random element of $\mathcal H$.
 When $4\rho_n+\tau_n< 4n\rho$, we
assume that $F_n$ is extended to take value $\tau_n$ on
$[4\rho_n+\tau_n, 4n\rho]$. 
    Let $S_{(n)}$ and $F_{(n)}$ denote the rescaled processes:
\begin{equation*}
\label{eq:equtionofF_n}
S_{(n)}=\left(\frac{S_n(2ns)}{\sqrt{3n}}\right)_{0\leq       s       \leq
  \frac{4\rho_n+\tau_n}{2n}}
  \end{equation*}
\begin{equation*}
F_{(n)}=\left(\frac{F_n(2ns)}{\sqrt{3n}}\right)_{0\leq       s
  \leq \max(\frac{4\rho_n+\tau_n}{2n},2\rho)}
\end{equation*}

  Let $\mathcal F_i=\sigma\left  \{ S_n(k),0\leq k \leq i  \right \}$ be
  the natural  filtration associated  with $S$.

  By
  Lemma~\ref{lemmasurprise}, the law of $W_n$ is the same of $F_n$.
  By
  Donsker's theorem  and Skorokhod's theorem,  we may  assume
  that  as $n\to\infty$,  $S_{(n)}$ converges  almost surely  toward a
  standard Brownian motion $(\beta_s)_{0\leq s\leq 2\rho}$ for the uniform
  topology. 

  \begin{claim}
    \label{cl:Ffirstpassage}
    Suppose $\rho >0$ and consider $0\leq \rho'<2\rho$. For $n$ large
    enough $ 2n\rho' < 4\rho_n+\tau_n$ and
    $(F_{(n)}(s))_{0\leq s \leq \rho'}$ converge in law toward
      $(F^{0\rightarrow \tau}_{[0,2\rho]})_{0\leq s\leq
          \rho'}$.
  \end{claim}

  \begin{proofclaim}
    It is clear that for $n$ large enough we have
    $ 2n\rho' < 4\rho_n+\tau_n$. Let $f$ be a continuous bounded
    function from $\mathcal H$ to $\mathbb{R}$. We have
    $$\mathbb{E}[f((F_{(n)}(s))_{0\leq              s             \leq
      \rho'})]=$$$$\mathbb{E}[f((S_{(n)}(s))_{0\leq       s\leq
      \rho'})|S_n({4\rho_n+\tau_n})=\tau_n,
    \overline{S_n}({4\rho_n+\tau_n-1})<\tau_n].$$
    
    By the  definition of  conditional probability  and the  fact that
    $(S_{(n)}(s))_{0\leq   s\leq   \rho'}$   is  measurable  with respect   to
    $\mathcal F_{2n\rho'}$, we have:
    $$\mathbb{E}[f((F_{(n)}(s))_{0\leq              s             \leq
      \rho'})]=$$

    $$\mathbb{E}\left[f((S_{(n)}(s))_{0\leq                s\leq
        \rho'})\frac{\mathbb{P}(S_n({4\rho_n+\tau_n})=\tau_n,
        \overline{S_n}({4\rho_n+\tau_n-1})<\tau_n|\mathcal
        F_{2n\rho'})}{\mathbb{P}(S_n({4\rho_n+\tau_n})=\tau_n,
        \overline{S_n}({4\rho_n+\tau_n-1})<\tau_n)}\right].$$   Recall   the
    notation              $Q^S_k(i)=\mathbb{P}(S_k=i)$;             by
    Lemmas~\ref{cyclelemmabertoin}, we have:
    $$\mathbb{P}(S_n({4\rho_n+\tau_n})=\tau_n,\overline{S_n}(4\rho_n+\tau_n-1)<\tau_n)=\frac{\tau_n}{4\rho_n+\tau_n}\mathbb{P}(S_n({4\rho_n+\tau_n})=\tau_n).$$


    Using  the  Markov  property,  we   obtain,  denoting  by  $T_n$  an
    independent copy of $S_n$:
    \begin{align*}
      \mathbb{P}&(S_n({4\rho_n+\tau_n})=\tau_n, \overline{S_n}(4\rho_n+\tau_n-1)<\tau_n|\mathcal F_{2n\rho'})
      \\
      =&\mathbb{P}(T_n({4\rho_n+\tau_n-2n\rho'})=\tau_n-S_n({2n\rho'}),
 \overline{T_n}({4\rho_n+\tau_n-2n\rho'-1})<\tau_n-S_n({2n\rho'}))
      \\&{\mathbbm 1}_{\overline{S_n}({2n\rho'})<\tau_n}
      \\
      =&\frac{\tau_n-S_n({2n\rho'})}{4\rho_n+\tau_n-2n\rho'}\mathbb{P}(T_n({4\rho_n+\tau_n-2n\rho'})=\tau_n-S_n({2n\rho'})){\mathbbm
      1}_{\overline{S_n}({2n\rho'})<\tau_n}.
      \end{align*}

    We now verify that the ratio
    $$                            \frac{\mathbb{P}(S_n({4\rho_n+\tau_n})=\tau_n,
      \overline{S_n}({4\rho_n+\tau_n-1})<\tau_n|\mathcal
      F_{2n\rho'})}{\mathbb{P}(S_n({4\rho_n+\tau_n})=\tau_n,
      \overline{S_n}({4\rho_n+\tau_n-1})<\tau_n)} $$
    converges almost surely
    to
    $\frac{p'_{\added{2}\rho-\rho'}(\tau-\beta_{\rho'})}{p'_{2\rho}(\tau)}{\mathbbm
      1}_{\overline{\beta}({\rho'})<\tau}$.    Indeed,   by    using   the
    Lemma~\ref{sumiid}    for    the    random   walk    $S$    with
    $(\eta,h)=(\sqrt{3},4)$, we obtain:
$$                          \frac{\sqrt{3}}{4}\sqrt{4\rho_n+\tau_n}\times
\mathbb{P}(S_n({4\rho_n+\tau_n})=\tau_n)    \rightarrow    p\left(\frac{\tau}{2}\right)
\text{ and }$$
$$\frac{\sqrt{3}}{4}\sqrt{4\rho_n+\tau_n-n\rho'}\times \mathbb{P}(T_n({4\rho_n+\tau_n-n\rho'})=\tau_n-S_n({n\rho'})) \rightarrow p\left(\frac{\tau-\beta_{\rho'}}{\sqrt{\added{2}\rho-\rho'}}\right).$$
We can see also that:
$$\frac{\tau_n-S_n({n\rho'})}{4\rho_n+\tau_n-n\rho'} \frac{4\rho_n+\tau_n}{\tau_n}  \frac{\sqrt{4\rho_n+\tau_n}}{\sqrt{4\rho_n+\tau_n-n\rho'}}$$ $$=\frac{4\rho_n+\tau_n}{4\rho_n+\tau_n-n\rho'}\frac{\sqrt{4\rho_n+\tau_n}}{\tau_n} \frac{\tau_n-S_n({n\rho'})}{\sqrt{4\rho_n+\tau_n-n\rho'}}$$
converges            toward            $\frac{8(\tau-\beta_{\rho'})}{\tau
  (\added{2}\rho-\rho')^{\frac{3}{2}}}$. This implies that
$$\frac{\mathbb{P}(S_n({4\rho_n+\tau_n})=\tau_n,
  \overline{S_n}(4\rho_n+\tau_n-1)<\tau_n|\mathcal
  F_{n\rho'})}{\mathbb{P}(S_n({4\rho_n+\tau_n})=\tau_n,
  \overline{S_n}(4\rho_n+\tau_n-1)<\tau_n)}     $$    converges     toward
$\frac{p'_{\added{2}\rho-\rho'}(\tau-\beta_{\rho'})}{p'_\rho(\tau)}{\mathbbm
  1}_{\overline{\beta}_{\rho'}<\tau}$,   and   the   Lemma~\ref{sumiid}
ensures that this convergence is dominated. So,
 $$\mathbb{E}[f((F_{(n)}(s))_{0\leq s \leq \rho'})] \rightarrow \mathbb{E}\left[f((\beta_s)_{0\leq s\leq \rho'})\frac{p'_{\added{2}\rho-\rho'}(\tau-\beta_{\rho'})}{p'_{2\rho}(\tau)}\mathbbm 1_{\overline{\beta}({\rho'})<\tau}\right] $$
 $$=\mathbb{E}\left[f\left(\left(F^{0\rightarrow \tau}_{[0,2\rho]}(s)\right)_{0\leq s\leq \rho'}\right)\right].$$
 
\end{proofclaim}

\begin{claim}
  \label{cl:tightF}
  There exists a constant $\alpha>0$ such that
  $$\forall \epsilon >0 \quad \exists C \quad
  \forall   n   \quad   \mathbb{P}\left(\|  F_{(n)}   \|_\alpha   \leq
    C\right)> 1-\epsilon.$$  In particular,
  the family of laws of $\left(F_{(n)}\right)_{n\geq 1}$ is tight for
  the space of probability measure on $\mathcal H$.
\end{claim}

\begin{proofclaim}
      For       any      $\alpha      \in       (0,1/2)$      and
  $X=(X(s))_{0\leq s\leq x} \in \mathcal H$, we write
 $$\left \| X \right \|_{\alpha}=\sup_{0\leq s<t\leq x}\frac{|X(t)-X(s)|}{|t-s|^\alpha}$$
 its $\alpha$-Holder  norm. We  prove a  stochastic domination  of the
 $\alpha$-Holder  norm   of  $F_{(n)}$   by  that  of   $B_{(n)}$,  where  $B_n$  denotes the  random  walk
 $S_n$ conditioned  to have  the appropriate  final value $\tau_n$ at time  $4\rho_n+\tau_n$ and $B_{(n)}$ is the rescaled version of $B_n$. By
 Lemma~\ref{bertoin},  we  can  assume  that $F_n$  is  realized  as
 $\theta_{m_{\nu_n}}(B_n)$ . We
 consider     the    following     two     cases,    noticing     that
 $$\left |  F_{(n)}(t)-F_{(n)}(s) \right  |=\frac{1}{\sqrt{3n}}\left |
   \theta_{m_{\nu_n}}(B_n)(2nt)-  \theta_{m_{\nu_n}}(B_n)(2ns)
 \right |.$$
 \begin{itemize}
 \item                                                              If
   $0\leq   s\leq t\leq   \frac{4\rho_n+\tau_n}{2n}-\frac{m_{\nu_n}(B_n)}{2n}$,
   then by the definition of $\theta$, we have:
    $$\theta_{m_{\nu_n}}(B_n)(2nt)=B_n(m_{\nu_n}(B_n)+2nt)-B_n(m_{\nu_n}(B_n)),$$
    $$\theta_{m_{\nu_n}}(B_n)(2ns)=B_n(m_{\nu_n}(B_n)+2ns)-B_n(m_{\nu_n}(B_n)) $$
    and we get:
    $$\left | F_{(n)}(t)-F_{(n)}(s) \right |=\frac{1}{\sqrt{3n}}|B_n(m_{\nu_n}(B_n)+2nt)-B_n(m_{\nu_n}(B_n)+2ns)|$$
    $$=\left|B_{(n)}\left(\frac{m_{\nu_n}(B_n)}{2n}+t\right)-B_{(n)}\left(\frac{m_{\nu_n}(B_n)}{2n}+s\right)\right|\leq \left \| B_{(n)} \right \|_{\alpha}\left | t-s \right |^{\alpha}$$
    
  \item                                                             If
    $\frac{4\rho_n+\tau_n}{2n}-\frac{m_{\nu_n}(B_n)}{2n}\leq
    s\leq t\leq
    \frac{4\rho_n+\tau_n}{2n}$,  then by  the definition  of $\theta$,  we
    have:
$$
                    \theta_{m_{\nu_n}(B_n)}(B_n)(2nt)=$$$$B_n(m_{\nu_n}(B_n)+2nt-(4\rho_n+\tau_n))-B_n(m_{\nu_n}(B_n)) +B_n(4\rho_n+\tau_n)),$$
$$
                    \theta_{m_{\nu_n}(B_n)}(B_n)(2ns)=$$$$B_n(m_{\nu_n}(B_n)+2ns-(4\rho_n+\tau_n))-B_n(m_{\nu_n}(B_n)) +B_n(4\rho_n+\tau_n)),$$

                    and we get:
    $$\left | F_{(n)}(t)-F_{(n)}(s) \right |$$
    $$=\frac{1}{\sqrt{3n}}|B_n(m_{\nu_n}(B_n)+2nt-(4\rho_n+\tau_n)))-B_n(m_{\nu_n}(B_n)+2ns-(4\rho_n+\tau_n)))|$$
    $$=\left|B_{(n)}\left(\frac{m_{\nu_n}(B_n)}{2n}+t-\frac{(4\rho_n+\tau_n)}{2n}\right)-B_n\left(\frac{m_{\nu_n}(B_n)}{2n}+s-\frac{(4\rho_n+\tau_n)}{2n}\right)\right|$$
    $$\leq \left \| B_{(n)} \right \|_{\alpha}\left | t-s \right |^{\alpha}.$$
    \end{itemize}
    Using the triangular inequality to deal with the third case, i.e. $0\leq
    s\leq \frac{4\rho_n+\tau_n}{2n}-\frac{m_{\nu_n}(B_n)}{2n}\leq        t\leq
    \frac{4\rho_n+\tau_n}{2n}$, we obtain
    $\left \| F_{(n)} \right  \|_{\alpha}\leq 2\left \| B_{(n)} \right
    \|_{\alpha}$.

Let  $\epsilon>0$,   thanks  to   Lemma~\ref{sumiid-2}  and
     Kolmogorov's criterion, we can find some constant $C$ such that
  $$\sup_n\mathbb{P}(\|F_{(n)}\|_\alpha > C)<\epsilon.$$
  By Ascoli's theorem,  this implies that the laws  of $F_{(n)}$'s are
  tight.     
\end{proofclaim}

Claim~\ref{cl:Ffirstpassage} shows that for any
$p\geq 1$ and $0\leq s_1<s_2\cdots<s_p<2\rho$,
    
    $$(F_{(n)}(s_1),F_{(n)}(s_2),\ldots,F_{(n)}(s_p))          \rightarrow
    \left(F^{0\rightarrow \tau}_{[0,2\rho]}(s_1),F^{0\rightarrow \tau}_{[0,2\rho]}(s_2),\ldots,F^{0\rightarrow \tau}_{[0,2\rho]}(s_p)\right).$$ It only remain to deal with the point
    $2\rho$. Consider
    $\epsilon>0$.
By Claim~\ref{cl:tightF}, there exists $\alpha$ and $C$ such that for
all $n$, we have $\mathbb P\left(\left \{ \|F_{(n)}\|_\alpha \leq C \right \}\right)>1-\epsilon$.

    Condition on  the  event
    $\left \{ \|F_{(n)}\|_\alpha \leq C \right \}$, we have

\begin{align*}\left |F_{(n)}\left (2\rho\wedge
    \frac{4\rho_n+\tau_n}{2n}\right)-F_{(n)}\left
    (\frac{4\rho_n+\tau_n}{2n}\right)\right|&\leq
C\,\left|2\rho\wedge
\frac{4\rho_n+\tau_n}{2n} - \frac{4\rho_n+\tau_n}{2n}\right|^\alpha\\
&\leq
C\,\left|2\rho - \frac{4\rho_n+\tau_n}{2n}\right|^\alpha
\end{align*}

Since  $\frac{4\rho_n+\tau_n}{2n}\rightarrow 2\rho$ and $\tau_n \rightarrow
\tau$, for $n$ large enough, we have:
\begin{align*}
  \left |F_{(n)}\left (2\rho\right )-\tau  \right|&
                                                \leq \epsilon
\end{align*}

Therefore we obtain for $n$ large enough:
$$\mathbb{P}\left(\left |F_{(n)}\left (2\rho\right)-\tau  \right|>\epsilon\right)\leq \mathbb{P}\left(\|  F_{(n)}   \|_\alpha
  > C \right)\leq \epsilon.$$ This implies that
$F_{(n)}\left (2\rho\right)$ converges in probability toward the
deterministic value $\tau$. So Slutzky's lemma shows that
$F_{(n)}\left (2\rho\right)$ converges in law toward $\tau$. Thus we have
proved the convergence of the finite-dimensional marginals of
$F_{(n)}$ toward
$\widetilde F^{0\rightarrow \tau}_{[0,2\rho]}$. By
Lemma~\ref{lemmatightofvn}, $F_{(n)}$ is tight so Prokhorov's lemma
give the result.
\end{proof}

\section{Convergence of the contour pair of well-labeled forests}
\label{section7}

Consider
$(\rho_n)\in \mathbb{N}^{\mathbb{N}}, (\tau_n)\in
\mathbb{N}^{\mathbb{N}}$ such that, there exists $\rho,\tau\in \mathbb
R_+^*$ satisfying:
$$\frac{\rho_n}{n}\longrightarrow \rho \text{ and }
\frac{\tau_n}{\sqrt{n}}\rightarrow \tau.$$

For $n\geq 1$, let $(F_n,\ell_n)$ be a random well-labeled forest
uniformly distributed in $\mathcal{F}^{\rho_n}_{\tau_n}$.  For
convenience, we write $({C}_n,{L}_n)$ the contour pair
$({C}_{ F_n},{L}_{( F_n, \ell_n)})$ of $(F_n,\ell_n)$ (see
Section~\ref{subsection:forestandwell-labelings} for the definitions).  Let the same notation $C_n,L_n$
denote its piecewise linear interpolation. When $2\rho_n+\tau_n< 2n\rho$, we
assume that $C_n$ is extended to take value $\tau_n$ on
$[2\rho_n+\tau_n, 2n\rho]$.  Then we define
the rescaled versions:

$$ C_{(n)}=\left(\frac{
    C_n(2ns)}{\sqrt{3n}}\right)_{0\leq s\leq \max\left(\frac{2\rho_n+\tau_n}{2n},\rho\right)}\,
\, \, \,  \text{and}\, \, \, \,   L_{(n)}=\left(\frac{
    L_n(2ns)}{n^{1/4}}\right)_{0\leq s\leq \frac{2\rho_n+\tau_n}{2n}} $$

The goal of this section is to prove the following lemma:

\begin{lemma}
  \label{main1}
In the  sense of
  weak convergence in the space $(\mathcal H,d_{\mathcal H})^2$ when
  $n$ goes to infinity, we have:
  $$\left(C_{(n)},L_{(n)}\right)           \to
  \left(\widetilde F^{0\rightarrow    \tau}_{\left    [     0,\rho    \right    ]},
    Z^\tau_{[0,\rho]}\right).$$
\end{lemma}

\subsection{Tightness of the contour function}

Recall that $\|\cdot\|_\alpha$ denotes the $\alpha$-Hölder norm.

\begin{lemma}[Tightness of contour function]
  \label{lemmatightofvn}
  There exists a constant $\alpha>0$ such that
  $$\forall \epsilon >0 \quad \exists C \quad
  \forall   n   \quad   \mathbb{P}\left(\|  C_{(n)}   \|_\alpha   \leq
    C\right)> 1-\epsilon.$$  In particular,  the family
  of laws of $\left(C_{(n)}\right)_{n \geq 1}$ is tight in the space of probability
  measures on $\mathcal H$.
\end{lemma}

\begin{proof}

  By the bijection of Lemma~\ref{bijection3dominating} and
  Section~\ref{section6}, we can consider $W_n$ the element of
  $\mathcal P_{3,3\rho_n+\tau_n, \rho_n}$ corresponding to $(F_n,\ell_n)$.
  Note that $W_n$ is a uniform random element of
  $\mathcal P_{3,3\rho_n+\tau_n, \rho_n}$.

  The convergence of $W_{(n)}$ (see Lemma~\ref{convergetowardbr}) implies that:
  $$\exists  \alpha   >0  \quad\forall   \epsilon  >0   \quad  \exists
  C \quad \forall n\quad \mathbb{P}(\|W_{(n)}\|_\alpha \leq
  C)> 1-\epsilon.$$
  
  Note that an integer $k$ such that $0\leq k \leq {2\rho_n+\tau_n}$
  corresponds to an angle $a(k)$ of the plane rooted tree representing
  $F$ (see Section~\ref{subsection:forestandwell-labelings}).  While
  encoding $(F,d)$ with a binary word of
  $\mathcal D_{3,3\rho_n+\tau_n, \rho_n}^{-1}$ starting from the root angle,
  we denote $\widetilde{k}$ the number of bits written before reaching
  angle $a(k)$.

 One can check that for all $k,k'\in[\![0,{2\rho_n+\tau_n}]\!]$, we
 have: $$|C_n(k)-C_n(k')|\leq
 |W_n(\widetilde{k})-W_n(\widetilde{k'})|,$$

 and $$|k-k'|\leq |\widetilde{k}-\widetilde{k'}|\leq 3|k-k'|.$$

  We use the definition of function $f$ and $r_F$ defined in
  Section~\ref{subsection:forestandwell-labelings}

  Consider $0\leq x < y \leq \frac{2\rho_n+\tau_n}{2n}$ such that $2nx, 2ny\in
\mathbb N$. Let $s=2nx$ and $t=2ny$.
    It is always possible
  to choose $u,v\in \mathbb N$, such that $s \leq u\leq v \leq t$, and satisfying:
  \begin{itemize}
  \item if $fl(r_F(s))\neq fl(r_F(t))$,
    then $r_F(u),r_F(v)\in F_1$, $r_F(u)=fl(r_F(s))$ and $r_F(v)=fl(r_F(t))$
      \item if  $fl(r_F(s))= fl(r_F(t))$, then $u=v$ and
    $r_F(u)$  is  the  nearest common ancestor of $r_F(s)$ and $r_F(t)$.
  \end{itemize}

  Using the triangular inequality, we get:
  \begin{equation*}
    \begin{split}
      \left    |   C_{n}(s)-C_{n}(t)    \right    |
  &\leq   \left    |
    C_{n}(s)-C_{n}(u)  \right  |+  \left  |  C_{n}(u)-C_{n}(v)
  \right   |+\left   |   C_{n}(v)-C_{n}(t)  \right   |\\
&
  \leq   \left    |
    W_{n}(\widetilde{s})-W_{n}(\widetilde{u})  \right  |+  \left  |  W_{n}(\widetilde{u})-W_{n}(\widetilde{v})
  \right   |+\left   |   W_{n}(\widetilde{v})-W_{n}(\widetilde{t})  \right   |
\end{split}
\end{equation*}

  We obtain
  \begin{align*}
    \left | C_{(n)}(x)-C_{(n)}(y) \right | \leq &\left |
                                             W_{(n)}(\widetilde{s}/2n)-   W_{(n)}(\widetilde{u})\right   |+\left   |
                                             W_{(n)}(\widetilde{u})-
                                                  W_{(n)}(\widetilde{v})\right
                                                  |\\ &+\left   |
                                             W_{(n)}(\widetilde{v})- W_{(n)}(\widetilde{t})\right |\\
                                           \leq &C (\left |
                                             \widetilde{s}/2n-\widetilde{u}
                                             \right |^{\alpha}+\left |
                                             \widetilde{u}-\widetilde{v}
                                             \right |^{\alpha}+\left |
                                             \widetilde{v}-\widetilde{t}
                                             \right |^{\alpha})\\
                                           \leq &C\left(\frac{3}{2n}\right)^\alpha (\left | {s}-{u} \right |^{\alpha}+\left | {u}-{v} \right |^{\alpha}+\left | {v}-{t} \right |^{\alpha}).
  \end{align*}

  Using     the    inequality
  $a^{\alpha}+b^{\alpha}+c^{\alpha}\leq 3(a+b+c)^{\alpha}$, we get

\begin{align*}
  \left | C_{(n)}(x)-C_{(n)}(y) \right | 
  &\leq 3C\left(\frac{
   3}{2n}\right)^\alpha (\left | {s}-{u} \right
   |+\left | {u}-{v} \right |+\left | {v}-{t} \right |)^\alpha\\
  &\leq
    3C\left(\frac{3}{2n}\right)^\alpha |s-t|^\alpha\\
  &\leq
    3^{\alpha+1}C |x-y|^\alpha
\end{align*}

This inequality is satisfied for
$0\leq x < y \leq \frac{2\rho_n+\tau_n}{2n}$ such that
$2nx, 2ny\in \mathbb N$. It is also satisfied for all
$0\leq x < y \leq \frac{2\rho_n+\tau_n}{2n}$ by linear interpolation. 
\end{proof}

\subsection{Conditioned Galton-Watson forest}

In this section, we introduce the notion of Galton-Watson forest which
allows us to present the law of uniform random well-labeled forests.

Let $(F,\ell)$ be a well-labeled forest in $\mathcal F^\rho_\tau$.
 For
 convenience, in this section, we extend the function $d$ to the set
 of tree-edges of $F$ by letting: for all $u\in F$
  such that $c_u(F)\geq 1$, for all $i\in \{1,\cdots, c_u(F)\}$, we define: 
  $$\ell(\{u, u\, i\})=\ell(u\, i)-\ell(u)$$

  Note that the value of $\ell$ on the set of tree-edges of $F$ is
  sufficient to recover $\ell$. 

For       $\tau        \in       \mathbb{N}$, let       
  $\mathbb{F}^{\infty}_\tau = \bigcup_{\rho\geq 0} \mathbb{F}^\rho_{\tau}$.
  
  Let $G$ be a random variable with geometric law of parameter $3/4$
  (i.e. $\mathbb{P}(G=c)=\frac{3}{4}\left(\frac{1}{4}\right)^{c}$ for
  $c\in \mathbb{N}$). Let $B$ be a random variable with law given
  by:
$$\mathbb{P}(B=c)=\frac{{\binom{c+2}{2}\, \mathbb{P}(G=c)}
}{\mathbb{E} \left[\binom{G+2}{2}\right]}, \text{ for }
c\in \mathbb{N} .$$

\begin{definition}
  For $\tau \in \mathbb N$, 
  a \emph{$\tau$-Galton-Watson forest} 
  is a random element $F'$ of $\mathbb{F'}^\infty_{\tau}$ such that,
  independent for each $u\in F'$, we have $c_u(F')$ has law $G$ if $u$
  is a floor and $c_u(F')$ has law $B$ if $u$ is a tree-vertex.
\end{definition}

  Let  $H$ be  a $\tau$-Galton-Watson  forest  conditioned to have $\rho$  tree-vertices. For each tree-vertex
  $v$ of $F'$, we add two stems incident to $v$, uniformly at random from
  among the
  $\binom{c_v(F')+1}{2}+\binom{c_v(F')+1}{1}=\binom{c_v(F')+2}{2}$
  possibilities. Let $(H,\ell)$ be the resulting forest of $\mathcal
  F^\rho_\tau$ (see Section~\ref{subsection:forestandwell-labelings} for the correspondence between stems
  and the function $\ell$).

\begin{lemma}
  \label{bijectionvalid2}
  $(H,\ell)$ is  uniformly distributed over $\mathcal  F^\rho_\tau$. 
\end{lemma}

\begin{proof}
  Let $(F,\ell')\in \mathcal F^\rho_\tau$. For each $1\leq i\leq \tau$,
  assume that the list of vertices of $F^i$ (the $i$-th tree of $F$ as
  in Section~\ref{subsection:forestandwell-labelings}) in
  lexicographic order is $v_{i1},v_{i2},...,v_{in_i}$.  Then $(H,\ell)$
  is equal to $(F,\ell')$ if and only if all the vertices of $H$ and $F$
  have the same number of children and the stems are inserted at the
  right place to obtain $(H,\ell)$ from $H$. Hence we have: \begin{align*} \mathbb{P}((H,\ell)=(F,\ell')) &\propto
    \prod_{i=1}^{\tau}\left[\mathbb{P}(G=c_{v_{i1}}(F))\prod_{j=2}^{n_i}\frac{\mathbb{P}(B=c_{v_{ij}}(F))}{\binom{c_{v_{ij}}(F)+2}{2}}\right]
                        \\
                                                                                                 &=\prod_{i=1}^{\tau}\left[\mathbb{P}(G=c_{v_{i1}}(F))\prod_{j=2}^{n_i}\frac{\binom{c_{v_{ij}}(F)+2}{2}}{\binom{c_{v_{ij}}(F)+2}{2}}\frac{\mathbb{P}(G=c_{v_{ij}}(F))}{\mathbb{E}\left[\binom{G+2}{2}\right
                                                                                                   ]}\right]\\
                                                           &=\frac{3^{\rho+\tau}}{4^{2\rho+\tau}\left(\mathbb{E}\left[\binom{G+2}{2}\right]\right)^\rho}.
  \end{align*}
  Since the last term does not depend on $(F,\ell')$, this concludes the proof
  of the Lemma.
\end{proof}

\begin{definition}
       Consider
$(\rho,\tau) \in  \mathbb{N}^2$, and $\mu=(\mu_k)_{k\geq 1}$ where
$\mu_k$ is a probability measure on $\mathbb R^k$. Let $LGW(\mu,\rho,\tau)$  be the
law of the well-labeled forest $(F,\ell)\in\mathcal F^\rho_\tau$ such that:
\begin{itemize}
\item $F$ has  the law of the  $\tau$-Galton-Watson forest
 conditioned  to  have  $\rho$  tree
  vertices,
\item Conditionally on
  $H$, independently for each tree-vertex $v$ of $H$ such that
  $c_{v}(H)\geq 1$, let
  $(\ell(\left \{ v,v\,j \right \}))_{1\leq j \leq c_{v}(H)}$ be a random
  vector with law $\mu_{ c_{v}(H)}$
\end{itemize}
\end{definition}

Consider $\nu=(\nu_k)_{k\geq 1}$ where $\nu_k$ is the uniform law over non-decreasing
vectors
$(X_1,X_2,...,X_k)\in \left \{ -1,0,1 \right \}^k$ (i.e. $X_1\leq
\ldots \leq X_k$).

\begin{remark}
\label{rem:galtonlaw}  
A consequence of Lemma~\ref{bijectionvalid2}, is that if $(F,\ell)$ is
uniformly distributed on $\mathcal{F}^\rho_\tau$, then the law of $(F,\ell)$
is $LGW(\nu,\rho,\tau)$.
\end{remark}

\subsection{Symmetrization of a forest}
\label{sec:symforest}
We  adapt  a notion first applied  in the case of  plane trees~\cite{berry2013scaling}  to well-labeled  forest. We  begin
this section with the following definition.

\begin{definition}
  Let  $\mu$   be  a   probability  measure  on   $\mathbb{R}^k$.  The
  \emph{symmetrization}  of  $\mu$,  denoted  by  $\widehat{\mu}$,  is
  obtained by uniformly permuting the marginals of $\mu$. In  other
  words,  if $(X_1,X_2,...,X_k)$  has  law $\mu$,  and  $\sigma$ is  a
  uniformly    random    in    the     set    of    permutations    of
  $\left       \{        1,2,...,k       \right        \}$,       then
  $(X_{\sigma(1)},X_{\sigma(2)},...,X_{\sigma(k)})$       has      law
  $\widehat{\mu}$.
\end{definition}

We now describe the symmetrization of $\nu=(\nu_k)_{k\geq 1}$ where $\nu_k$ is the uniform law over non-decreasing
vectors of
$\left \{ -1,0,1 \right \}^k$  (as in previous section).
Assume
that $(X_1,X_2,...,X_k)$  has law $\nu_k$,  and $\bm{\sigma}$ is  a uniform
random element of the  set of permutations of $\left \{  1,2,...,k
\right \}$. Then, for $x=(x_1,x_2,...,x_k)\in       \{-1,0,1\}^k$, we have:
\begin{equation*}
\begin{split}
\widehat{\nu}_k\{x\}=\mathbb{P}\left \{(X_1,X_2,...,X_k)=(x_{\bm{\sigma_x}(1)},x_{\bm{\sigma_x}(2)},...,x_{\bm{\sigma_x}(k)})
\, ;\,  \bm \sigma^{-1}=\bm \sigma_x\right\},
\end{split}
\end{equation*}
where 
$\bm{\sigma_x}$  is a  permutation
of            $\{1,2,...,k\}$
such that
$(x_{\bm{\sigma_x}(1)},x_{\bm{\sigma_x}(2)},...,x_{\bm{\sigma_x}(k)})$
is non-decreasing. Thus, for $x=(x_1,x_2,...,x_k)\in
\{-1,0,1\}^k$, we have:
\begin{equation*}
\begin{split}
\widehat{\nu}_k\{x\}\propto (n_{-1}(x))!(n_0(x))!(n_1(x))!,
\end{split}
\end{equation*}
where
$n_{-1}(x)$, $n_0(x), n_1(x)$ denotes the  number of
occurrences   of    $-1$, $0$, $1$    in $x$, respectively.
Note that the marginals
of $\widehat{\nu}_k$ are  not i.i.d, but  that each of them  has uniform law
on $\{-1,0,1\}$.

Let $(F,\ell)$ be a well-labeled forest in $\mathcal F_\tau^\rho$ for
$\rho,\tau\in \mathbb N$.
We define the following
set of vectors of permutations:
$$\mathcal P(F)=\{  (\bm{p}_v)_{v  \in  F,\,  c_v(F)>0}\, : \,
\bm{p}_v \textrm{ is a permutation of } \left \{ 1,2,...,c_v(F) \right
\}\}.$$

The \emph{symmetrization of $F$ with respect to
  $\bm{p}\in \mathcal P (F)$} is the forest $F_{\bm{p}}$ obtained from
$F$ by permuting the order of the children at each tree-vertex
$v$ according to $\bm{p}_v$. More formally, we
have $$F_{\bm{p}}=\{\overline{\bm{p}}(v)\, : v\in F\,\},$$ where for
$v=v_1\,\ldots\,v_k$ in $F$, we define
$$\overline{\bm{p}}(v)=v_1\,\bm{p}_{v_1}(v_2)\,\bm{p}_{v_1v_2}(v_3)...\,\bm{p}_{v_1\ldots
  v_{k-1}}(v_k).$$

Note that $F$ and $F_{\bm{p}}$ are isomorphic in terms of
(non-embedded) graphs (the image of a vertex $v$ of $F$ is precisely
$\overline{\bm{p}}(v)$ in $F_{\bm{p}}$). We now define two variants of
labeling function $\ell^0_{\bm{p}},\ell^1_{\bm{p}}$ of $F_{\bm{p}}$ by the
following: for each tree-edge $\left \{ u,u\,i \right \}$ of $F$, let
$$\ell^0_{\bm{p}}(\overline{\bm{p}}(u),\overline{\bm{p}}(u)i)=\ell(u,ui)$$
$$\ell^1_{\bm{p}}(\overline{\bm{p}}(u),\overline{\bm{p}}(ui))=\ell(u,ui).$$

Informally, for $\ell^1_{\bm{p}}$, the labels of $F$ are attached to
edges during the permutation of the children and for $\ell^0_{\bm{p}}$,
the labels stay at their initial position and do not move.

The \emph{partial symmetrization of $(F,\ell)$ with respect to
  $\bm{p}\in \mathcal P (F)$} is the well-labeled forest
$(F_{\bm{p}},\ell^0_{\bm{p}})$.  The \emph{complete symmetrization of
  $(F,\ell)$ with respect to $\bm{p}\in \mathcal P (F)$} is the
labeled forest
$(F_{\bm{p}},\ell^1_{\bm{p}})$. Note that
$(F_{\bm{p}},\ell^1_{\bm{p}})$ is not necessarily a well-labeled
forest.

\begin{lemma}
  \label{lem:toto}
  Let $(F,\ell)$ be a random element on $\mathcal F_\tau^\rho$ with law
  $LGW(\nu,\rho,\tau)$ and $\bm{p}$ be a uniform element on
  $\mathcal P(F)$, then $(F_{\bm{p}},\ell_{\bm{p}}^0)$ has law
  $LGW({\nu},\rho,\tau)$ and $(F_{\bm{p}},\ell_{\bm{p}}^1)$ has law
  $LGW(\widehat{\nu},\rho,\tau)$.
\end{lemma}

\begin{proof}
It follows from the branching property of Galton-Watson processes that $F$   and
$F_{\bm{p}}$ have the  same law. The rest follows from the definitions
of $\ell_{\bm{p}}^0,\ell_{\bm{p}}^1,\widehat{\nu}$.
\end{proof}

Recall some notations from
Section~\ref{subsection:forestandwell-labelings}.  For $u\in F$, with
$|u|\geq 2$, $pa(u)$ denotes the parent of $u$ in $F$.  For $u\in F$,
$A_u(F)$ denotes the set of ancestors of $u$ in $F$.  
For $u,v\in F$, we say that  $v< u$ if
 $v\in A_u(F)$. Similarly, we say that $v\leq u$  if
 $v\in (A_u(F)\cup \{u\})$.
 
 Let $U$ be a set of tree-vertices of $F$. We denote
$A_U(F)=\cup_{u\in U}A_u(F)$.
 Let $O_U(F)$ denote the set
 of vertices of $F$ that have exactly one child in $A_U(F)$. Note that
 $O_U(F) \subseteq A_U(F)$. 
We define $\mathcal P_U(F)$ as the subset of vectors $\bm{p}$ of
$\mathcal P(F)$ such that for all $v\in (F\setminus O_U(F))$, we have
$\bm{p}_v$ is equal to identity.  For $\bm{p}\in \mathcal P_U(F)$, we
define $\overline{\bm{p}}(U)=\{\overline{\bm{p}}(u)\, :\, u\in U\}$.

\begin{lemma}
\label{partialsymmetrization}
Let $(F,\ell)$ be a random element on $\mathcal F_\tau^\rho$ with law
$LGW(\nu,\rho,\tau)$. Let $k\in [\![0,\rho+\tau+1]\!]$ and $U$ be a set of $k$
independent and uniformly random vertices of $F$. Let $\bm{p}$ be a
uniformly random element of $\mathcal P_U(F)$. Then $(F,\ell,U)$ and
$(F_{\bm{p}},\ell^0_{\bm{p}},\overline{\bm{p}}(U))$ have the same law.
\end{lemma}

\begin{proof}
  Let $(F',\ell')\in \mathcal{F}^\rho_\tau$, $U'$
be a set of
$k$ vertices of $F'$.
We have:
$$
    \mathbb{P}[
    (F,\ell,U)=(F',\ell',U')]
           =\mathbb{P}[(F,\ell)=(F',\ell')]\times
      \frac{1}{(\rho+\tau+1)^k}
            $$
            \begin{equation*}
              \begin{split}
&    \mathbb{P}[
    (F_{\bm{p}},\ell^0_{\bm{p}},\overline{\bm{p}}(U))=(F',\ell',U')]\\
&    =\sum_{\bm{p'}\in\mathcal P(F,U)}
    \left[
    \mathbb{P}[(F_{\bm{p}},\ell^0_{\bm{p}})=(F',\ell')\, ;
    \,\overline{\bm{p}}(U)=U' ;
    \,{\bm{p}}={\bm{p'}} ]
  \right]\\
  &    =\sum_{\bm{p'}\in\mathcal P(F,U)}
    \left[
    \mathbb{P}[(F_{\bm{p}},\ell^0_{\bm{p}})=(F',\ell')\, |
    \,\overline{\bm{p}}(U)=U' ;
    \,{\bm{p}}={\bm{p'}} ]
    \times
\mathbb{P}[\overline{\bm{p}}(U)=U' |
    \,{\bm{p}}={\bm{p'}} ]
    \times
    \mathbb{P}[
    {\bm{p}}={\bm{p'}} ]
  \right]\\
\end{split}
\end{equation*}
  
for all ${\bm{p'}\in\mathcal P_U(F)}$
we have
$
\mathbb{P}[(F_{\bm{p}},\ell^0_{\bm{p}})=(F',\ell')\, |
    \,\overline{\bm{p}}(U)=U' ;
    \,{\bm{p}}={\bm{p'}} ]
=
\mathbb{P}[(F,\ell)=(F',\ell')]$ and

$\mathbb{P}[\overline{\bm{p}}(U)=U' |
    \,{\bm{p}}={\bm{p'}} ]= \frac{1}{(\rho+\tau+1)^k}$
thus we obtain the result.
\end{proof}

We obtain the following lemma (similar to~\cite[Corollary~6.7]{berry2013scaling}).

\begin{lemma}
  \label{corollarypartialsymmetrization}
Let $(F,\ell)$ be a random element on $\mathcal F_\tau^\rho$ with law
 $LGW(\nu,\rho,\tau)$. Let $k\in [\![0,\rho+\tau+1]\!]$ and $U$ be a set of $k$
 independent and uniformly random vertices of $F$. 
  Let $(\widehat F,\widehat \ell)$ be a random element with law
  $LGW(\widehat \nu,\rho,\tau)$. Let $\widehat U$ be a set of $k$
 independent and uniformly random vertices of $\widehat F$.
 Let $U=\{u_1,\ldots,u_k\}$  and $\widehat U=\{\widehat
 u_1,\ldots,\widehat u_k\}$   such that  $u_1,\ldots,u_k$
  and $\widehat
  u_1,\ldots,\widehat u_k$ are  lexicographically ordered.
  For $1\leq i \leq k$, let
  
  $$S_i=\sum_{\substack{v\leq u_i\\
    w\in O_U(F) \\ w=pa(v)}}\ell(w,v)$$
  
  $$\widehat S_i=\sum_{\substack{v\leq {\widehat u_i}\\
    w\in O_{\widehat U}(\widehat F) \\ w=pa(v)}}\widehat \ell(w,v).$$
    Then 
   $(|u_1|,\ldots,|u_k|,S_1,\ldots,S_k)$ and $\left(|\widehat{u}_1|,\ldots,|\widehat{u}_k|,\widehat{S}_1,\ldots,\widehat{S}_k\right)$ have the same law.
\end{lemma}

\begin{proof}
  Let $\bm p$ be a uniformly random element of $\mathcal{P}_U(F)$ and
  consider $(F_{\bm{p}},\ell^0_{\bm{p}},\overline{\bm{p}}(U))$. For
  $v\in F_{\geq 2}$, if $\{pa(v),v\}$ is a tree-edge of $F$ such that
  $pa(v)\in O_U(F)$, then the partial symmetrization of $(F,\ell)$ with
  respect to $\bm p$ uniformly permutes the children of $pa(v)$ but the
  labels are not permuted.  Consider two distinct vertices $u,v\in F$ such that
  $c_u(F),c_v(F)$ are at least $1$. If $u',v'$ are children of $u,v$,
  respectively, then the  values of $\ell(u,u')$ and $\ell(v,v')$ are independent.
It follows that
  the random variables
  $$\left\{\ell_{\bm
      p}^0(\overline{\bm{p}}(w),\overline{\bm{p}}(v))\,:\, 
v\in F,
w\in O_U(F)
 \,\, \,  \text{and} \,           \, \, 
w=pa(v)\right\}
$$
  are independent and uniformly distributed on $\{-1,0,1\}$.

  Thus,
  by Lemma~\ref{partialsymmetrization}, the random variables
  $$\left\{{\ell}(w,v)\,:\,
v\in F,
w\in O_U(F)
 \,\, \,  \text{and} \,           \, \, 
w=pa(v)\right\}
$$
  are
  independent and uniformly distributed on $\{-1,0,1\}$.\\
  Finally, the trees $F$ and $\widehat{F}$ have the same law, so
  $(|u_1|,\ldots ,|u_k|)\overset{(\mathrm
    d)}{=}(|\widehat{u}_1|,\ldots ,|\widehat{u}_k|)$. Moreover, by the
  definition of $\widehat{\nu}$, the random variables
  $$\left\{\widehat {\ell}( w, v)\,:\,
v\in \widehat F,
w\in O_{\widehat U}(\widehat F)
 \,\, \,  \text{and} \,           \, \, 
w=pa(v)\right\}
$$ are
  independent and uniformly distributed on $\{-1,0,1\}$, and the
  result follows.
\end{proof}

\subsection{Tightness of  the labeling function of  a symmetrized Galton-Watson
  forest}
\label{TightnessofthelabelingfunctionofaGaltonWatsonforest}
Recall that $\nu=(\nu_k)_{k\geq 1}$ where $\nu_k$ is the uniform law
over non-decreasing vectors of $\left \{ -1,0,1 \right \}^k$ and
$\widehat \nu$ is the symmetrization of $\nu$ as defined in previous
section.

By Remark~\ref{rem:galtonlaw}, $(F_n,\ell_n)$ is a
random element with law $LGW({\nu},\tau_n,\rho_n)$.
Now consider $(\widehat F_n,\widehat \ell_n)$   a
random element with law $LGW(\widehat{\nu},\tau_n,\rho_n)$. For convenience, we write
$(\widehat{C}_n,\widehat{L}_n)$ the contour pair $({C}_{\widehat F_n},{L}_{(\widehat F_n,\widehat \ell_n)})$ of
$(\widehat F_n,\widehat \ell_n)$.  As before, we consider that
$\widehat{C}_n$ and $\widehat{L}_n$ are linearly
interpolated. We extend $\widehat{C}_n$ to be equal to
$\tau_n$ on $[2\rho_n+\tau_n,2n\rho]$ when $2\rho_n+\tau_n<2n\rho$. Then we define
the rescaled versions:
$$\widehat C_{(n)}=\left(\frac{\widehat
    C_n(2ns)}{\sqrt{3n}}\right)_{0\leq s\leq \max\left(\frac{2\rho_n+\tau_n}{2n},\rho\right)}\,
\, \, \,  \text{and}\, \, \, \,  \widehat L_{(n)}=\left(\frac{\widehat
    L_n(2ns)}{n^{1/4}}\right)_{0\leq s\leq \frac{2\rho_n+\tau_n}{2n}} $$

The aim  of this  section is  to prove the  tightness of  the labeling
function    $\widehat L_{(n)}$.

Since $F_n$ and $\widehat F_n$ do not depend on $\nu$ and
$\widehat \nu$, they have the same law. So the contour functions
$\widehat{C}_{n}$ and $C_{n}$ have the same law (but not necessarily
$\widehat{L}_{n}$ and $L_{n}$).  Thus we can couple the two labeled
forests $(\widehat F_n,\widehat \ell_n)$ and $(F_n,\ell_n)$ so that
$\widehat{C}_{n}= C_{n}$.

We  need the following classical inequality:
\begin{lemma}[Rosenthal's inequality, \cite{petrov1995limit}]
  \label{Rosenthal}
  For each $p\geq 2$, there exists a constant $C_p>0$ such that for
  $k\geq 1$ we have the following. Consider $X,X_1, \ldots, X_k$ a
  sequence of $i.i.d.$ centered random variables in $\mathbb R$.  Let
  $\Sigma=\sum_{i=1}^{k}X_i$. Then:
  $$\mathbb{E}(\left |  \Sigma \right |^p)\leq C_p\left(k\, \mathbb{E}(\left  | X^p
    \right |)+(k\, \mathbb{E}(X^2))^{p/2}\right).$$
\end{lemma}

 We now prove the main result of this section:

\begin{lemma}[Tightness of the labeling function]
  \label{tightofz_n}
  The family of laws of $\left(\widehat L_{(n)}\right)_{n\geq  1}$ is tight for the space of
  probability measure on $\mathcal H$.
\end{lemma}

\begin{proof}
  By  Lemma~\ref{lemmatightofvn}, there exists a constant $\alpha>0$ such that
  $$\forall \epsilon >0 \quad \exists C \quad
  \forall   n   \quad   \mathbb{P}\left(\|  C_{(n)}   \|_\alpha   \leq
    C\right)> 1-\epsilon.$$  

Let $\epsilon >0$ and $C$ that satisfies the above inequality.

We assume that $C_{(n)}$ is conditioned on
$\| C_{(n)} \|_\alpha \leq C$.

Let $X$ be uniformly distributed in $\{-1,0,1\}$.
Recall that  the marginals of $\widehat \nu_k$ for $k\geq 1$, have the
same law as $X$. So
for all $a,b\in \widehat F$ with $a=p(b)$, we have
$\widehat  \ell_n(a,b)$ and $X$ have the same law.

One can check that for all $i,j\in[\![0,{2\rho_n+\tau_n}]\!]$, with
$u=r_{\widehat F_n}(i)$, $v=r_{\widehat F_n}(j)$,
$u\in A_{v}(\widehat F_n)$, we have:
$$\widehat L_{n}(j)-\widehat L_{n}(i)=\sum_{\substack{u< b\leq v\\
    a=p(b)}} \widehat \ell_n(a,b)$$
Let $k=|v|-|u|$. Note that $k=C_n(j)-C_n(i)$.
Then by Lemma~\ref{Rosenthal}, we have, for $p\geq 2$,
 there exists a constant $C_p>0$ such that:

\begin{align*}
  \mathbb{E}\left(\left |  \widehat L_{n}(j)-\widehat L_{n}(i) \right |^p\right)
  &\leq C_p\left(k\,\mathbb{E}(\left  | X^p
    \right |)+(k\,\mathbb{E}(X^2))^{p/2}\right)\\
    &\leq C_p\left((C_n(j)-C_n(i))\,\mathbb{E}(\left  | X^p
    \right |)+((C_n(j)-C_n(i))\,\mathbb{E}(X^2))^{p/2}\right) \\
  &\leq C_p\, C \, \sqrt{3n} \left(\left |\frac{j-i}{2n}\right |^\alpha\,\mathbb{E}(\left  | X^p
    \right |)+\left (\left |\frac{j-i}{2n}\right
      |^\alpha\,\mathbb{E}(X^2)\right )^{p/2}\right)   
\end{align*}

As in the proof of Lemma~\ref{lemmatightofvn},
  consider $0\leq x < y \leq \frac{2\rho_n+\tau_n}{2n}$ such that $2nx, 2ny\in
\mathbb N$. Let $s=2nx$ and $t=2ny$. Let $u=r_{\widehat F}(s)$ and $v=r_{\widehat F}(t)$.
    It is always possible
  to choose $p,q\in \mathbb N$, such that $s \leq i\leq j \leq t$, and satisfying:
  \begin{itemize}
  \item if $fl(u)\neq fl(v)$,
    then $r_{\widehat F}(i),r_{\widehat F}(j)\in ({\widehat F})_1$, $r_{\widehat F}(i)=fl(u)$ and $r_{\widehat F}(j)=fl(v)$
      \item if  $fl(u)= fl(v)$, then $i=j$ and
    $r_{\widehat F}(i)$  is  the  nearest common ancestor of $u$ and $v$.
  \end{itemize}

  Note that ${\widehat L}_{n}(i)={\widehat L}_{n}(j)=0$, so we have:

  \begin{align*}
   \mathbb{E}\left[\left     |     {\widehat L}_{n}(s)-{\widehat
    L}_{n}(t)    \right     |^p\right ]
    \leq&
    3^p\left(
    \mathbb{E}\left[\left     |     {\widehat L}_{n}(s)-{\widehat
    L}_{n}(i)    \right     |^p\right ]+
    \mathbb{E}\left[\left     |     {\widehat L}_{n}(i)-{\widehat
    L}_{n}(j)    \right     |^p\right ]\right.\\ & \left.+
        \mathbb{E}\left[\left     |     {\widehat L}_{n}(j)-{\widehat
    L}_{n}(t)    \right     |^p\right ]
                                          \right )\\
     \leq&
      3^p\, C_p\, C \, \sqrt{3n}
      \left(
      \left |\frac{s-i}{2n}\right |^\alpha\,\mathbb{E}\left(\left  | X^p
    \right |\right)+\left (\left |\frac{s-i}{2n}\right
      |^\alpha\,\mathbb{E}\left(X^2\right)\right )^{p/2}
      \right.
    \\
    & +
      \left.
      \left |\frac{j-t}{2n}\right |^\alpha\,\mathbb{E}\left(\left  | X^p
    \right |\right)+\left (\left |\frac{j-t}{2n}\right
      |^\alpha\,\mathbb{E}\left(X^2\right)\right )^{p/2}
      \right)   
     \end{align*}

     Thus for the rescaled version, we have:
  \begin{align*}
   \mathbb{E}\left[\left     |     {\widehat L}_{(n)}(x)-{\widehat
    L}_{(n)}(y)    \right     |^p\right ]
     \leq&
      n^{-p/4}3^p\, C_p\, C \, \sqrt{3n}
      \left(
      \left |\frac{s-i}{2n}\right |^\alpha\,\mathbb{E}\left(\left  | X^p
    \right |\right)+\left (\left |\frac{s-i}{2n}\right
      |^\alpha\,\mathbb{E}\left(X^2\right)\right )^{p/2}
      \right.
    \\
    & +
      \left.
      \left |\frac{j-t}{2n}\right |^\alpha\,\mathbb{E}\left(\left  | X^p
    \right |\right)+\left (\left |\frac{j-t}{2n}\right
      |^\alpha\,\mathbb{E}\left(X^2\right)\right )^{p/2}
      \right)
\\
       \leq &
      n^{-p/4}3^p\, C_p\, C \, \sqrt{3n}
      \left(
      \left |x-y\right |^\alpha\,\mathbb{E}\left(\left  | X^p
    \right |\right)+\left (\left |x-y\right
      |^\alpha\,\mathbb{E}\left(X^2\right)\right )^{p/2}
      \right)
     \end{align*}
     
     Consider $p$ such that $p >10$. So we have
     $n^{-p/4+1/2}\leq 1/n^{2}$.  Since $2nx,2ny\in\mathbb N$, and
     $x\neq y$, we have $\left | x-y \right |\geq \frac{1}{2n}$. So
     $n^{-p/4+1/2}\leq 4 \left | x-y \right |^2$.
     Moreover, we have $\left | x-y \right |\leq
     \frac{2\rho_n+\tau_n}{2n}$ which converge to $\rho$.
So  there exists a constant $C'$, such that:
       \begin{align*}
   \mathbb{E}\left[\left     |     {\widehat L}_{(n)}(x)-{\widehat
    L}_{(n)}(y)    \right     |^p\right ]
\leq C' |x-y|^2 
     \end{align*}

     Since $\widehat L_{n}$ is linearly interpolated, the above inequality
     holds for for all
     $x,y \in \left[0,\frac{2\rho_n+\tau_n}{2n}\right]$.  By Billingsley
     (\cite{billingsley2013convergence}, Theorem $12.3$), the family
     of laws of $\left(\widehat L_{(n)}\right)_{n\geq 1}$ is tight,
     which completes the proof of the Lemma.
\end{proof} 

\subsection{Convergence of the contour function}

We consider  $\widehat F_n$, $\widehat \ell_n$
$\widehat{L}_n$, $\widehat{L}_{(n)}$ as in previous section.

Here, we prove the convergence of the contour function by
using the convergence of uniformly random $3$-dominating binary words from
Section~\ref{section6} and the tightness of ${\widehat L}_{(n)}$ from
Section~\ref{TightnessofthelabelingfunctionofaGaltonWatsonforest}.

We need the following bound:

\begin{lemma}
  \label{faibletightofz_n}
  For  all  $\epsilon>0$,  there   exists  a  constant  $C$
such that
  $$\sup_{n}\ \mathbb{P}\left(
    \sup_{v \in  F_n}| \ell_n(v)|
    \geq C\,n^{1/4}\right) < \epsilon.$$
  \end{lemma}

\begin{proof}
    For  any   $\epsilon  >0$,  by
  Lemma~\ref{tightofz_n}, there exists  a constant $C$ such
  that:
  $$\sup_{n}\ \mathbb{P}\left(
    \sup_{v \in \widehat F_n}|\widehat \ell_n(v)|
    \geq C\,n^{1/4}\right) < \epsilon.$$
 
  Let $\bm{p}$ be a uniform random element of $\mathcal
  P(F_n)$. Denote by $(\ell_n)^1_{\bm{p}}$ the labeling function of
  $(F_n)_{\bm{p}}$ as defined in Section~\ref{sec:symforest}. For
  convenience we write $\ell'_n=(\ell_n)^1_{\bm{p}}$ and
  $F'_n=(F_n)_{\bm{p}}$.
  Note
  that for all $v \in F_n$, we have
  $\ell_n(v)=\ell'_n(\overline {\bm{p}}(v))$.  Then we have
  \begin{equation*}
  \sup_{v \in F_n}|\ell_n(v)|=\sup_{v \in F'_n}|\ell'_n(v)|.
  \end{equation*}
  By Lemma~\ref{lem:toto}, we have $(F'_n, \ell'_n)$
  has law $LGW(\widehat{\nu},\tau_n,\rho_n)$, i.e.
  $(F'_n, \ell'_n)$ and
  $(\widehat{F}_n,\widehat{\ell}_n)$ have the same law.
So 
$$\sup_{v \in F_n}|\ell_n(v)|=\sup_{v \in \widehat F_n}|\widehat \ell_n(v)|.$$

This completes the proof of the Lemma.
\end{proof}

\begin{lemma}[Convergence of contour function]
  \label{convergenceofcontour}
  The   process  $C_{(n)}$  converges   in  law   toward
  $\widetilde F^{0\rightarrow  \tau}_{\left  [ 0,
      \rho  \right  ]}$  in the  space
  $(\mathcal            H,d_\mathcal             H)$, when $n$ goes to infinity.
\end{lemma}

\begin{remark}
  Note    that    the    limit     in    this    lemma    is    indeed
  $\widetilde  F^{0\to\tau}_{[0,\rho]}$ and  not $F^{0\to\tau}_{[0,\rho]}$
  as  the corresponding  result in~\cite{bettinelli2010scaling}  would
  seem to indicate. This is due  to the fact that our decomposition of
  the unicellular map  into Motzkin paths and  well-labeled forests is
  not exactly the same as in the case of quadrangulations.
\end{remark}

\begin{proof}

    Let $f$ be a bounded continuous function from $\mathbb R$ to
    $\mathbb R$. Let $0\leq t <\rho$ and $\epsilon >0$.
    Since $\frac{2\rho_n+\tau_n}{2n}$ converge toward $\rho$,
    there
  exists $N$ such that $t\leq \min_{n\geq N}{\frac{2\rho_n+\tau_n}{2n}}$.
For $n\geq N$, we define
    $$T_n(t)=\min\left \{ k\in [\![0,2\rho_n+\tau_n]\!]\,:\, r_{F_n}(k)=fl
      (r_{F_n}(\lfloor 2nt\rfloor ))\right \}.$$
Note that $r_{F_n}(T_n(t))$ is an integer that we denote  by 
$i_n$.

As in the proof of Lemma~\ref{lemmatightofvn}, we consider $W_n$ the
element of $\mathcal P_{3,3\rho_n+\tau_n, \rho_n}$ corresponding to
$(F_n,\ell_n)$.  Note that for $k\in [\![0,2\rho_n+\tau_n]$, such that
$r_{F_n}(k)$ is a floor of $F_n$, we have
$C_n(k)=W_n(2k-r_{F_n}(k))$. So in particular:
$$C_n(T_n(t))=W_n(2T_n(t)+i_n).$$

For convenience, let $j_n=L_n(\lfloor 2nt\rfloor)$ and
$k_n=i_n-j_n+|r_{F_n}(\lfloor 2nt\rfloor)|$. Note that we have:
 $$C_n(\lfloor 2nt\rfloor)-C_n(T_n(t))=\frac{1}{2}\left(W_n(2\lfloor
   2nt\rfloor+k_n)-W_n(2T_n(t)+i_n)-j_n\right).$$

 Thus we have:
 $$C_n(\lfloor 2nt\rfloor)=\frac{1}{2}\left(W_n(2\lfloor 2nt\rfloor+k_n)+W_n(2T_n(t)+i_n)-j_n\right).$$

Note that $W_n(2T_n(t)+i_n)=\max_{s\leq 2\lfloor 2nt\rfloor+k_n}W_n(s)$, therefore:
$$C_n(\lfloor 2nt\rfloor)=\frac{1}{2}\left(W_n(2\lfloor 2nt\rfloor+k_n)+\max_{s\leq 2\lfloor 2nt\rfloor+k_n}W_n(s)-j_n\right).$$

  By Lemma~\ref{lemmatightofvn},  there exists  a constant  $c_1$ such
  that
  \begin{equation}
    \label{equation2345}
    \inf_{n\geq N}\mathbb{P}\left(\sup_{k\in[\![0, 2\rho_n+\tau_n]\!]}\left | C_n(k) \right |<c_1n^{1/2}\right)\geq 1-\epsilon.
  \end{equation}
  Moreover, by  Lemma~\ref{faibletightofz_n}, there exists  a constant
  $c_2$ such that
  \begin{equation}
    \label{equation3456}
    \inf_{n\geq N}\mathbb{P}\left (\sup_{k\in[\![0, 2\rho_n+\tau_n]\!]}\left | L_n(k) \right |<c_2n^{1/4}\right)\geq 1-\epsilon.
  \end{equation}
  
 By \eqref{equation2345}  and \eqref{equation3456}, there
  exists a constant $c>0$ such that:
 
   \begin{equation*}
       \inf_{n\geq N}\mathbb{P}\left(\sup_{k\in[\![0, 2\rho_n+\tau_n]\!]}\left |
        C_n(k) \right |<cn^{1/2}\ ; \ \sup_{k\in[\![0, 2\rho_n+\tau_n]\!]}\left | L_n(k) \right |<cn^{1/4} \right)\geq 1-2\epsilon.
  \end{equation*}

  So we have:
  
  \begin{equation}
    \label{equation5678}
    \inf_{n\geq N}\mathbb{P}\left(\left | i_n \right |\leq c\,n^{1/2}, \left | j_n \right |\leq c\,n^{1/4},\left | k_n \right |\leq c\,n^{1/2}\right)\geq 1-2\epsilon. 
  \end{equation}


  Let $\mathcal E_n$ the event:
  $$\left \{ \left | i_n \right |\leq c\,n^{1/2}, \left | j_n \right
    |\leq c\,n^{1/4},\left | k_n  \right |\leq c\,n^{1/2} \right \}.$$

  Now we  define  a   random   variable  $Y_n$   as   follows:
    $$Y_n=\frac{1}{2}\left(W_n(2\lfloor 2nt\rfloor+k_n\mathbbm 1_{\mathcal E_n})+\max_{s\leq 2\lfloor 2nt\rfloor+k_n}W_n(s)-j_n\mathbbm 1_{\mathcal E_n}\right).$$

    By Lemma~\ref{convergetowardbr}, we have
    $\left(\frac{Y_n}{\sqrt{3n}}\right)_{n\geq N}$ converge toward
    $\frac{1}{2}\left(F^{0\rightarrow
          \tau}_{[0,2\rho]}(2t)+\overline {F^{0\rightarrow
            \tau}_{[0,2\rho]}}(2t)\right
      )=  \widetilde F^{0\rightarrow \tau}_{[0,\rho]}(t)$ when $n$ goes to
    infinity.  Thus by (\ref{equation5678}), there exists $n_0\geq N$ such
    that for all $n\geq n_0$:

  \begin{equation*}\begin{split}
  &    \left | \mathbb{E}[f(C_{(n)}(t))]
    -\mathbb{E}\left[f\left(
\widetilde F^{0\rightarrow \tau}_{[0,\rho]}(t)
          \right)\right]\right    |
\\&\leq
     \left | \mathbb{E}[f(C_{(n)}(t))]
        -\mathbb{E}\left[f\left(\frac{Y_n}{\sqrt{3n}}\right)\right]\right    |+
       \left | \mathbb{E}\left[f\left(\frac{Y_n}{\sqrt{3n}}\right)\right]
         -\mathbb{E}\left[f\left(
\widetilde F^{0\rightarrow \tau}_{[0,\rho]}(t)
          \right)\right]\right    |
      \\
      &\leq 2\,\mathbb{E}[1-\mathbbm 1_{\mathcal E_n}]\,\left    \|f    \right
    \|_{\infty}+\epsilon.
      \\
      &\leq
    (4\left    \|f    \right
    \|_{\infty}+1)\epsilon.
  \end{split}
  \end{equation*}

      This implies that $\left(\mathbb{E}[f(C_{(n)}(t))]\right)_{n\geq N}$
  converge toward $\mathbb{E}\left[f\left(\widetilde F^{0\rightarrow
        \tau}_{[0,\rho]}(t)
        \right)\right]$.

        We now prove the finite dimensional convergence of
        $C_{(n)}$. Let $k\geq 1$ and consider
        $0\leq t_1< t_2<...< t_k < \rho$.  Let $N$ such that
        $t_{k}\leq \min_{n\geq N}{\frac{2\rho_n+\tau_n}{2n}}$.  By
        above arguments, for $1\leq i\leq k$, we have
        $(C_{(n)}(t_i))_{n\geq N}$ converge in law toward
        $\widetilde F^{0\rightarrow \tau}_{[0,\rho]}(t_i)$

        It remains to
          deal with the point $\rho$.  
\begin{align*}
  \left |C_{(n)}\left (\rho\right )-\tau  \right|&
  =\left |C_{(n)}\left (\rho\wedge \frac{2\rho_n+\tau_n}{2n}\right)-\tau  \right|\\
  &=\left |C_{(n)}\left (\rho\wedge \frac{2\rho_n+\tau_n}{2n}\right)-C_{(n)}\left (\frac{2\rho_n+\tau_n}{2n}\right)+C_{(n)}\left (\frac{2\rho_n+\tau_n}{2n}\right)-\tau  \right|\\
    &\leq \left |C_{(n)}\left (\rho\wedge \frac{2\rho_n+\tau_n}{2n}\right)-C_{(n)}\left (\frac{2\rho_n+\tau_n}{2n}\right)\right|+\left |\tau_n-\tau  \right| \\
\end{align*}

Suppose that 

Consider $\epsilon>0$.
By Lemma~\ref{lemmatightofvn}, there exists $\alpha$ and $C$ such that
for all $n$:
$\mathbb{P}\left(\|  C_{(n)}   \|_\alpha   \leq
  C\right)> 1-\epsilon$. Condition on the event $\{\|  C_{(n)}   \|_\alpha
\leq C\}$,
  we have
\begin{align*}\left |C_{(n)}\left (\rho\wedge
    \frac{2\rho_n+\tau_n}{2n}\right)-C_{(n)}\left
    (\frac{2\rho_n+\tau_n}{2n}\right)\right|&\leq
C\,\left|\rho\wedge
\frac{2\rho_n+\tau_n}{2n} - \frac{2\rho_n+\tau_n}{2n}\right|^\alpha\\
&\leq
C\,\left|\rho - \frac{2\rho_n+\tau_n}{2n}\right|^\alpha
\end{align*}

Since  $\frac{2\rho_n+\tau_n}{2n}\rightarrow \rho$ and $\tau_n \rightarrow
\tau$, for $n$ large enough, we have:
\begin{align*}
  \left |C_{(n)}\left (\rho\right )-\tau  \right|&
                                                \leq \epsilon
\end{align*}

Therefore we obtain for $n$ large enough:
$$\mathbb{P}\left(\left |C_{(n)}\left (\rho\right)-\tau  \right|>\epsilon\right)\leq \mathbb{P}\left(\|  C_{(n)}   \|_\alpha
  > C \right)\leq \epsilon.$$ This implies that
$C_{(n)}\left (\rho\right)$ converges in probability toward the
deterministic value $\tau$. So Slutzky's lemma shows that
$C_{(n)}\left (\rho\right)$ converges in law toward $\tau$. Thus we have
proved the convergence of the finite-dimensional marginals of
$C_{(n)}$ toward
$\widetilde F^{0\rightarrow \tau}_{[0,\rho]}$. By
Lemma~\ref{lemmatightofvn}, $C_{(n)}$ is tight so Prokhorov's lemma
give the result.
\end{proof}

\begin{remark}
  In  the  case  when  $\tau_n=1$   for  all  $n$,  this  provides  an
  alternative  proof of  a  particular  case of  a  theorem of  Aldous
  (\cite{aldous1991continuum}, Theorem $2$).
\end{remark}

\subsection{Convergence of the contour pair}
\label{proofofmain1}
We consider  $ F_n$, $ \ell_n$,
 ${C}_n$,
 ${L}_n$,
 $\widehat F_n$, $\widehat \ell_n$,
 $\widehat{L}_n$,
 as in previous sections.

By Lemma~\ref{convergenceofcontour},  the rescaled contour function $C_{(n)}$ converge. 
{So as in \cite[Corollary~16]{bettinelli2010scaling} one obtain} the
following lemma which proof is omitted:

\begin{lemma}
   \label{convergeofforest2}
In the  sense of
  weak convergence in the space $(\mathcal H,d_{\mathcal H})^2$ when
  $n$ does to infinity, we have:
  $$\left(C_{(n)},\widehat L_{(n)}\right)           \to
  \left(\widetilde F^{0\rightarrow    \tau}_{\left    [     0,\rho    \right    ]},
    Z^\tau_{[0,\rho]}\right).$$
\end{lemma}

\begin{lemma}
  \label{tightofZ_n}
  The  family of laws  of  $(L_{(n)})_{n\geq  1}$  is tight  in  the  space  of
  probability measures on $\mathcal H$.
\end{lemma}

\begin{proof}
    We  prove  that for all $\epsilon >0$, there  exists $\delta >0$
  such that
  \begin{equation}
    \label{equation7.2.3.1}
    \limsup_{n}\mathbb{P}(  \sup_{\left  |  i-j  \right  |\leq  \delta
      (2\rho_n+\tau_n)}\left     |      L_n(i)-L_n(j)     \right     |>\epsilon
    n^{1/4})<\epsilon
  \end{equation}
  For $n \geq 1$, let ${\bm p}_n$ be a uniformly random element of
  $\mathcal P(F_n)$ and let
  $(F'_n,\ell'_n)=((F_n)_{{\bm p}_n},(\ell_n)^1_{{\bm p}_n})$ be the
  complete symmetrization of $F_n$ with respect to ${\bm p}_n$ (see
  Section~\ref{sec:symforest} for the definition).

  By Lemma~\ref{convergeofforest2}, we have
  \begin{equation}
    \label{equation7.2.3.2}
    ((3n)^{-1/2}{C}_n,
    n^{-1/4}\widehat{L}_n) \to (F^{0\rightarrow  \tau}_{\left [ 0,\rho
      \right ]}, Z^\tau_{[0,\rho]}).
  \end{equation}
  This   implies   that   for   all  $\epsilon   >0$,   there   exist
  $\alpha>0$ and $\beta>0$ such that:
  \begin{equation}
    \label{equation7.2.3.3}
    \sup_{n}\mathbb{P}\left(\sup_{     |      i-j     |\leq     \alpha
        (2\rho_n+\tau_n)}|           \widehat{L}_n(i)-\widehat{L}_n(j)
      |>\epsilon n^{1/4}\right)<\epsilon \quad \text{and}
  \end{equation}
  \begin{equation}
    \label{equation7.2.3.4}
    \sup_{n}\mathbb{P}\left(\sup_{\substack{i,j\in[\![0,2\rho_n+\tau_n]\!]\\
          d_{F'_n}(\overline{\bm p}\left(r_{F_n}(i) \right),\overline{\bm p}\left(r_{F_n}(j) \right))\leq \beta n^{1/2}}}| \widehat{L}_n(i)-\widehat{L}_n(j)
        |>\epsilon n^{1/4}\right)<\epsilon.
  \end{equation}

  Indeed, the  existence of  $\alpha$ is a  direct consequence  of the
  convergence of the sequence $(n^{-1/4}  \hat L_n)$ seen as functions
  on the  integers, while  the existence of  $\beta$ follows  from the
  continuity           of           $Z_{[0,\rho]}^{\tau}$           on
  $\mathcal T = {\mathcal T}_{F_{[0,\rho]}^{0\to \tau}}$ equipped with
  the distance  $d_{\mathcal T}$  (see Remark~\ref{domainofZ}  and the
  paragraphs before  it): fix  $\epsilon>0$ and $\eta>0$,  $n_0$ after
  which
  $d_{\mathcal H}( (3n)^{-1/2} C_n  , F_{[0,\rho]}^{0\to\tau}) < \eta$
  and
  $d_{\mathcal  H}( n^{-1/4}  \widehat L_n  , Z_{0,\rho}^\tau)<\epsilon/3$
  and use the  domination of $d_{\mathcal T}$ by $d_{F}$ (the limit of
  $d_{F_n}$) to write for
  $n \geq n_0$
  \begin{equation}
    \sup_{  \substack{i,j\in[\![0,2\rho_n+\tau_n]\!]\\  d_{F_n}(r_{F_n}(i),r_{F_n}(j))\leq \beta n^{1/2}}}| \widehat{L}_n(i)-\widehat{L}_n(j)
    | \leq \frac {2\epsilon}3  + \sup_{\substack{u,v\in[0,\rho]\\ d_{\mathcal T}(u,v)\leq \beta
      + 2 \eta}}
    | Z_{[0,\rho]}^\tau(u) - Z_{[0,\rho]}^\tau(v)|
  \end{equation}
  which can  be made smaller  that $\epsilon$ by choosing  $\beta$ and
  $\eta$ appropriately;  the (finitely  many) cases $n  < n_0$  can be
  taken into account by making $\beta$ even smaller if needed.

Next, one can see that, for all $i,j\in[\![0,2\rho_n+\tau_n]\!]$: 
$$ d_{F_n}(r_{F_n}(i),r_{F_n}(j))=d_{{F'}_n}(\overline{\bm p}\left(r_{F_n}(i) \right), \overline{\bm p}\left(r_{F_n}(j) \right)), \text{ and }$$
$$|L_n(i)-L_n(j)|= | \ell_{n}(r_{F_n}(i))-\ell_{n}(r_{F_n}(j))  |= | \ell'_{n}(\overline{\bm p}\left(r_{F_n}(i) \right))-\ell'_{n}(\overline{\bm p}\left(r_{F_n}(j) \right)) |=|\widehat{L}_n(i)-\widehat{L}_n(j)  |.$$

We have for all $n\geq 1$ and $\delta \in [0,1]$:
\begin{align*}\ &\mathbb{P}\left(\sup_{\substack{i,j\in[\![0,2\rho_n+\tau_n]\!] \\\left | i-j \right |\leq \delta
                  (2\rho_n+\tau_n)}} \left | L_n(i)-L_n(j) \right
                  |>\epsilon n^{1/4} \right)\\
  = &\mathbb{P}\left(\sup_{\substack{i,j\in[\![0,2\rho_n+\tau_n]\!] \\
  \left | i-j \right |\leq \delta (2\rho_n+\tau_n)}}
  |\widehat{L}_n(i)-\widehat{L}_n(j)|
    >\epsilon n^{1/4} \right)\\
\leq &\mathbb{P}\left(\exists i,j:  i,j\in[\![0,2\rho_n+\tau_n]\!]\, ,
          d_{F'_n}(\overline{\bm p}\left(r_{F_n}(i) \right),\overline{\bm p}\left(r_{F_n}(j) \right))\leq \beta n^{1/2}, |\widehat{L}_n(i)-\widehat{L}_n(j)  |>\epsilon n^{1/4}  \right)\\
&+\mathbb{P}\left(\exists i,j: i,j\in[\![0,2\rho_n+\tau_n]\!]\, ,\left
                                                                                                                                                                                             |i-j\right | \leq \delta (2\rho_n+\tau_n), d_{F'_n}(\overline{\bm p}\left(r_{F_n}(i) \right),\overline{\bm p}\left(r_{F_n}(j) \right)) \geq \beta n^{1/2} \right).
  \\
                                                    \leq &\epsilon + \mathbb{P}\left(\exists i,j: i,j\in[\![0,2\rho_n+\tau_n]\!]\, ,\left |i-j\right | \leq \delta (2\rho_n+\tau_n), d_{F_n}(r_{F_n}(i) ,r_{F_n}(j) ) \geq \beta n^{1/2} \right).
\end{align*}

Moreover, we can see that 
\begin{align*}&\sup
  \left\{d_{F_n}(r_{F_n}(i),r_{F_n}(j)):i,j\in[\![0,2\rho_n+\tau_n]\!]\,
                , \left |i-j \right | \leq \delta (2\rho_n+\tau_n) \right\}
  \\
  \leq&
        3\sup\left \{\left |C_{n}(i)-C_{n}(j)\right | : i,j\in[\![0,2\rho_n+\tau_n]\!]\,
    , \left |i-j
        \right | \leq \delta (2\rho_n+\tau_n) \right\}
   \\ \leq&
 3\sqrt{3n}\sup\left \{\left |C_{(n)}(x)-C_{(n)}(y)\right | : x,y\in\left[0,\frac{2\rho_n+\tau_n}{2n}\right]\,
     ,\left | x-y \right
   |\leq \delta \frac{2\rho_n+\tau_n}{2n}\right\}.
\end{align*}

By Lemma~\ref{convergenceofcontour}, $C_{(n)}$ converges in law toward
$\widetilde F^{0\rightarrow \tau}_{\left [ 0, \rho \right ]}$.  Since
$\widetilde F_{[0,\rho]}^{0\rightarrow \tau}$ is almost surely continuous on
$[0,\rho]$, there exists $\delta$ small enough such that:
$$\sup_{n}\mathbb{P}(\exists i,j: i,j\in[\![0,2\rho_n+\tau_n]\!]\, ,\left |i-j\right | \leq \delta (2\rho_n+\tau_n), d_{F_n}(r_{F_n}(i), r_{F_n}(j)) \geq \beta n^{1/2} )< \epsilon.$$
For this $\delta$, we have:
$$\sup_{n}\mathbb{P}\left(\sup_{\substack{i,j\in[\![0,2\rho_n+\tau_n]\!] \\\left | i-j \right |\leq \delta
                  (2\rho_n+\tau_n)}}\left | L_n(i)-L_n(j)  \right |>\epsilon n^{1/4}\right)<2\epsilon,$$
this completes the proof of the Lemma. 
\end{proof}

{Then the proof of Lemma~\ref{main1} follows from
  Lemmas~\ref{convergeofforest2} and~\ref{tightofZ_n} by applying
  exactly the same steps as in \cite{berry2013scaling}. We omit the
  details.}

\section{Convergence of uniformly random toroidal triangulations}
\label{sec:proof}

In this section, we prove our main theorem. Combining the results
of previous sections, we have all the necessary tools to adapt the
method of Addario-Berry and Albenque~(\cite{berry2013scaling}, lemma
$6.1$); we extend the arguments of Bettinelli
(\cite{bettinelli2010scaling}, Theorem $1$) and Le Gall
\cite{le2007topological} to obtain Theorem~\ref{main3}.

For $n\geq 1$, let $G_n$ be a uniformly random element of
$\mathcal G(n)$.  Let $V_n$ be the vertex set of $G_n$. Recall that
$\Phi$ denotes the bijection from $\mathcal T_{r,s,b}(n)$ to
$ \mathcal G(n)$ of Theorem~\ref{them:bijectionbenjamin}.  Let
$T_n=\Phi^{-1}(G_n)$.  Therefore $T_n$ is a uniformly random element
of $\mathcal T_{r,s,b}(n)$.

We now consider $t_n$ that is uniformly distributed over
$[\![1,3]\!]$.  So the random pair $(t_n,T_n)$ is uniformly
distributed over the set $[\![1,3]\!]\times\mathcal{T}_{r,s,b}(n)$.
Then we consider $(r_n,R_n)$ the image of $(t_n,T_n)$ by the bijection
of Lemma~\ref{bij:decomposition}.  Let $k_n\in [\![0,9]\!]$ be such
that $R_n\in \mathcal R^{k_n}(n)$, so that we have
$r_n\in [\![1,3]\!]$ if $k_n=0$ (i.e. $T_n$ is a square) and
$r_n\in [\![1,2]\!]$ otherwise (i.e. $T_n$ is hexagonal).  By
Lemma~\ref{structureofu}, almost surely $k\neq 0$ so we can consider
that $T_n$ is always hexagonal.

By the discussion on the decomposition of unicellular map in
Section~\ref{dec}, the elements of
$\cup_{0\leq j\leq 9}\mathcal R^{j}(n)$ are in bijection with
$\mathcal U_{r,b}(n)$.  Let $U_n$ be the element of
$\mathcal U_{r,b}(n)$ that is decomposed into $R_n$.

As in Section~\ref{sec:labeldistance}, we define $Q_n$ as the
unicellular map obtained from $U_n$ by removing all its stems and let
$r_n=r_{Q_n}$ be the vertex contour function of $Q_n$.

We  define a pseudo-distance $d_n$ on $[\![0,2n+1]\!]$ by the following:
for $i,j\in [\![0,2n+1]\!]^2$, let
$$d_n(i,j)=d_{G_n}(r_n(i),r_n(j)).$$
Then we define the associated equivalence relation: for
$i,j\in [\![0,2n+1]\!]$, we say that $i\sim_{n}j$ if
$d_n(i,j)=0$. Thus we can see $d_n$ as a metric on
$[\![0,2n+1]\!]/\sim_{n}$.

We extend the definition of $d_n$ to non-integer values by the
following linear interpolation: for $s,t\in [0,2n+1]$, let
\begin{equation*}
\label{wybis}
d_{n}(s,t)=\underline{s}\,\underline{t}d_n(\lceil  s \rceil  , \lceil  t \rceil   )+\underline{s}\,\overline{t}d_n(\lceil  s \rceil  , \lfloor  t \rfloor   )+\overline{s}\,\underline{t}d_n(\lfloor s \rfloor  , \lceil  t \rceil   )+\overline{s}\,\overline{t}d_n(\lfloor  s \rfloor  , \lfloor  t \rfloor),  
\end{equation*}
where
$\lfloor x \rfloor=\sup\left\{k \in \mathbb Z:\, k\leq x \right\}$,
$\lceil x \rceil =\lfloor x \rfloor+1$,
$\underline{x}=x-\lfloor x \rfloor$ and
$\overline{x}=\lceil x \rceil -x$. We define its rescaled version by
the following: 
\begin{equation*}
\label{wy}
d_{(n)}=\left(\frac{d_n((2n+1)s,(2n+1)t)}{n^{1/4}}\right)_{s,t\in[0,1]^2}.
\end{equation*}

Note that the metric space $\left(\frac{1}{2n+1}[\![0,2n+1]\!]/\sim_{n}, d_{(n)}\right)$ is isometric to $\left(V_{n},n^{-1/4}d_{G_{n}}\right)$. Therefore we obtain
\begin{equation}
\label{eq:hausdorft}
d_{GH}\left(\left(\frac{1}{2n+1}[\![0,2n+1]\!]/\sim_{n}, d_{(n)}\right), \left(V_{n},n^{-1/4}d_{G_{n}}\right) \right)=0.
\end{equation}

The goal of this section is to prove the following lemma which implies Theorem~\ref{main3}

\begin{lemma}
\label{lem:convergenceofd}
  There exists a subsequence $(n_k)_{k\geq 0}$  and a pseudo-metric
  $d$ on $[0,1]$ such that 

$$\left(\frac{1}{2n_k+1}[\![0,2n_k+1]\!]/\sim_{n_k},
  d_{(n_k)}\right)\xrightarrow[{k\rightarrow \infty}]{{(d)}}
\left([0,1]/\sim_d,d\right) $$ for the Gromov-Hausdorff distance,
where for $x,y\in[0,1]^2$, we say that $x\sim_{d} y$ if $d(x,y)=0$.
\end{lemma}

\subsection{Convergence of the shifted labeling function of the
  unicellular map}

As in Section~\ref{section5}, let $(\rho^1_n,\ldots,\rho^6_n)\in \mathbb N^{6}$,
$(\tau^1_n,\ldots,\tau^6_n)\in (\mathbb N^*)^{6}$,
$(\gamma^1_n,\gamma^2_n,\gamma^3_n)\in \mathbb Z^{3}$,
$(\sigma^1_n,\sigma^2_n,\sigma^3_n)\in \mathbb N^{3}$,
$((F^{1}_{n},\ell^{1}_{n}), \ldots ,(F^{6}_{n},\ell^{6}_{n}))\in
\mathcal F^{\rho^1_n}_{\tau^1_n}\times \cdots \times \mathcal
F^{\rho^6_n}_{\tau^6_n}$,
$(M^{1}_{n},M^{2}_{n}, M^{3}_{n})\in \mathcal
M_{\sigma^1_n}^{\gamma^1_n}\times \mathcal M_{\sigma^2_n}^{\gamma^2_n}
\times \mathcal M_{\sigma^3_n}^{\gamma^3_n}$ be such that
$R_n=((F^{1}_{n},\ell^{1}_{n}), \ldots
,(F^{6}_{n},\ell^{6}_{n}),M^{1}_{n},M^{2}_{n}, M^{3}_{n})$ (see
Definition~\ref{def:defofR1to9}).
For $i\in [\![4,6]\!]$, let
$\gamma^i_n=-\gamma^{i-3}_n$ and
$\sigma^i_n=\sigma^{i-3}_n$. Moreover, for every $n>0$, we define
the renormalized version $\rho_{(n)},\gamma_{(n)},\sigma_{(n)}$ by
letting $\rho_{(n)}=\frac{\rho_n}{n}$,
$\gamma_{(n)}=(\frac{9}{8n})^{1/4}\gamma_n$ and
$\sigma_{(n)}=\frac{\sigma_n}{\sqrt{2n}}$.
For $1\leq k\leq 9$ and $1\leq i\leq 6$, let $c_i(k)$ be the value
  of $c_i$ given by line $k$ of Table~\ref{tab:gamma}.

As  in
Section~\ref{sec:labeldistance}, we introduce several definitions.
For $i\in[\![1,6]\!]$ and  $j\in[\![0, 2\rho_n^i+\tau_n^i]\!]$, we define
  $$S_n^i(j)=L_{(F^i_n,\ell^i_n)}(j)+\widetilde
  {M_n^i}^{c_i(k_n)}(\overline{C_{F^i_n}}(j)),$$

  Let $S_n^\bullet=S_n^1\bullet \cdots \bullet S_n^{6}$.  Let $P_n$ be
  the the unicellular map obtained from $U_n$ by removing all the
  stems that are not incident to proper vertices and let $r_{P_n}$ be
  its vertex contour function.  We see $S_n^\bullet$ as a function
  from the angles of $P_n$ to $\mathbb Z$

 Note that $P_n$ contains exactly $2\times(\sigma_1+\cdots+\sigma_{3})+2\times
 \mathbbm 1_{k\neq 0}$ stems.

We define the sequence $(S_n(i))_{0\leq i\leq 2n+1}$ as the sequence
that is obtained from $S_n^\bullet$ by removing all the values that
appear in an angle of $P_n$ that is just after a stem of $P_n$ in \cw
order around its incident vertex.   So $S_n$ is the {shifted
  labeling function} of the unicellular map $U_n$ (as defined in
Section~\ref{sec:labeldistance}) and is seen as a function from the
angles of $Q_n$ to $\mathbb Z$.

We consider that $S_n$ is linearly
interpolated between its integer values and define its rescaled
version:
$${S}_{(n)}=\left ( \frac{ S_n((2n+1)x)}{n^{1/4}} \right )_{0\leq x \leq 1}$$

We have the following lemma

\begin{lemma}
  \label{lem:convergenceshift}
  $S_{(n)}$ converge  in law in the space
$(\mathcal  H,d_{\mathcal  H})$  when
  $n$ goes to infinity.
\end{lemma}

\begin{proof}
 By Lemma~\ref{structureofu},
 the vector
$\left(k_n,\rho_{(n)},\gamma_{(n)},\sigma_{(n)}\right)$
converges  in
law toward a random vector
$\left(k,\rho,\gamma,
  \sigma\right)$
whose law is the 
probability measure $\mu$ of Section~\ref{section5}.

For convenience, for $1\leq i\leq 6$, let $(C_n^i,L_n^i)$ denote the
contour pair $(C_{F^i_n},L_{(F^i_n,\ell^i_n)})$ of the well-labeled
forest $(F^i_n,\ell^i_n)$. As usual, $(C_n^i,L_n^i)$ is linearly
interpolated and we denote the rescaled version by
$(C_{(n)}^i,L_{(n)}^i)$ as in Section~\ref{section7}.  By Lemma~\ref{main1},
conditionally on $ \left(k,\rho,\gamma, \sigma\right)$, we have
$(C_{(n)}^i,L_{(n)}^i)$ converge in law toward
$(C^i,L^i)=\left(\widetilde F^{0\rightarrow \tau^i}_{\left [ 0,\rho^i \right ]},
  Z^{\tau^i}_{[0,\rho^i]}\right)$.

Similarly as in Section~\ref{sec:motzkin} we consider that
$\widetilde{{M_{n}^i}}$ and $\widetilde{{M_{n}^i}}^{c_i({k_n})}$ are
linearly interpolated and we define their rescaled versions:

\begin{equation*}
\label{eq:rescaledtilde}
\widetilde{M_{(n)}^i}=\left(\left(\frac{9}{8n}\right)^{1/4}\widetilde{M_n^i}(\sqrt{2n}s)\right)_{0\leq       s       \leq
 \frac{2\sigma_n^i+\gamma_n^i}{\sqrt {2n}}}
\end{equation*}

\begin{equation*}
\label{eq:rescaledtilde}
\widetilde{M_{(n)}^i}^{c_i({k_n})}=\left(\left(\frac{9}{8n}\right)^{1/4}\widetilde{M_n^i}^{c_i({k_n})}(\sqrt{2n}s)\right)_{0\leq       s       \leq
 \frac{\tau_n^i}{\sqrt {2n}}}
\end{equation*}

By Lemma~\ref{convergeofmotzkinpathex}, $\widetilde{{M_{(n)}^i}}$
converges in law toward
$\widetilde{{M^i}}=B^{0\rightarrow \gamma^i}_{[0,2\sigma^i]}$ thus $\widetilde{{M_{(n)}^i}}^{c_i({k_n})}$  converge toward the
same limit.

Note that the processes
$C_{(n)}^i,L_{(n)}^i,\widetilde{{M_{(n)}^i}}^{c_i({k_n})}$, for $i\in [\![1,6]\!]$, are independent.  Moreover, by
Skorokhod's theorem, we can assume that these convergences hold almost
surely.

  We consider that $S_n^i$ is linearly interpolated between its
  integer values and we define  its rescaled version:
$${S}_{(n)}^i=\left ( \frac{ S_n^i(2ns)}{n^{1/4}} \right )_{0\leq
  s \leq \frac{2\rho_n^i+\tau_n^i}{2n}}.$$

We have
\begin{align*}
  {S}_{(n)}^i(s)&=\frac{1}{n^{1/4}} S_n^i(2ns)
  \\
                &=
\frac{1}{n^{1/4}}L_{n}(2ns)+\frac{1}{n^{1/4}}\widetilde
                  {M_n^i}^{c_i(k_n)}\left(\overline{C_{n}}(2ns)\right)
                  \\
                &=
L_{(n)}(s)+\left(\frac{8}{9}\right)^{1/4}\widetilde
  {M_{(n)}^i}^{c_i(k_n)}\left(\sqrt{\frac{3}{2}}\overline{C_{(n)}}(s)\right)
\end{align*}

So ${S}_{(n)}^i$ converge in law toward
a
limit ${S}^i:[0,\rho^i]\rightarrow \mathbb R$ in the space $(\mathcal
H,d_{\mathcal H})$, where, for $t\in[0,\rho^i]$, we have:
$$S^i(t)= L^i(t)
+\left(\frac{8}{9}\right)^{1/4}
\widetilde{M^i}
\left(\sqrt{\frac{3}{2}}
  \overline{
    C^i}(t)\right
)$$

 We consider that ${S}^\bullet_{n}$ is linearly interpolated between its
  integer values and we define  its rescaled version:
 
$${S}^\bullet_{(n)}=\left ( \frac{ S^\bullet_n((2n)s)}{n^{1/4}} \right )_{0\leq s \leq \frac{2n+\sum_i\sigma_i+4}{2n}}$$

Therefore we have that  ${S}^\bullet_{(n)}$ converge in law toward
 ${S}^\bullet={S}^{1}\bullet \ldots\bullet {S}^{6}$ in the space $(\mathcal H,d_{\mathcal H})$.

 It remains to show the convergence of $S_{(n)}$ given that of $S_{(n)}^\bullet$. This is done by noticing that $S_n$ is within bounded distance (in the uniform topology on continuous functions) from a time-change of $S_n^\bullet$, where the time change itself is within $O(\sqrt n)$ from the identity. This and the tension of both sequences (or a priori bounds on their moduli of continuity) imply that $S_{(n)}$ and $S_{(n)}^\bullet$ converge toward the same limit.
\end{proof}

\subsection{Subsequential convergence of the pseudo-distance function
  of the unicellular map}

We now introduce several definitions similar to those in Section~\ref{sec:label}.  Let
$a_n^0$ be the root angle of $T_n$ and $v^0_n$ be its root vertex. Let
$\ell_n=4n+1$.  Let $\Gamma_n$ be the unicellular map obtained
from $T_n$ by adding a special dangling half-edge, called the root
half-edge, incident to the root angle of $T_n$.
Let $\lambda_n$ be the labeling function of th angles of $\Gamma_n$ as
defined in Section~\ref{sec:label}.     For all $u,v \in V$,
let $m(u)$ and $\overline m(u,v)$ be  defined  as  in
Section~\ref{sec:label}.

As in Section~\ref{sec:labeldistance}, we define the 
 following pseudo-distance: for $i,j\in[\![0,2n+1]\!]$,
$$d_n^{o}(i,j)=m_n(r_n(i))+         m_n(r_n(j))-
2\overline{m_n}(r_n(i),r_n(j)).$$
We  extend the definition of $d_n^{o}$ to
non-integer values and define its rescaled version $d_{(n)}^{o}$ as for $d_n$.

By Lemma~\ref{lem:convergenceshift}, $S_{(n)}$  converge in law toward
a limit $S:[0,1]\rightarrow \mathbb R$  in the space
$(\mathcal  H,d_{\mathcal  H})$  when
$n$ goes to infinity. For $s,t\in [0,1]$, we define:
$$d^o(s,t)={S}(s)+ {S}(t)-2\min_{x\in [s, t]}
S(x).$$

\begin{lemma}
  \label{lem:d0nconverge}
    $d_{(n)}^{o}$ converges in law toward $d^o$  when
  $n$ goes to infinity.
\end{lemma}

\begin{proof}

By~(\ref{eq:d0S}) we have: for $i,j\in[\![0,2n+1]\!]$,
\begin{equation}
  \label{eq:dnsn}
|d_n^{o}(i,j)-(S_n(i)+S_n(j)-2\overline S_n(i,j))| \leq 64
\end{equation}

By Lemma~\ref{lem:mm7}
for any $i,j\in[\![0,2n
+1]\!]$, we have
$$d_n^{o}(i+1,j),d_n^{o}(i,j+1),d_n^{o}(i+1,j+1)\in
[\![d_n^{o}(i,j)-28,d_n^{o}(i,j)+28]\!].$$
Thus, for $s,t\in [0,2n+1]$, we have

\begin{equation*}
|d_{n}^{o}(s,t)-d_{n}^{o}(\lfloor s\rfloor,\lfloor t\rfloor)|\leq 28
\end{equation*}

So, for $s,t\in[0,1]^2$, we have:
$$\left |d_{(n)}^{o}(s,t)-d_{(n)}^{o}\left(\frac{\lfloor(2n+1)s\rfloor}{2n+1},\frac{\lfloor(2n+1)t\rfloor}{2n+1}\right)\right|\leq
\frac{28}{n^{1/4}}$$

Since every vertex is incident to at most two stems and the variation
of $S^\bullet$ is at most $1$, we have
for $s,t\in [0,2n+1]$:

\begin{equation*}
  |S_{n}(s)-S_{n}(\lfloor s\rfloor)|\leq 3
  \end{equation*}
  \begin{equation*}
  |\overline S_{n}(s,t)-\overline S_{n}(\lfloor s\rfloor,\lfloor
  t\rfloor)|\leq 6
\end{equation*}

So, for $s,t\in[0,1]^2$, we have:
$$\left |S_{(n)}(s)-S_{(n)}\left(\frac{\lfloor(2n+1)s\rfloor}{2n+1}\right)\right|\leq
\frac{3}{n^{1/4}}$$
$$\left |\overline S_{(n)}(s,t)-\overline S_{(n)}\left(\frac{\lfloor(2n+1)s\rfloor}{2n+1},\frac{\lfloor(2n+1)s\rfloor}{2n+1}\right)\right|\leq
\frac{6}{n^{1/4}}$$

Then by (\ref{eq:dnsn}), for $C=28+3+3+2\times 6+64=110$, we have, for
all $s,t\in[0,1]^2$:

\begin{equation}
  \label{eq:dnsnrescaled}
|d_{(n)}^{o}(s,t)-(S_{(n)}(s)+S_{(n)}(t)-2\overline S_{(n)}(s,t)| \leq \frac{C}{n^{1/4}}.
\end{equation}

By Lemma~\ref{lem:convergenceshift}, ${S}_{(n)}$ converge in law toward
 ${S}$ in the space $(\mathcal H,d_{\mathcal H})$.
 So
$d_{(n)}^o$   converges   in law  toward
$d^o$.
\end{proof}

\subsection{Convergence for the Gromov-Haussdorf distance}

We use the same notations as in previous sections. We first prove the
tightness of $d_{(n)}$ and then the convergence for the Gromov-Haussdorf distance.

\begin{lemma}
\label{tightofd_n}
 The sequence of the laws of the processes $$\left(d_{(n)}(s,t)\right)_{0\leq s,t\leq 1}$$
 is tight in the space of probability measures on $C([0,1]^2, \mathbb{R})$.
\end{lemma}

\begin{proof}

For every $s,s',t,t' \in [0,1]$, by triangular inequality for $d_{G_n}$, we have:
\begin{align*}
d_{(n)}(s,t)\leq d_{(n)}(s,s')+d_{(n)}(s',t')+d_{(n)}(t',t)\\
d_{(n)}(s',t')\leq d_{(n)}(s',s)+d_{(n)}(s,t)+d_{(n)}(t,t')
\end{align*}
Therefore we obtain:
$$\left|d_{(n)}(s,t)-d_{(n)}(s',t')\right|\leq d_{(n)}(s,s')+d_{(n)}(t,t').$$

By Lemma~\ref{distancebetweentwovertices}, we
have, for $s,t\in[0,1]$
$$d_{(n)}(s,t)\leq d_{(n)}^{o}(s,t)+\frac{14}{n^{1/4}}.$$
So we have:
$$\left|d_{(n)}(s,t)-d_{(n)}(s',t')\right|\leq d^o_{(n)}(s,s')+d^o_{(n)}(t,t')+\frac{28}{n^{1/4}}.$$

Consider $\epsilon, \eta>0$. By Lemma~\ref{lem:d0nconverge}, $d_{(n)}^{o}$ converge toward $d^o$, so by using Fatou's
lemma, we have for every $\delta>0$,
\begin{equation}
\label{eq:proofoftightofdn1}
\limsup_{n\rightarrow \infty}\mathbb{P}\left(\sup_{|s-s'|\leq \delta}d^o_{(n)}(s,s')\geq \eta\right)\leq \mathbb{P}\left(\sup_{|s-s'|\leq \delta}d^o(s,s')\geq \eta\right) .
\end{equation}

Since $d^o$ is continuous and null on the diagonal, therefore there exists $\delta_\epsilon>0$ such that:
\begin{equation}
\label{eq:proofoftightofdn2}
\mathbb{P}\left(\sup_{|s-s'|\leq \delta_\epsilon}d^o(s,s')\geq
  \eta\right)\leq \epsilon.
\end{equation}
By~\eqref{eq:proofoftightofdn1},\eqref{eq:proofoftightofdn2} there exists $n_0\in \mathbb{N}$ such that for every $n\geq n_0$ we have:

$$\mathbb{P}\left(\sup_{|s-s'|\leq \delta_\epsilon}d^o_{(n)}(s,s')\geq \eta \right)\leq \epsilon.$$

By taking $n_0$ large enough (if necessary) such that
$\frac{28}{n^{1/4}}\leq \eta$, we have for every $n\geq n_0$: 

$$\mathbb{P}\left(\sup_{|s-s'|\leq \delta_\epsilon; |t-t'|\leq
    \delta_\epsilon}\left|d_{(n)}(s,t)-d_{(n)}(s',t')\right|\geq
  3\eta\right)\leq 2\epsilon.$$

By Ascoli's theorem, this completes the proof of lemma. 
\end{proof}

We are now able to prove the main result of this section.

\begin{proof}[Proof of Lemma~\ref{lem:convergenceofd}]
By Lemma~\ref{tightofd_n}, there exists a subsequence $(n_k)_{k\geq 0}$ and a function $d\in C([0,1]^2, \mathbb{R})$ such that 
\begin{equation}
\label{eq:convergenceofd}
d_{(n_k)} \overset{(d)}{\longrightarrow} d.
\end{equation}

By the Skorokhod theorem, we  assume that this convergence holds
almost surely. As the triangular inequality holds for each $d_{(n)}$
function, the function $d$ also satisfies the triangular
inequality. On the other hand, for $s\in [\![0,2n+1]\!]$, note that  we have
$d_{(n)}(s,s)\leq 1$. So for $x\in [0,1]$, we have
$d_{(n)}(x,x)=O(n^{-1/4})$. Therefore the function $d$ is actually a
pseudo-metric. For $x,y\in[0,1]^2$, we say that $x\sim_{d} y$ if $d(x,y)=0$. 

We use the characterization of the Gromov-Hausdorff distance via
correspondence. Recall that a \emph{correspondence} between two metric spaces
$(S,\delta)$ and $(S',\delta')$ is a subset $R\subseteq S\times S'$ such
that for all $x\in S$, there exists at least one $x'\in S'$ such that
$(x,x')\in R$ and vice-versa. The distortion of $R$ is defined by:
$$\mathrm{dis}(R)=\sup\left\{\left|\delta(x,y)-\delta'(x',y')\right|: (x,x'), (y,y')\in R \right\}.$$

Therefore we have (see~\cite{burago2001course})
$$d_{GH}((S,\delta),(S',\delta'))=\frac{1}{2}\inf_{R} \mathrm{dis}(R),$$

where the infimum is taken over all correspondence $R$ between $S$ and $S'$. 

We define the correspondence $R_{n}$ between $\left(\frac{1}{2n+1}[\![0,2n+1]\!]/\sim_{n}, d_{(n)}\right)$ and $\left([0,1]/\sim_d,d\right)$ as the set
$$R_{n}=\left\{\left(\frac{\pi_{n}\left(\lfloor(2n+1)x\rfloor\right)}{2n+1}, \pi(x)\right), x\in [0,1]\right\},$$
where
$\pi_n$ the canonical projection from $[\![0,2n+1]\!]$ to
$[\![0,2n+1]\!]/\sim_{n}$ and $\pi$ is the canonical projection from
$[0,1]$ to $[0,1]/\sim_d$.

We have
$$\mathrm{dis}(R_{n})=\sup_{0\leq x,y\leq 1}\left|d_{(n)}\left(\frac{\lfloor (2n+1)x\rfloor}{2n+1}, \frac{\lfloor (2n+1)y\rfloor}{2n+1}\right)-d(x,y)\right| $$

By~\eqref{eq:convergenceofd}, we have $\mathrm{dis}(R_{n_k})$ converges
toward $0$ and thus the following convergence  for the Gromov-Hausdorff distance:

$$\left(\frac{1}{2n_k+1}[\![0,2n_k+1]\!]/\sim_{n_k},
  d_{(n_k)}\right)\xrightarrow[{k\rightarrow \infty}]{{(d)}}
\left([0,1]/\sim_d,d\right).$$
\end{proof}
\newpage

\appendix

\section{Approximation of distance by labels}
\label{rightmostpath}

In this  appendix, we  show that with  high probability,  the labeling
function  defined in  Section~\ref{subsection:forestandwell-labelings}
approximates the  distance to  the root up  to a  uniform $o(n^{1/4})$
correction. As we mentioned in  the introduction, we believe that this
is an essential  step toward proving uniqueness  of the subsequential
limit in  Theorem~\ref{main3}. The  proof is  quite technical  and the
estimate itself  is not  needed in  the proof  of Theorem~\ref{main3};
since it exploits the same rather involved combinatorial construction,
we chose to include it here as  an appendix rather than to write it as
a separate article.

\subsection{Rightmost walks and distance properties}
\label{rightmostpathsub}

\subsubsection{Definition and properties of rightmost walks}

We use the
same notations as in Sections~\ref{sec:preliminaries} and~\ref{sec:label}.

For $n\geq 1$, let $T$ be an element of $\mathcal T_{r,s,b}(n)$, and
$G=\Phi(T)$ the corresponding element of $\mathcal G(n)$.  The
canonical orientation of $G$ is noted $D_0$. Recall that, as already
mentioned, every vertex of $G$ as outdegree exactly three in $D_0$.

For an (oriented) edge $e$ of $D_0$, we define the \emph{rightmost
  walk} from $e$ as the sequence of edges starting by following $e$,
and at each step taking the rightmost outgoing edge among the three
outgoing edges at the current vertex. Note that a rightmost walk is
necessarily ending on a periodic closed walk since $G$ is finite.

We have the following essential lemma concerning rightmost walks:

\begin{lemma}
\label{lem:rightwalkroot}
For any edge $e$ of $D_0$, the ending part of the rightmost walk from
$e$ is the root triangle with the interior of the triangle on its
right side.
\end{lemma}

\begin{proof}
  The proof is based on results from~\cite{leveque2017generalization}.
  Let $e$ be an edge of $D_0$.
  By~\cite[Lemma~37]{leveque2017generalization}, i.e. by the balanced
  property of the orientation $D_0$, the end of the rightmost walk
  from $e$ is a triangle $A$ with the interior of the triangle on its
  right side.  By~\cite[Lemma~25]{leveque2017generalization}, i.e. by
  minimality of the orientation $D_0$, the interior of $A$ must
  contain the root face $f_0$ of $G$. The root face is incident to the
  root triangle $A_0$ by definition.  Since the outdegree of all the
  edges is three, a classic counting argument using Euler's formula
  gives that all the edges in the interior of $A_0$ and incident to it
  are entering $A_0$. So it is not possible that $A$ is entering in
  the interior of $A_0$. Since $f_0$ is in the interior of both $A$
  and $A_0$, we have that the interior of $A$ contains the interior of
 $A_0$. Then by maximality of $A_0$, we have that $A=A_0$.
\end{proof}

By Lemma~\ref{lem:rightwalkroot}, any rightmost walk visit the root vertex. For an edge $e$ of  $D_0$, we define the \emph{right-to-root walk}, noted $W_R(e)$, as the subwalk of the rightmost walk started from $e$ that stops at the first visit of the root vertex $v_0$.

Recall that, for $0\leq i \leq \ell$, the set $\mathcal{A}(i)$ denote the set of angles of $G^+$ which are splited from $a_i$ by the complete closure procedure.
Let $f$ be the mapping that associate to an angle $\alpha$ of $G^+$ the integer $i$ such that $\alpha\in \mathcal A(i)$.
Let $g$ be the mapping that associate to an angle $\alpha$ of $\Gamma$ the integer $i$ such that $\alpha=a_i$.

Depending of the type of the unicellular map, i.e. hexagonal or
square, and the fact that $r_s$ is special or not, we define three
particular angles $x_1$, $x_2$ and $x_3$ of $\Gamma$, as represented
on Figure~\ref{fig:partitionA}. Note that in the particular case where
$r_s\in S$, we have $x_1=x_2$. Moreover, let $x_0=a_0$ and
$x_4=a_\ell$. Then, for $1\leq j \leq 4$, let
$X_j=\bigcup_{g(x_{j-1})\leq i < g(x_j)} \mathcal A(i)$. Note that
$X_2=\emptyset$ if $x_1=x_2$. Thus the set of angles of $G^+$ is
partitioned into the four sets $X_1,\ldots, X_{4}$ such that if
$\alpha \in X_i$ and $\alpha' \in X_j$, with $i<j$, then
$f(\alpha)<f(\alpha')$.

\begin{figure}[h]
\center
\begin{tabular}{cc}
  \scalebox{0.4}{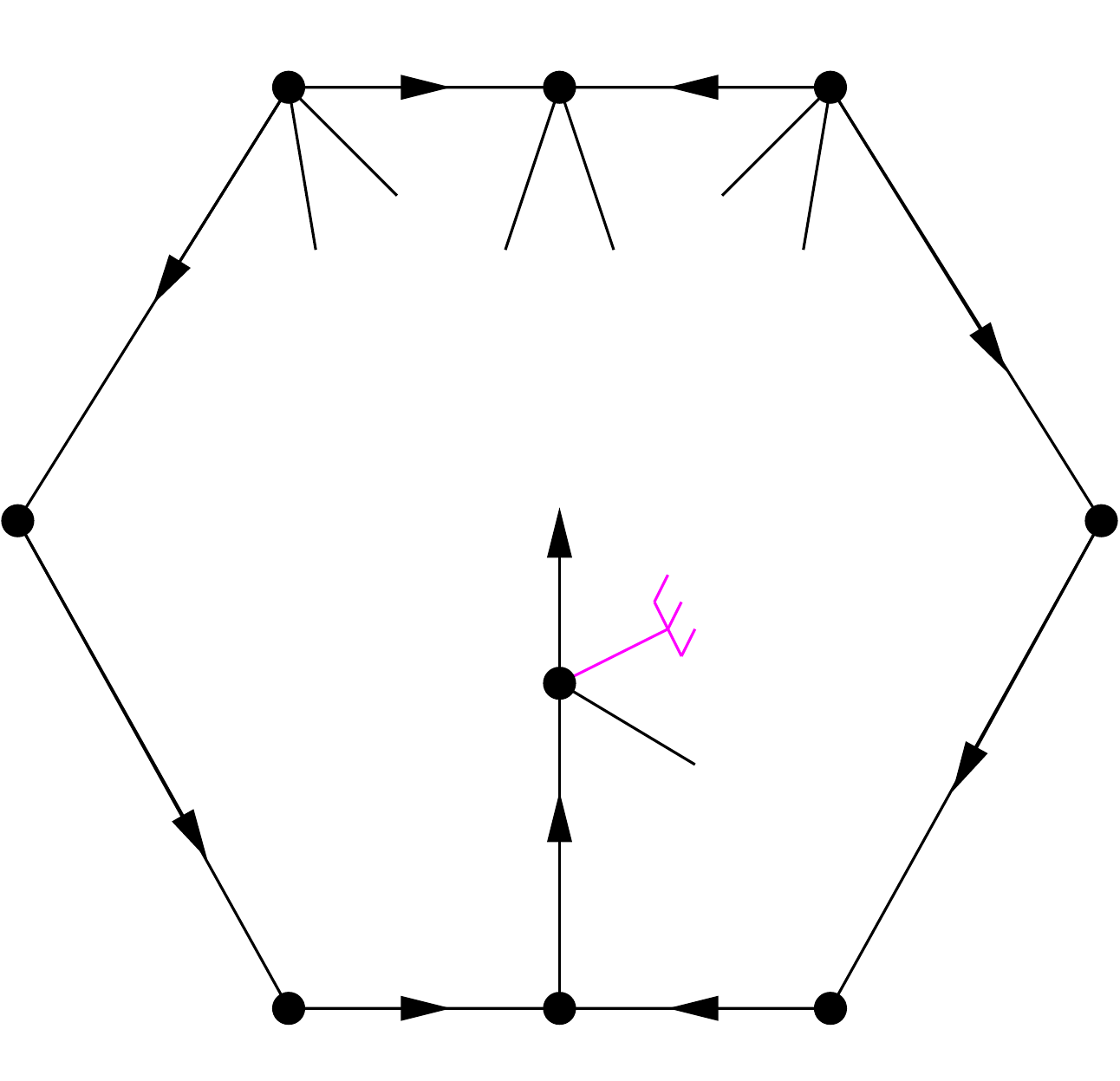}
\ \ & \ \ 
  \scalebox{0.4}{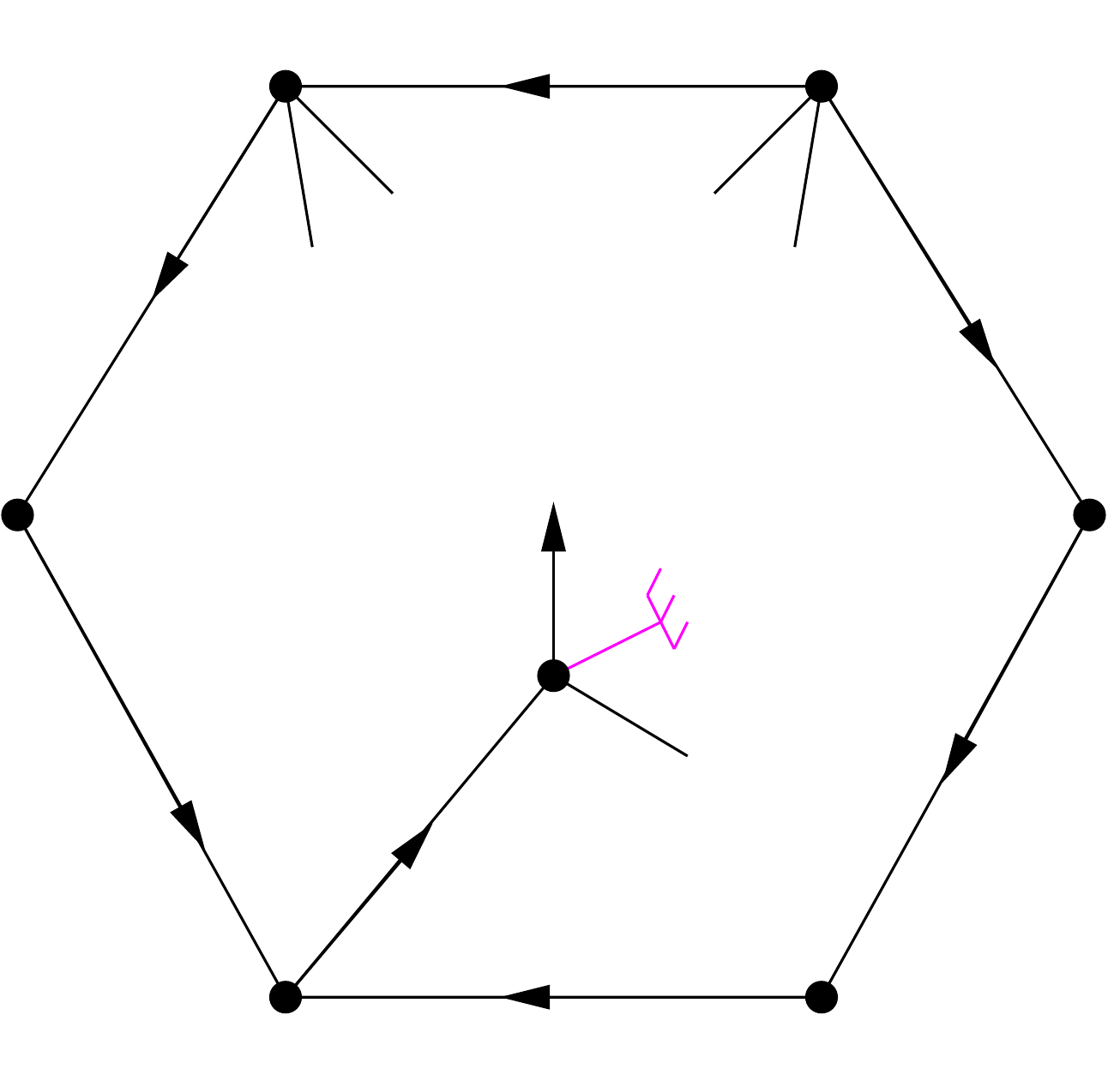}
  \\
   & \\
    \scalebox{0.4}{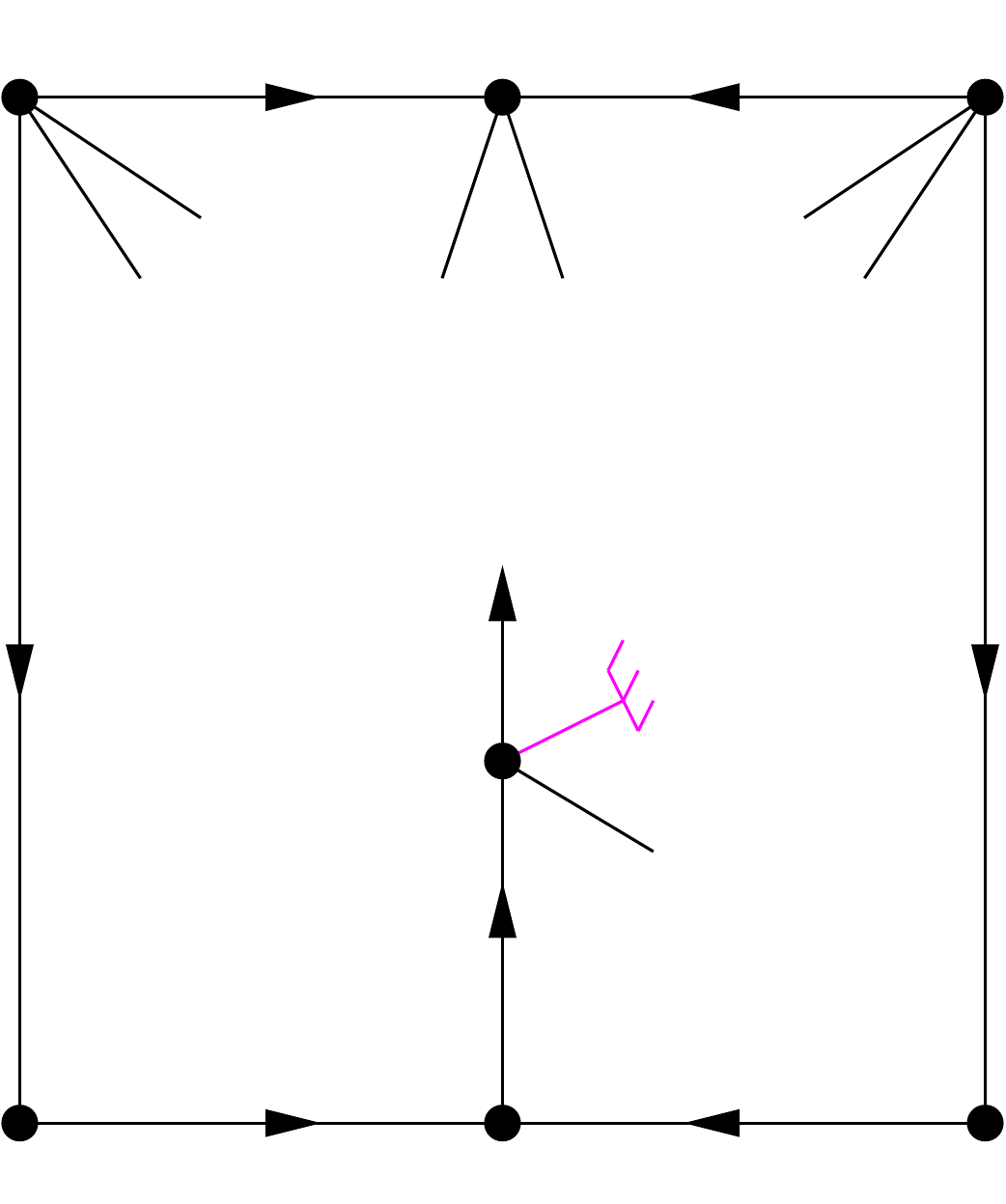}
\ \ & \ \ 
  \scalebox{0.4}{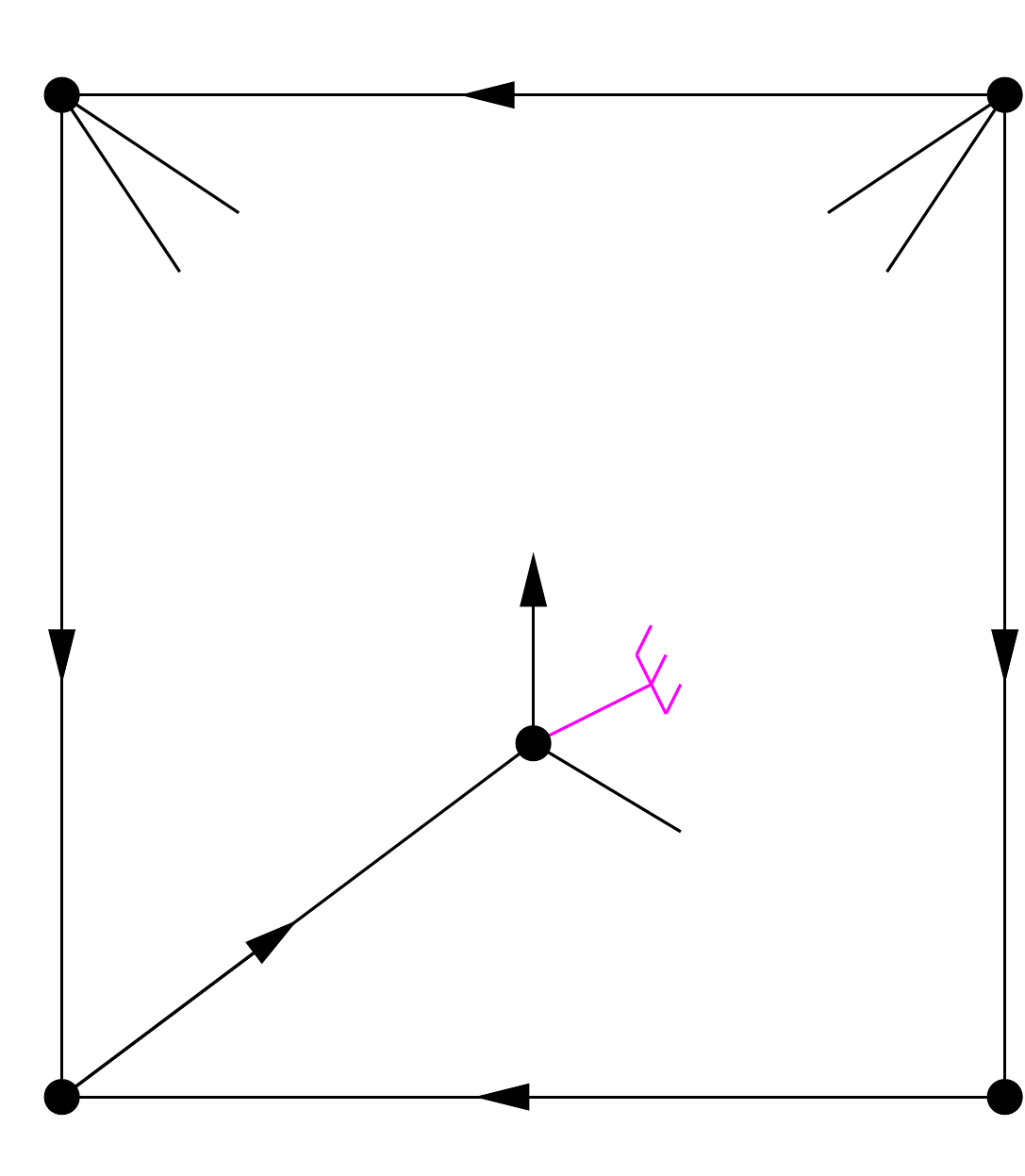}
\end{tabular}
\caption{Definition of the angles $x_1$, $x_2$ and $x_3$ depending on the type of unicellular map.}
\label{fig:partitionA}
\end{figure}

The partition $(X_1,\cdots, X_{4})$ has been defined to satisfied the
following property.  Consider an edge $e=uv$ of $E_P\cup E_R$,
oriented from $u$ to $v$ in the canonical orientation, with angles
$a,a'$ of $G^+$ incident to $e$ that appears in \ccw order around $v$.
Then, one can see on Figure~\ref{fig:partitionA} that
$a\in (X_1\cup X_3)$. Moreover if $a$ in $X_1$ (resp. in $X_3$), then
$a'$ is in $X_{2}\cup X_{3}\cup X_{4}$ (resp. in $X_4$).

Given an edge $\{ u,v \}$ of $G^+$, we note $a^\ell(u,v)$
(respectively $a^r(u,v)$ ) the angle incident to $u$ that is just
after $\{ u,v \}$ in \ccw order (resp \cw order) around $u$.

Consider $e\in D_0$, and $W_R(e)$ the right-to-root walk starting from
$e$, whose sequence of vertices is $(u_j)_{0\leq j\leq k}$, with
$k> 0$.  We define two sequence of angles of $G^+$ incident to the
right side of $W_R$. For $0\leq i \leq k-1$, let
$\alpha_i=a^r(u_i,u_{i+1})$. For $1\leq i \leq k$, let
$\beta_i=a^\ell(u_{i},u_{i-1})$. Note that, for $0< i < k$, we might
have $\alpha_i=\beta_i$ if there is no edges incident to the right
side of $W_R(e)$ at $u_i$.

\begin{lemma}
\label{lem:alphabetadecrease}
For  $0\leq i\leq k-1$, we have
$\lambda(\beta_{i+1})-\lambda(\alpha_i)=-1$.
For $1\leq i\leq k-1$, we have
$-6\leq \lambda(\alpha_{i})-\lambda(\beta_i)\leq 0$. 
Moreover  $|\{ i\in [\![ 1,k-1 ]\!]: \lambda(\alpha_{i})<\lambda(\beta_i)\} |\leq 2$ and
$f(\alpha_0)< f(\beta_1)\leq f(\alpha_1)<\cdots< f(\beta_{k-1})\leq f(\alpha_{k-1})<f(\beta_k)$.
\end{lemma}

\begin{proof}
Let $0\leq i\leq k-1$ and consider the edge $\left \{ u_i,u_{i+1} \right \}$. We have  $\left \{ u_i,u_{i+1} \right \}$ is either in $E(\Gamma)$ or not. If $\left \{ u_i,u_{i+1} \right \} \notin E(\Gamma)$, let $s$ be a stem such that we attach $s$ to an angle that comes from $a(s)$ to form the edge $\{u_i,u_{i+1} \}$ of $G$. By Lemma~\ref{lemma1}, we have $\lambda(\beta_{i+1})=\lambda(a(s))=\lambda(s)-1=\lambda(\alpha_{i})-1$. Moreover since $U$ is safe, we have $f(\beta_{i+1})> f(\alpha_{i})$.
If $\left \{ u_i,u_{i+1} \right \} \in E(\Gamma)$, we also have
$\lambda(\beta_{i+1})=\lambda(\alpha_{i})-1$ and $f(\beta_{i+1})> f(\alpha_{i})$.

Consider $1\leq i\leq k-1$.
By Lemma~\ref{lem:Mm6}, we have $-6\leq \lambda(\alpha_{i})-\lambda(\beta_i)$.
Let $(\gamma^i_1,\ldots ,\gamma^i_{p_i})$, with $p_i\geq 1$, be the set of consecutive angles of $G^+$ between $\beta_i=\gamma^i_1$ and $\alpha_i=\gamma^i_{p_i}$ in \ccw order around $u_i$. Since $W_R(e)$ is a right-to-root walk, if $p_i> 1$, then all the edges that are incident to $u_i$ between two consecutive angles $\gamma^i_j$ and $\gamma^i_{j+1}$, with $1\leq j<p$, are entering $u_i$. So, by Lemma~\ref{lem:labelvariation}, for $1\leq j<p_i$, we have $\lambda(\gamma^i_{j+1})-\lambda(\gamma^i_j)\leq 0$. Moreover, we have
$\lambda(\gamma^i_{j+1})-\lambda(\gamma^i_j)< 0$ if and only if the edge entering $u_i$ between
$\gamma^i_j$ and $\gamma^i_{j+1}$ is in $E_P\cup E_R$. Thus we have $\lambda(\alpha_{i})-\lambda(\beta_i)\leq 0$, and, for $1\leq j<p_i$, we have $f(\gamma^i_{j+1}) \geq f(\gamma^i_j)$.

We obtain that the sequence 
$$(f_p)_{0\leq p\leq r}=(f(\alpha_0),f(\gamma^1_1),\ldots,f(\gamma^1_{p_1}),\ldots,f(\gamma^{k-1}_1),\ldots,f(\gamma^{k-1}_{p_{k-1}}),f(\beta_k))$$
is increasing and thus $f(\alpha_0)< f(\beta_1)\leq f(\alpha_1)<\cdots< f(\beta_{k-1})\leq f(\alpha_{k-1})<f(\beta_k)$. 
This also implies that the sequence $I=(\{i\, :\, f_p\in X_i\})_{0\leq p\leq r}$ is increasing. 

If there is a couple $(i,j)$, with 
 $1\leq i\leq k-1$, and $1 \leq j < p_i$, such that the edge incident to $\gamma^{i}_{j}$ and $\gamma^{i}_{j+1}$ is in $E_P\cup E_R$, then
either $\gamma^i_j \in X_1$ and $\gamma^i_{j+1} \in X_2\cup X_3\cup X_4$, or, $\gamma^i_j \in X_3$ and  $\gamma^i_{j+1} \in X_4$.
Since $I$ is increasing, this implies that there is at most two such couples $(i,j)$.
So  $|\{ i\in [ 1,k-1 ]: \lambda(\alpha_{i})<\lambda(\beta_i)\} |\leq 2$.
\end{proof}

\begin{lemma}
\label{lem:wr186}
For all $e=uv \in D_0$, we have
$$m(u)-18\leq\left | W_R(e) \right |\leq m(u)+6.$$
\end{lemma}

\begin{proof}
  By Lemma~\ref{lem:alphabetadecrease}, the sequence
  $(\lambda(\alpha_0), \lambda(\beta_1),\lambda(\alpha_1),\ldots,
  \lambda(\beta_{k-1}),\lambda(\alpha_{k-1}),\lambda(\beta_k))$ is
  decreasing by one between $\alpha_i$ and $\beta_{i+1}$, for
  $0\leq i\leq k-1$, it is constant between $\beta_i$ and
  $\alpha_{i}$, for $1\leq i\leq k-1$, except for at most two value
  $1\leq i\leq k-1$ where it can decrease by at most $6$.  So
  $\lambda(\alpha_0)-\lambda(\beta_k)-2\times6\leq |W_R(e)|\leq
  \lambda(\alpha_0)-\lambda(\beta_k)$.  By Lemma~\ref{lem:Mm6}, we
  have $m(u)\leq \lambda(\alpha_0)\leq m(u)+6$ and
  $0\leq \lambda(\beta_k)\leq 6$.  So
  $m(u)-18\leq\left | W_R(e) \right |\leq m(u)+6$.
\end{proof}

 We define 
 
 $$
t=
\begin{cases}
3 & \textrm{if $\Gamma$ is hexagonal and $r_s\notin S$}\\
4 & \textrm{if  $\Gamma$ is hexagonal and $r_s\in S$}\\
4 & \textrm{if  $\Gamma$ is square and $r_s\notin S$}\\
5 & \textrm{if  $\Gamma$ is square and $r_s\in S$}\\
\end{cases}
$$

and $t-1$ particular angles $y_1,\ldots,y_{t-1}$ of $\Gamma$, as represented on Figure~\ref{fig:partitionB}. Moreover, let $y_0=a_0$ and $y_t=a_\ell$. Then, for $1\leq j \leq t$, let 
$Y_j=\bigcup_{g(y_{j-1})\leq i < g(y_j)} \mathcal A(i)$. 
Thus the set of angles of $G^+$ is partitioned into the $t$ sets $(Y_1,\cdots, Y_{t})$ such that if $\alpha \in Y_i$ and $\alpha' \in Y_j$, with $i<j$, then $f(\alpha)<f(\alpha')$.

\begin{figure}[h]
\center 
\begin{tabular}{cc}
  \scalebox{0.5}{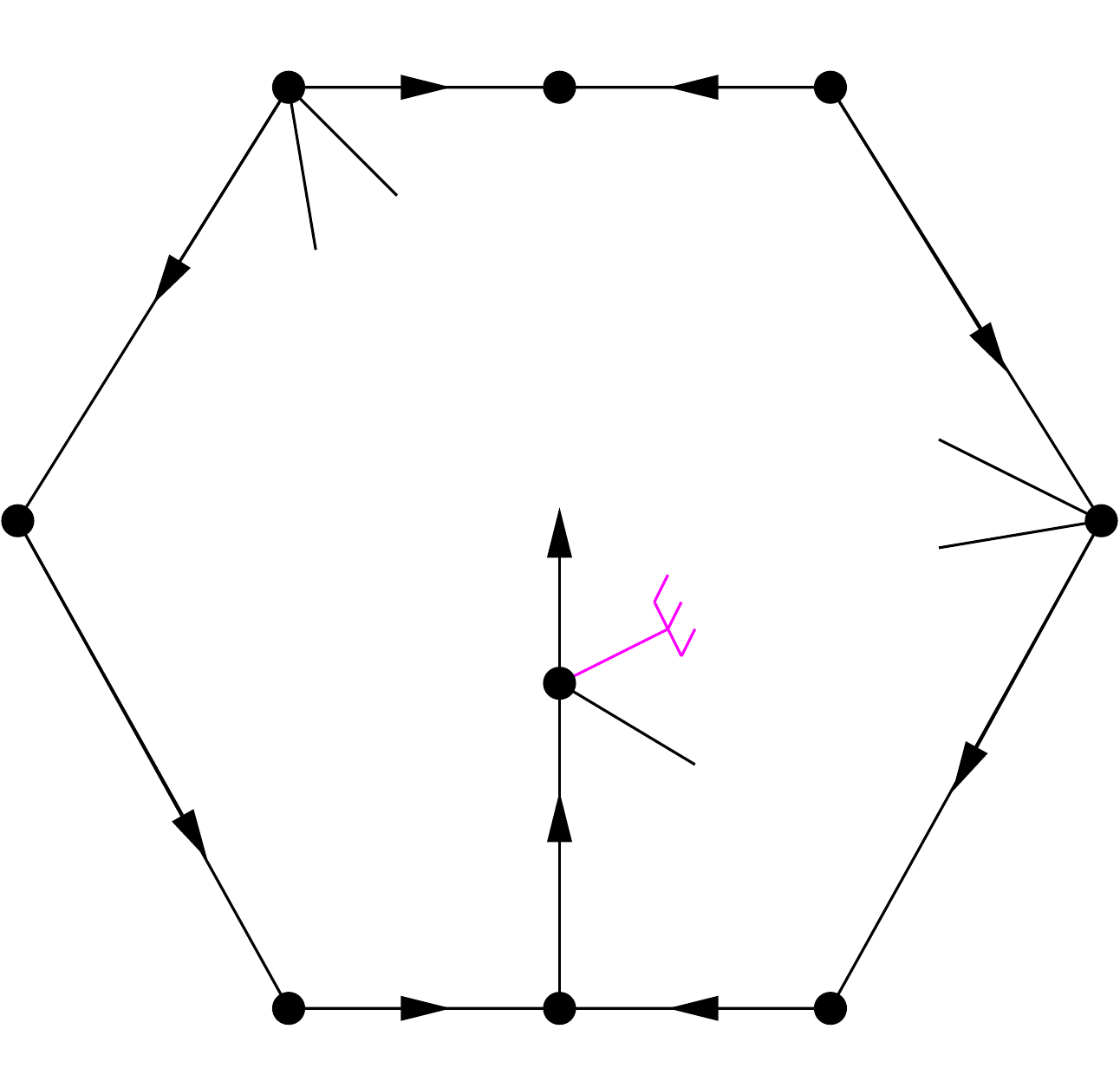}
\ \ & \ \ 
  \scalebox{0.5}{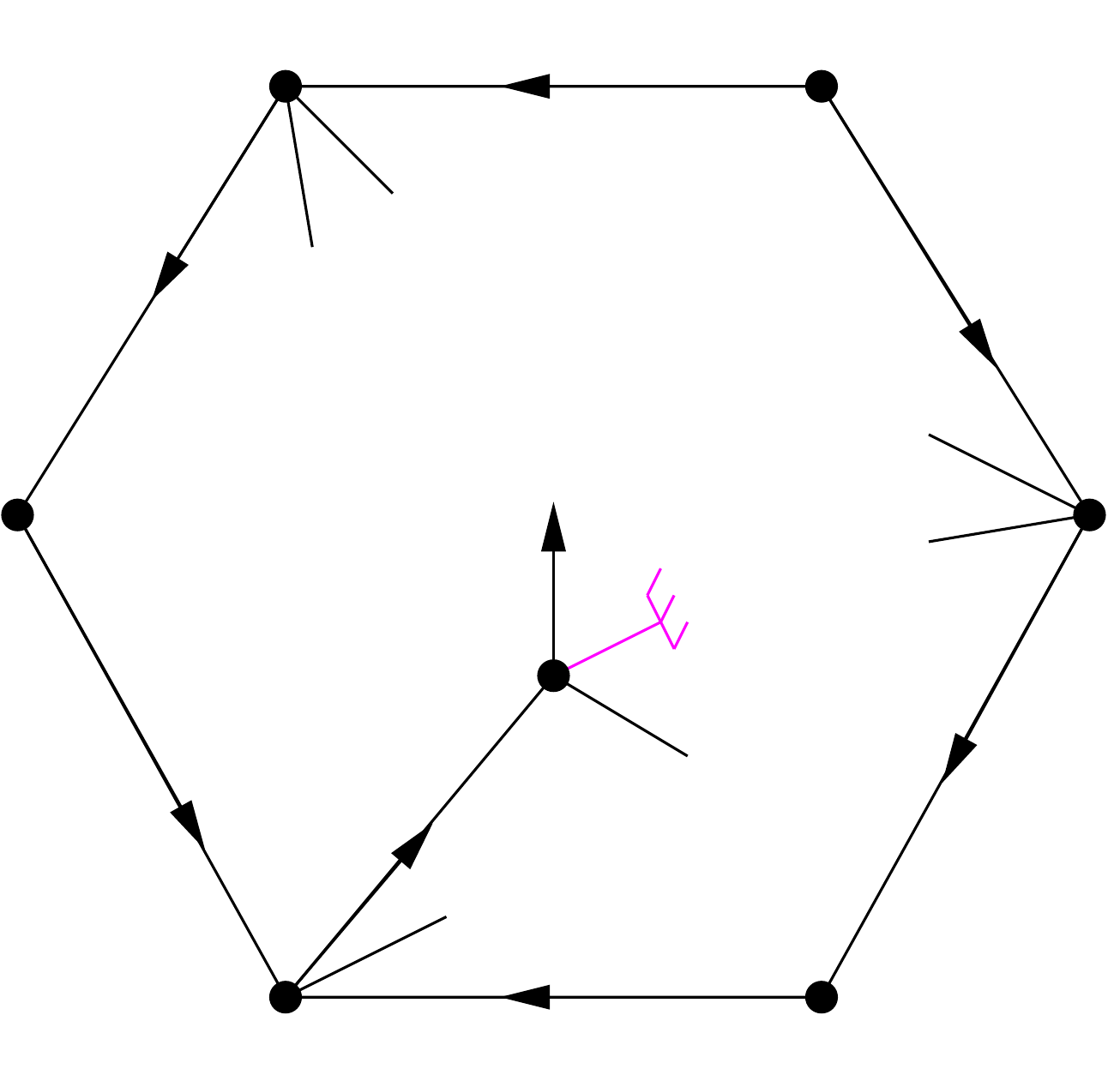}
  \\
   & \\
    \scalebox{0.5}{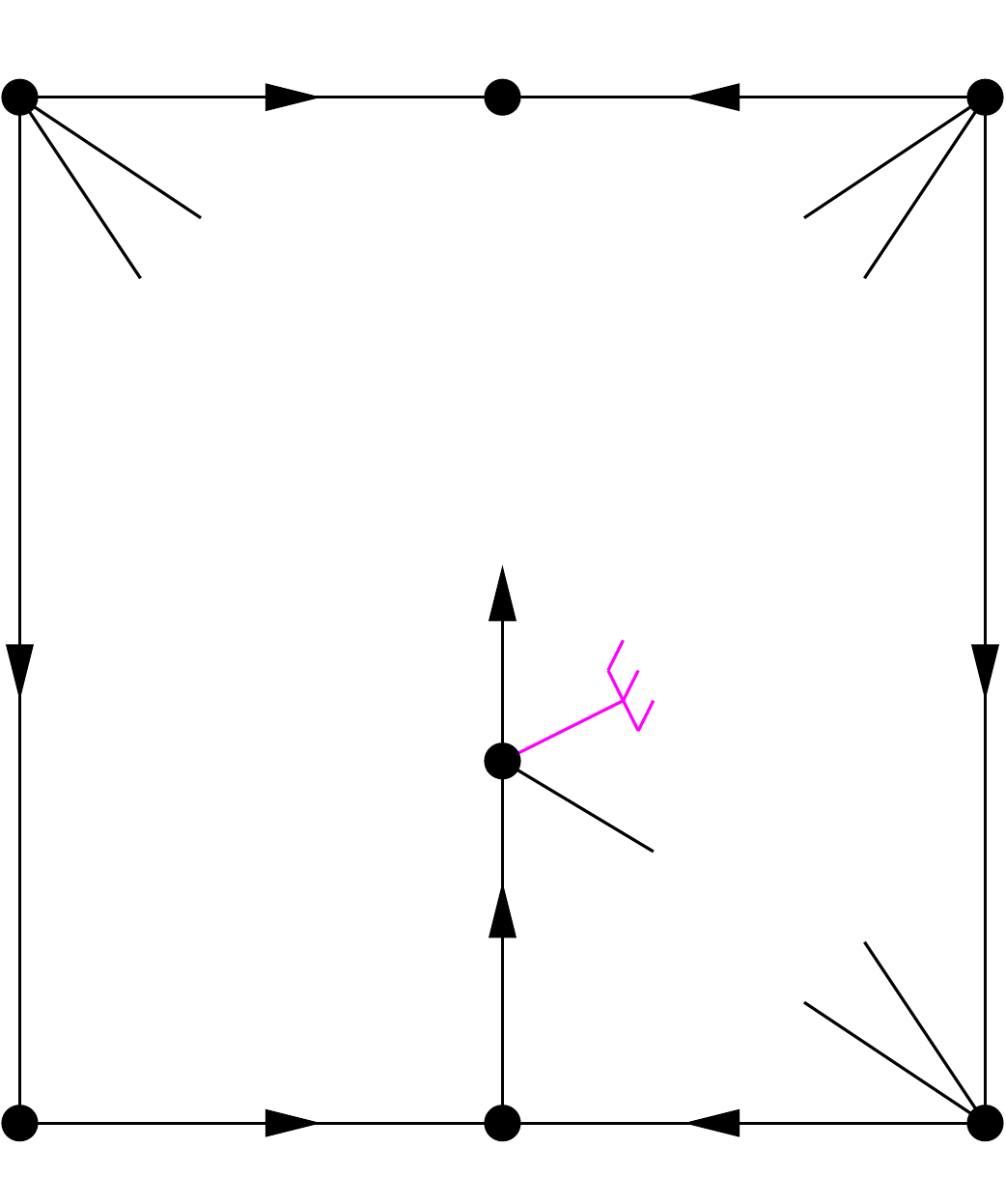}
\ \ & \ \   
  \scalebox{0.5}{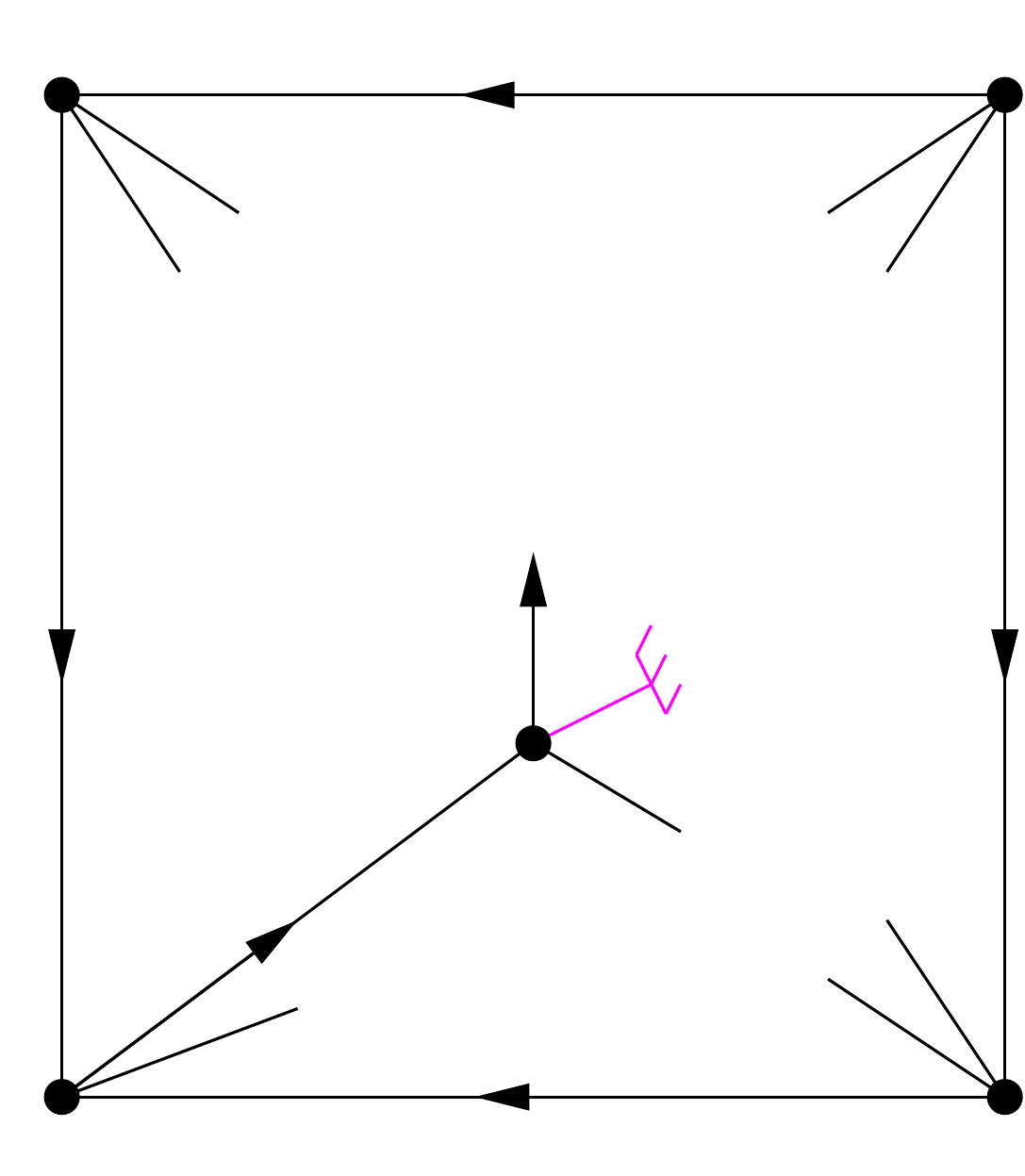}
\end{tabular}
\caption{Definition of the angles $y_1,\ldots,y_{t-1}$.}
\label{fig:partitionB}
\end{figure}
 
The partition $(Y_1,\cdots, Y_{t})$ has been defined to satisfied the
following property.  For any vertex $v$, each set of consecutive
angles around $v$ that is delimited by edges of $E_P\,\cup\, E_R$ lies
in a different set $Y_j$.

We define the \emph{right-to-root path} $P_R(e)$ starting at $e$ and
ending at $v_0$, obtained by deleting edges from $W_R(e)$ by the
following method. We follow $W_R(e)$ from $e$, the first time we meet
a vertex $v$ that appears twice in the sequence of vertices
$(u_i)_{0\leq i \leq k}$ of $W_R(e)$. Let $m=\min \{i\, : \, u_i=v\}$
and $M=\max \{i\, : \, u_i=v\}$.  Then we delete all the edges of
$W_R(e)$ between $u_m$ and $u_M$. We repeat the process until reaching
$v_0$. Note that  $P_R(e)$ is not ``rightmost''.
For $e\in D_0$, let $h(e)$ be the set of inner vertices of $P_R(e)$ that
have outgoing edges on the right side of $P_R(e)$.

\begin{lemma}
  \label{lem:geometryrtor}
  $|P_R(e)|\leq |W_R(e)|\leq |P_R(e)|+24$ and  $|h(e)|\leq 4$.
\end{lemma}

\begin{proof}
Consider a vertex $v$ appearing at least twice in the sequence $(u_i)_{0\leq i \leq k}$.
Let $m=\min \{i\, : \, u_i=v\}$ and $M=\max \{i\, : \, u_i=v\}$. We have $0\leq m < M \leq k$.
By Lemma~\ref{lem:alphabetadecrease}, we have $f(\alpha_m) <f(\beta_M)$ and $\lambda(\beta_M)\leq \lambda(\alpha_m)- (M-m)$.
By Lemma~\ref{lem:Mm6}, we have $\lambda(\alpha_m)-6\leq  \lambda(\beta_M)$. So $M-m\leq 6$.

Suppose by contradiction that there is $1\leq p\leq t$, such that
$\alpha_m$ and $\beta_m$ are in $Y_p$. Then $\alpha$ and $\beta$ lie
in the same set of consecutive angles around $v$ delimited by edges of
$E_P\,\cup\, E_R$. Since $f(\alpha_m) <f(\beta_M)$, there is no edge
of $E_P\,\cup\, E_R$ incident to $v$ in the \ccw sector from
$\alpha_m$ to $\beta_M$. Moreover, all the edges of $E_N$ incident to
$v$ in this sector are entering $v$. So By
Lemma~\ref{lem:labelvariation}, the sequence of labels from $\alpha_m$
to $\beta_M$ is increasing around $v$ in \ccw order. So
$\lambda(\alpha_m) \leq \lambda(\beta_M)$, a contradiction.  So there
exists $1\leq p<q\leq t$, such that $\alpha_m\in Y_p$ and
$\beta_M\in Y_q$.

With the same notations as in Lemma~\ref{lem:alphabetadecrease}, the sequence 
$$(f_p)_{0\leq p\leq r}=(f(\alpha_0),f(\gamma^1_1),\ldots,f(\gamma^1_{p_1}),\ldots,f(\gamma^{k-1}_1),\ldots,f(\gamma^{k-1}_{p_{k-1}}),f(\beta_k))$$
is increasing. Thus the sequence $I=(\{i\, :\, f_p\in Y_i\})_{0\leq p\leq r}$ is increasing. 

The path $P_R(e)$ is obtained by following $W_R(e)$ from $e$, each
time we meet a vertex $v$ that appears twice in the sequence of
vertices of $W_R(e)$, then we delete all the edges of $W_R(e)$ between
$u_m$ and $u_M$. Since $M-m\leq 6$, we have deleted at most $6$ edges
from $W_R(e)$.  Since there exists $1\leq p<q\leq t\leq 5$ with
$\alpha_m\in Y_p$ and $\beta_M\in Y_q$, and the sequence $I$ is
increasing, there is at most $4$ such steps of deletions. Thus in
total, we have deleted at most $24$ edges to obtained $P_R(e)$ from
$W_R(e)$ and there are at most $4$ inner vertices of $P_R(e)$ that
have outgoing edges on the right side of $P_R(e)$.
\end{proof}

Finally we obtain the following lemma by combining Lemmas~\ref{lem:wr186} and~\ref{lem:geometryrtor}:

\begin{lemma}
\label{lem:Pr426}
For all $e=uv \in D_0$, we have
$$m(u)-42\leq\left | P_R(e) \right |\leq m(u)+6$$
\end{lemma}

\subsubsection{Relation with shortest paths}
\label{sec:relationwith}

Let $e=uv \in D_0$. Consider $P_R(e)=(u_0=u,u_1=v,...,u_k=v_0)$ the
right-to-root path starting at $e$ and $h(e)$ the set of inner
vertices of $P_R(e)$ that have outgoing edges on the right side of
$P_R(e)$.  Recall that $|h(e)|\leq 4$ by
Lemma~\ref{lem:geometryrtor}.

Let $S=(w_0,w_1,...,w_p)$ be a path of
$G$ with distinct extremities and meeting $P(e)$ only at $w_0$ and
$w_p$, such that $w_0=u_i$ and $w_p=u_j$ for $0\leq i < j \leq k$. Let
$C=(w_0,\ldots, w_p=u_j,\ldots,u_i)$ be the cycle formed by the union
of $S$ and $(u_i,\ldots,u_j)$, given with the traversal direction
corresponding to $S$ oriented from $w_0$ to $w_p$.

We say that $S$ \emph{leaves $P_R(e)$ from the right} if $i>0$ and $S$
leaves $P_R(e)$ by its right side.  Otherwise, we say that $S$
\emph{leaves $P_R(e)$ from the left}. In particular, if $i=0$, then
$S$ leaves $P_R(e)$ from the left, by convention.
Likewise, we say that $S$ \emph{enters $P_R(e)$ from the right} if $j<k$ and $S$ enters $P_R(e)$ by its right side. Otherwise, we say that $S$ \emph{enters  $P_R(e)$ from the left}. In particular, if $j=k$, then $S$ enters $P_R(e)$ from the left, by convention.  

We define different possible types for $S$, depending on whether $S$ is
leaving/entering on the left or right side of $P_R(e)$, whether $C$ is
contractible or not, and whether $C$ contains some vertices of $V(e)$
or not.  We say that $S$ has \emph{type $LR$} (respectively type $RR$,
type $RL$, type $LL$) if $S$ leaves $P_R(e)$ from the left
(respectively right, right, left), enters $P_R(e)$ from the right
(respectively right, left, left). When $C$ is contractible, we add the
subscript $\ell$ or $r$ depending on whether $C$ delimits a region
homeomorphic to an open disk on its left or right side. When $C$ is
non-contractible, we add the subscript $n$. When $C$ contains some
vertices of $h(e)$, we add the superscript $h$.  Thus we have define
twenty-four types $LR_\ell$, $RR_\ell$, $RL_\ell$, $LL_\ell$, $LR_r$,
$RR_r$, $RL_r$, $LL_r$, $LR_n$, $RR_n$, $RL_n$, $LL_n$, $LR_\ell^h$,
$RR_\ell^h$, $RL_\ell^h$, $LL_\ell^h$, $LR_r^h$, $RR_r^h$, $RL_r^h$,
$LL_r^h$, $LR_n^h$, $RR_n^h$, $RL_n^h$, $LL_n^h$ so that a path $S$
as defined above is of exactly one type.

We show the following  inequality between $p$, $i$ and $j$
depending on the type:

\begin{lemma}
\label{lem:contractible}
We  have $p\geq j-i+c$ where $c$ is a constant given in Table~\ref{tab:type} that depends on the type of $S$.
\end{lemma}

\begin{table}[h!]
\center
\begin{tabular}{|c|c|c|c|c|c|c|c|c|c|c|c|}
\hline
         $LR_\ell$ &
{$RR_\ell$} &
{$RL_\ell$} &
{$LL_\ell$} &
{$LR_r$} &
{$RR_r$} &
{$RL_r$} &
{$LL_r$} &
{$LR_n$} &
{$RR_n$} &
{$RL_n$} &
{$LL_n$} \\ \hline
      -2 & 0 & -3 & -5 & 4 & 6 & 3 & 1 & 1 & 3 & 0 & -2 \\
\hline
\end{tabular}

\ \\

\ \\

  \begin{tabular}{|c|c|c|c|c|c|c|c|c|c|c|c|}
\hline
          $LR_\ell^h$ &
{$RR_\ell^h$} &
{$RL_\ell^h$} &
{$LL_\ell^h$} &
{$LR_r^h$} &
{$RR_r^h$} &
{$RL_r^h$} &
{$LL_r^h$} &
{$LR_n^h$} &
{$RR_n^h$} &
{$RL_n^h$} &
{$LL_n^h$} \\ \hline
        -10 & -8 & -11 & -13 & -4 & -2 & -5 & -7 & -3 & -1 & -4 & -6 \\
\hline
\end{tabular}
\caption{Values of $c$ in Lemma~\ref{lem:contractible}.}
\label{tab:type}
\end{table}

\begin{proof}

Suppose first that $C$ is contractible. Let $R$ be the region homeomorphic to an open disk that is delimited by $C$. Let $t$ be the size of $C$, so $t=j-i+p$.
Let $G'$ be the planar map formed by all the vertices and edges that lie in $R$ (including its border). Let $n',m',f'$ be the number of vertices, edges, faces of $G'$ respectively. By Euler's formula, we have $n'-m'+f'=2$.
 All inner faces of $G'$ have degree three and its outer face has degree $t$, so $3(f'-1)=2m'-t$.
 Let $y$ be the number of edges in the interior of $R$ incident to $C$ and leaving $C$.
  Since $G$ is $3$-orientation, it follows that $m'=3(n'-t)+y+t$.
  So, by combining the three equalities, we have
\begin{equation}
\label{contractile}
y=t-3
\end{equation}

Assume that $S$ is of type $LR_\ell$. For $i<m\leq j$, the number of edges that
are in the interior of $R$ and leaving $u_m$ is $0$. Then we obtain $y\leq 3p-p-1$. By (\ref{contractile}), we obtain $p\geq j-i-2$. 

Assume that $S$ is of type $RR_\ell$. For $i\leq m\leq j$, the number of edges that
are in the interior of $R$ and leaving $u_m$ is $0$. Then we obtain $y\leq 3(p-1)-p$. By (\ref{contractile}), we obtain $p\geq j-i$.

Assume that $S$ is of type $RL_\ell$. For $i\leq m< j$, the number of edges that
are in the interior of $R$ and leaving $u_m$ is $0$. Then we obtain $y\leq 3p-p$. By (\ref{contractile}), we obtain $p\geq j-i-3$.

Assume that $S$ is of type $LL_\ell$. For $i< m< j$, the number of edges that
are in the interior of $R$ and leaving $u_m$ is $0$. Then we obtain $y\leq 3(p+1)-p-1$. By (\ref{contractile}), we obtain $p\geq j-i-5$. 

When $S$ is of type $LR_\ell^h$, $RR_\ell^h$, $RL_\ell^h$,
$LL_\ell^h$. The argument is exactly the same as above except that
there might be some vertices of $h(e)$ along $C$. Each such vertex
has at most $2$ edges leaving in the interior of $R$ and there is at
most $4$ such vertices along $C$. So we obtain a difference of $8$
between the  two rows of Table~\ref{tab:type} for these cases.

Assume that $S$ is of type $LR_r$. For $i<m\leq j$, the number of edges that
are in the interior of $R$ and leaving $u_m$ is $2$ if $m<j$ and $3$ if $m=j$. Then we obtain $y\geq 2(j-i-1)+3$. By (\ref{contractile}), we obtain $p\geq j-i+4$.

Assume that $S$ is of type $RR_r$. 
 For $i\leq m\leq j$, the number of edges that
are in the interior of $R$ and leaving $u_m$ is $2$ if $m<j$ and $3$ if $m=j$. Then we obtain $y\geq 2(j-i)+3$. By (\ref{contractile}), we obtain $p\geq j-i+6$.

Assume that $S$ is of type $RL_r$. For $i\leq m < j$, the number of edges that
are in the interior of $R$ and leaving $u_m$ is $2$. Then we obtain $y\geq 2(j-i)$. By (\ref{contractile}), we obtain $p\geq j-i+3$. 

Assume that $S$ is of type $LL_r$. For $i<m<j$, the number of edges that
are in the interior of $R$ and leaving $u_m$ is $2$. Then we obtain $y\geq 2(j-i-1)$. By (\ref{contractile}), we obtain $p\geq j-i+1$.

Again, when $S$ is of type $LR_r^h$, $RR_r^h$, $RL_r^h$,
$LL_r^h$. The argument is exactly the same as above except that
there might be some vertices of $h(e)$ along $C$. Each such vertex
has at most $2$ edges leaving on the right side of $P_R(e)$,
i.e. outside $R$, and there is at
most $4$ such vertices along $C$. So we obtain a difference of $8$
between the  two rows of Table~\ref{tab:type} for these cases.

Suppose now that $C$ is non-contractible

Assume that $S$ is of type $LR_n$. For $i<m\leq j$, the number of outgoing edges that are incident to $u_m$ and leaving $C$ by its right side is equal to $2$ if $m<j$ and $3$ if $m=j$. So the number of edges leaving $C$ by its right is at least $2(j-i-1)+3$. Moreover the number of edges leaving $C$ by its left side is at most $3p-p-1$. 
Since $D_0$ is balanced, we have exactly the same number of outgoing edges incident to each side of $C$. Then we obtain $p\geq j-i+1$. 

Assume that $S$ is of type $RR_n$. For $i\leq m\leq j$, the number of outgoing edges that are incident to $u_m$ and leaving $C$ by its right side is equal to $2$ if $m<j$ and $3$ if $m=j$. So the number of edges leaving $C$ by its right is at least $2(j-i)+3$. Moreover the number of edges leaving $C$ by its left side is at most $3(p-1)-p$. 
Since $D_0$ is balanced, we obtain $p\geq j-i+3$. 

Assume that $S$ is of type $RL_n$. For $i\leq m< j$, the number of outgoing edges that are incident to $u_m$ and leaving $C$ by its right side is equal to $2$. So the number of edges leaving $C$ by its right is at least $2(j-i)$. Moreover the number of edges leaving $C$ by its left side is at most $3p-p$. 
Since $D_0$ is balanced, we obtain $p\geq j-i$. 

Assume that $S$ is of type $LL_n$. For $i< m< j$, the number of outgoing edges that are incident to $u_m$ and leaving $C$ by its right side is equal to $2$. So the number of edges leaving $C$ by its right is at least $2(j-i-1)$. Moreover the number of edges leaving $C$ by its left side is at most $3(p+1)-p-1$. 
Since $D_0$ is balanced, we obtain $p\geq j-i-2$.

Again, when $S$ is of type $LR_n^h$, $RR_n^h$, $RL_n^h$,
$LL_n^h$. The argument is exactly the same as above except that
there might be some vertices of $h(e)$ along $C$. There is a division
by two in the computation of  these cases that results
in a difference of $4$
between the two rows of Table~\ref{tab:type} for these cases.
\end{proof}

Let $Q$ be a shortest path from $u$ to $v_0$ that maximizes the number of common edges with $P_R(e)$. Subdivide  $Q$ into edge-disjoint sub-paths $S_1,S_2,...,S_t$, each of which 
meets $P_R(e)$ only at its (distinct) endpoints. For $1\leq q \leq t$,
note that $S_q$ is not necessarily edge-disjoint from $P_R(e)$, but if
$S_q$ share an edge with $P(e)$ then it has length $1$.  We assume
that $S_1,S_2,...,S_t$ are ordered so that $Q$ is the concatenation of
$S_1,S_2,...,S_t$, so in particular, $u_0$ is the first vertex of
$S_1$ and $u_k$ is the last vertex of $S_t$. For $1\leq q \leq t$,
note that $S_q$ is not necessarily of a type define previously since
such a path might starts (resp. ends) at a vertex $u_i$ (resp. $u_j$) of
$P_R(e)$ such that $j<i$.

For $i,j$ in $\{0,k\}$, the sub-path of $P_R(e)$ between $u_i$ and $u_j$ is denoted by $P_R(e)[i,j]$. Likewise, if $u_i,u_j$ are vertices of $Q$, then the sub-path of $Q$ between $u_i$ and $u_j$ is denoted by $Q[i,j]$. \\

\begin{lemma}
\label{lem:zone}
Consider $1\leq q\leq q+9\leq q' \leq t$ such that $S_q$ starts at a
vertex $u_i$, ends at vertex $u_j$, with $i<j$, and $(u_{i'}, u_{j'})$
are the extremities of $S_{q'}$ with $i'\leq j'$ (note that $S_{q'}$
may starts at $u_{i'}$ or $u_{j'}$). Then we have $j< i'$.
\end{lemma}
\begin{proof}
  Suppose by contradiction that $i'\leq j$.
Let $q_1=\min\{q\in[\![1,t]\!]\ : \textrm{$S_q$ starts at $u_{i''}$, ends
  at $u_{j''}$ with }i''\leq i' \leq j''\}$. Note that $q_1\leq q$.
Let $(i_1, j_1)$ be such that $S_{q_1}$
starts at $u_{i_1}$, ends at $ u_{j_1}$.
For $2\leq r \leq 8$, let
$q_r=q_1+r$. Note that $q_8<q+9\leq q'$.
Let $p_1,  \cdots, p_8$ be the lengths of
$S_{q_1},  \cdots, S_{q_8}$ respectively.  By
Lemma~\ref{lem:contractible}, we have $p_1\geq j_1-i_1-13$. Moreover,
we have $|Q[i_1,i']|\geq p_1+\cdots + p_8\geq p_1+7$. Since
$Q$ is a shortest path, we have
$|P_R(e)[i',j_1]|\geq |Q[j_1,i']|\geq p_2+\cdots +p_8\geq 7$.
We obtain the following contradiction:
$$|P_R(e)[i_1,i']|=|P_R(e)[i_1,j_1]|-|P_R(e)[i',j_1]|\leq j_1-i_1 -7\leq p_1+6 \leq |Q[i_1,i']|-1.$$
\end{proof}

For all types
$\xi \in \{LR_\ell, RR_\ell, RL_\ell, LL_\ell, LR_r, RR_r, RL_r, LL_r,
LR_n, RR_n, RL_n, LL_n\}$, let
$n_{\xi}(Q,e)=|\{j \in \{1,...,t\}: S_j \text{ has type }
\xi\}|$. 

\begin{lemma}
\label{lem:nlll2}
$n_{LL_\ell}(Q,e)\leq 2$ 
\end{lemma}

\begin{proof}
  Suppose by contradiction that $n_{LL_\ell}(Q,e)\geq 3$. Let
  $q_1,q_2,q_3$ be three distinct elements of $\{1, \cdots, t\}$ such
  that $S_{q_1}$, $S_{q_2}$ and $S_{q_3}$ have type $LL_\ell$. For
  $1\leq r \leq 3$, let $(u_{i_r}, u_{j_r})$, be the extremities of
  $S_{q_r}$, such that $S_{q_r}$ starts at $u_{i_r}$ and ends at
  $u_{j_r}$.  Let $p_1$, $p_2$ and $p_3$ be the length of $S_{q_1}$,
  $S_{q_2}$ and $S_{q_3}$. We assume, w.l.o.g., that
  $i_1<i_2<i_3$. Then, one can see that $i_1<i_2<i_3<j_3<j_2<j_1$. By
  Lemma~\ref{lem:contractible}, we have
  $p_1\geq j_1-i_1-5$. Let $q_{m}=\min\{q_1, q_2,
  q_3\}$ and $q_{M}=\max\{q_1,q_2, q_3\}$.  Since
  $Q$ is a shortest path we have $|P_R(e)[i_1,i_m]| +
  |P_R(e)[j_{M},j_1]| \geq
  |Q[i_m,i_1]|+|Q[j_1,j_M]|$.  Moreover, whenever $q_1= q_m$, $q_1
  =q_M$ or
  $q_m<q_1<q_M$, one can check that $|Q[i_m,i_1]|+|Q[j_1,j_M]|\geq
  4$.  We also have $|Q[i_{m},j_{M}]|\geq p_1+p_2+p_3+2 \geq p_1+4$.

Then we obtain the following contradiction:
\begin{align*}
|P_R(e)[i_{m},j_{M}]| & = |P_R(e)[i_1,j_1]|- |P_R(e)[i_1,i_m]|-|P_R(e)[j_{M},j_1]|\\
 & \leq (j_1-i_1)-|Q[i_m,i_1]|-|Q[j_1,j_M]|\\
 & \leq (j_1-i_1)-4\\
 & \leq p_1+1\\
 & \leq |Q[i_{m},j_{M}]|-3
\end{align*}
\end{proof}

For $1\leq z\leq |h(e)|$, let
$t_z=min\{q\in [\![1, t]\!]: S_q \textrm{ ends at $u_j$} \textrm{ with
  $P_R(e)[0,u_j]$} \textrm{ contains at least $z$ elements} \textrm{ of
  $h(e)$} \}$.  Let $X=\cup_{1\leq z\leq h(e)}[\![t_z, t_z+18[\![$ and
$Y=[\![1, t]\!]\setminus X$ and $Z=[\![1, t]\!]\setminus Y$. So $[\![1, t]\!]$ is
partitioned into $Y,Z$.
 By Lemma~\ref{lem:geometryrtor}, we have $h(e)\leq 4$, so
 $|Z|\leq 4\times 18=72$.
Note that
$Y$ has been defined so that it satisfies the
following by Lemma~\ref{lem:zone}: if $q,q'\in [\![1, t]\!]$ are such that
$q\in Y$, $q-9\leq q'\leq q$, and
$S_{q'}$ has extremities $(u_i,u_j)$, then $P_R(e)[i,j]$ contains no
vertex of $h(e)$.\\
For $q \in \{1,...,t\}$, we say that $S_q $ has type $h$ if $S_q$ is
of one of the type $LR_\ell^h, RR_\ell^h, RL_\ell^h, LL_\ell^h,$
$LR_r^h, RR_r^h, RL_r^h, LL_r^h,$ $LR_n^h, RR_n^h, RL_n^h, LL_n^h$.

\begin{lemma}
\label{lem:jpluspetiti}
Consider $q_1,q_2\in Y$, such that $q_1<q_2$
and $S_{q_1},S_{q_2}$ are of type $LL_n$. If $i_1,j_1,i_2,j_2$
are such that $S_{q_1}$, $S_{q_2}$ have extremities
$(u_{i_1},u_{j_1})$ and $(u_{i_2}, u_{j_2})$ with $i_1<j_1$ and
$i_2<j_2$,
then $j_1\leq i_2$.
\end{lemma}

\begin{proof}
 Suppose by contradiction that $i_2<j_1$. 
Let $p_1$, $p_2$ be the length of $S_{q_1}$ and $S_{q_2}$.  By Lemma~\ref{lem:contractible}, we have $p_1\geq j_1-i_1-2$.
Since $q_1<q_2$ we have $i_1\neq j_2$,  $i_1\neq i_2$ and $j_1\neq j_2$.
 We consider the four following cases: $j_2< i_1$ or $i_1<j_2<j_1$ or $i_2<i_1<j_1<j_2$ or $i_1<i_2<j_1<j_2$.

\begin{itemize}
\item \emph{If $j_2< i_1$:} Let
  $q_0=\max\{q\in[\![1,q_1[\![\ : \text{$S_q$ starts at $u_i$, ends at
    $u_j$ with } i\leq j_2 \leq j\}$. Let $(u_{i_0},u_{j_0})$ be the
  extremities of $S_{q_0}$ with $ i_0\leq j_2 \leq j_0$.  Let $p_0$ be
  the length of $S_{q_0}$. Since $ i_0\leq j_2 \leq j_0$, by
  definition of $Y$ and Lemma~\ref{lem:zone}, we have that $S_{q_0}$
  is not of type $h$.  By Lemma~\ref{lem:contractible}, we have
  $p_0\geq j_0-i_0-5$.  Moreover, we have
  $|Q[i_0,j_2]|\geq p_0+p_1+p_2+1\geq p_0+3$. Since $Q$ is a shortest
  path, we have
  $|P_R(e)[j_2,j_0]|\geq |Q[j_0,j_2]|\geq p_1+p_2+1\geq 3$.  We obtain
  the following contradiction:
$$|P_R(e)[i_0,j_2]|=|P_R(e)[i_0,j_0]|-|P_R(e)[j_2,j_0]|\leq j_0-i_0
-3\leq p_0+2 \leq |Q[i_0,j_2]|-1$$

\item \emph{If $i_1<j_2<j_1$:}  We have $|Q[i_1,j_2]|\geq p_1+p_2+1\geq p_1+2$. Since $Q$ is a shortest path, we have $|P_R(e)[j_2,j_1]|\geq |Q[j_1,j_2]| \geq 
1+p_2\geq 2$. 
We obtain the following contradiction:
$$|P_R(e)[i_1,j_2]|=|P_R(e)[i_1,j_1]|-|P_R(e)[j_2,j_1]|\leq j_1-i_1-2 \leq p_1\leq |Q[i_1,j_2]|-2$$

\item If $i_2<i_1<j_1<j_2$: Let
  $q_0=\max\{q\in[\![1,q_1[\![\ : \text{$S_q$ starts at $u_i$, ends at
    $u_j$ with } i\leq i_2 \leq j\}$. Let $(u_{i_0},u_{j_0})$ be the
  extremities of $S_{q_0}$ with $ i_0\leq i_2 \leq j_0$.  Let $p_0$
  be the length of $S_{q_0}$. Since $ i_0\leq i_2 \leq j_0$, by
  definition of $Y$ and Lemma~\ref{lem:zone}, we have that $S_{q_0}$
  is not of type $h$.  We consider two cases depending on whether
  $j_2 \leq j_0$ or not.

\begin{itemize}
\item \emph{$j_2 \leq j_0$:}
By Lemma~\ref{lem:contractible}, we have $p_0\geq j_0-i_0-5$. Moreover, we have $|Q[i_0,j_2]|\geq p_0+p_1+p_2+2\geq p_0+4$. Since $Q$ is a shortest path, we have $|P_R(e)[j_0,j_2]|\geq p_1+p_2+2\geq 4$. We obtain the following contradiction:
$$|P_R(e)[i_0,j_2]|=|P_R(e)[i_0,j_0]|-|P_R(e)[j_2,j_0]|\leq j_0-i_0-4\leq p_0+1 $$$$\leq |Q[i_0,j_2]|-3$$

\item \emph{$j_0 < j_2$:} We have $i_0<i_2<j_0<j_2$ so one can remark
  that $S_{q_0}$ is not of type $LL_\ell$. By
  Lemma~\ref{lem:contractible}, we have $p_0\geq j_0-i_0-3$. Moreover,
  we have $|Q[i_0,i_2]|\geq p_0+p_1+1\geq p_0+2$. Since $Q$ is a
  shortest path, we have
  $|P_R(e)[i_2,j_0]|\geq |Q[j_0,i_2]|\geq p_1+1\geq 2$.  We obtain the
  following contradiction:
$$|P_R(e)[i_0,i_2]|=|P_R(e)[i_0,j_0]|-|P_R(e)[i_2,j_0]|\leq  j_0-i_0 -2 \leq p_0 +1 \leq |Q[i_0,i_2]|-1$$
\end{itemize}

\item {$i_1<i_2<j_1<j_2$:}
We have $|Q[i_1,i_2]|\geq p_1+1$. Since $Q$ is a shortest path, we have $|P_R(e)[i_2,j_1]|\geq |Q[j_1,i_2]| \geq 1$. 
We obtain the following:
$$|P_R(e)[i_1,i_2]|=|P_R(e)[i_1,j_1]|-|P_R(e)[i_2,j_1]|\leq j_1-i_1-1 \leq p_1+1\leq |Q[i_1,i_2]|$$
Since $Q$ is a shortest path, we obtain $|P_R(e)[i_1,i_2]|= |Q[i_1,i_2]|$.
Consider the walk $Q'$ obtain by replacing the part $Q[i_1,i_2]$ in $Q$ by $P_R(e)[i_1,i_2]$. Thus $Q'$ is a walk from $u_0$ to $v_0$ that have the same length as $Q$, so $Q'$ is a shortest path. Moreover $Q'$ has strictly more edges of $P_R(e)$ than $Q$, a contradiction.
\end{itemize}

\end{proof}

Let $n_{LL_n}^Y(Q,e)$ be the number of integers in $q\in Y$ such that $S_q$ has type $LL_n$.

\begin{lemma}
\label{lemma11bis}
$n_{LL_n}^Y(Q,e)\leq 2$ 
\end{lemma}

\begin{proof}
Suppose by contradiction that $n_{LL_n}^Y(Q,e)\geq 3$. Let $q_1,q_2,q_3$ be three distinct elements of $Y$ such that $S_{q_1}$, $S_{q_2}$ and $S_{q_3}$ are of type $LL_n$ and $q_1<q_2<q_3$.
Let $(u_{i_1}, u_{j_1})$, $(u_{i_2}, u_{j_2})$ and $(u_{i_3}, u_{j_3})$ be the extremities of $S_{q_1}$, $S_{q_2}$ and $S_{q_3}$.
Then by Lemma~\ref{lem:jpluspetiti}, we have
$i_1<j_1\leq i_2<j_2 \leq i_3< j_3$.
Let $C_1$ (resp. $C_2$, $C_3$) be the cycle formed by the union of $S_1$ (resp. $S_2$, $S_3$) and $P_R(e)[i_1, j_1]$  (resp. $P_R(e)[i_2, j_2]$, $P_R(e)[i_3, j_3]$).
The two non contractible cycle $C_1$ and $C_3$ are vertex disjoint. Thus we are in the situation of Figure~\ref{fig:nLLn}, where $C_1,C_3$ are homotopic but with opposite traversal direction.
Then the union of $C_1$, $C_3$ and $P_R(e)[j_1,i_3]$  delimit a contractible region whose interior contain all the edges of $S_2$. Then $C_2$ is contractible, a contradiction.
\end{proof}

\begin{figure}[h]
\center
\includegraphics[scale=0.5]{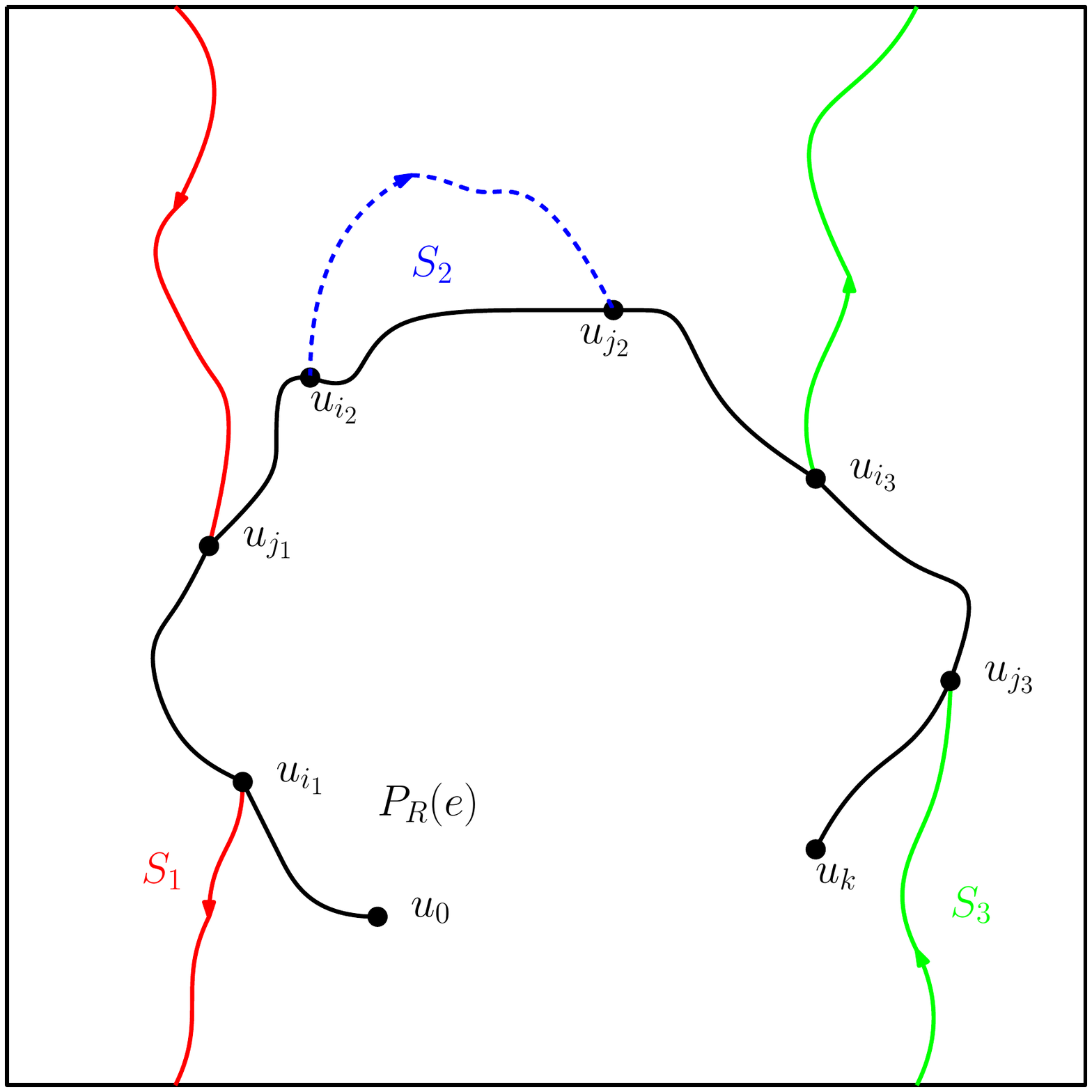}
\caption{Situation of Lemma~\ref{lemma11bis}.}
\label{fig:nLLn}
\end{figure}

\begin{lemma}
\label{bornonrightmostpath}
$$\left | Q \right | \geq \left | P_R(e) \right |-2n_{LR_\ell}(Q,e)-3n_{RL_\ell}(Q,e)-922.$$
\end{lemma}

\begin{proof}
By Lemmas~\ref{lem:contractible},  we have 
$$|Q|=\sum_{q=1}^{t}\left | S_q \right |\geq \left | P_R(e) \right |
-2n_{LR_\ell}(Q,e) -3n_{RL_\ell}(Q,e) -5n_{LL_\ell}(Q,e)
-2n_{LL_n}^Y(Q,e)-{13\times |Z|}.$$ Thus we obtain the lemma by
Lemmas~\ref{lem:nlll2} and~\ref{lemma11bis} and since $|X|\leq 72$.
\end{proof}

\begin{lemma}
\label{lem:geometryoftypelr_l}
Consider $q_1,q_2\in Y$, such that $q_1\neq q_2$
and $S_{q_1},S_{q_2}$ are both of type
$LR_l$ or $RL_\ell$. If $S_{q_1}$, $S_{q_2}$ have
extremities $(u_{i_1},u_{j_1})$ and $(u_{i_2}, u_{j_2})$ with
$i_1<j_1$, $i_2<j_2$ and $i_1<i_2$, then $q_1< q_2$.
\end{lemma}

\begin{proof}
  Suppose by
  contradiction that $q_1>q_2$. Let $p_1$, $p_2$ be the length of
  $S_{q_1}$ and $S_{q_2}$. By Lemma~\ref{lem:contractible}, we have
  $p_1\geq j_1-i_1-3$ and $p_2\geq j_2-i_2-3$. Since $q_2<q_1$, we
  have $i_2\neq j_1$. We consider the two following cases: $i_2<j_1$
  or $j_1<i_2$.
\begin{itemize}
\item \emph{If $i_2<j_1$:} We have
  $|Q[i_2,j_1]|\geq p_1+p_2+1\geq p_2+2$. Since $Q$ is a shortest
  path, we have $|P_R(e)[j_1,j_2]|\geq |Q[j_2,j_1]| \geq p_1+1\geq 2$.
  We obtain the following contradiction:
$$|P_R(e)[i_2,j_1]|=|P_R(e)[i_2,j_2]|-|P_R(e)[j_1,j_2]|\leq j_2-i_2-2 \leq p_2+1\leq |Q[i_2,j_1]|-1.$$

\item \emph{If $j_1<i_2$:} Let
  $q_0=\max\{q\in[\![q_1,q_2[\![\ : \text{the extremities }
  i,j\textrm{ of }S_q\textrm{ are such that }i\leq j_1 \leq j\}$. Let
  $(u_{i_0},u_{j_0})$ be the extremities of $S_{q_0}$ with
  $ i_0\leq j_1 \leq j_0$.  Let $p_0$ be the length of
  $S_{q_0}$. Since $ i_0\leq j_1 \leq j_0$, by definition of $Y$ and
  Lemma~\ref{lem:zone}, we have that $S_{q_0}$ is not of type $h$. By
  Lemma~\ref{lem:contractible}, we have $p_0\geq j_0-i_0-5$. Moreover,
  we have $|Q[i_0,j_1]|\geq p_0+p_1+p_2+1\geq p_0+3$. Since $Q$ is a
  shortest path, we have $|P_R(e)[j_0,j_1]|\geq p_1+p_2+1\geq 3$. We
  obtain the following contradiction:
$$|P_R(e)[i_0,j_1]|=|P_R(e)[i_0,j_0]|-|P_R(e)[j_0,j_1]|\leq j_0-i_0-3\leq p_0+2 \leq |Q[i_0,j_1]|-1.$$
\end{itemize}
\end{proof}

We now state a lemma which is analogous to Proposition $11$ and Proposition $12$ of~\cite{berry2013scaling}. 

Consider $C$ a contractible cycle of $G$, given with a traversal
direction. Then $C$ separates the map $G$ into two regions. We define
$V_\ell(C)$ (respectively $V_r(C)$) the set of vertices lying in the
region on the left (resp. right) side of $C$, including $C$. The graphs
$G[V_\ell(C)]$ and $G[V_r(C)]$ denotes the subgraph of $G$ induced by
these set of vertices.

\begin{lemma}
\label{corollaryimportant}
If $n_{LR_\ell}(Q,e)>3$ (resp. $n_{RL_\ell}(Q,e)>3$), then there
exists a contractible cycle $C$ in $G$, given with a direction of
traversal, of length at most $\frac{6|Q|}{n_{LR_\ell}(Q,e)-3}+2$
(resp. $\frac{6|Q|}{n_{RL_\ell}(Q,e)-3}+3$) such that 
for all
$\iota\in \{\ell,r\}$, we have
$\max_{u\in V_\iota(C)}m(u)- \min_{u\in V_\iota(C)}m(u)$ is at least
$\lfloor n_{LR_\ell}(Q,e)/3\rfloor -{79}$ (resp.
$\lfloor n_{RL_\ell}(Q,e)/3\rfloor-{79}$).
\end{lemma}

\begin{proof}

  We prove the lemma for $n_{LR_\ell}(Q,e)>3$ (the proof for
  $n_{RL_\ell}(Q,e)>3$ is similar).   For $1\leq q\leq t$, let
  $n_{LR_\ell}(q)$ be the number of sub-paths of type $LR_\ell$ among
  $\{S_1,\cdots, S_q\}$. Let
  $s=\left \lfloor n_{LR_\ell}(Q,e)/3 \right \rfloor$.  Let $Z$ be the
  set of elements $1\leq q\leq t$, such that $S_q$ is of type
  $LR_\ell$ and $s+1\leq n_{LR_\ell}(q)\leq 2s$. Let $q^\star \in Z$
  such that $|S_{q^\star}|=\min\{|S_{q}|\, : \, q\in Z\}$.  Let
  $S_{q^\star}=(w_0,...,w_p)$ with $w_0=u_i$, $w_p=u_j$ for some
  $0\leq i<j\leq k$ and let $C=(w_0,\ldots, w_p=u_j,\ldots,u_i)$.
  Then
$$|Q|\geq s\,p\geq \frac{n_{LR_\ell}(Q,e)-3}{3}p$$

By Lemma~\ref{lem:contractible}, we have $p\geq j-i-2$. Then
 $|C|=p+j-i \leq 2p+2 \leq \frac{6|Q|}{n_{LR_\ell}(Q,e)-3}+2$ edges. 
 
Finally, by Lemma~\ref{lem:geometryoftypelr_l}, one of $G[V_\ell(C)]$
or $G[V_r(C)]$ contains all sub-paths of type $LR_\ell$ among
$(S_1,\cdots, S_{q^\star}){\bigcap \left\{\bigcup_{i\in Y}S_i\right\}}$
and the other contains all sub-paths of type $LR_\ell$ among
$(S_{q\star},\cdots,S_t){\bigcap \left\{\bigcup_{i\in
      Y}S_i\right\}}$. Therefore, each of $G[V_\ell(C)]$ and
$G[V_r(C)]$ contains at least $s- 18\times 4$ vertices of $P_R(e)$. By
Lemmas~\ref{lem:alphabetadecrease} and~\ref{lem:Mm6}, we obtain
$\max_{u\in V_\iota(C)}m(u)- \min_{u\in V_\iota(C)}m(u)\geq s-72-7$ for all
$\iota\in \{\ell,r\}$.
\end{proof}

\subsection{Approximation of distances by labels}
\label{section8}

As in Section~\ref{sec:proof}, for $n\geq 1$, let $G_n$ be a uniformly
random element of $\mathcal G(n)$.  Let $d_n$ denote the graph
distance $d_{G_n}$. Recall that $\Phi$ denotes the bijection from
$\mathcal T_{r,s,b}(n)$ to $ \mathcal G(n)$ of
Theorem~\ref{them:bijectionbenjamin}.  Let $T_n=\Phi^{-1}(G_n)$.
Therefore $T_n$ is a uniformly random element of
$\mathcal T_{r,s,b}(n)$.

We need several definitions similar to Section~\ref{sec:label}. Let
$V_n$ be the set of vertices of $T_n$. Let $a_n^0$
be the root angle of $T_n$ and $v^0_n$ be its root vertex. Let
$\ell_n=4n+1$.  We define $\Gamma_n$ as the unicellular map obtained
from $T_n$ by adding a special dangling half-edge, called the root
half-edge, incident to the root angle of $T_n$.  The root angle of
$\Gamma_n$, still noted $a_n^0$, is the angle of $\Gamma_n$ just after the root half-edge in
counterclockwise order around its incident vertex.  Let
$A_n=(a_n^0,\ldots,a_n^\ell)$ be the sequence of consecutive angles of
$\Gamma_n$ in clockwise order around the unique face of $\Gamma_n$
starting from $a_n^0$.  Let $\lambda_n$ be the labeling
function of $\Gamma_n$ as defined in Section~\ref{sec:label}.
For each vertex $u$ of $V_n$, let $m_n(u)$ be the minimum of the
labels incident to $u$.

The main result of this section is the following:

\begin{theorem}
  \label{main2bis}
  For all $\epsilon>0$, we have 
  $$\lim_{n\rightarrow \infty}\mathbb{P}\left(\exists u\in V_n:|d_n(u,v_n^0)- m_n(u)|>\epsilon\, n^{1/4}\right)=0. $$
\end{theorem}

Before going into the proof, we need some additional notations.  For
$0\leq i\leq \ell_n$, let $r_n(i)$ be the vertex of $V_n$ incident to
angle $a_n^i$ (i.e the vertex contour function of $\Gamma$).  Given an
integer $0\leq i \leq \ell_n$ and $\Delta>0$, we denote
\begin{align*}
  p_n(i,\Delta) &=\max (\{0\}\cup \left \{ j<i: \left | m_n(r_n(j))-m_n(r_n(i))  \right |\geq \Delta\right
                  \}), \\
  q_n(i,\Delta)&=\min (\{\ell_n\} \cup \left  \{ j>i:  \left |  m_n(r_n(j))-m_n(r_n(i)) \right |\geq \Delta
                 \right \}) \; \text{and}\\
  N_n(i,\Delta) &= | \left \{ r_n(j): \exists j \in \llbracket p_n(i,\Delta),
                  q_n(i,\Delta) \rrbracket \right \} |. 
\end{align*}

{The proof of the following lemma is omitted, it is almost identical to
\cite[Lemma $8.2$]{berry2013scaling}:}

\begin{lemma}
  \label{lemma25}
  For  all $\epsilon>0$  and  $\beta>0$, there  exists $\alpha>0$  and
  $n_0 \in \mathbb{N}$ such that for every $n \geq n_0$,
  $$\mathbb{P}\left ( \inf\left \{ N_n(i,\beta n^{1/4}): 0\leq i\leq 2n+1 \right \}\geq \alpha n \right )\geq 1-\epsilon.$$
\end{lemma}

We are now ready to prove the main theorem of this section.

\begin{proof}[Proof of Theorem~\ref{main2bis}]
   By Lemma~\ref{prop:bornesupdistance}, for $n\geq 1$ and $u\in V_n$,
  we have $d_{n}(v_n^0,u)\leq m_n(u)$.  So it suffices to prove that
  for all $\epsilon >0$
 $$\lim_{n\rightarrow \infty}\mathbb{P}\left(\exists u\in
   V_n:d_n(v_n^0,u)<m_n(u)-\epsilon\, n^{1/4}\right)=0. $$
 
 This is equivalent to show that for all $\epsilon>0$,
  $$\limsup_{n\rightarrow \infty}\mathbb{P}\left(\exists u \in V_n: d_{n}(v^0_n,u)<m(u)-15\epsilon n^{1/4}+964 \right)\leq 4\epsilon.$$
 Denote by $\diam(G_n)$ the diameter of the graph $G_n$. Consider
 $\epsilon >0$. By
 Lemma~\ref{tightofd_n}, there exists $y>0$ such
 that $\mathbb{P}\{\diam(G_n)\geq yn^{1/4} \}<\epsilon$.

  Now, assume that there exists $n_0\in \mathbb N$, such that for all
  $n\geq n_0$, there exists $u_n\in V_n$ such that
  $d_{G_n}(u_n,v_n^0)<m(u_n)-15\epsilon n^{1/4}-964$. Consider the
  canonical orientation of $G_n$ and let $e_n$ be an outgoing edge of
  $u_n$. With the notations of Section~\ref{rightmostpathsub}, let $P_n=P_R(e_n)$ be the
right-to-root path starting at $e_n$. Let $Q_n$ be a shortest path from $u_n$ to $v_n^0$ that maximizes the number of common edges with $P_n$.

  By
  Lemmas~\ref{lem:Pr426} and \ref{bornonrightmostpath},              we              have

  \begin{align*}
 2n_{LR_\ell}(Q_n,e_n)+3n_{RL_\ell}(Q_n,e_n) \geq &\left | P_n \right |-\left | Q_n \right | -922\\
  \geq& (m_n(u)-42)-(m_n(u)-15\epsilon  n^{1/4}-964)-922\\
 \geq &15\epsilon  n^{1/4}
  \end{align*}

  Thus for $n_0$ large enough we have (for each
  $n\geq n_0$)  either
  $n_{LR_\ell}(Q_n,e_n)\geq \max(3,3\epsilon n^{1/4})$ or
  $n_{RL_\ell}(Q_n,e_n)\geq \max(3,3\epsilon n^{1/4})$.
We call $B_n$ the event $G_n$ contains a contractile cycle
 $C$ of length at most $(2y/\epsilon+4)$, given with a traversal direction, such that for both
 $\iota \in \{l,r\}$, we have
  $$\max_{u\in   V_\iota(C)}m_n(u)-   \min_{u\in  V_\iota(C)}m_n(u)\geq   \epsilon
  n^{1/4}-79.$$
  We deduce from
  Lemma~\ref{corollaryimportant} that, for $n_0$ large enough and all
  $n\geq n_0$, either
  $\diam(G_n)\geq y n^{1/4}$ or $B_n$ occurs.

  Therefore it suffices to prove that
  $$ \mathbb{P}(B,\diam(G_n)\leq yn^{1/4}) \leq 3\epsilon.$$

  Consider $n\geq  n_0$ such  that $B$  occurs. Let $C$  be as  in the
  definition  of $B$.  Let $F$  be the  subgraph of  $T_n$ induced  by
  $V(G_n)\setminus   V(C)$.    Recall   that    $G_n[V_l(C)]$   (resp.
  $G_n[V_r(C)]$) is the sub-graph of  $G_n$ induced by $V_l(C)$ (resp.
  $V_r(C)$). Then each component of  $F$ is contained in $G_n[V_l(C)]$
  or $G_n[V_r(C)]$. By Lemma~\ref{lem:mm7}, for $\{u,v\}\in E(G_n)$ we
  have $|m(u)-m(v)|\leq 7$. It follows  that, for $\iota \in \{l,r\}$,
  there exists one component $F_\iota$ of $F$ such that
  $$\max_{u\in  V(F_\iota)}m_n(u)-   \min_{u\in  V(F_\iota)}m_n(u)\geq  \epsilon^2
  n^{1/4}/(2y+4\epsilon)-79. $$ 
  Using again Lemma~\ref{lem:mm7},
  then  for $\iota\in  \{l,r\}$, there
  exists $v_\iota\in F_\iota$ such that
  \begin{align*}
  \min_{v\in                                  V(C)}|m(v_\iota)-m(v)|\geq
 & \left(\frac{\epsilon^2n^{1/4}}{2y+4\epsilon}-79\right)\big{/}2-7-(2y/\epsilon+2)\times 7\\
 &\geq \frac{\epsilon^2n^{1/4}}{4y+8\epsilon}-19-14y/\epsilon.
  \end{align*}
 
  Now           for          $\iota\in           \{\ell,r\}$,          let
  $j_\iota=\inf\{0\leq   i\leq   \ell_n:  r_{n}(i)=v_\iota\}$.   Fix   any
  $\beta\in (0,\epsilon^2/(4y+8\epsilon))$. By Lemma~\ref{lemma25}, there exists
  $\alpha>0$ such that for $n$ large enough,
  $$\mathbb{P}\left(\min      \{|N_n(j_\ell,\beta     n^{1/4})|,|N_n(j_r,\beta
    n^{1/4})| \}\leq \alpha n \right)\leq \epsilon.$$ For $n$
  sufficiently large, we have
  $\frac{\epsilon^2n^{1/4}}{4y+8\epsilon}-19-14y/\epsilon>\beta
  n^{1/4}$.  Then we have for          $\iota\in           \{\ell,r\}$,
  $N(j_\iota,\beta n^{1/4}) \subset V_\iota(C)$. It follows that for $n$
  large enough,
  $$\mathbb{P}(B,\diam(G_n)\leq y n^{1/4})
  \leq $$$$ \epsilon + \mathbb{P}\left(\exists C \text{ contractile
      cycle}, |C|\leq 2y/\epsilon+4, \min\{|V_l(C)|,|V_r(C)| \}\geq
    \alpha n \right).$$ The event $\{\exists C \text{ contractile
    cycle}, |C|\leq 2y/\epsilon+4, \min\{|V_l(C)|,|V_r(C)| \}\geq \alpha
  n \}$ means that
  $G_n$ contains a separating contractile cycle of length at most
  $2y/\epsilon+4$ that separates
  $G_n$ into two sub-triangulations both of size at least $\alpha
  n$. It remains to prove that this has probability going to
  $0$ when $n$ goes to infinity. Let
  $p_{n,m}$ (resp.
  $t_{n,m}$) be the number of simple triangulation of an
  $m$-gon with
  $n$ inner vertices (resp.  the number of essentially simple toroidal
  maps on the torus with
  $n$ vertices, such that all faces have size three except one that
  has size $m$),  rooted at a maximal triangle.  From previously known
  estimates, there exist two constants
  $A_m$ (see~\cite{brown1964enumeration}) and
  $B_m$ (by Corollary~\ref{cor:nosymmetric}) such that
  \begin{equation*}
    p_{n,m}\leq A_mn^{-5/2}\left(\frac{256}{27}\right)^n \quad\text{and}\quad
    {t}_{n,m}\leq B_m \left(\frac{256}{27}\right)^n
  \end{equation*}
  (the                 upper                 bound                 for
  $p_{n,m}$ estimates the number of arbitrarily rooted triangulations,
  of which there are more than the type counted by $p_{n,m}$ itself).

Let $\Gamma_n$ be the event $G_n$
  contains  a   separating  contractile   cycle  of  length   at  most
  $2y/\epsilon+4$ that separates $G_n$  into two sub-triangulations both
  of size at least $\alpha n$. We have:
  
\begin{align*}
  \mathbb{P}(\Gamma_n) &\leq
                         \Upsilon^{-1}
                         \left(\frac{256}{27}\right)^{-n}\sum_{k=3}^{\lfloor
                         2y/\epsilon+4\rfloor}  \sum_{\ell  =  \lfloor
                         \alpha  n  \rfloor}  ^{\lfloor  (1-\alpha)  n
                         \rfloor} p_{\ell,k}{t}_{n-\ell,k} \\
                       &\leq
                         \Upsilon^{-1}
                         \left(\frac{256}{27}\right)^{-n}\sum_{k=3}^{\lfloor
                         2y/\epsilon+4\rfloor}  \sum_{\ell  =  \lfloor
                         \alpha  n  \rfloor}  ^{\lfloor  (1-\alpha)  n
                         \rfloor} A_k  \ell^{-5/2} \left(  \frac {256}
                         {27} \right)^\ell B_k   \left(  \frac {256}
                         {27} \right)^{n-\ell}  \\
                       &\leq
                         \Upsilon^{-1}
                         \sum_{k=3}^{\lfloor
                         2y/\epsilon+4\rfloor} A_k B_k\sum_{\ell = \lfloor
                         \alpha  n  \rfloor}  ^{\lfloor  (1-\alpha)  n
                         \rfloor} \ell^{-5/2} \leq
                         \Upsilon^{-1}
                         \sum_{k=3}^{\lfloor
                         2y/\epsilon+4\rfloor} A_k B_k \; n \; (\alpha
                         n)^{-5/2}.
\end{align*}
Therefore $\mathbb{P}(\Gamma_n)$  converges towards $0$ when  $n$ goes
to infinity, which concludes the proof of the Theorem.
\end{proof}

\bibliographystyle{siam}
\bibliography{abc}

\end{document}